\documentclass[english,leqno,11pt,a4paper]{memoir}
\usepackage{pdfpages}

\usepackage[utf8]{inputenc}
\usepackage{amsmath, amsthm, amssymb, dsfont}
\usepackage{tikz}
\usepackage{calc}

\usepackage{multirow}
\usepackage{rotating}
\usepackage[mathscr]{eucal}
\usepackage{graphicx}

\newtheoremstyle{slant}{\topsep}{\topsep}{\slshape}{0pt}%
                  {\bfseries}{.}{ }{}

\theoremstyle{slant}
\newtheorem{thm}{Theorem}[chapter]
\newtheorem{coro}[thm]{Corollary}
\newtheorem{lem}[thm]{Lemma}
\newtheorem{prop}[thm]{Proposition}
\newtheorem*{claims}{Claim}
\newtheorem{assum}{Assumption}

\theoremstyle{definition}
\newtheorem{rem}[thm]{Remark}
\newtheorem{defn}[thm]{Definition}

\newtheorem*{rems}{Remark}

\newcommand{\ifff}{if and only if}
\newcommand{\bbar}{\ensuremath{\|}}
\newcommand{\rv}[1]{{#1}}
\newcommand{\F}{\ensuremath{\mathds{F}}}
\newcommand{\N}{\mathds{N}}
\newcommand{\eps}{\epsilon}
\renewcommand{\varepsilon}{\epsilon}
\newcommand{\veps}{\epsilon}

\newcommand{\U}{\mathcal{U}}
\newcommand{\C}{\mathcal{C}}
\newcommand{\Z}{\mathds{Z}}
\newcommand{\cI}{\mathcal{I}}
\newcommand{\cC}{\mathcal{C}}

\newcommand{\Ex}{\mathds{E}}
\newcommand{\supp}{\mathsf{supp}}
\newcommand{\List}{\mathsf{LIST}}
\newcommand{\eqdef}{:=}
\newcommand{\cS}{\mathcal{S}}
\newcommand{\cX}{\mathcal{X}}
\newcommand{\cY}{\mathcal{Y}}
\newcommand{\cD}{\mathcal{D}}
\newcommand{\cG}{\mathcal{G}}
\newcommand{\cE}{\mathcal{E}}
\newcommand{\cB}{\mathcal{B}}
\newcommand{\cM}{{M}}
\newcommand{\cZ}{\mathcal{Z}}
\newcommand{\PI}{p}

\newcommand{\dist}{{\mathsf{dist}}}
\newcommand{\rdist}{{\mathsf{rdist}}}
\newcommand{\rk}{{\mathsf{rank}}}
\newcommand{\poly}{{\mathsf{poly}}}
\newcommand{\qpoly}{{\mathsf{quasipoly}}}
\newcommand{\tre}{{\mathsf{Tre}}}

\newcommand{\Exp}{\mathbb{E}}
\newcommand{\wgt}{\mathsf{wgt}}
\newcommand{\bou}{\mathsf{AExt}}
\newcommand{\kz}{\mathsf{SFExt}}
\newcommand{\extr}{\mathsf{Ext}}

\newcommand{\sm}{\setminus}
\newcommand{\agr}{\mathsf{Agr}}
\newcommand{\TRASH}[1]{}
\newcommand{\Th}{\text{th}}

\newcommand{\floor}[1]{\lfloor #1 \rfloor}

\providecommand{\eqref}[1]{(\ref{#1})}

\newcommand{\nextLine}{\\}

\newenvironment{Proof}{\begin{proof}}{\end{proof}}

\newcommand{\cL}{\mathcal{L}}

\newcommand{\tn}{{\tilde{n}}}

\newcommand{\td}{{\tilde{d}}}
\newcommand{\tk}{k}
\newcommand{\Tk}{{\tilde{k}}}
\newcommand{\tl}{{\tilde{\ell}}}
\newcommand{\tee}{t}
\newcommand{\tx}{{\tilde{x}}}
\newcommand{\ty}{{\tilde{y}}}

\newcommand{\zo}{\{0,1\}}

\newcommand{\Wlog}{Without loss of generality}
\newcommand{\rep}{\odot}

\newcommand{\Title}{Applications of Derandomization Theory in Coding}
\newcommand{\Author}{Mahdi Cheraghchi Bashi Astaneh}
\author{\Author}
\title{\Title}

\makeindex
\usepackage[
     backref,
     pagebackref,
     plainpages=false,
     pdftitle={Applications of Derandomization Theory in Coding (PhD Thesis)},
     pdfsubject={PhD Thesis of Mahdi Cheraghchi},
     pdfstartview=FitH,
     bookmarks=true,
     colorlinks=false,
     bookmarksopen=false,
     bookmarksnumbered=true,
     pdfhighlight=/I,
     hyperfigures=true,
     pdfborder=0 0 0,
     pdfauthor={Mahdi Cheraghchi}]{hyperref}

\usepackage[backrefs]{amsrefs}
\providecommand{\cites}[1]{\cite{#1}}

\newfloat[chapter]{constr}{cns}{Construction}
\newcommand{\newConstruction}[6][tbp]{%
\begin{constr}[#1] 
  \begin{framed}
  \begin{itemize}
  \item \textit{Given:} #3

  \item \textit{Output:} #4

  \item \textit{Construction:} #5
  \end{itemize}
  \end{framed}
  \caption{#2}
  \label{#6}
\end{constr}
}

\newcommand{\musicBox}[2]{
\vfill
\noindent
\begin{minipage}{\textwidth}
\vspace{2cm}
\hrule \vspace{1pt} \hrule \vspace{2pt}
\begin{centering}
\includegraphics[width=\textwidth]{#1}
\end{centering}
\begin{flushright}
\begin{minipage}{10cm}
\begin{flushleft}
{\textsf{\textsl{\footnotesize #2}}}
\end{flushleft}
\end{minipage}
\end{flushright}
\end{minipage}
}
\newcommand{\musicBoxIntro}{
\musicBox{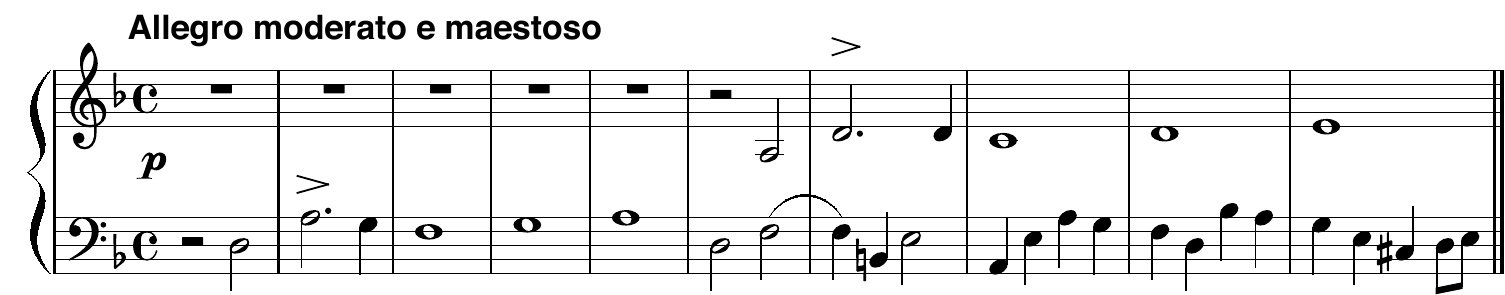}
{Johann Sebastian Bach (1685--1750): The Art of Fugue BWV~1080, Contrapunctus~XIV.}}
\newcommand{\musicBoxExtractor}{
\musicBox{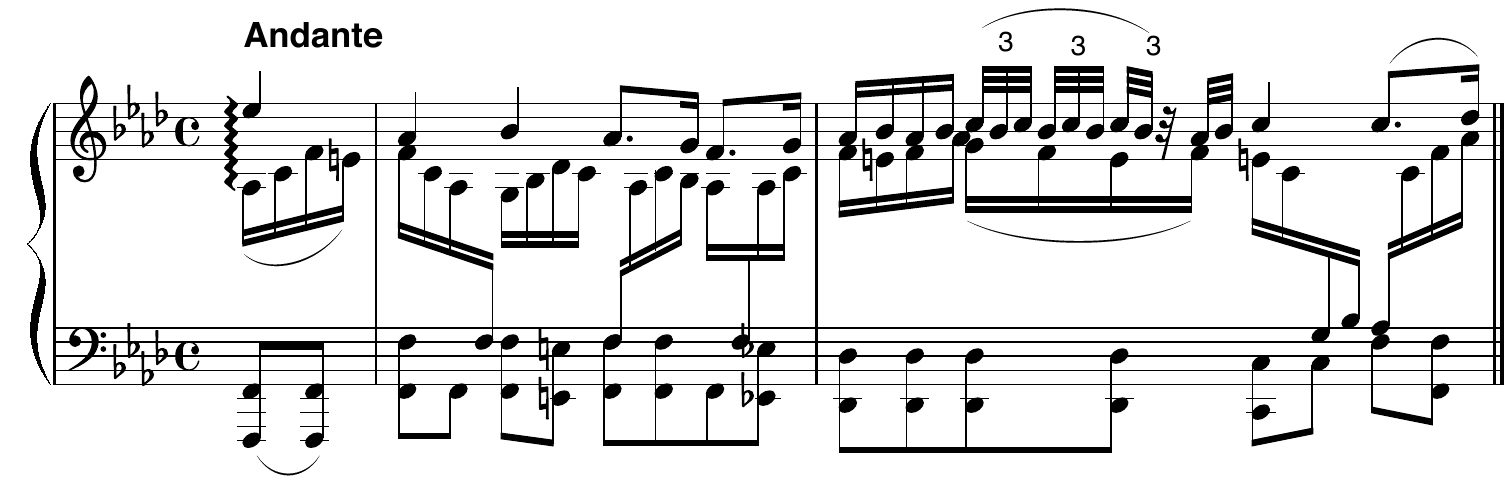}
{Johann Sebastian Bach (1685--1750): Chorale Prelude in F~minor
BWV~639 ``Ich ruf zu dir, Herr Jesu Christ''.
Piano transcription by Ferruccio Busoni (1866--1924).}}
\newcommand{\musicBoxWiretap}{
\musicBox{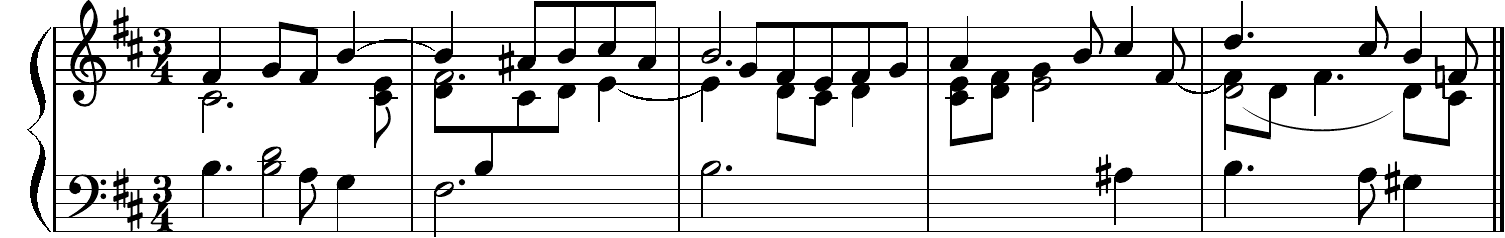}
{Domenico Scarlatti (1685--1757): Keyboard Sonata \\ in B~minor K.~87 (L.~33).}}
\newcommand{\musicBoxTesting}{
\musicBox{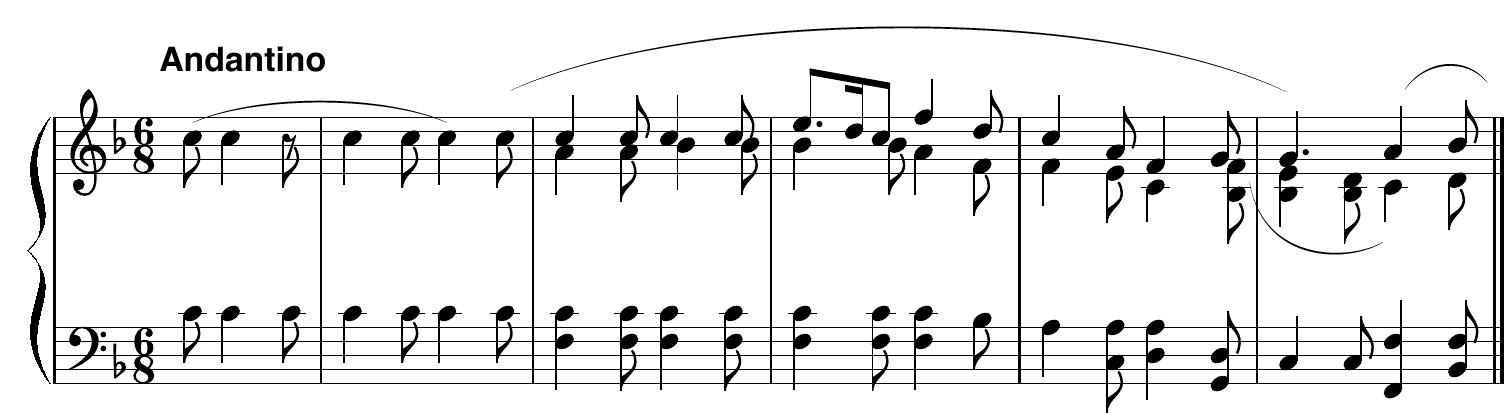}
{Fr\'ed\'eric Chopin (1810--1849): Ballade Op.~38 No.~2 in F~major.}}
\newcommand{\musicBoxCapacity}{
\musicBox{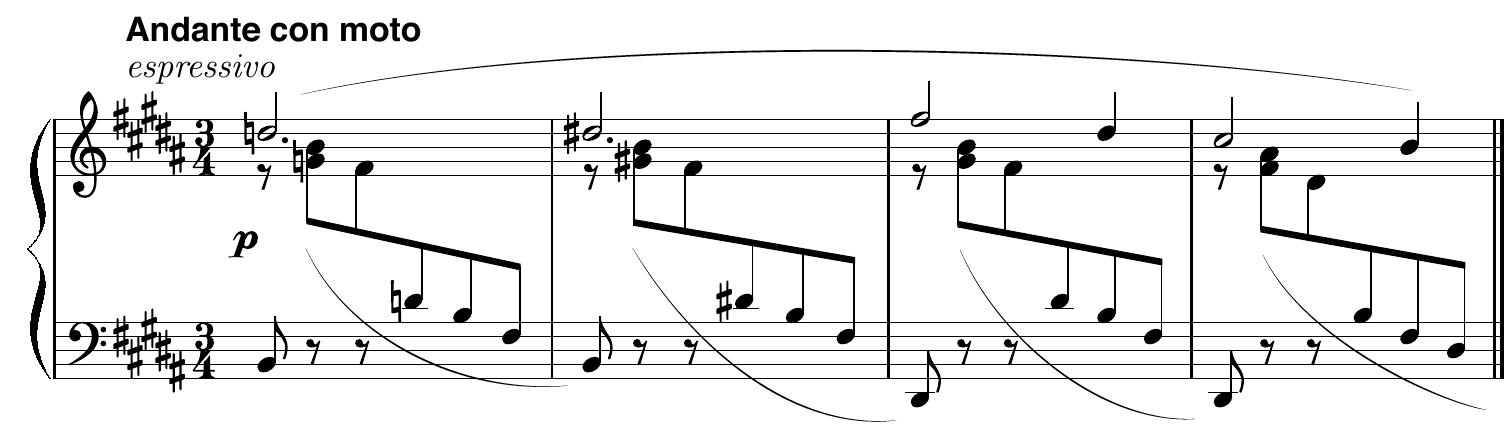}
{Johannes Brahms (1833--1897): Ballade Op.~10 No.~4 in B~major.}}
\newcommand{\musicBoxGV}{
\musicBox{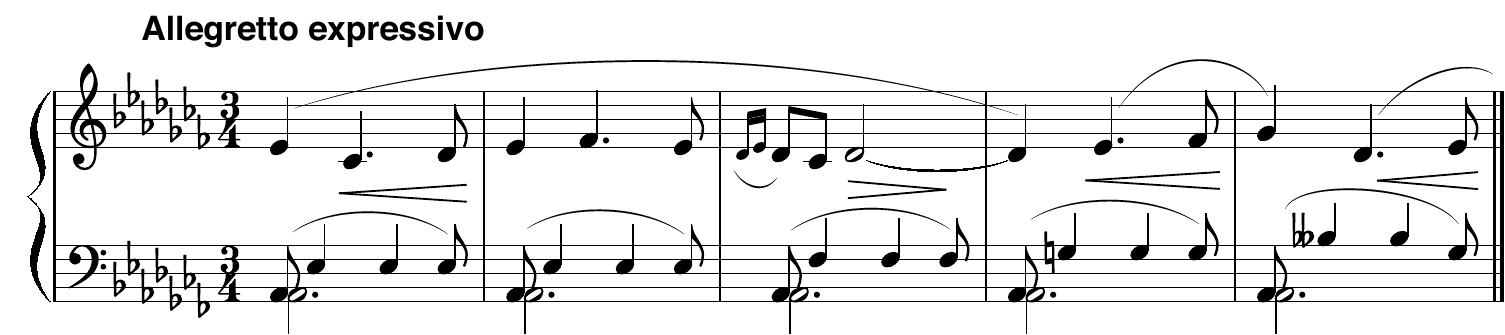}
{Isaac Alb\'eniz (1860--1909): Iberia Suite for Piano, Book~1, \\ Evocaci\'on in A~flat.}}
\newcommand{\musicBoxConclusion}{
\musicBox{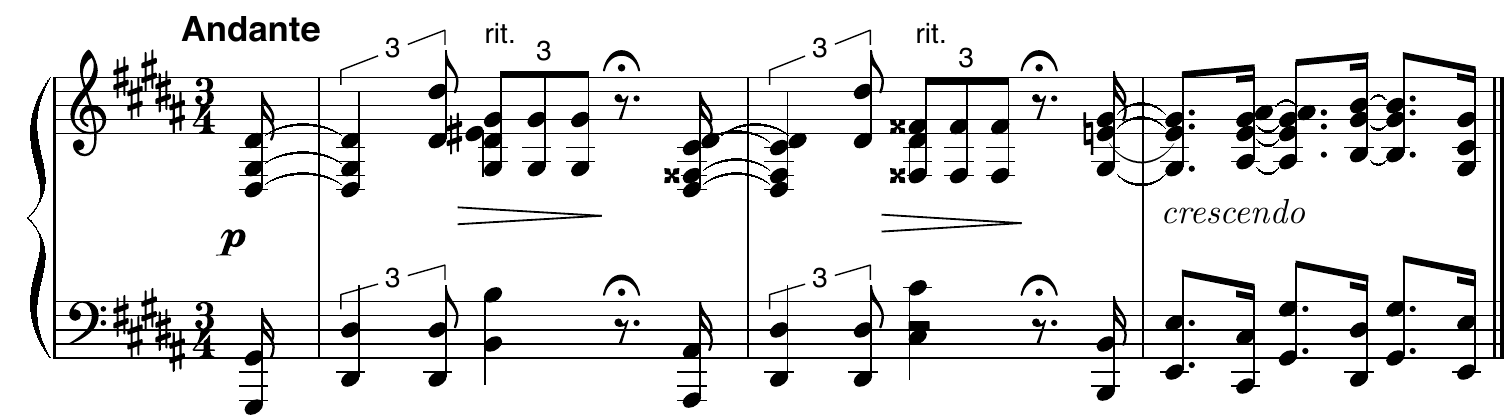}
{Alexander Scriabin (1872--1915): Piano Sonata No.~2 in G~sharp minor (Op.~19, ``Sonata-Fantasy'').}}
\newcommand{\musicBoxAppendixCode}{
\musicBox{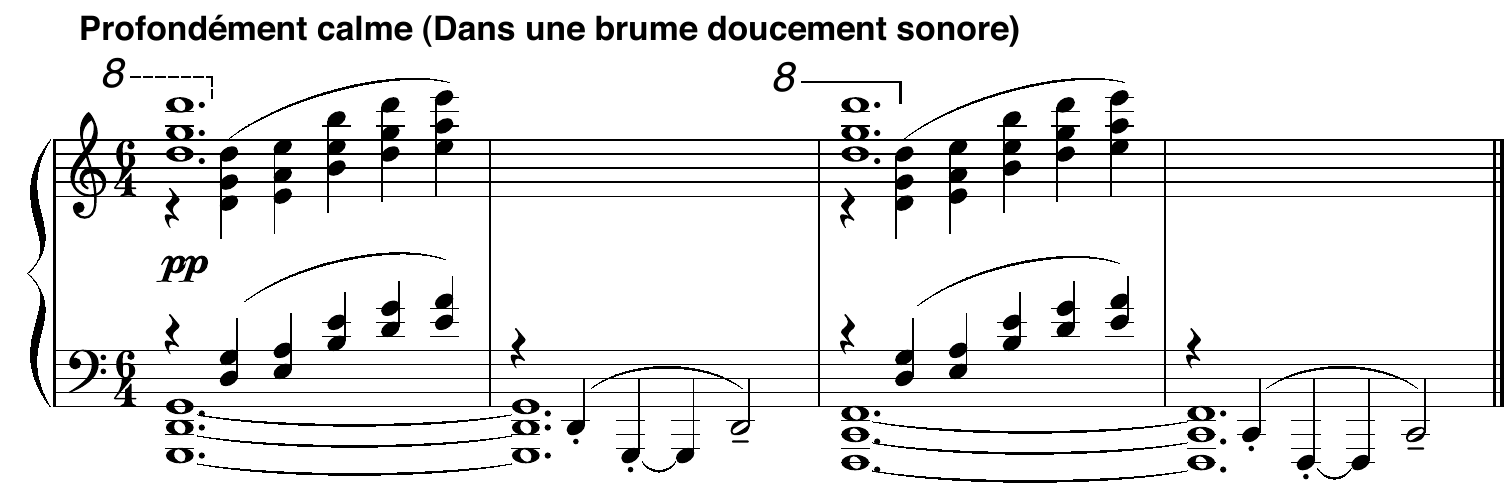}
{Claude Debussy (1862--1918): Preludes, Book~I, No.~X \\ (La cath\'edrale engloutie).}}



\copypagestyle{headings-new}{headings}
\makeheadrule{headings-new}{\textwidth}{\normalrulethickness}
\makeevenhead{headings}{\thepage}{}{\textsc{\slshape\leftmark}}
\makeoddhead{headings}{\slshape\rightmark}{}{\thepage}

\newcommand{\pagestyleselection}{headings-new}

\newcommand{\Chapter}[1]{\openleft \chapter{#1}}

\begin{document}

\setcounter{secnumdepth}{3} \setcounter{tocdepth}{3}



\newpage \thispagestyle{empty}
\begin{centering}
  \textbf{\Large \Title} \vfill
  {\large by \\[7mm]
    \Large{\Author}} \vfill {\large Master of Science (\'Ecole
    Polytechnique F\'ed\'erale de Lausanne), 2005} \vfill
  {\large A dissertation submitted in partial fulfillment of the \\
    requirements for the degree of \\
    Doctor of Philosophy} \vfill
  {\large in \\[7mm]
    \large{Computer Science}} \vfill
  {\large at the \\[7mm]
    School of Computer and Communication Sciences \\[3mm]
    \textsc{\large \'Ecole Polytechnique F\'ed\'erale de Lausanne}}
    \vfill
  {Thesis Number: 4767}
  \vfill
  {\large Committee in charge: \\[3mm]
    Emre Telatar, Professor (President) \\
    Amin Shokrollahi, Professor (Thesis Director) \\
    R\"udiger Urbanke, Professor \\
    Venkatesan Guruswami, Associate Professor \\
    Christopher Umans, Associate Professor } \vfill {\large
    \centerline{July 2010}}
\end{centering}
\newpage \thispagestyle{empty} \phantom{.}

\newpage

\newpage \thispagestyle{empty}
\vspace*{2cm}
\begin{centering}
  \textbf{\Large \Title}
\end{centering}

\newpage \thispagestyle{empty} \phantom{.}

\pagenumbering{roman} \setcounter{page}{0} \newpage

\vspace{3mm}%
\noindent%
\centerline{\large \textsf{\textbf{Abstract}}} \vspace{10mm}

Randomized techniques play a fundamental role in theoretical computer
science and discrete mathematics, in particular for the design of
efficient algorithms and construction of combinatorial objects. The
basic goal in derandomization theory is to eliminate or reduce the
need for randomness in such randomized constructions.  Towards this
goal, numerous fundamental notions have been developed to provide a
unified framework for approaching various derandomization problems and
to improve our general understanding of the power of randomness in
computation.  Two important classes of such tools are
\emph{pseudorandom generators} and \emph{randomness
  extractors}. Pseudorandom generators transform a short, purely
random, sequence into a much longer sequence that \emph{looks} random,
while extractors transform a weak source of randomness into a
perfectly random one (or one with much better qualities, in which case
the transformation is called a \emph{randomness condenser}).

In this thesis, we explore some applications of the fundamental
notions in derandomization theory to problems outside the core of
theoretical computer science, and in particular, certain problems
related to coding theory.  First, we consider the \emph{wiretap
  channel problem} which involves a communication system in which an
intruder can eavesdrop a limited portion of the transmissions. We
utilize randomness extractors to construct efficient and
infor\-mation-theoretically optimal communication protocols for this
model.

Then we consider the \emph{combinatorial group testing} problem. In
this classical problem, one aims to determine a set of defective items
within a large population by asking a number of queries, where each
query reveals whether a defective item is present within a specified
group of items. We use randomness condensers to explicitly construct
optimal, or nearly optimal, group testing schemes for a setting where
the query outcomes can be highly unreliable, as well as the
\emph{threshold model} where a query returns positive if the number of
defectives pass a certain threshold.

Next, we use randomness condensers and extractors to design ensembles
of error-correcting codes that achieve the infor\-mation-theoretic
capacity of a large class of communication channels, and then use the
obtained ensembles for construction of explicit capacity achieving
codes. Finally, we consider the problem of explicit construction of
error-correcting codes on the \emph{Gilbert-Varshamov bound} and
extend the original idea of Nisan and Wigderson to obtain a small
ensemble of codes, mostly achieving the bound, under suitable
computational hardness assumptions.

\vspace{1mm}%
\noindent\textit{Keywords: Derandomization theory, randomness
  extractors, pseudorandomness, wiretap channels, group testing,
  error-correcting codes.}

\newpage \vspace{3mm}%
\noindent%
\centerline{\large \textsf{\textbf{R\'esum\'e}}} \vfill 

Les techniques de randomisation jouent un r\^ole fondamental en
informatique th\'eorique et en math\'ematiques discr\`etes, en
particulier pour la conception d'algorithmes efficaces et pour la
construction d'objets combinatoires. L'objectif principal de la
th\'eorie de d\'erandomisation est d'\'eliminer ou de r\'eduire le
besoin d'al\'ea pour de telles constructions. Dans ce but, de
nombreuses notions fondamentales ont \'et\'e d\'evelopp\'ees, d'une
part pour cr\'eer un cadre unifi\'e pour aborder diff\'erents
probl\`emes de d\'erandomisation, et d'autre part pour mieux
comprendre l'apport de l'al\'ea en informatique. Les
\emph{g\'en\'erateurs pseudo-al\'eatoires} et les \emph{extracteurs}
sont deux classes importantes de tels outils. Les g\'en\'erateurs
pseudo-al\'eatoires transforment une suite courte et purement
al\'eatoire en une suite beaucoup plus longue qui \emph{parait}
al\'eatoire. Les \emph{extracteurs d'al\'ea} transforment une source
faiblement al\'eatoire en une source parfaitement al\'eatoire (ou en
une source de meilleure qualit\'e. Dans ce dernier cas, la
transformation est appel\'ee un \emph{condenseur d'al\'ea}).

Dans cette th\`ese, nous explorons quelques applications des notions
fondamentales de la th\'eorie de d\'erandomisation \`a des probl\`emes
p\'eriph\'eriques \`a l'informatique th\'eorique et en particulier \`a
certains probl\`emes relevant de la th\'eorie des codes.  Nous nous
int\'eressons d'abord au \emph{probl\`eme du canal \`a jarreti\`ere},
qui consiste en un syst\`eme de communication o\`u un intrus peut
intercepter une portion limit\'ee des transmissions. Nous utilisons
des extracteurs pour construire pour ce mod\`ele des protocoles de
communication efficaces et optimaux du point de vue de la th\'eorie de
l'information.

Nous \'etudions ensuite le probl\`eme du \emph{test en groupe
  combinatoire}.  Dans ce probl\`eme classique, on se propose de
d\'eterminer un ensemble d'objets d\'efectueux parmi une large
population, \`a travers un certain nombre de questions, o\`u chaque
r\'eponse r\'ev\`ele si un objet d\'efectueux appartient \`a un
certain ensemble d'objets. Nous utilisons des condenseurs pour
construire explicitement des tests de groupe optimaux ou
quasi-optimaux, dans un contexte o\`u les r\'eponses aux questions
peuvent \^etre tr\`es peu fiables, et dans le \emph{mod\`ele de seuil}
o\`u le r\'esultat d'une question est positif si le nombre d'objets
d\'efectueux d\'epasse un certain seuil.

Ensuite, nous utilisons des condenseurs et des extracteurs pour
concevoir des ensembles de codes correcteurs d'erreurs qui atteignent
la capacit\'e (dans le sens de la th\'eorie de l'information) d'un
grand nombre de canaux de communications. Puis, nous utilisons les
ensembles obtenus pour la construction de codes explicites qui
atteignent la capacit\'e. Nous nous int\'eressons finalement au
probl\`eme de la construction explicite de codes correcteurs d'erreurs
qui atteignent la \emph{borne de Gilbert--Varshamov} et reprenons
l'id\'ee originale de Nisan et Wigderson pour obtenir un petit
ensemble de codes dont la plupart atteignent la borne, sous certaines
\emph{hypoth\`eses de difficult\'e computationnelle}.

\vspace{1mm}%
\noindent\textit{Mots-cl\'es: Th\'eorie de d\'erandomisation,
  extracteurs d'al\'ea, pseudo-al\'ea, ca\-naux \`a jarreti\`ere, test
  en groupe, codes correcteurs d'erreurs.}

\newpage\pagestyle{plain}%
\noindent{\Large\textsf{\textbf{Acknowledgments}}} \vspace{1cm}

During my several years of study at EPFL, both as a Master's student
and a Ph.D.~student, I have had the privilege of interacting with so
many wonderful colleagues and friends who have been greatly
influential in my graduate life.  Despite being thousands of miles
away from home, thanks to them my graduate studies turned out to be
one of the best experiences of my life.  These few paragraphs are an
attempt to express my deepest gratitude to all those who made such an
exciting experience possible.

My foremost gratitude goes to my adviser, Amin Shokrollahi, for not
only making my academic experience at EPFL truly enjoyable, but also
for numerous other reasons. Being not only a great adviser and an
amazingly brilliant researcher but also a great friend, Amin is
undoubtedly one of the most influential people in my life.  Over the
years, he has taught me more than I could ever imagine. Beyond his
valuable technical advice on research problems, he has thought me how
to be an effective, patient, and confident researcher. He would always
insist on picking research problems that are worth thinking, thinking
about problems for the joy of thinking and without worrying about the
end results, and publishing only those results that are worth
publishing.  His mastery in a vast range of areas, from pure
mathematics to engineering real-world solutions, has always greatly
inspired for me to try learning about as many topics as I can and
interacting with people with different perspectives and interests. I'm
especially thankful to Amin for being constantly available for
discussions that would always lead to new ideas, thoughts, and
insights. Moreover, our habitual outside-work discussions in
restaurants, on the way for trips, and during outdoor activities
turned out to be a great source of inspiration for many of our
research projects, and in fact some of the results presented in this
thesis!  I'm also grateful to Amin for his collaborations on several
research papers that we coauthored, as well as the technical substance
of this thesis.  Finally I thank him for numerous small things, like
encouraging me to buy a car which turned out to be a great idea!

Secondly, I would like to thank our secretary Natascha Fontana for
being so patient with too many inconveniences that I made for her over
the years! She was about the first person I met in Switzerland, and
kindly helped me settle in Lausanne and get used to my new life
there. For several years I have been constantly bugging her with
problems ranging from administrative trouble with the doctoral school
to finding the right place to buy curtains. She has also been a great
source of encouragement and support for my graduate studies.

Besides Amin and Natascha, I'm grateful to the present and past
members of our Laboratory of Algorithms (ALGO) and Laboratory of
Algorithmic Mathematics (LMA) for creating a truly active and
enjoyable atmosphere: Bertrand Meyer, Ghid Maatouk, Giovanni Cangiani,
Harm Cronie, Hesam Salavati, Luoming Zhang, Masoud Alipour, Raj Kumar
(present members), and Andrew Brown, Bertrand Ndzana Ndzana, Christina
Fragouli, Fr\'ed\'eric Didier, Fr\'ed\'erique Oggier, Lorenz Minder,
Mehdi Molkaraie, Payam Pakzad, Pooya Pakzad, Zeno Crivelli (past
members), as well as Alex Vardy, Emina Soljanin, Martin F\"urer, and
Shahram Yousefi (long-term visitors).  Special thanks to:

\begin{description}
\item Alex Vardy and Emina Soljanin: For fruitful discussions on the
  results presented in Chapter~\ref{chap:wiretap}.

\item Fr\'ed\'eric Didier: For numerous fruitful discussions and his
  collaboration on our joint paper \cite{ref:CDS09}, on which
  Chapter~\ref{chap:wiretap} is based.

\item Giovanni Cangiani: For being a brilliant system administrator
  (along with Damir Laurenzi), and his great help with some technical
  problems that I had over the years.

\item Ghid Maatouk: For her lively presence as an endless source of
  fun in the lab, for taking student projects with me prior to joining
  the lab, helping me keep up my obsession about classical music,
  encouraging me to practice the piano, and above all, being an
  amazing friend. I also thank her and Bertrand Meyer for translating
  the abstract of my thesis into French.

\item Lorenz Minder: For sharing many tech-savvy ideas and giving me a
  quick campus tour when I visited him for a day in Berkeley, among
  other things.

\item Payam Pakzad: For many fun activities and exciting discussions
  we had during the few years he was with us in ALGO.

\item Zeno Crivelli and Bertrand Ndzana Ndzana: For sharing their
  offices with me for several years!  I also thank Zeno for countless
  geeky discussions, lots of fun we had in the office, and for
  bringing a small plant to the office, which quickly grew to reach
  the ceiling and stayed fresh for the entire duration of my
  Ph.D.~work.
\end{description}

I'm thankful to professors and instructors from whom I learned a great
deal attending their courses as a part of my Ph.D.~work: I learned
Network Information Theory from Emre Telatar, Quantum Information
Theory from Nicolas Macris, Algebraic Number Theory from Eva Bayer,
Network Coding from Christina Fragouli, Wireless Communication from
Suhas Diggavi, and Modern Coding Theory from my adviser Amin. As a
teaching assistant, I also learned a lot from Amin's courses (on
algorithms and coding theory) and from an exciting collaboration with
Monika Henzinger for her course on advanced algorithms.

During summer 2009, I spent an internship at KTH working with Johan
H{\aa}stad and his group.  What I learned from Johan within this short
time turned out far more than I had expected. He was always available
for discussions and listening to my countless silly ideas with extreme
patience, and I would always walk out of his office with new ideas
(ideas that would, contrary to those of my own, always work!). Working
with the theory group at KTH was more than enjoyable, and I'm
particularly thankful to Ola Svensson, Marcus Isaksson, Per Austrin,
Cenny Wenner, and Luk\'a\v{s} Pol\'a\v{c}ek for numerous delightful
discussions.

Special thanks to cool fellows from the Information Processing Group
(IPG) of EPFL for the fun time we had and also countless games of
Foosball we played (brought to us by Giovanni).

I'm indebted to my great friend, Soheil Mohajer, for his close
friendship over the years.  Soheil has always been patient enough to
answer my countless questions on information theory and communication
systems in a computer-science-friendly language, and his brilliant
mind has never failed to impress me.  We had tons of interesting
discussions on virtually any topic, some of which coincidentally (and
finally!) contributed to a joint paper \cite{ref:CKMS10}.  I also
thank Amin Karbasi and Venkatesh Saligrama for this work.  Additional
thanks goes to Amin for our other joint paper \cite{ref:CHKV09} (along
with Ali Hormati and Martin Vetterli whom I also thank) and in
particular giving me the initial motivation to work on these projects,
plus his unique sense of humor and amazing friendship over the years.

I'd like to extend my warmest gratitude to Pedram Pedarsani, for too
many reasons to list here, but above all for being an amazingly caring
and supportive friend and making my graduate life even more
pleasing. Same goes to Mahdi Jafari, who has been a great friend of
mine since middle school! Mahdi's many qualities, including his
humility, great mind, and perspective to life (not to mention great
photography skills) has been a big influence on me.  I feel extremely
lucky for having such amazing friends.

I take this opportunity to thank three of my best, most brilliant, and
most influential friends; Omid Etesami, Mohammad Mahmoody, and Ehsan
Ardes\-tani\-zadeh, whom I'm privileged to know since high school.  In
college, Omid showed me some beauties of complexity theory which
strongly influenced me in pursuing my post-graduate studies in
theoretical computer science. He was also influential in my decision
to study at EPFL, which turned out to be one of my best decisions in
life. I had the most fascinating time with Mohammad and Ehsan during
their summer internships at EPFL. Mohammad thought me a great deal
about his fascinating work on foundations of cryptography and
complexity theory and was always up to discuss anything ranging from
research ideas to classical music and cinema.  I worked with Ehsan on
our joint paper \cite{ref:ACS09} which turned out to be one of the
most delightful research collaborations I've had.  Ehsan's unique
personality, great wit and sense of humor, as well as musical
talents---especially his mastery in playing Santur---has always filled
me with awe.  I also thank the three of them for keeping me company
and showing me around during my visits in Berkeley, Princeton, and San
Diego.

In addition to those mentioned above, I'm grateful to so many amazing
friends who made my study in Switzerland an unforgettable stage of my
life and full of memorable moments: Ali Ajdari Rad, Amin Jafarian,
Arash Golnam, Arash Salarian, Atefeh Ma\-sha\-tan, Banafsheh Abasahl,
Elham Ghadiri, Faezeh Malakouti, Fereshteh Baghe\-rimiyab, Ghazale
Hosseinabadi, Hamed Alavi, Hossein Afshari, Hossein Rouhani, Hossein
Taghavi, Javad Ebrahimi, Laleh Goles\-tanirad, Mani Bastani Parizi,
Marjan Hamedani, Marjan Sedighi, Maryam Javanmardy, Maryam Zaheri,
Mina Karzand, Mohammad Karzand, Mona Mahmoudi, Morteza Zadimoghaddam,
Nasibeh Pouransari, Neda Salamati, Nooshin Hadadi, Parisa Ha\-ghani,
Pooyan Abouzar, Pouya Dehghani, Ramtin Pedarsani, Sara Kherad Pajouh,
Shi\-rin Saeedi, Vahid Aref, Vahid Majidzadeh, Wojciech Galuba, and
Zahra Sinaei. Each name should have been accompanied by a story
(ranging from a few lines to a few pages); however, doing so would
have made this section exceedingly long. Moreover, having prepared the
list rather hastily, I'm sure I have missed a lot of nice friends on
it.  I owe them a coffee (or tea, if they prefer) each!  Additional
thanks to Mani, Ramtin, and Pedram, for their musical presence.

Thanks to Alon Orlitsky, Avi Wigderson, Madhu Sudan, Rob Calderbank,
and Umesh Vazirani for arranging my short visits to UCSD, IAS, MIT,
Princeton, and U.C.~Berkeley, and to Anup Rao, Swastik Kopparty, and
Zeev Dvir for interesting discussions during those visits.

I'm indebted to Chris Umans, Emre Telatar, R\"udiger Urbanke, and
Venkat Guruswami for giving me the honor of having them in my
dissertation committee. I also thank them (and Amin) for carefully reading the
thesis and their comments on an earlier draft of this work.
Additionally, thanks to Venkat for numerous illuminating
discussions on various occasions, in particular on my papers \cites{ref:Che09,ref:Che10}
that form the basis of the material presented in
Chapter~\ref{chap:testing}.

My work was in part funded by grants from
the Swiss National Science Foundation (Grant No.~200020-115983/1) and
the European Research Council (Advanced Grant No.~228021) that I
gratefully acknowledge.

Above all, I express my heartfelt gratitude to my parents, sister
Azadeh, and brother Babak who filled my life with joy and
happiness. Without their love, support, and patience none of my
achievements---in particular this thesis---would have been possible. I
especially thank my mother for her everlasting love, for all she went
through until I reached this point, and her tremendous patience during
my years of absence while I was only able to go home for a short visit
each year. This thesis is dedicated with love to her.

\newpage\pagestyle{\pagestyleselection}
\tableofcontents*
\eject \listoffigures \eject \listoftables \eject


\setlength{\epigraphrule}{0pt}

\Chapter{Introduction}
\epigraphhead[70]{\epigraph{\textsl{``You are at the wheel of your car, waiting at a traffic
light, you take the book out of the bag, rip off the transparent wrapping,
start reading the first lines. A storm of honking breaks over you;
the light is green, you're blocking traffic.''}}{\textit{--- Italo Calvino}}}
\label{chap:intro} \setcounter{page}{0}
\pagenumbering{arabic}

Over the decades, the role of randomness in computation has proved to be
one of the most intriguing subjects of study in computer science.
Considered as a fundamental computational resource, randomness has
been extensively used as an indispensable tool in design and analysis
of algorithms, combinatorial constructions, cryptography, and
computational complexity.

As an illustrative example on the power of randomness in algorithms,
consider a \emph{clustering} problem, in which we wish to partition a
collection of items into two groups. Suppose that some pairs of items
are marked as \emph{inconsistent}, meaning that they are best be
avoided falling in the same group. Of course, it might be simply
impossible to group the items in such a way that no inconsistencies
occur within the two groups.  For that reason, it makes sense to
consider the objective of minimizing the number of inconsistencies
induced by the chosen partitioning.  Suppose that we are asked to
color individual items red or blue, where the items marked by the same
color form each of the two groups. How can we design a strategy that
maximizes the number of inconsistent pairs that fall in different
groups? The basic rule of thumb in randomized algorithm design
suggests that

\begin{quote}
  \textsl{When unsure making decisions, try flipping coins!}
\end{quote}

Thus a naive strategy for assigning color to items would be to flip a
fair coin for each item. If the coin falls Heads, we mark the item
blue, and otherwise red.

How can the above strategy possibly be any reasonable? After all we
are defining the groups without giving the slightest thought on the
given structure of the inconsistent pairs! Remarkably, a simple
analysis can prove that the coin-flipping strategy is in fact a quite
reasonable one.  To see why, consider any inconsistent pair. The
chance that the two items are assigned the same color is exactly one
half. Thus, we expect that half of the inconsistent pairs end up
falling in different groups. By repeating the algorithm a few times
and checking the outcomes, we can be sure that an assignment
satisfying half of the inconsistency constraints is found after a few
trials.

We see that, a remarkably simple algorithm that does not even read its
input can attain an approximate solution to the clustering problem in
which the number of inconsistent pairs assigned to different groups is
no less than half the maximum possible. However, our algorithm used a
valuable resource; namely random coin flips, that greatly simplified
its task. In this case, it is not hard to come up with an efficient
(i.e., polynomial-time) algorithm that does equally well without using
any randomness. However, designing such an algorithm and analyzing its
performance is admittedly a substantially more difficult task that
what we demonstrated within a few paragraphs above.

As it turns out, finding an optimal solution to our clustering problem
above is an intractable problem (in technical terms, it is
$\mathsf{NP}$-hard), and even obtaining an approximation ratio better
than $16/17 \approx .941$ is so \cite{HastadOptimal}.  Thus the
trivial bit-flipping algorithm indeed obtains a reasonable solution.
In a celebrated work, Goemans and Williamson \cite{ref:GW95} improve
the approximation ratio to about $.878$, again using randomization\footnote{
Improving upon the approximation ration obtained by this algorithm
turns out to be $\mathsf{NP}$-hard under a well-known conjecture \cite{KKMOD04}.
}. A
deterministic algorithm achieving the same quality was later
discovered \cite{ref:MR95}, though it is much more complicated to
analyze.

Another interesting example demonstrating the power of randomness in
algorithms is the \emph{primality testing} problem, in which the goal
is to decide whether a given $n$-digit integer is prime or
composite. While efficient (poly\-nomial-time in $n$) randomized
algorithms were discovered for this problem as early as 1970's (e.g.,
Solovay-Strassen's \cite{ref:SS77} and Miller-Rabin's algorithms
\cites{ref:Mil76,ref:Rab80}), a deterministic polynomial-time
algorithm for primality testing was found decades later, with the
breakthrough work of Agrawal, Kayal, and Saxena~\cite{ref:AKS04},
first published in 2002.  Even though this algorithm provably works in
polynomial time, randomized methods still tend to be more favorable
and more efficient for practical applications.

The primality testing algorithm of Agrawal et al.~can be regarded as a
derandomization of a particular instance of the \emph{polynomial
  identity testing} problem.  Polynomial identity testing generalizes
the high-school-favorite problem of verifying whether a pair of
polynomials expressed as closed form formulae expand to identical
polynomials. For example, the following is an 8-variate identity
\begin{multline*}
  (x_1^2+x_2^2+x_3^2+x_4^2)(y_1^2+y_2^2+y_3^2+y_4^2) \stackrel{?}{\equiv} \\
  (x_1 y_1 - x_2 y_2 - x_3 y_3 - x_4 y_4)^2 +
  (x_1 y_2 + x_2 y_1 + x_3 y_4 - x_4 y_3)^2 + \\
  (x_1 y_3 - x_2 y_4 + x_3 y_1 + x_4 y_2)^2 + (x_1 y_4 + x_2 y_3 - x_3
  y_2 + x_4 y_1)^2
\end{multline*}
which turns out to be valid. When the number of variables and the
complexity of the expressions grow, the task of verifying identities
becomes much more challenging using naive methods.

This is where the power of randomness comes into play again. A
fundamental idea due to Schwartz and Zippel
\cites{ref:Schwartz,ref:Zippel} shows that the following approach
indeed works:
\begin{quote}
  \textsl{Evaluate the two polynomials at sufficiently many randomly
    chosen points, and identify them as identical if and only if all
    evaluations agree.}
\end{quote}
It turns out that the above simple idea leads to a randomized
efficient algorithm for testing identities that may err with an
arbitrarily small probability.  Despite substantial progress, to this
date no polynomial-time deterministic algorithms for solving general
identity testing problem is known, and a full \emph{derandomization}
of Schwartz-Zippel's algorithm remains a challenging open problem in
theoretical computer science.

The discussion above, among many other examples, makes the strange
power of randomness evident. Namely, in certain circumstances the
power of randomness makes algorithms more efficient, or simpler to
design and analyze. Moreover, it is not yet clear how to perform
certain computational tasks (e.g., testing for general polynomial
identities) without using randomness.

Apart from algorithms, randomness has been used as a fundamental tool
in various other areas, a notable example being combinatorial
constructions. Combinatorial objects are of fundamental significance
for a vast range of theoretical and practical problems.  Often solving
a practical problem (e.g., a real-world optimization problem) reduces
to construction of suitable combinatorial objects that capture the
inherent structure of the problem. Examples of such combinatorial
objects include graphs, set systems, codes, designs, matrices, or even
sets of integers. For these constructions, one has a certain
structural property of the combinatorial object in mind (e.g., mutual
intersections of a set system consisting of subsets of a universe) and
seeks for an instance of the object that optimizes the property in
mind in the best possible way (e.g., the largest possible set system
with bounded mutual intersections).

The task of constructing suitable combinatorial objects turns out
quite challenging at times. Remarkably, in numerous situations the
power of randomness greatly simplifies the task of constructing the
ideal object. A powerful technique in combinatorics, dubbed as
\emph{the probabilistic method} (see \cite{probMethod}) is based on
the following idea:
\begin{quote}
  \textsl{When out of ideas finding the right combinatorial object,
    try a random one!}
\end{quote}

Surprisingly, in many cases this seemingly naive strategy
significantly beats the most brilliant constructions that do not use
any randomness. An illuminating example is the problem of constructing
\emph{Ramsey graphs}. It is well known that in a group of six or more
people, either there are at least three people who know each other or
three who do not know each other. More generally, \emph{Ramsey theory}
shows that for every positive integer $K$, there is an integer $N$
such that in a group of $N$ or more people, either there are at least
$K$ people who mutually know each other (called a \emph{clique} of
size $K$) or $K$ who are mutually unfamiliar with one another (called
an \emph{independent set} of size $K$). Ramsey graphs capture the
reverse direction:

\begin{quote}
\textsl{For a given $N$, what is the smallest $K$ such
that there is a group of $N$ people with no cliques or independent sets of
size $K$ or more? And how can an example of such a group be constructed?}
\end{quote}

In graph-theoretic terms (where mutual acquaintances are captured by
edges), an undirected graph with $N := 2^n$ vertices
is called a Ramsey graph with \emph{entropy $k$} if it has no clique
or independent set of size $K := 2^k$ (or larger). The Ramsey graph
construction problem is to efficiently construct a graph with smallest
possible entropy $k$.

Constructing a Ramsey graph with entropy $k = (n+1)/2$ is already
nontrivial.  However, the following \emph{Hadamard graph} does the job
\cite{CG88}: Each vertex of the graph is associated with a binary
vector of length $n$, and there is an edge between two vertices if
their corresponding vectors are orthogonal over the binary field.  A
much more involved construction, due to Barak et al.~\cite{BRSW06}
(which remains the best deterministic construction to date) attain an
entropy $k = n^{o(1)}$.

A brilliant, but quite simple, idea due to Erd{\H o}s~\cite{Erdos}
demonstrates the power of randomness in combinatorial constructions:
Construct the graph randomly, by deciding whether to put an edge
between every pair of vertices by flipping a fair coin. It is easy to
see that the resulting graph is, with overwhelming probability, a
Ramsey graph with entropy $k = \log n + 2$.  It also turns out that
this is about the lowest entropy one can hope for!  Note the
significant gap between what achieved by a simple, probabilistic
construction versus what achieved by the best known deterministic
constructions.

Even though the examples discussed above clearly demonstrate the power
of randomness in algorithm design and combinatorics, a few issues are
inherently tied with the use of randomness as a computational
resource, that may seem unfavorable:

\begin{enumerate}
\item A randomized algorithm takes an abundance of fair, and
  independent, coin flips for granted, and the analysis may fall apart
  if this assumption is violated. For example, in the clustering
  example above, if the coin flips are biased or correlated, the $.5$
  approximation ratio can no longer be guaranteed. This raises a
  fundamental question:
  \begin{quote}
    \textsl{ Does ``pure randomness'' even exist? If so, how can we
      instruct a computer program to produce purely random coin flips?
    }
  \end{quote}

\item Even though the error probability of randomized algorithms (such
  as the primality testing algorithms mentioned above) can be made
  arbitrarily small, it remains nonzero. In certain cases where a
  randomized algorithm never errs, its running time may vary depending
  on the random choices being made. We can never be completely sure
  whether an error-prone algorithm has really produced the right
  outcome, or whether one with a varying running time is going to
  terminate in a reasonable amount of time (even though we can be
  almost confident that it does).

\item As we saw for Ramsey graphs, the probabilistic method is a
  powerful tool in showing that combinatorial objects with certain
  properties exist, and it most cases it additionally shows that a
  random object almost surely achieves the desired properties. Even
  though for certain applications a randomly produced object is good
  enough, in general there might be no easy way to certify whether a
  it indeed satisfies the properties sought for. For the example of
  Ramsey graphs, while almost every graph is a Ramsey graph with a
  logarithmically small entropy, it is not clear how to certify
  whether a given graph satisfies this property.  This might be an
  issue for certain applications, when an object with
  \emph{guaranteed} properties is needed.
\end{enumerate}

The basic goal of \emph{derandomization theory} is to address the
above-mentioned and similar issues in a systematic way. A central
question in derandomization theory deals with efficient ways of
\emph{simulating randomness}, or relying on \emph{weak randomness}
when perfect randomness (i.e., a steady stream of fair and independent
coin flips) is not available. A mathematical formulation of randomness
is captured by the notion of \emph{entropy}, introduced by Shannon
\cite{ref:Shannon}, that quantifies randomness as the amount of
\emph{uncertainty} in the outcome of a process.  Various sources of
``unpredictable'' phenomena can be found in nature. This can be in
form of an electric noise, thermal noise, ambient sound input, image
captured by a video camera, or even a user's input given to an input
device such as a keyboard.  Even though it is conceivable to assume
that a bit-sequence generated by all such sources contains a certain
amount of entropy, the randomness being offered might be far from
perfect.  Randomness \emph{extractors} are fundamental combinatorial,
as well as computational, objects that aim to address this issue.

As an example to illustrate the concept of extractors, suppose that we
have obtained several independent bit-streams $X_1, X_2, \ldots, X_r$
from various physically random sources. Being obtained from physical
sources, not much is known about the structure of these sources, and
the only assumption that we can be confident about is that they
produce a substantial amount of entropy.  An extractor is a function
that combines these sources into one, \emph{perfectly random},
source. In symbols, we have
\[
f(X_1, X_2, \ldots, X_r) = Y,
\]
where the output source $Y$ is purely random provided that the input
sources are reasonably (but not fully) random. To be of any practical
use, the extractor $f$ must be efficiently computable as well. A more
general class of functions, dubbed \emph{condensers} are those that do
not necessarily transform imperfect randomness into perfect one, but
nevertheless substantially purifies the randomness being given.  For
instance, as a condenser, the function $f$ may be expected to produce
an output sequence whose entropy is $90\%$ of the optimal entropy
offered by perfect randomness.

Intuitively, there is a trade-off between \emph{structure} and
\emph{randomness}. A sequence of fair coin flips is extremely
unpredictable in that one cannot \emph{bet} on predicting the next
coin flip and expect to gain any advantage out of it. On the other
extreme, a sequence such as what given by digits of $\pi =
3.14159265\dots$ may look random but is in fact perfectly
structured. Indeed one can use a computer program to perfectly predict
the outcomes of this sequence. A physical source, on the other hand,
may have some inherent structure in it. In particular, the outcome of
a physical process at a certain point might be more or less
predictable, dictated by physical laws, from the outcomes observed
immediately prior to that time. However, the degree of predictability
may of course not be as high as in the case of $\pi$.

From a combinatorial point of view, an extractor is a combinatorial
object that neutralizes any kind of structure that is inherent in a
random source, and, \emph{extracts} the ``random component'' out (if
there is any).  On the other hand, in order to be any useful, an
extractor must be computationally efficient. At a first sight, it may
look somewhat surprising to learn that such objects may even exist! In
fact, as in the case of Ramsey graphs, the probabilistic method can be
used to show that a randomly chosen function is almost surely a decent
extractor.  However, a random function is obviously not good enough as
an extractor since the whole purpose of an extractor is to eliminate
the need for pure randomness. Thus for most applications, an extractor
(and more generally, condenser) is required to be efficiently
computable and utilize as small amount of auxiliary pure randomness as
possible.

While randomness extractors were originally studied for the main
purpose of eliminating the need for pure randomness in randomized
algorithms, they have found surprisingly diverse applications in
different areas of combinatorics, computer science, and related
fields.  Among many such developments, one can mention construction of
good expander graphs \cite{WZ99} and Ramsey graphs \cite{BRSW06} (in
fact the best known construction of Ramsey graphs can be considered a
byproduct of several developments in extractor theory), communication
complexity \cite{CG88}, Algebraic complexity theory \cite{ref:RY08},
distributed computing (e.g., \cites{Zuc97,GVZ06,RZ01}), data
structures (e.g., \cite{ref:Ta02}), hardness of optimization problems
\cites{MU01,Zuc96}, cryptography (see, e.g., \cite{ref:Dodis}), coding
theory \cite{ref:TZ04}, signal processing \cite{ref:Ind08},
and various results in structural complexity theory (e.g.,
\cite{GZ97}).


In this thesis we extend such connections to several fundamental
problems related to coding theory. In the following we present a brief
summary of the individual problems that are studied in each chapter.

\subsection*{The Wiretap Channel Problem}

The wiretap channel problem studies reliable transmission of messages
over a communication channel which is partially observable by a
\emph{wiretapper}.  As a basic example, suppose that we wish to
transmit a sensitive document over the internet. Loosely speaking, the
data is transmitted in form of packets, consisting of blocks of
information, through the network.

Packets may be transmitted along different paths over the network
through a cloud of intermediate transmitters, called \emph{routers},
until delivered at the destination. Now an adversary who has access to
a set of the intermediate routers may be able to learn a substantial
amount of information about the message being transmitted, and thereby
render the communication system insecure.

A natural solution for assuring secrecy in transmission is to use a
standard cryptographic scheme to encrypt the information at the
source.  However, the infor\-mation-theoretic limitation of the
adversary in the above scenario (that is, the fact that not all of the
intermediate routers, but only a limited number of them are being
eavesdropped) makes it possible to provably guarantee secure
transmission by using a suitable encoding at the source. In
particular, in a wiretap scheme, the original data is encoded at the
source to a slightly redundant sequence, that is then transmitted to
the recipient.  As it turns out, the scheme can be designed in such a
way that no information is leaked to the intruder and moreover no
secrets (e.g., an encryption key) need to be shared between the two
parties prior to transmission.

We study this problem in Chapter~\ref{chap:wiretap}. The main
contribution of this chapter is a construction of
infor\-mation-theoretically secure and optimal wiretap schemes that
guarantee secrecy in various settings of the problem.  In particular
the scheme can be applied to point-to-point communication models as
well as networks, even in presence of noise or active intrusion (i.e.,
when the adversary not only eavesdrops, but also alters the
information being transmitted). The construction uses an explicit
family of randomness extractors as the main building block.

\subsection*{Combinatorial Group Testing}

Group testing is a classical combinatorial problem that has
applications in surprisingly diverse and seemingly unrelated areas,
from data structures to coding theory to biology.

Intuitively, the problem can be described as follows: Suppose that
blood tests are taken from a large population (say hundreds of
thousands of people), and it is suspected that a small number (e.g.,
up to one thousand) carry a disease that can be diagnosed using costly
blood tests. The idea is that, instead of testing blood samples one by
one, it might be possible to pool them in fairly large groups, and
then apply the tests on the groups without affecting reliability of
the tests.  Once a group is tested negative, all the samples
participating in the group must be negative and this may save a large
number of tests. Otherwise, a positive test reveals that at least one
of the individuals in the group must be positive (though we do not
learn which).

The main challenge in group testing is to design the pools in such a
way to allow identification of the exact set of infected population
using as few tests as possible, thereby economizing the identification
process of the affected individuals.  In Chapter~\ref{chap:testing} we
study the group testing problem and its variations.  In particular, we
consider a scenario where the tests can produce highly unreliable
outcomes, in which case the scheme must be designed in such a way that
allows correction of errors caused by the presence of unreliable
measurements. Moreover, we study a more general \emph{threshold}
variation of the problem in which a test returns positive if the
number of positives participating in the test surpasses a certain
threshold. This is a more reasonable model than the classical one,
when the tests are not sufficiently sensitive and may be affected by
dilution of the samples pooled together.  In both models, we will use
randomness condensers as combinatorial building blocks for
construction of optimal, or nearly optimal, explicit measurement
schemes that also tolerate erroneous outcomes.

\subsection*{Capacity Achieving Codes}

The theory of error-correcting codes aims to guarantee reliable
transmission of information over an unreliable communication medium,
known in technical terms as a \emph{channel}. In a classical model,
messages are encoded into sequences of bits at their source, which are
subsequently transmitted through the channel.  Each bit being
transmitted through the channel may be flipped (from $0$ to $1$ or
vice versa) with a small probability.

Using an error-correcting code, the encoded sequence can be designed
in such a way to allow correct recovery of the message at the
destination with an overwhelming probability (over the randomness of
the channel). However, the cost incurred by such an encoding scheme is
a loss in the transmission rate, that is, the ratio between the
information content of the original message and the length of the
encoded sequence (or in other words, the effective number of bits
transmitted per channel use).

A \emph{capacity achieving code} is an error correcting code that
essentially maximizes the transmission rate, while keeping the error
probability negligible. The maximum possible rate depends on the
channel being considered, and is a quantity given by the \emph{Shannon
  capacity} of the channel.

In Chapter~\ref{chap:capacity}, we consider a general class of
communication channels (including the above example) and show how
randomness condensers and extractors can be used to design capacity
achieving \emph{ensembles of} codes for them.  We will then use the
obtained ensembles to obtain \emph{explicit} constructions of capacity
achieving codes that allow efficient encoding and decoding as well.

\subsection*{Codes on the Gilbert-Varshamov Bound}

While randomness extractors aim for eliminating the need for
\emph{pure} randomness in algorithms, a related class of objects known
as \emph{pseudorandom generators} aim for eliminating randomness
altogether. This is made meaningful by a fundamental idea saying that
randomness should be defined \emph{relative to the observer}. The idea
can be perhaps best described by an example due to Goldreich
\cite{ref:Gol08}*{Chapter~8}, quoted below:

\begin{quote}
  ``Alice and Bob play \textsc{head} or \textsc{tail} in one of the
  following four ways. In all of them Alice flips a coin high in the
  air, and Bob is asked to guess its outcome \emph{before} the coin
  hits the floor. The alternative ways differ by the knowledge Bob has
  before making his guess.

  In the first alternative, Bob has to announce his guess before Alice
  flips the coin. Clearly, in this case Bob wins with probability
  $1/2$.

  In the second alternative, Bob has to announce his guess while the
  coin is spinning in the air. Although the outcome is
  \emph{determined in principle} by the motion of the coin, Bob does
  not have accurate information on the motion. Thus we believe that,
  also in this case Bob wins with probability $1/2$.

  The third alternative is similar to the second, except that Bob has
  at his disposal sophisticated equipment capable of providing
  accurate information on the coin's motion as well as on the
  environment affecting the outcome. However, Bob cannot process this
  information in time to improve his guess.

  In the fourth alternative, Bob's recording equipment is directly
  connected to a \emph{powerful computer} programmed to solve the
  motion equations and output a prediction. It is conceivable that in
  such a case Bob can improve substantially his guess of the outcome
  of the coin.''
\end{quote}

Following the above description, in principle the outcome of a coin
flip may well be deterministic. However, as long as the observer does
not have enough resources to gain any advantage predicting the
outcome, the coin flip should be considered random for him. In this
example, what makes the coin flip random for the observer is the
inherent \emph{hardness} (and not necessarily \emph{impossibility}) of
the prediction procedure. The theory of pseudorandom generators aim to
express this line of thought in rigorous ways, and study the
circumstances under which randomness can be \emph{simulated} for a
particular class of observers.

\newcommand{\BPP}{\mathsf{BPP}} \newcommand{\sP}{\mathsf{P}} The advent of
probabilistic algorithms that are unparalleled by deterministic
methods, such as randomized primality testing (before the AKS
algorithm \cite{ref:AKS04}), polynomial identity testing and the like
initially made researchers believe that the class of problems solvable
by randomized polynomial-time algorithms (in symbols, $\BPP$) might be
strictly larger than those solvable in polynomial-time without the
need for randomness (namely, $\sP$) and conjecture $\sP \neq \BPP$.
To this date, the ``$\sP$ vs. $\BPP$'' problem remains one
of the most challenging problems in theoretical computer science.

Despite the initial belief, more recent research has led most
theoreticians to believe otherwise, namely that $\sP = \BPP$.
This is supported by recent discovery of deterministic
algorithms such as the AKS primality test, and more importantly, the
advent of strong pseudorandom generators.  In a seminal
work~\cite{NW}, Nisan and Wigderson showed that a ``hard to compute''
function can be used to efficiently transform a short sequence of
random bits into a much longer sequence that looks
\emph{indistinguishable} from a purely random sequence to any
efficient algorithm.  In short, they showed how to construct
\emph{pseudorandomness} from \emph{hardness}.  Though the underlying
assumption (that certain hard functions exists) is not yet proved, it
is intuitively reasonable to believe (just in the same way that, in
the coin flipping game above, the hardness of gathering sufficient
information for timely prediction of the outcome by Bob is reasonable
to believe without proof).

In Chapter~\ref{chap:gv} we extend Nisan and Wigderson's method
(originally aimed for probabilistic algorithms) to combinatorial
constructions and show that, under reasonable hardness assumptions, a
wide range of probabilistic combinatorial constructions can be
substantially derandomized.

The specific combinatorial problem that the chapter is based on is the
construction of error-correcting codes that attain the rate versus
error-tolerance trade-off shown possible using the probabilistic
method (namely, construction of codes on the so-called
\emph{Gilbert-Varshamov bound}). In particular, we demonstrate a small
ensemble of efficiently constructible error-correcting codes almost
all of which being as good as random codes (under a reasonable
assumption). Even though the method is discussed for construction of
error-correcting codes, it can be equally applied to numerous other
probabilistic constructions; e.g., construction of optimal Ramsey
graphs.

\subsection*{Reading Guidelines}

The material presented in each of the technical chapters of this
thesis (Chapters~3--6) are presented independently so they can be read
in any order.  Since the theory of randomness extractors plays a
central role in the technical content of this thesis,
Chapter~\ref{chap:extractor} is devoted to an introduction to this
theory, and covers some basic constructions of extractors and
condenser that are used as building blocks in the main chapters. Since
the extractor theory is already an extensively developed area, we will
only touch upon basic topics that are necessary for understanding the
thesis.

Apart from extractors, we will extensively use fundamental notions of
coding theory throughout the thesis. For that matter, we have provided
a brief review of such notions in Appendix~\ref{app:coding}.

The additional mathematical background required for each chapter is
provided when needed, to the extent of not losing focus.  For a
comprehensive study of the basic tools being used, we refer the reader
to \cites{ref:MR,ref:MU,probMethod} (probability, randomness in
algorithms, and probabilistic constructions), \cite{HLW06} (expander
graphs), \cites{ref:AB09,ref:Gol08} (modern complexity theory),
\cites{ref:MS,ref:vanl,ref:Roth} (coding theory and basic algebra
needed), \cite{ref:Venkat} (list decoding), and
\cites{ref:groupTesting,ref:DH06} (combinatorial group testing).

\emph{Each chapter of the thesis is concluded by the opening notes of a
piece of music that I truly admire.}

\musicBoxIntro


\Chapter{Extractor Theory}
\epigraphhead[70]{\epigraph{\textsl{``Art would be useless if the world were perfect,
as man wouldn't look for harmony but would simply live in it.''}}{\textit{--- Andrei Tarkovsky}}}
\label{chap:extractor}

Suppose that you are given a possibly biased coin that falls heads
some $p$ fraction of times ($0 < p < 1$) and are asked to use it to
``simulate'' fair coin flips.  A natural approach to solve this
problem would be to first try to ``learn'' the bias $p$ by flipping
the coin a large number of times and observing the fraction of times
it falls heads during the experiment, and then using this knowledge to
encode the sequence of biased flips to its infor\-mation-theoretic
entropy.

Remarkably, back in 1951 John von~Neumann \cite{vN51} demonstrated a
simple way to solve this problem without knowing the bias $p$: flip
the coin twice and one of the following cases may occur:

\begin{enumerate}
\item The first flip shows Heads and the second Tails: output ``H''.
\item The first flip shows Tails and the second Heads: output ``T''.
\item Otherwise, repeat the experiment.
\end{enumerate}

Note that the probability that the output symbol is ``H'' is precisely
equal to it being ``T'', namely, $p(1-p)$. Thus, the outcome of this
process represents a perfectly fair coin toss. This procedure might be
somewhat wasteful; for instance, it is expected to waste half of the
coin flips even if $p=1/2$ (that is, if the coin is already fair) and
that is the cost we pay for not knowing $p$.  But nevertheless, it
transforms an imperfect, not fully known, source of randomness into a
perfect source of random bits.

This example, while simple, demonstrates the basic idea in what is
known as ``extractor theory''. The basic goal in extractor theory is to
improve randomness, that is, to efficiently transform a ``weak''
source of randomness into one with better qualities; in particular,
having a higher entropy per symbol. The procedure shown above, seen as
a function from the sequence of coin flips to a Boolean function (over
$\{ H,T \}$) is known as an \emph{extractor}. It is called so since it
``extracts'' pure randomness from a weak source.

When the distribution of the weak source is known, it is possible to
use techniques from \emph{source coding} (say Huffman or Arithmetic
Coding) to compress the information to a number of bits very close to
its actual entropy, without losing any of the source information. What
makes extractor theory particularly challenging is the following
issues:

\begin{enumerate}
\item An extractor knows little about the exact source
  distribution. Typically nothing more than a lower bound on the
  source entropy, and no structure is assumed on the source. In the
  above example, even though the source distribution was unknown, it
  was known to be an i.i.d. sequence (i.e., a sequence of independent,
  identically distributed symbols).  This need not be the case in
  general.

\item The output of the extractor must ``strongly'' resemble a uniform
  distribution (which is the distribution with maximum possible
  entropy), in the sense that no statistical test (no matter how
  complex) should be able to distinguish between the output
  distribution and a purely random sequence. Note, for example, that a
  sequence of $n-1$ uniform and independent bits followed by the
  symbol ``$0$'' has $n-1$ bits of entropy, which is only slightly
  lower than that of $n$ purely random bits (i.e., $n$). However, a
  simple statistical test can trivially distinguish between the two
  distributions by only looking at the last bit.
\end{enumerate}

Since extractors and related objects (in particular, \emph{lossless
  condensers}) play a central role in the technical core of this thesis,
we devote this chapter to a formal treatment of extractor theory,
introducing the basic ideas and some fundamental constructions. In
this chapter, we will only cover basic notions and discuss a few of
the results that will be used as building blocks in the rest of
thesis.

\section{Probability Distributions}

\subsection{Distributions and Distance}

In this thesis we will focus on probability distributions over finite
domains. Let $(\Omega, \mathsf{E}, \cX)$ be a probability space, where
$\Omega$ is a finite sample space, $\mathsf{E}$ is the set of events
(that in our work, will always consist of the set of subsets of
$\Omega$), and $\cX$ is a probability measure.  The probability
assigned to each outcome $x \in \Omega$ by $\cX$ will be denoted by
$\cX(x)$, or $\Pr_\cX(x)$\index{notation!$\cX(x),
  \Pr_\cX(x)$}. Similarly, for an event $T \in \mathsf{E}$, we will
denote the probability assigned to $T$ by $\cX(T)$, or
$\Pr_\cX[T]$\index{notation!$\cX(T), \Pr_\cX[T]$} (when clear from the
context, we may omit the subscript $\cX$). The
\emph{support}\index{probability distribution!support of} of $\cX$ is
defined as
\[
\supp(\cX) \eqdef \{ x\in \Omega\colon \cX(x) > 0 \}.
\]
A particularly important probability measure is defined by the uniform
distribution, which assigns equal probabilities to each element of
$\Omega$.  We will denote the uniform distribution over $\Omega$ by
$\U_\Omega$\index{notation!$\U_\Omega, \U_n$}, and use the shorthand
$\U_n$, for an interger $n \geq 1$, for $\U_{\{0,1\}^n}$.  We will use
the notation $X \sim \cX$ to denote that the random variable $X$ is
drawn from the probability distribution $\cX$.

It is often convenient to think about the probability measure as a
real vector of dimension $|\Omega|$, whose entries are indexed by the
elements of $\Omega$, such that the value at the $i$th entry of the
vector is $\cX(i)$.

An important notion for our work is the distance between
distributions.\index{probability distribution!distance} There are
several notions of distance in the literature, some stronger than the
others, and often the most suitable choice depends on the particular
application in hand.  For our applications, the most important notion
is the $\ell_p$ distance:

\begin{defn}
  Let $\cX$ and $\cY$ be probability distributions on a finite domain
  $\Omega$.  Then for every $p \geq 1$, their $\ell_p$ distance,
  denoted by $\| \cX - \cY \|_p$, is given by \index{notation!{$\bbar
      \cX - \cY \bbar_p$}}
  \[
  \| \cX - \cY \|_p \eqdef \left(\sum_{x\in\Omega} |\cX(x) -
    \cY(y)|^p\right)^{1/p}.
  \]
  We extend the distribution to the special case $p = \infty$, to
  denote the \emph{point-wise} distance:
  \[
  \| \cX - \cY \|_\infty \eqdef \max_{x\in\Omega}{ |\cX(x) - \cY(y)|
  }.
  \]
  The distributions $\cX$ and $\cY$ are called
  $\eps$-close\index{probability distribution!$\eps$-close} with
  respect to the $\ell_p$ norm \ifff\ $\| \cX - \cY \|_p \leq \eps$.
\end{defn}
We remark that, by the Cauchy-Schwarz inequality, the following
relationship between $\ell_1$ and $\ell_2$ distances holds:
\[
\| \cX - \cY \|_2 \leq \| \cX - \cY \|_1 \leq \sqrt{|\Omega|} \cdot \|
\cX - \cY \|_2.
\]

Of particular importance is the \emph{statistical} (or \emph{total
  variation}) distance\index{statistical distance}.  This is defined
as half the $\ell_1$ distance between the distributions:
\[
\| \cX - \cY \| \eqdef \textstyle \frac{1}{2} \| \cX - \cY \|_1.
\]
We may also use the notation $\dist(\cX,\cY)$ to denote the
statistical distance.  We call two distributions $\eps$-close \ifff\
their statistical distance is at most $\eps$.  When there is no risk
of confusion, we may extend such notions as distance to the random
variables they are sampled from, and, for instance, talk about two
random variables being $\eps$-close.

This is in a sense, a very strong notion of distance since, as the
following proposition suggests, it captures the worst-case difference
between the probability assigned by the two distributions to any
event:

\begin{prop} \label{prop:statDist} Let $\cX$ and $\cY$ be
  distributions on a finite domain $\Omega$. Then $\cX$ and $\cY$ are
  $\eps$-close \ifff\ for every event $T \subseteq \Omega$, $|
  \Pr_\cX[T] - \Pr_\cY[T] | \leq \eps$.
\end{prop}

\begin{proof}
  Denote by $\Omega_\cX$ and $\Omega_\cY$ the following partition of
  $\Omega$:
  \[
  \Omega_\cX \eqdef \{ x \in \Omega\colon \cX(x) \geq \cY(x) \}, \quad
  \Omega_\cY \eqdef \Omega\setminus T_\cX.
  \]
  Thus, $\| \cX - \cY \| = 2(\Pr_\cX(\Omega_\cX)-\Pr_\cY(\Omega_\cX))
  = 2(\Pr_\cY(\Omega_\cY)-\Pr_\cX(\Omega_\cY))$.  Let $p_1 \eqdef
  \Pr_\cX[T \cap \Omega_\cX] - \Pr_\cY[T \cap \Omega_\cX]$, and $p_2
  \eqdef \Pr_\cY[T \cap \Omega_\cY] - \Pr_\cX[T \cap
  \Omega_\cY]$. Both $p_1$ and $p_2$ are positive numbers, each no
  more than $\eps$.  Therefore,
  \[
  |\Pr_\cX[T] - \Pr_\cY[T]| = |p_1 - p_2| \leq \eps.
  \]

  For the reverse direction, suppose that for every event $T \subseteq
  \Omega$, $| \Pr_\cX[T] - \Pr_\cY[T] | \leq \eps$. Then,
  \[
  \| \cX - \cY \|_1 = |\Pr_\cX[\Omega_X]-\Pr_\cY[\Omega_X]| +
  |\Pr_\cX[\Omega_Y]-\Pr_\cY[\Omega_Y]| \leq 2\eps.
  \]
\end{proof}

An equivalent way of looking at an event $T \subseteq \Omega$ is by
defining a predicate $P\colon \Omega \to \zo$ whose set of accepting
inputs is $T$; namely, $P(x) = 1$ if and only if $x \in T$. In this
view, Proposition~\ref{prop:statDist} can be written in the following
equivalent form.

\begin{prop} \label{prop:disting} Let $\cX$ and $\cY$ be distributions
  on the same finite domain $\Omega$. Then $\cX$ and $\cY$ are
  $\eps$-close if and only if, for every distinguisher $P\colon \Omega
  \to \zo$, we have
  \[
  \left|\Pr_{X \sim \cX}[P(X) = 1] - \Pr_{Y \sim \cY}[P(Y) = 1]\right|
  \leq \eps.
  \] \qed
\end{prop}

The notion of \emph{convex combination} of distributions is defined as
follows:
\begin{defn}
  Let $\cX_1, \cX_2, \ldots, \cX_n$ be probability distributions over
  a finite space $\Omega$ and $\alpha_1, \alpha_2, \ldots, \alpha_n$
  be nonnegative real values that sum up to $1$. Then the \emph{convex
    combination} \[ \alpha_1 \cX_1 + \alpha_2 \cX_2 + \cdots +
  \alpha_n \cX_n \] is a distribution $\cX$ over $\Omega$ given by the
  probability measure
  \[ \Pr_\cX(x) \eqdef \sum_{i=1}^n \alpha_i \Pr_{\cX_i}(x), \] for
  every $x \in \Omega$.
\end{defn}

When regarding probability distributions as vectors of probabilities
(with coordinates indexed by the elements of the sample space), convex
combination of distributions is merely a linear combination
(specifically, a point-wise average) of their vector forms.  Thus
intuitively, one expects that if a probability distribution is close
to a collection of distributions, it must be close to any convex
combination of them as well.  This is made more precise in the
following proposition.

\begin{prop} \label{prop:ConvexOfClose} Let $\cX_1, \cX_2, \ldots,
  \cX_n$ be probability distributions, all defined over the same
  finite set $\Omega$, that are all $\eps$-close to some distribution
  $\cY$. Then any convex combination
  \[
  \cX := \alpha_1 \cX_1 + \alpha_2 \cX_2 + \cdots + \alpha_n \cX_n
  \]
  is $\eps$-close to $\cY$.
\end{prop}
\begin{proof}
  We give a proof for the case $n=2$, which generalizes to any larger
  number of distributions by induction.  Let $T \subseteq \Omega$ be
  any nonempty subset of $\Omega$.  Then we have
  \begin{eqnarray*}
    | \Pr_\cX[T] - \Pr_\cY[T] | &=& | \alpha_1 \Pr_{\cX_1}[T] + (1-\alpha_1) \Pr_{\cX_2}[T] - \Pr_\cY[T] | \\
    &=& | \alpha_1 (\Pr_{\cY}[T] + \eps_1) + (1-\alpha_1) (\Pr_{\cY}[T] + \eps_2) - \Pr_\cY[T] |,
  \end{eqnarray*}
  where $|\eps_1|, |\eps_2| \leq \eps$ by the assumption that $\cX_1$
  and $\cX_2$ are $\eps$-close to $\cY$. Hence the distance simplifies
  to
  \[
  |\alpha_1 \eps_1 + (1-\alpha_1) \eps_2|,
  \]
  and this is at most $\eps$.
\end{proof}
In a similar manner, it is straightforward to see that a convex
combination $(1-\eps)\cX + \eps \cY$ is $\eps$-close to $\cX$.

Sometimes, in order to show a claim for a probability distribution it
may be easier, and yet sufficient, to write the distribution as a
convex combination of ``simpler'' distributions and then prove the
claim for the simpler components. We will examples of this technique
when we analyze constructions of extractors and condensers.

\subsection{Entropy}

A central notion in the study of randomness is related to the
\emph{information content} of a probability distribution. Shannon
formalized this notion in the following form:

\begin{defn}
  Let $\cX$ be a distribution on a finite domain $\Omega$. The
  \emph{Shannon entropy} of $\cX$ (in bits) is defined as
  \[
  H(\cX) := \sum_{x \in \supp(\cX)} -\cX(x) \log_2 \cX(x) = \Ex_{X
    \sim \cX}[ - \log_2 \cX(X) ].
  \]
\end{defn}

Intuitively, Shannon entropy quantifies the number of bits required to
specify a sample drawn from $\cX$ \emph{on average}. This intuition is
made more precise, for example by Huffman coding that suggest an
efficient algorithm for encoding a random variable to a binary
sequence whose expected length is almost equal to the Shannon entropy
of the random variable's distribution (cf.\ \cite{ref:cover}).  For
numerous applications in computer science and cryptography, however,
the notion of Shannon entropy--which is an \emph{average-case}
notion--is not well suitable and a \emph{worst-case} notion of entropy
is required.  Such a notion is captured by \emph{min-entropy}, defined
below.

\begin{defn}
  Let $\cX$ be a distribution on a finite domain $\Omega$. The
  \emph{min-entropy} of $\cX$ (in bits) is defined as
  \[
  H_\infty(\cX) := \min_{x \in \supp(\cX)} - \log_2 \cX(x).
  \]
\end{defn}
Therefore, the min-entropy of a distribution is at least $k$ if and
only if the distribution assigns a probability of at most $2^{-k}$ to
any point of the sample space (such a distribution is called a
$k$-source).  It also immediately follows by definitions that a
distribution having min-entropy at least $k$ must also have a Shannon
entropy of at least $k$.  When $\Omega = \zo^n$, we define the
\emph{entropy rate} of a distribution $\cX$ on $\Omega$ as
$H_\infty(\cX)/n$.

A particular class of probability distributions for which the notions
of Shannon entropy an min-entropy coincide is \emph{flat
  distributions}\index{distribution!flat}.  A distribution on $\Omega$
is called flat if it is uniformly supported on a set $T \subseteq
\Omega$; that is, if it assigns probability $1/|T|$ to all the points
on $T$ and zeros elsewhere.  The Shannon- and min-entropies of such a
distribution are both $\log_2 |T|$ bits.

An interesting feature of flat distributions is that their convex
combinations can define any arbitrary probability distribution with a
nice preservence of the min-entropy, as shown below.

\begin{prop} \label{prop:convex} Let $K$ be an integer. Then any
  distribution $\cX$ with min-entropy at least $\log K$ can be
  described as a convex combination of flat distributions with
  min-entropy $\log K$.
\end{prop}

\begin{proof}
  Suppose that $\cX$ is distributed on a finite domain $\Omega$.  Any
  probability distribution on $\Omega$ can be regarded as a real
  vector with coordinates indexed by the elements of $\Omega$,
  encoding its probability measure.  The set of distributions
  $(p_i)_{i \in \Omega}$ with min-entropy at least $\log K$ form a
  simplex
  \[
  \begin{array}{c}
    (\forall i \in \Omega)\ 0 \leq p_i \leq 1/K,\\
    \sum_{i \in \Omega} p_i = 1,
  \end{array}
  \]
  whose corner points are flat distributions. The claim follows since
  every point in the simplex can be written as a convex combination of
  the corner points.
\end{proof}

\section{Extractors and Condensers}

\subsection{Definitions}

Intuitively, an extractor is a function that transforms impure
randomness; i.e., a random source containing a sufficient amount of
entropy, to an almost uniform distribution (with respect to a suitable
distance measure; e.g., statistical distance).

Suppose that a source $\cX$ is distributed on a sample space $\Omega
:= \zo^n$ with a distribution containing at least $k$ bits of
min-entropy. The goal is to construct a function $f\colon \zo^n \to
\zo^m$ such that $f(\cX)$ is $\eps$-close to the uniform distribution
$\U_m$, for a negligible distance $\eps$ (e.g., $\eps =
2^{-\Omega(n)}$). Unfortunately, without having any further knowledge
on $\cX$, this task becomes impossible. To see why, consider the
simplest nontrivial case where $k = n-1$ and $m=1$, and suppose that
we have come up with a function $f$ that extracts one almost unbiased
coin flip from any $k$-source. Observe that among the set of
pre-images of $0$ and $1$ under $f$; namely, $f^{-1}(0)$ and
$f^{-1}(1)$, at least one must have size $2^{n-1}$ or more. Let $\cX$
be the flat source uniformly distributed on this set. The distribution
$\cX$ constructed this way has min-entropy at least $n-1$ yet $f(\cX)$
is always constant. In order to alleviate this obvious impossibility,
one of the following two solutions is typically considered:

\begin{enumerate}
\item Assume some additional structure on the source: In the
  counterexample above, we constructed an opportunistic choice of the
  source $\cX$ from the function $f$.  However, in general the source
  obtained this way may turn out to be exceedingly complex and
  unstructured, and the fact that $f$ is unable to extract any
  randomness from this particular choice of the source might be of
  little concern.  A suitable way to model this observation is to
  require a function $f$ that is expected to extract randomness only
  from a \emph{restricted class} of randomness sources.

  The appropriate restriction in question may depend on the context
  for which the extractor is being used. A few examples that have been
  considered in the literature include:

  \begin{itemize}
  \item Independent sources: In this case, the source $\cX$ is
    restricted to be a product distribution with two or more
    components. In particular, one may assume the source to be the
    product distribution of $r \geq 2$ independent random variables
    $X_1, \ldots, X_r \in \zo^{n'}$ that are each sampled from an
    arbitrary $k'$-source (assuming $n = rn'$ and $k=rk'$).

  \item Affine sources: We assume that the source $\cX$ is uniformly
    supported on an arbitrary translation of an unknown
    $k$-dimensional vector subspace of\footnote{Throughout the thesis,
      for a prime power $q$, we will use the notation
      $\F_q$\index{notation!$\F_q := \mathrm{GF}(q)$} to denote the
      finite field with $q$ elements.} $\F_2^n$.  A further
    restriction of this class is known as \emph{bit-fixing sources}. A
    bit-fixing source is a product distribution of $n$ bits $(X_1,
    \ldots, X_n)$ where for some unknown set of coordinates positions
    $S \subseteq [n]$ of size $k$, the variables $X_i$ for $i \in S$
    are independent and uniform bits, but the rest of the $X_i$'s are
    fixed to unknown binary values. In Chapter~\ref{chap:wiretap}, we
    will discuss these classes of sources in more detail.

  \item Samplable sources: This is a class of sources first studied by
    Trevisan and Vadhan \cite{TV00}.  In broad terms, a samplable
    source is a source $\cX$ such that a sample from $\cX$ can
    produced out of a sequence of random and independent coin flips by
    a restricted computational model. For example, one may consider
    the class of sources of min-entropy $k$ such that for any source
    $\cX$ in the class, there is a function $f\colon \zo^r \to \zo^n$,
    for some $r \geq k$, that is computable by polynomial-size Boolean
    circuits and satisfies $f(\U_r) \sim \cX$.
  \end{itemize}

  For restricted classes of sources such as the above examples, there
  are deterministic functions that are good extractors for all the
  sources in the family. Such deterministic functions are known as
  \emph{seedless extractors} for the corresponding family of
  sources. For instance, an affine extractor for entropy $k$ and error
  $\eps$ (in symbols, an affine $(k,\eps)$-extractor) is a mapping
  $f\colon \F_2^n \to \F_2^m$ such that for every affine $k$-source
  $\cX$, the distribution $f(\cX)$ is $\eps$-close to the uniform
  distribution $\U_m$.

  In fact, it is not hard to see that for any family of not ``too
  many'' sources, there is a function that extracts almost the entire
  source entropy of the sources (examples include affine $k$-sources,
  samplable $k$-sources, and two independent sources\footnote{For this
    case, it suffices to count the number of independent \emph{flat}
    sources.}).  This can be shown by a probabilistic argument that
  considers a random function and shows that it achieves the desired
  properties with overwhelming probability.

\item Allow a short \emph{random seed}: The second solution is to
  allow extractor to use a small amount of pure randomnness as a
  ``catalyst''.  Namely, the extractor is allowed to require two
  inputs: a sample from the unknown source and a short sequence of
  random and independent bits that is called the \emph{seed}.  In this
  case, it turns out that extracting almost the entire entropy of the
  weak source becomes possible, without any structural assumptions on
  the source and using a very short independent seed.  Extractors that
  require an auxiliary random input are called \emph{seeded
    extractors}.  In fact, an equivalent of looking at seeded
  extractors is to see them as \emph{seedless} extractors that assume
  the source to be structured as a product distribution of two
  sources: an arbitrary $k$-source and the uniform distribution.

\end{enumerate}

For the rest of this chapter, we will focus on seeded
extractors. Seedless extractors (especially affine extractors) are
treated in Chapter~\ref{chap:wiretap}.  A formal definition of
(seeded) extractors is as follows.

\begin{defn} \index{extractor} A function $f\colon \zo^n \times \zo^d
  \to \zo^m$ is a \emph{$(k, \eps)$-extractor} if, for every
  $k$-source $\cX$ on $\zo^n$, the distribution $f(\cX, \U_d)$ is
  $\eps$-close (in statistical distance) to the uniform distribution
  on $\zo^m$.  The parameters $n$, $d$, $k$, $m$, and $\eps$ are
  respectively called the \emph{input length}, \emph{seed length},
  \emph{entropy requirement}, \emph{output length}, and \emph{error}
  of the extractor.
\end{defn}

An important aspect of randomness extractors is their computational
complexity.  For most applications, extractors are required to be
efficiently computable functions.  We call an extractor
\emph{explicit} \index{extractor!explicit} if it is computable in
polynomial time (in its input length). Though it is rather
straightforward to show existence of good extractors using
probabilistic arguments, coming up with a nontrivial explicit
construction can turn out a much more challenging task.  We will
discuss and analyze several important explicit constructions of seeded
extractors in Section~\ref{sec:extrConstr}.

Note that, in the above definition of extractors, achieving an output
length of up to $d$ is trivial: the extractor can merely output its
seed, which is guaranteed to have a uniform distribution! Ideally the
output of an extractor must be ``almost independent'' of its seed, so
that the extra randomness given in the seed can be ``recycled''. This
idea is made precise in the notion of \emph{strong extractors} given
below.

\begin{defn} \index{extractor!strong} A function $f\colon \zo^n \times
  \zo^d \to \zo^m$ is a strong \emph{$(k, \eps)$-extractor} if, for
  every $k$-source $\cX$ on $\zo^n$, and random variables $X \sim
  \cX$, $Z \sim \U_d$, the distribution of the random variable $(X,
  f(X, Z))$ is $\eps$-close (in statistical distance) to $\U_{d+m}$.
\end{defn}

A fundamental property of strong extractors that is essential for
certain applications is that, the extractor's output remains close to
uniform for almost all fixings of the random seed. This is made clear
by an ``averaging argument'' stated formally in the proposition below.

\begin{prop} \label{prop:strongDist} Consider joint distributions
  $\tilde{\cX} := (\cZ, \cX)$ and $\tilde{\cY} := (\cZ, \cY)$ that are
  $\eps$-close, where $\cX$ and $\cY$ are distributions on a finite
  domain $\Omega$, and $\cZ$ is uniformly distributed on $\zo^d$. For
  every $z \in \zo^d$, denote by $\cX_z$ the distribution of the
  second coordinate of $\tilde{\cX}$ conditioned on the first
  coordinate being equal to $z$, and similarly define $\cY_z$ for the
  distribution $\tilde{\cY}$.  Then, for every $\delta > 0$, at least
  $(1-\delta)2^d$ choices of $z \in \zo^d$ must satisfy
  \[
  \| \cX_z - \cY_z \| \leq \eps/\delta.
  \]
\end{prop}

\begin{proof}
  Clearly, for every $\omega \in \Omega$ and $z \in \zo^d$, we have
  $\cX_z(\omega) = 2^d \cX(z,\omega)$ and similarly, $\cY_z(\omega) =
  2^d \cY(z,\omega)$.  Moreover from the definition of statistical
  distance,
  \[
  \sum_{z \in \zo^d} \sum_{\omega \in \Omega} |\cX(z,\omega) -
  \cY(z,\omega)| \leq 2\eps.
  \]
  Therefore,
  \[
  \sum_{z \in \zo^d} \sum_{\omega \in \Omega} |\cX_z(\omega) -
  \cY_z(\omega)| \leq 2^{d+1} \eps,
  \]
  which can be true only if for at least $(1-\delta)$ fraction of the
  choices of $z$, we have
  \[
  \sum_{\omega \in \Omega} |\cX_z(\omega) - \cY_z(\omega)| \leq 2
  \eps/\delta,
  \]
  or in other words,
  \[
  \| \cX_z - \cY_z \| \leq \eps/\delta.
  \]
  This shows the claim.
\end{proof}

Thus, according to Proposition~\ref{prop:strongDist}, for a strong
$(k, \eps)$-extractor \[f\colon \zo^n \times \zo^d \to \zo^m\] and a
$k$-source $\cX$, for $1-\sqrt{\eps}$ fraction of the choices of $z
\in \zo^d$, the distribution $f(\cX, z)$ must be $\eps$-close to
uniform.

Extractors are specializations of the more general notion of
\emph{randomness condensers}.  Intuitively, a condenser transforms a
given weak source of randomness into a ``more purified'' but possibly
imperfect source. In general, the output entropy of a condenser might
be substantially less than the input entropy but nevertheless, the
output is generally required to have a substantially higher
\emph{entropy rate}. For the extremal case of extractors, the output
entropy rate is required to be $1$ (since the output is required to be
an almost uniform distribution).  Same as extractors, condensers can
be seeded or seedless, and also seeded condensers can be required to
be strong (similar to strong extractors). Below we define the general
notion of strong, seeded condensers.

\begin{defn} \index{condenser} A function $f\colon \zo^n \times \zo^d
  \to \zo^m$ is a \emph{strong $k \to_\eps k'$ condenser} if for every
  distribution $\cX$ on $\zo^n$ with min-entropy at least $k$, random
  variable $X \sim \cX$ and a \emph{seed} $Y \sim \U_d$, the
  distribution of $(Y, f(X, Y))$ is $\eps$-close to a distribution
  $(\U_d, \cZ)$ with min-entropy at least $d+k'$. The parameters $k$,
  $k'$, $\eps$, $k - k'$, and $m - k'$ are called the \emph{input
    entropy}, \emph{output entropy}, \emph{error}, the \emph{entropy
    loss} and the \emph{overhead} of the condenser, respectively.  A
  condenser is \emph{explicit} if it is polynomial-time computable.
\end{defn}

Similar to strong extractors, strong condensers remain effective under
almost all fixings of the seed. This follows immediately from
Proposition~\ref{prop:strongDist} and is made explicit by the
following corollary:

\begin{coro} \label{coro:strongExt} Let $f\colon \zo^n \times \zo^d
  \to \zo^m$ be a strong $k \to_\eps k'$ condenser.  Consider an
  arbitrary parameter $\delta > 0$ and a $k$-source $\cX$. Then, for
  all but at most a $\delta$ fraction of the choices of $z \in \zo^d$,
  the distribution $f(\cX, z)$ is $(\eps/\delta)$-close to a
  $k'$-source. \qed
\end{coro}

Typically, a condenser is only interesting if the output entropy rate
$k'/m$ is considerably larger than the input entropy rate $k/n$.  From
the above definition, an extractor is a condenser with zero overhead.
Another extremal case corresponds to the case where the entropy loss
of the condenser is zero. Such a condenser is called
\emph{lossless}\index{condenser!lossless}.  We will use the
abbreviated term \emph{$(k, \eps)$-condenser} for a lossless condenser
with input entropy $k$ (equal to the output entropy) and error $\eps$.
Moreover, if a function is a $(k_0, \eps)$-condenser for every $k_0
\leq k$, it is called a $(\leq k, \eps)$-condenser. Most known
constructions of lossless condensers (and in particular, all
constructions used in this thesis) are $(\leq k, \eps)$-condensers for
their entropy requirement $k$.

Traditionally, lossless condensers have been used as intermediate
building blocks for construction of extractors. Having a good lossless
condenser available, for construction of extractors it would suffice
to focus on the case where the input entropy is large.  Nevertheless,
lossless condensers have been proved to be useful for a variety of
applications, some of which we will discuss in this thesis.

\subsection{Almost-Injectivity of Lossless
  Condensers} \label{sec:inject}

Intuitively, an extractor is an almost ``uniformly surjective''
mapping.  That is, the extractor mapping distributes the probability
mass of the input source almost evenly among the elements of its
range.

On the other hand, a lossless condenser preserves the entire source
entropy on its output and intuitively, must be an almost injective
function when restricted to the domain defined by the input
distribution. In other words, in the mapping defined by the condenser
``collisions'' rarely occur and in this view, lossless condensers are
useful ``hashing'' tools.  In this section we formalize this intuition
through a simple practical application.

Given a source $\cX$ and a function $f$, if $f(\cX)$ has the same
entropy as that of $\cX$ (or in other words, if $f$ is a perfectly
lossless condenser for $\cX$) we expect that from the outcome of the
function, its input when sampled from $\cX$ must be
reconstructible. For flat distributions (that is, those that are
uniform on their support) and considering an error for the condenser,
this is shown in the following proposition. We will use this simple
fact several times throughout the thesis.

\begin{prop} \label{prop:flatmap} Let $\cX$ be a flat distribution
  with min-entropy $\log K$ over a finite sample space $\Omega$ and
  $f\colon \Omega \to \Gamma$ be a mapping to a finite set $\Gamma$.
  \begin{enumerate}
  \item If $f(\cX)$ is $\eps$-close to having min-entropy $\log K$,
    then there is a set $T \subseteq \Gamma$ of size at least
    $(1-2\eps)K$ such that
    \[
    (\forall y \in T\text{ and }\forall x,x'\in \supp(\cX))\quad
    f(x)=y \land f(x') = y \Rightarrow x = x'.
    \]
  \item Suppose $|\Gamma| \geq K$. If $f(\cX)$ has a support of size
    at least $(1-\eps)K$, then it is $\eps$-close to having
    min-entropy $\log K$.
  \end{enumerate}
\end{prop}

\begin{Proof}
  Suppose that $\cX$ is uniformly supported on a set $S \subseteq
  \Omega$ of size $K$, and denote by $\mu$ the distribution $f(\cX)$
  over $\Gamma$. For each $y \in \Gamma$, define \[n_y := |\{x \in
  \supp(\cX)\colon f(x) = y\}|.\] Moreover, define $T := \{ y \in
  \Gamma\colon n_y = 1\}$, and similarly, $T' := \{ y \in \Gamma\colon
  n_y \geq 2\}$. Observe that for each $y \in \Gamma$ we have $\mu(y)
  = n_i/K$, and also $\supp(\mu) = T \cup T'$.  Thus,
  \begin{equation} \label{eqn:TTp} |T| + \sum_{y \in T'} n_y = K.
  \end{equation}

  Now we show the first assertion. Denote by $\mu'$ a distribution on
  $\Gamma$ with min-entropy $K$ that is $\eps$-close to $\mu$, which
  is guaranteed to exist by the assumption.  The fact that $\mu$ and
  $\mu'$ are $\eps$-close implies that
  \[
  \sum_{y \in T'} | \mu(y) - \mu'(y) | \leq \eps \Rightarrow \sum_{y
    \in T'} (n_y - 1) \leq \eps K.
  \]
  In particular, this means that $|T'| \leq \eps K$ (since by the
  choice of $T'$, for each $y \in T'$ we have $n_y \geq 2$).
  Furthermore,
  \[
  \sum_{y \in T'} (n_y - 1) \leq \eps K \Rightarrow \sum_{y \in T'}
  n_y \leq \eps K + |T'| \leq 2\eps K.
  \]
  This combined with \eqref{eqn:TTp} gives
  \[
  |T| = K - \sum_{y \in T'} n_y \geq (1-2\eps) K
  \]
  as desired.

  For the second part, observe that $|T'| \leq \eps K$. Let $\mu'$ be
  any flat distribution with a support of size $K$ that contains the
  support of $\mu$. The statistical distance between $\mu$ and $\mu'$
  is equal to the difference between the probability mass of the two
  distributions on those elements of $\Gamma$ to which $\mu'$ assigns
  a bigger probability, namely,
  \[
  \frac{1}{K}(\supp(\mu') - \supp(\mu)) = \frac{\sum_{y \in T'}
    (n_y-1)}{K} = \frac{\sum_{y \in T'} n_y - |T'|}{K} =
  \frac{K-|T|-|T'|}{K},
  \]
  where we have used \eqref{eqn:TTp} for the last equality.  But
  $|T|+|T'| = |\supp(\mu)| \geq (1-\eps)K$, giving the required bound.
\end{Proof}

As a simple application of this fact, consider the following ``source
coding'' problem.  Suppose that Alice wants to send a message $x$ to
Bob through a noiseless communication channel, and that the message is
randomly sampled from a distribution $\cX$. Shannon's source coding
theorem roughly states that, there is a compression scheme that
encodes $x$ to a binary sequence $y$ of length $H(X)$ bits on average,
where $H(\cdot)$ denotes the Shannon entropy, such that Bob can
perfectly reconstruct $y$ from $x$ (cf.\
\cite{ref:cover}*{Chapter~5}).  If the distribution $\cX$ is known to
both Alice and Bob, they can use an efficient coding scheme such as
Huffman codes or Arithmetic coding to achieve this bound (up to a
small constant bits of redundancy).

On the other hand, certain \emph{universal} compression schemes are
known that guarantee an optimal compression provided that $\cX$
satisfies certain statistical properties. For instance, Lempel-Ziv
coding achieves the optimum compression rate without exact knowledge
of $\cX$ provided that is defined by a stationary, ergodic process
(cf.\ \cite{ref:cover}*{Chapter~13}).

Now consider a situation where the distribution $\cX$ is arbitrary but
only known to the receiver Bob. In this case, it is known that there
is no way for Alice to substantially compress her information without
interaction with Bob \cite{ref:AM01}.  On the other hand, if we allow
interaction, Bob may simply send a description of the probability
distribution $\cX$ to Alice so she can use a classical source coding
scheme to compress her information at the entropy.

Interestingly, it turns out that this task is still possible if the
amount of information sent to Alice is substantially lower than what
needed to fully encode the probability distribution $\cX$. This is
particularly useful if the bandwidth from Alice to Bob is
substantially lower than that of the reverse direction (consider, for
example, an ADSL connection) and for this reason, the problem is
dubbed as the \emph{asymmetric communication channel
  problem}\index{asymmetric communication}.  In particular, Watkinson
et al.~\cite{ref:WAF01} obtain a universal scheme with $H(\cX)+2$ bits
of communication from Alice to Bob and $n(H(\cX)+2)$ bits from Bob to
Alice, where $n$ is the bit-length of the message. Moreover, Adler et
al.~\cite{ref:ADHP06} obtain strong lower bounds on the number of
rounds of communication between Alice and Bob.

Now let us impose a further restriction on $\cX$ that it is uniformly
supported on a set $S \subseteq \zo^n$, and Alice knows nothing about
$S$ but its size.  If we disallow interaction between Alice and Bob,
there would still be no deterministic way for Alice to
deterministically compress her message.  This is easy to observe by
noting that any deterministic, and compressing, function
$\varphi\colon \zo^n \to \zo^m$, where $m<n$, has an output value with
as many as $2^{n-m}$ pre-images, and an adversarial choice of $S$ that
concentrates on the set of such pre-images would force the compression
scheme to fail.

However, let us allow the encoding scheme to be randomized, and err
with a small probability over the randomness of the scheme and the
message.  In this case, Alice can take a strong lossless condenser
$f\colon \zo^n \times \zo^d \to \zo^m$ for input entropy $k := \log
|S|$, choose a uniformly random seed $z \in \zo^d$, and transmit $y :=
(z,f(x,z))$ to Bob. Now we argue that Bob will be able to recover $x$
from $y$.

Let $\eps$ denote the error of the condenser. Since $f$ is a lossless
condenser for $\cX$, we know that, for $Z \sim \U_d$ and $X \sim \cX$,
the distribution of $(Z, f(X, Z))$ is
$\eps$-close to some distribution $(\U_d, \cY)$, with min-entropy at
least $d+k$.  Thus by Corollary~\ref{coro:strongExt} it follows that, for
at least $1-\sqrt{\eps}$ fraction of the choices of $z \in \zo^d$, the
distribution $\cY_z := f(\cX,z)$ is $\sqrt{\eps}$-close to having
min-entropy $k$. For any such ``good seed'' $z$,
Proposition~\ref{prop:flatmap} implies that only for at most
$2\sqrt{\eps}$ fraction of the message realizations $x \in S$ can the
encoding $f(x,z)$ be confused with a different encoding $f(x',z)$ for
some $x' \in S$, $x' \neq x$.  Altogether we conclude that, from the
encoding $y$, Bob can uniquely deduce $x$ with probability at least
$1-3\sqrt{\eps}$, where the probability is taken over the randomness
of the seed and the message distribution $\cX$.

The amount of communication in this encoding scheme is $m+d$ bits.
Using an optimal lossless condenser for $f$, the encoding length
becomes $k+O(\log n)$ with a polynomially small (in $n$) error
probability (where the exponent of the polynomial is arbitrary and
affects the constant in the logarithmic term).  On the other hand,
with the same error probability, the explict condenser of
Theorem~\ref{thm:CRVW} would give an encoding length $k+O(\log^3
n)$. Moreover, the explicit condenser of Theorem~\ref{thm:GUVcond}
results in length $k(1+\alpha) + O_\alpha(\log n)$ for any arbitrary
constant $\alpha > 0$.

\section{Constructions} \label{sec:extrConstr}

We now turn to explicit constructions of strong extractors and
lossless condensers.

Using probabilistic arguments, Radhakrishan and Ta-Shma
\cite{ref:lowerbounds} showed that, for every $k, n, \eps$, there is a
strong $(k, \eps)$-extractor with seed length $d = \log (n-k) + 2 \log
(1/\eps) + O(1)$ and output length $m = k - 2\log (1/\eps) - O(1)$. In
particular, a random function achieves these parameters with
probability $1-o(1)$. Moreover, their result show that this trade-off
is almost the best one can hope for.

Similar trade-offs are known for lossless condensers as well.
Specifically, the probabilistic construction of Radhakrishan and
Ta-Shma has been extended to the case of lossless condensers by
Capalbo et al.~\cite{ref:CRVW02}, where they show that a random
function is with high probability a strong lossless $(k,
\eps)$-condenser with seed length $d = \log n + \log(1/\eps) + O(1)$
and output length $m = k+\log(1/\eps)+ O(1)$. Moreover, this tradeoff
is almost optimal as well.

In this section, we introduce some important \emph{explicit}
constructions of both extractors and lossless condensers that are used
as building blocks of various constructions in the thesis. In
particular, we will discuss extractors and lossless condensers
obtained by the Leftover Hash Lemma, Trevisan's extractor, and a
lossless condenser due to Guruswami, Umans, and Vadhan.

\subsection{The Leftover Hash Lemma}

One of the foremost explicit constructions of extractors is given by
the \emph{Leftover Hash Lemma} first stated by Impagliazzo, Levin, and
Luby \cite{ref:ILL89}. This extractor achieves an optimal output
length $m = k - 2\log(1/\eps)$ albeit with a substantially large seed
length $d = n$.  Moreover, the extractor is a linear function for
every fixing of the seed.  In its general form, the lemma states that
any \emph{universal family of hash functions} can be transformed into
an explicit extractor. The universality property required by the hash
functions is captured by the following definition.

\newcommand{\cH}{\mathcal{H}}
\newcommand{\hLin}{\mathcal{H}_{\mathsf{lin}}}

\begin{defn} \label{defn:pwindep} \index{universal hash family}
  A family of functions $\cH = \{h_1, \ldots, h_D\}$ where $h_i\colon
  \zo^n \to \zo^m$ for $i=1, \ldots, D$ is called \emph{universal} if,
  for every fixed choice of $x, x' \in \zo^n$ such that $x \neq x'$
  and a uniformly random $i \in [D] := \{1, \ldots,
  D\}$\index{notation!$[n] := \{1,\ldots,n\}$} we have
  \[
  \Pr_i [h_i(x) = h_i(x')] \leq 2^{-m}.
  \]
\end{defn}

One of the basic examples of universal hash families is what we call
\emph{the linear family}, defined as follows. Consider an arbitrary
isomorphism $\varphi\colon \F_2^n \to \F_{2^n}$ between the vector
space $\F_2^n$ and the extension field $\F_{2^n}$, and let $0 < m \leq
n$ be an arbitrary integer.  The linear family $\hLin$ is the set $\{
h_\alpha\colon \alpha \in \F_{2^n} \}$ of size $2^n$ that contains a
function for each element of the extension field $\F_{2^n}$.  For each
$\alpha$, the mapping $h_\alpha$ is given by
\[
h_\alpha(x) := (y_1, \ldots, y_m), \text{ where $(y_1, \ldots, y_n) :=
  \varphi^{-1}(\alpha \cdot \varphi(x))$}.
\]
Observe that each function $h_\alpha$ can be expressed as a linear
mapping from $\F_2^n$ to $\F_2^m$.  Below we show that this family is
pairwise independent.

\begin{prop}
  The linear family $\hLin$ defined above is universal.
\end{prop}

\begin{proof}
  Let $x, x'$ be different elements of $\F_{2^n}$. Consider the
  mapping $f\colon \F_{2^n} \to \F_2^m$ defined as
  \[
  f(x) := (y_1, \ldots, y_m), \text{ where $(y_1, \ldots, y_n) :=
    \varphi^{-1}(x)$},
  \]
  which truncates the binary representation of a field element from
  $\F_{2^n}$ to $m$ bits.  The probability we are trying to estimate
  in Definition~\ref{defn:pwindep} is, for a uniformly random $\alpha
  \in \F_{2^n}$,
  \[
  \Pr_{\alpha \in \F_{2^n}} [f(\alpha \cdot x) = f(\alpha \cdot x')] =
  \Pr_{\alpha \in \F_{2^n}} [f(\alpha \cdot (x-x')) = 0].
  \]
  But note that $x-x'$ is a nonzero element of $\F_{2^n}$, and thus,
  for a uniformly random $\alpha$, the random variable $\alpha x$ is
  uniformly distributed on $\F_{2^n}$. It follows that
  \[
  \Pr_{\alpha \in \F_{2^n}} [f(\alpha \cdot (x-x')) = 0] = 2^{-m},
  \]
  implying that $\hLin$ is a universal family.
\end{proof}

Now we are ready to state and prove the Leftover Hash Lemma. We prove
a straightforward generalization of the lemma which shows that
universal hash families can be used to construct not only strong
extractors, but also lossless condensers.

\begin{thm} (Leftover Hash Lemma) \index{Leftover Hash
    Lemma} \label{lem:leftover} Let $\cH = \{ h_i\colon \F_2^n \to
  \F_2^m \}_{i \in \F_2^d}$ be a universal family of hash functions
  with $2^d$ elements indexed by binary vectors of length $d$, and
  define the function $ f\colon \F_2^n \times \F_2^d \to \F_2^m $ as
  $f(x, z) := h_z(x)$. Then
  \begin{enumerate}
  \item For every $k, \eps$ such that $m \leq k - 2 \log(1/\eps)$, the
    function $f$ is a strong $(k,\eps)$-extractor, and

  \item For every $k, \eps$ such that $m \geq k + 2 \log(1/\eps)$, the
    function $f$ is a strong lossless $(k,\eps)$-condenser.
  \end{enumerate}
  In particular, by choosing $\cH = \hLin$, it is possible to get
  explicit extractors and lossless condensers with seed length $d=n$.
\end{thm}

\begin{proof}
  Considering Proposition~\ref{prop:convex}, it suffices to show the
  claim when $\cX$ is a flat distribution on a support of size $K :=
  2^k$.  Define $M := 2^m$, $D := 2^d$, and let $\mu$ be any flat
  distribution over $\F_2^{d+m}$ such that $\supp(\cX) \subseteq
  \supp(\mu)$, and denote by $\cY$ the distribution of $(Z, f(X, Z))$
  over $\F_2^{d+m}$ where $X \sim \cX$ and $Z \sim \U_d$.  We will
  first upper bound the $\ell_2$ distance of the two distributions
  $\cY$ and $\mu$, that can be expressed as follows:
  {\allowdisplaybreaks
    \begin{eqnarray}
      \| \cY - \mu \|_2^2 &=& \sum_{x \in \F_2^{d+m}} (\cY(x) - \mu(x))^2 \nonumber \\
      &=& \sum_{x} \cY(x)^2 + \sum_{x} \mu(x)^2 -2 \sum_{x} \cY(x) \mu(x) \nonumber \\
      &\stackrel{\mathrm{(a)}}{=}& \sum_{x} \cY(x)^2 + \frac{1}{|\supp(\mu)|} -\frac{2}{|\supp(\mu)|} \sum_{x} \cY(x) \nonumber \\
      &=& \sum_{x} \cY(x)^2 - \frac{1}{|\supp(\mu)|} \label{eqn:LLLa},
    \end{eqnarray}}
  where $\mathrm{(a)}$ uses the fact that $\mu$ assigns probability $1/|\supp(\mu)|$
  to exactly $|\supp(\mu)|$ elements of $\F_2^{d+m}$ and zeros elsewhere.

  Now observe that $\cY(x)^2$ is the probability that two independent
  samples drawn from $\cY$ turn out to be equal to $x$, and thus,
  $\sum_{x} \cY(x)^2$ is the \emph{collision probability} of two
  independent samples from $\cY$, which can be written as
  \begin{equation*}
    \sum_{x} \cY(x)^2 = \Pr_{Z,Z',X,X'}[(Z, f(X, Z)) = (Z', f(X', Z'))],
  \end{equation*}
  where $Z, Z' \sim \F_2^d$ and $X, X' \sim \cX$ are independent
  random variables.  We can rewrite the collision probability as
  {\allowdisplaybreaks
    \begin{eqnarray*}
      \sum_{x} \cY(x)^2 &=& \Pr[Z = Z'] \cdot \Pr[ f(X, Z) = f(X', Z') \mid Z = Z'] \\
      &=& \frac{1}{D} \cdot \Pr_{Z,X,X'}[ h_{Z}(X) = h_{Z}(X') ] \\
      &=& \frac{1}{D} \cdot (\Pr[X=X'] + \frac{1}{K^2} \sum_{\substack{x, x' \in \supp(\cX) \\ x \neq x'}} \Pr_{Z}[ h_{Z}(x) = h_{Z}(x') ]) \\
      &\stackrel{\mathrm{(b)}}{\leq}& \frac{1}{D} \cdot \big(\frac{1}{K} + \frac{1}{K^2} \sum_{\substack{x, x' \in \supp(\cX) \\ x \neq x'}} \frac{1}{M}\big)
      \leq \frac{1}{DM} \cdot \big(1 + \frac{M}{K}\big),
    \end{eqnarray*}}
  \noindent where $\mathrm{(b)}$ uses the assumption that $\cH$ is a
  universal hash family.  Plugging the bound in \eqref{eqn:LLLa}
  implies that
  \[
  \| \cY - \mu \|_2 \leq \frac{1}{\sqrt{DM}} \cdot \sqrt{1 -
    \frac{DM}{|\supp(\mu)|} + \frac{M}{K}}.
  \]
  Observe that both $\cY$ and $\mu$ assign zero probabilities to
  elements of $\zo^{d+m}$ outside the support of $\mu$. Thus using
  Cauchy-Schwarz on a domain of size $\supp(\mu)$, the above bound
  implies that the statistical distance between $\cY$ and $\mu$ is at
  most
  \begin{equation} \label{eqn:LLLb} \frac{1}{2} \cdot
    \sqrt{\frac{|\supp(\mu)|}{DM}} \cdot \sqrt{1 -
      \frac{DM}{|\supp(\mu)|} + \frac{M}{K}}.
  \end{equation}
  Now, for the first part of the theorem, we specialize $\mu$ to the
  uniform distribution on $\zo^{d+m}$, which has a support of size
  $DM$, and note that by the assumption that $m \leq k -
  2\log(1/\eps)$ we will have $M \leq \eps^2 K$.  Using
  \eqref{eqn:LLLb}, it follows that $\cY$ and $\mu$ are
  $(\eps/2)$-close.

  On the other hand, for the second part of the theorem, we specialize
  $\mu$ to any flat distribution on a support of size $DK$ containing
  $\supp(\cY)$ (note that, since $\cX$ is assumed to be a flat
  distribution, $\cY$ must have a support of size at most $DK$). Since
  $m \geq k + 2\log(1/\eps)$, we have $K = \eps^2 M$, and again
  \eqref{eqn:LLLb} implies that $\cY$ and $\mu$ are $(\eps/2)$-close.
\end{proof}

\subsection{Trevisan's Extractor} \label{sec:Tre}

One of the most important explicit constructions of extractors is due
to Trevisan \cite{ref:Tre}.  Since we will use this extractor at
several points in the thesis, we dedicate this section to sketch the
main ideas behind this important construction.

Trevisan's extractor can be thought of as an ``infor\-mation-theoretic''
variation of Nisan-Wigderson's pseudorandom generator that will be
discussed in detail in Chapter~\ref{chap:gv}.  For the purpose of this
exposition, we will informally demonstrate how Nisan-Wigderson's
generator works and then discuss Trevisan's extractor from a
coding-theoretic perspective.

Loosely speaking, a pseudorandom generator is an efficient and
deterministic function (where the exact meaning of ``efficient'' may
vary depending on the context) that transforms a statistically uniform
distribution on $d$ bits to a distribution on $m$ bits, for some $m
\gg d$, that ``looks random'' to any ``restricted''
distinguisher. Again the precise meaning of ``looking random'' and the
exact restriction of the distinguisher may vary. In particular, we
require the output distribution $\cX$ of the pseudorandom generator to
be such that, for every restricted distinguisher $D\colon \zo^m \to
\zo$, we have
\[
\left|\Pr_{X \sim \cX}[D(X) = 1] - \Pr_{Y \sim \U_m}[D(Y) = 1]\right|
\leq \eps,
\]
where $\eps$ is a negligible bias.  Recall that, in in light of
Proposition~\ref{prop:disting}, this is very close to what we expect
from the output distribution of an extractor, except that for the case
of pseudorandom generators the distinguisher $D$ cannot be an
arbitrary function. Indeed, when $m > d$, the output distribution of a
pseudorandom generator cannot be close to uniform and is always
distinguishable by \emph{some} distinguisher.  The main challenge in
construction of a pseudorandom gnerator is to exclude the possibility
of such a distinguisher to be included in the restricted class of
functions into consideration.  As a concrete example, one may require a
pseudorandom generator to be a polynomial-time computable function
whose output is a sequence of length $d^2$ that is indistinguishable by
linear-sized Boolean circuits with a bias better than $d^{-2}$.

Nisan and Wigderson observed that the hardness of distinguishing the
output distribution from uniform can be derived from a hardness
assumption that is inherent in the way the pseudorandom generator
itself is computed. In a way, their construction shows how to
``trade'' computational hardness with pseudorandomness. In a
simplified manner, a special instantiation of this generator can be
described as follows: Suppose that a Boolean predicate $f\colon \zo^d
\to \zo$ is hard to compute on average by ``small'' Boolean circuits;
meaning that no circuit consisting of a sufficiently small number of
gates (as determined by a \emph{security parameter}) is able to
compute $f$ substantially better than a trivial circuit that always
outputs a constant value. Then, given a random seed $Z \in \zo^d$, the
sequence $(Z,f(Z))$ is pseudorandom for small circuits. The reason can
be seen by contradiction. Let us suppose that for some distinguisher
$D$, we have
\[
\left|\Pr_{X \sim \cX}[D(X) = 1] - \Pr_{Y \sim \U_m}[D(Y) = 1]\right|
> \eps.
\]
By the following simple proposition, such a distinguisher can be
transformed into a predictor for the hard function $f$.

\begin{prop} \label{prop:predictor} Consider predicates $f\colon
  \F_2^d \to \F_2$ and $D\colon \F_2^{d+1} \to \F_2$ and suppose that
  \[
  \left|\Pr_{X \sim \U_d}[D(X,f(X)) = 1] - \Pr_{Y \sim \U_{d+1}}[D(Y)
    = 1]\right| > \eps.
  \]
  Then, there are fixed choices of $a_0, a_1 \in \F_2$ such that
  \[
  \Pr_{X \sim \U_d}[D(X,a_0)+a_1 = f(X)] > \frac{1}{2} + \eps.
  \]
\end{prop}

\begin{proof}
  Without loss of generality, assume that the quantity inside the
  absolute value is non-negative (otherwise, one can reason about the
  negation of $D$).  Consider the following randomized algorithm $A$
  that, given $x\in \F_2^d$, tries to predict $f(X)$: Flip a random
  coin $r \in F_2$. If $r = 1$, output $r$ and otherwise output
  $\bar{r}$.

  Intuitively, the algorithm $A$ tries to make a random guess for
  $f(X)$, and then feeds it to the distinguisher.  As $D$ is more
  likely to output $1$ when the correct value of $f(X)$ is supplied,
  $A$ takes the acceptance of $x$ as an evidence that the random guess
  $r$ has been correct (and vice versa). The precise analysis can be
  however done as follows.

  {\allowdisplaybreaks
    \begin{eqnarray*}
      \Pr_{X,r}[A(X) = f(X)] &=& \frac{1}{2} \Pr_{X,r}[A(X) = f(X) \mid r=f(X)] + \\&& \frac{1}{2} \Pr_{X,r}[A(X) = f(X) \mid r \neq f(X)]\\
      &=& \frac{1}{2} \Pr_{X,r}[D(X,r) = 1 \mid r=f(X)] + \\&& \frac{1}{2} \Pr_{X,r}[D(X,r) = 0 \mid r \neq f(X)]\\
      &=& \frac{1}{2} \Pr_{X,r}[D(X,r) = 1 \mid r=f(X)] + \\&& \frac{1}{2} (1-\Pr_{X,r}[D(X,r) = 1 \mid r \neq f(X)])\\
      &=& \frac{1}{2} + \Pr_{X,r}[D(X,r) = 1 \mid r=f(X)] - \\&& \frac{1}{2} \left(\Pr_{X,r}[D(X,r) = 1 \mid r=f(X)]+ \right. \\&& \left. \Pr_{X,r}[D(X,r) = 1 \mid r \neq f(X)]\right)\\
      &=& \frac{1}{2} + \Pr_{X}[D(X,f(X)) = 1] - \Pr_{X,r}[D(X,r) = 1] \\
      &>& \frac{1}{2} + \eps.
    \end{eqnarray*}}

  Therefore, by averaging, for some fixed choice of $r$ the
  probability must remain above $\frac{1}{2} + \eps$, implying that
  one of the functions $D(X,0)$, $D(X,1)$ or their negations must be
  as good a predictor for $f(X)$ as $A$ is.
\end{proof}

Since the complexity of the predictor is about the same as that of the
distinguisher $D$, and by assumption $f$ cannot be computed by small
circuits, we conclude that the outcome of the generator must be
indistinguishable from uniform by small circuits. Nisan and Wigderson
generalized this idea to obtain generators that output a long sequence
of bits that is indistinguishable from having a uniform distribution. In
order to obtain more than one pseudorandom bit from the random seed,
they evaluate the hard function $f$ on carefully chosen subsequences
of the seed (for this to work, the input length of $f$ is assumed to
be substantially smaller than the seed length $d$).

An important observation in Trevisan's work is that Nisan-Wigderson's
pseudorandom generator is a \emph{black-box} construction. Namely, the
generator merely computes the hard function $f$ at suitably chosen
points without caring much about how this computation is
implemented. Similarly, the analysis uses the distinguisher $D$ as a
black-box. If $f$ is computable in polynomial time, then so is the
generator (assuming that it outputs polynomially many bits), and if
$f$ is hard against small circuits, the class of circuits of about the
same size must be fooled by the generator.

\newConstruction{Trevisan's extractor $E\colon \zo^n \times \zo^d \to
    \zo^m$.}{A random sample $X \sim \cX$, where $\cX$ is a
    distribution on $\zo^n$ with min-entropy at least $k$, and a
    uniformly distributed random seed $Z \sim \U_d$ of length $d$.
    Moreover, the extractor assumes a $(\frac{1}{2} - \delta, \ell)$
    list-decodable binary code $\C$ of length $N$ (a power of two) and
    size $2^n$, and a combinatorial design $\mathcal{S} := \{ S_1,
    \ldots, S_m \}$, where
    \begin{itemize}
    \item For all $i \in [m]$, $S_i \subseteq [d]$, $|S_i| = \log_2
      N$, and
    \item For all $1 \leq i < j \leq m$, $|S_i \cap S_j| \leq r$.
    \end{itemize}}%
    {A binary string $E(X,Z)$ of length $m$.}%
    {Denote the encoding of $X$ under $\C$ by
    $C(X)$.  For each $i \in [m]$, the subsequence of $Z$ picked by
    the coordinate positions in $S_i$ (denoted by $Z|_i$) is a string
    of length $\log_2 N$ and can be regarded as an integer in
    $[N]$. Let $C_i(X)$ denote the bit at the $(Z|_i)$th position of
    the encoding $C(X)$. Then,
    \[
    E(X,Z) := (C_1(X), \ldots, C_m(X)).
    \]}{constr:Tre}

How can we obtain an extractor from Nisan-Wigderson's construction?
Recall that the output distribution of an extractor must be
indistinguishable from uniform by \emph{all} circuits, and not only
small ones. Adapting Nisan-Wigderson's generator for this requirement
means that we will need a function $f$ that is hard for all circuits,
something which is obviously impossible.  However, this problem can be
resolved if we take \emph{many} hard functions instead of one, and
enforce the predictor to simultaneously predict all functions with a
reasonable bias.  More precisely, statistical indistinguishability can
be obtained if the function $f$ is sampled from a random distribution,
and that is exactly how Trevisan's extractor uses the supplied weak
source. In particular, the extractor regards the sequence obtained
from the weak source as the truth table of a randomly chosen function,
and then applies Nisan-Wigderson's construction relative to that
function.

The exact description of the extractor is given in
Construction~\ref{constr:Tre}. The extractor assumes the existence of
a suitable list-decodable code (see Appendix~\ref{app:coding} for the
terminology) as well as a \emph{combinatorial design}. Intuitively, a
combinatorial design is a collection of subsets of a universe such
that their pairwise intersections are small. We will study designs
more closely in Chapter~\ref{chap:testing}. In order to obtain a
polynomial-time computable extractor, we need an efficient
construction of the underlying list-decodable code and combinatorial
design.

An analysis of Trevisan's construction is given by the following
theorem, which is based on the original analysis of \cite{ref:Tre}.

\begin{thm}
  Trevisan's extractor (as described in Construction~\ref{constr:Tre})
  is a strong $(k, \eps)$-extractor provided that $\eps \geq 2m
  \delta$ and $k > d+m2^{r+1}+\log(\ell/\eps)+3$.
\end{thm}

\begin{proof}
  In light of Proposition~\ref{prop:convex}, it suffices to show the
  claim when $\cX$ is a flat distribution.  Suppose for the sake of
  contradiction that the distribution of $(Z,E(X,Z))$ is not
  $\eps$-close to uniform.  Without loss of generality, and using
  Proposition~\ref{prop:disting}, this means that there is a
  distinguisher $D\colon \zo^{m} \to \zo$ such that
  \begin{equation} \label{eqn:TreDisA} \Pr_{X,Z}[D(Z,E(X,Z)) = 1] -
    \Pr_{Z,U \sim \U_m}[D(Z,U) = 1] > \eps,
  \end{equation}
  where $U=(U_1, \ldots, U_m)$ is a sequence of uniform and
  independent random bits.  Let $X' \subseteq \supp(\cX)$ denote the
  set of inputs on the support of $\cX$ that satisfy
  \begin{equation} \label{eqn:TreDisB} \Pr_{Z}[D(Z,E(x,Z)) = 1] -
    \Pr_{Z,U}[D(Z,U) = 1] > \frac{\eps}{2},
  \end{equation}
  Observe that the size of $X'$ must be at least $\frac{\eps}{2}
  |\supp(\cX)| = \eps 2^{k-1}$, since otherwise \eqref{eqn:TreDisA}
  cannot be satisfied. In the sequel, fix any $x \in X'$.

  For $i=0, \ldots, m$, define a \emph{hybrid sequence} $H_i$ as the
  random variable $H_i := (Z, C_1(x), \ldots, C_i(x), U_{i+1}, \ldots,
  U_m)$. Thus, $H_0$ is a uniformly random bit sequence and $H_m$ has
  the same distribution as $(Z,E(x,Z))$. For $i \in [m]$, define
  \[
  \delta_i := \Pr[D(H_i) = 1] - \Pr[D(H_{i-1})=1],
  \]
  where the probability is taken over the randomness of $Z$ and $U$.
  Now we can rewrite \eqref{eqn:TreDisB} as
  \[
  \Pr[D(H_m) = 1] - \Pr[D(H_0)=1] > \frac{\eps}{2},
  \]
  or equivalently,
  \[
  \sum_{i=1}^m \delta_i > \frac{\eps}{2}.
  \]
  Therefore, for some $i \in [m]$, we must have $\delta_i > \eps/(2m)
  =: \eps'$. Fix such an $i$, and recall that we have
  \begin{multline} \label{eqn:TreDisC} \Pr[D(Z,C_1(x), \ldots, C_i(x),
    U_{i+1}, \ldots, U_m) = 1] - \\ \Pr[D(Z,C_1(x), \ldots,
    C_{i-1}(x), U_{i}, \ldots, U_m) = 1] > \eps'.
  \end{multline}
  Now observe that there is a fixing $U_{i+1} = u_{i+1}, \ldots, U_m =
  u_m$ of the random bits $U_{i+1}, \ldots, U_m$ that preserves the
  above bias. In a similar way as we defined the subsequence $Z|_i$,
  denote by $Z|_{\bar{i}}$ the subsequence of $Z$ obtained by removing
  the coordinate positions of $Z$ picked by $S_i$. Now we note that
  $C_i(x)$ depends only on $x$ and $Z|_i$ and is in particular
  independent of $Z|_{\bar{i}}$. Furthermore, one can fix
  $Z|_{\bar{i}}$ (namely, the portion of the random seed outside
  $S_i$) such that the bias in \eqref{eqn:TreDisC} is preserved. In
  other words, there is a string $z' \in \zo^{d-|S_i|}$ such that
  \begin{multline*}
    \Pr[D(Z,C_1(x), \ldots, C_i(x), u_{i+1}, \ldots, u_m) = 1 \mid (Z|_{\bar{i}})=z'] - \\
    \Pr[D(Z,C_1(x), \ldots, C_{i-1}(x), U_{i}, u_{i+1}, \ldots, u_m) =
    1 \mid (Z|_{\bar{i}})=z'] > \eps',
  \end{multline*}
  where the randomness is now only over $U_i$ and $Z|_i$, and all
  other random variables are fixed to their appropriate values.  Now,
  Proposition~\ref{prop:predictor} can be used to show that, under the
  above fixings, there is a fixed choice of bits $a_0, a_1 \in \F_2$
  such that $D$ can be transformed into a predictor for $C_i(x)$;
  namely, so that
  \[
  \Pr_Z[D(Z,C_1(x), \ldots, C_{i-1}(x), a_0, u_{i+1}, \ldots, u_m)+a_1
  = C_i(x) \mid (Z|_{\bar{i}})=z'] > \frac{1}{2} + \eps'.
  \]
  Since $Z|_i$ is a uniformly distributed random variable, the above
  probability can be interpreted in coding-theoretic ways as follows:
  By running through all the $N$ possibilities of $Z|_i$, the
  predictor constructed from $D$ can correctly recover the encoding
  $C(x)$ at more than $\frac{1}{2} + \eps'$ fraction of the
  positions. Therefore, the distinguisher $D$ can be transformed into
  a word $w \in \F_2^N$ that has an agreement above $\frac{1}{2} +
  \frac{\eps}{2m}$ with $C(x)$.

  Now a crucial observation is that the word $w$ can be obtained from
  $D$ without any knowledge of $x$, as long a correct ``advice''
  string consisting of the appropriate fixings of $i, u_{i+1}, \ldots,
  u_m, a_0, a_1, z'$, and the truth tables of $C_1(x), \ldots,
  C_{i-1}(x)$ as functions of $Z|_i$ are available. Here is where the
  small intersection property of the design $\mathcal{S}$ comes to
  play: Each $C_j(x)$ (when $j \neq i$) depends on at most $r$ of the
  bits in $Z|_i$, and therefore, $C_j(x)$ as a function of $Z|_i$ can
  be fully described by its evaluation on at most $2^r$ points (that
  can be much smaller than $2^{|S_i|} = N$).  This means that the
  number of possibilities for the advice string is at most
  \[
  m \cdot 2^m \cdot 4 \cdot 2^{d-\log N} \cdot 2^{m2^r} = \frac{m}{N}
  \cdot 2^{d+m(2^r+1)+2} \leq 2^{d+m(2^{r+1})+2} =: T.
  \]
  Therefore, regardless of the choice of $x \in X'$, there are words
  $w_1, \ldots, w_T \in \F_2^N$ (one for each possibility of the
  advice string) such that at least one (corresponding to the
  ``correct'' advice) has an agreement better than $\frac{1}{2} +
  \eps'$ with $C(x)$. This, in turn, implies that there is a set $X''
  \subseteq X'$ of size at least $|X'|/T \geq \eps 2^{k-1}/T$ and a
  fixed $j \in [T]$ such that, for every $x \in X''$, the codeword
  $C(x)$ has an agreement better than $\frac{1}{2} + \eps'$ with
  $w_j$. As long as $\delta \leq \eps'$, the number of such codewords
  can be at most $\ell$ (by the list-decodability of $\C$), and we
  will reach to the desired contradiction (completing the proof) if
  the list size $\ell$ is small enough; specifically, if
  \[
  \ell < \frac{\eps 2^{k-1}}{T},
  \]
  which holds by the assumption of the theorem.
\end{proof}

By an appropriate choice of the underlying combinatorial design
$\mathcal{S}$ and the list-decodable code $\C$ (namely, concatenation
of the Reed-Solomon code and the Hadamard code as described in
Section~\ref{sec:concat}), Trevisan \cite{ref:Tre} obtained a strong
extractor with output length $k^{1-\alpha}$, for any fixed constant
$\alpha > 0$, and seed length $d = O(\log^2(n/\eps)/\log_k)$. In a
subsequent work, Raz, Reingold and Vadhan observed that a weaker
notion of combinatorial designs suffice for this construction to work.
Using this idea and a careful choice of the list-decodable code $\C$,
they managed to improve Trevisan's extractor so that it extracts
almost the entire source entropy.  Specifically, their imrpovement can
be summarized as follows.

\begin{thm} \cite{ref:RRV} \label{thm:Tre} For every $n, k, m \in \N$,
  $(m \leq k \leq n)$ and $\eps > 0$, there is an explicit strong $(k,
  \eps)$-extractor $\tre\colon \zo^n \times \zo^d \to \zo^m$ with $d =
  O(\log^2 (n/\eps) \cdot \log (1/\alpha))$, where $\alpha := k/(m-1)
  - 1$ must be less than $1/2$. \qed
\end{thm}

Observe that, as long as the list-decodable code $\C$ is linear,
Trevisan's extractor (as well as its improvement above) becomes linear
as well, meaning that it can be described as a linear function of the
weak source for every fixed choice of the seed. We will make crucial
use of this observation at several points in the thesis.

\subsection{Guruswami-Umans-Vadhan's Condenser}

One of the important constructions of lossless condensers that we will
use in this thesis is the coding-theoretic construction of Guruswami,
Umans and Vadhan \cite{ref:GUV09}. In this section, we discuss the
construction (Construction~\ref{constr:GUV}) and its analysis
(Theorem~\ref{thm:GUVcond}).

\newConstruction{Guruswami-Umans-Vadhan's Condenser $C\colon \F_q^n \times
    \F_q \to \F_q^m$.}%
    {A random sample $X \sim \cX$, where $\cX$ is a
    distribution on $\F_q^n$ with min-entropy at least $k$, and a
    uniformly distributed random seed $Z \sim \U_{\F_q}$ over $\F_q$.}%
    {A vector $C(X,Z)$ of length $\ell$ over $\F_q$.}%
    {Take any irreducible univariate
    polynomial $g$ of degree $n$ over $\F_q$, and interpret the input
    $X$ as the coefficient vector of a random univariate polynomial
    $F$ of degree $n-1$ over $\F_q$. Then, for an integer parameter
    $h$, the output is given by
    \[
    C(X,Z) := (F(Z), F_1(Z), \ldots, F_{\ell-1}(Z)),
    \]
    where we have used the shorthand $F_i := F^{h^i} \mod g$.}%
    {constr:GUV}

We remark that this construction is inspired by a variation of
Reed-Solomon codes due to Parvaresh and Vardy
\cite{ref:PV05}. Specifically, for a given $x \in \F_q^n$, arranging
the outcomes of the condenser $C(x, z)$ for all possibilities of the
seed $z \in \F_q$ results in the encoding of the input $x$ using a
Parvaresh-Vardy code.  Moreover, Parvaresh-Vardy codes are equipped
with an efficient list-decoding algorithm that is implicit in the
analysis of the condenser. The main technical part of the analysis is
given by the following theorem.

\begin{thm} \cite{ref:GUV09} \label{thm:GUVcondGeneral} The mapping defined in
  Construction~\ref{constr:GUV} is a strong $(k, \eps)$ lossless
  condenser with error $\eps := (n-1)(h-1)\ell/q$, provided that $\ell \geq
  k/\log h$ (thus, under the above conditions the mapping becomes a
  strong $(\leq k, \eps)$-condenser as well).
\end{thm}

\begin{proof}
  Without loss of generality (using Proposition~\ref{prop:convex}),
  assume that $\cX$ is uniformly distributed on a subset of $\F_q^n$
  of size $K := 2^k$.  Let $D := q - (n-1)(h-1)\ell$. Define the random
  variable
  \[
  Y := (Z, F(Z), F_1(Z), \ldots, F_{\ell-1}(Z)),
  \]
  and denote by $T \subseteq \F_q^{\ell+1}$ the set that supports the
  distribution of $Y$; i.e., the set of vectors in $\F_q^{\ell+1}$ for
  which $Y$ has a nonzero probability of being assigned to.  Our goal
  is to show that $|T| \geq DK$.  Combined with the second part of
  Proposition~\ref{prop:flatmap}, this will prove the theorem, since
  we will know that the distribution of $(Z, C(X,Z))$ has a support of
  size at least $(1-\eps) q2^k$.

  Assume, for the sake of contradiction, that $|T| < DK$. Then the set
  of points in $T$ can be interpolated by a nonzero multivaraite
  low-degree polynomial of the form
  \[
  Q(z, z_1, \ldots, z_\ell) = \sum_{i=0}^{D-1} z^i Q'_i(z_1, \ldots,
  z_\ell),
  \]
  where each monomial $z_1^{j_1} \cdots z_\ell^{j_\ell}$ in every $Q'_i$ has
  weighted degree $j_1 + h j_2 + h^2 j_3 + \cdots + h^{\ell-1} j_\ell$ at
  most $K-1 < h^\ell$ and individual degrees less than $h$ (this
  condition can be assured by taking $j_1, \ldots, j_\ell$ to be the
  integer representation of an integer between $0$ and $K-1$).  Note
  that $Q$ can be described by its $DK$ coefficients, and each point
  on $T$ specifies a linear constraint on their choice. Since the
  number of constraints is less than the number of unknowns, we know
  that a nonzero polynomial $Q$ vanishes on the set $T$. Fix a nonzero
  choice of $Q$ that has the lowest degree in the first variable
  $z$. This assures that if we write down $Q$ as
  \[
  Q(z, z_1, \ldots, z_\ell) = \sum_{j=(j_1, \ldots, j_{\ell})} Q_j(z)
  z_1^{j_1}\cdots z_\ell^{j_\ell},
  \]
  the polynomials $Q_j(z)$ do not have common irreducible factors
  (otherwise we could divide by the common factor and contradict
  minimality of the degree). In particular at least one of the $Q_j$'s
  must be nonzero modulo the irreducible polynomial $g$.

  Now consider the set $S$ of univariate polynomials of degree less
  than $n$ chosen so that
  \[
  f \in S \Leftrightarrow (\forall z \in \F_q)\colon (z, f(z), f_1(z),
  \ldots, f_{\ell-1}(z)) \in T,
  \]
  where, similarly as before, we have used the shorthand $f_i$ for
  $(f^{h^{i}}\ \mathrm{mod}\ g)$.  Note that, if we regard
  $\supp(\cX)$ as a set of low-degree univariate polynomials, by
  construction of the condenser this set must be contained in
  $S$. Therefore, to reach the desired contradiction, it suffices to
  show that $|S| < K$.

  Let $f$ be any polynomial in $S$. By the definition of $S$, the
  univariate polynomial $Q(z, f(z), f_1(z), \ldots, f_{\ell-1}(z))$ must
  have $q$ zeros (namely, all the elements of $\F_q$).  But the total
  degree of this polynomial is at most $D-1+(n-1)(h-1)\ell=q-1$, and
  thus, the polynomial must be identically zero, and in particular,
  identically zero modulo $g$.  Thus, we have the polynomial identity
  \[
  Q(z, f(z), f^2(z), \ldots, f^{h^{\ell-1}}(z)) \equiv 0 \mod g(z),
  \]
  and by expanding the identity, that
  \begin{equation*}
    \sum_{j=(j_1, \ldots, j_{\ell})} (Q_j(z)\ \mathrm{mod}\ g(z)) \cdot (f(z))^{j_1} (f^{h}(z))^{j_2} \cdots (f^{h^{\ell-1}}(z))^{j_\ell} \equiv 0,
  \end{equation*}
  which simplifies to the identity
  \begin{equation} \label{eqn:polyIden} \sum_{j=(j_1, \ldots, j_{\ell})}
    (Q_j(z)\ \mathrm{mod}\ g(z)) \cdot (f(z))^{{j_1}+j_2 h+\cdots+j_\ell
      h^{\ell-1}} \equiv 0.
  \end{equation}

  Consider the degree $n$ field extension $\F = \F_q[z]/g(z)$ of
  $\F_q$, that is isomorphic to the set of $\F_q$-polynomials of
  degree smaller than $n$.  Under this notation, for every $j$ let
  $\alpha_j \in \F$ to be the extension field element corresponding to
  the $\F_q$-polynomial $(Q_j(z)\ \mathrm{mod}\ g(z))$. Recall that,
  by our choice of $Q$, at least one of the $\alpha_j$'s is nonzero,
  and \eqref{eqn:polyIden} implies that the nonzero univariate
  $\F$-polynomial
  \[
  \sum_{j=(j_1, \ldots, j_{\ell})} \alpha_j z^{{j_1}+j_2 h+\cdots+j_\ell
    h^{\ell-1}}
  \]
  has $f$, regarded as an element of $\F$, as one of its zeros.  The
  degree of this polynomial is less than $K$ and thus it can have less
  than $K$ zeros. Thus we conclude that $|S| < K$ and get the desired
  contradiction.
\end{proof}

By a careful choice of the parameters $h$ and $q$ in the above
construction (roughly, $h \approx (2nk/\eps)^{1/\alpha}$ and $q
\approx h^{1+\alpha}$ for arbitrary constant $\alpha > 0$ and error
$\eps$), Guruswami et al.\ derived the following corollary of the
above theorem:

\begin{thm} \cite{ref:GUV09} \label{thm:GUVcond} For all constants
  $\alpha \in (0,1)$ and every $k \leq n \in \N$, $\eps > 0$ there is
  an explicit strong $(k, \eps)$ lossless condenser with seed length
  $d=(1+1/\alpha) \log (nk/\eps) + O(1)$ and output length
  $m=d+(1+\alpha)k$. \qed
\end{thm}

Using a straightforward observation, we slightly strengthen this result
and show that in fact the parameters can be set up in such a way that
the resulting lossless condenser becomes linear. Linearity of the condenser
is a property that is particularly useful for the results obtained
in Chapter~\ref{chap:capacity}.

\begin{coro} \label{coro:GUVcondLinear}
 Let $p$ be a fixed prime power and $\alpha > 0$ be an arbitrary constant. Then, for
 parameters $n \in \N$, $k \leq n \log p$, and $\eps > 0$, there is an explicit strong $(\leq k, \eps)$ lossless condenser
 $f\colon \F_{p}^n \times \zo^d \to \F_{p}^m$ with seed length
  $d \leq (1+1/\alpha) (\log (nk/\eps) + O(1))$ and output length satisfying\footnote{All unsubscripted logarithms are to the base $2$.}
  $m \log p \leq d+(1+\alpha)k$. Moreover, $f$ is a linear function (over $\F_{p}$) for every fixed choice of the seed.
\end{coro}

\begin{proof}
  We set up the parameters of the condenser $C$ given by Construction~\ref{constr:GUV} and
  apply Theorem~\ref{thm:GUVcondGeneral}. The range of the parameters is mostly similar to what
  chosen in the original result of Guruswami et al.~\cite{ref:GUV09}.

  Letting $h_0 := (2p^2 nk/\eps)^{1/ \alpha}$, we take $h$ to be an integer power of $p$ in range
  $[h_0, p h_0]$. Also, let $\ell := \lceil k/\log h \rceil$ so that the condition $\ell \geq k/\log h$ required by Theorem~\ref{thm:GUVcondGeneral}
  is satisfied. Finally, let $q_0 := nh \ell/\eps$ and choose the field size $q$ to be an integer power of $p$ in range
  $[q_0, p q_0]$.

  We choose the input length of the condenser $C$ to be equal to $n$. Note that
  $C$ is defined over $\F_q$, and we need a condenser over $\F_p$. Since $q$ is a fixed parameter, we can
  ensure that $q \geq p$ (for large enough $n$), so that $\F_p$ is a subfield of $\F_q$.  For $x \in \F_p^n$ and $z \in \zo^d$,
  let $y := C(x,y) \in \F_q^\ell$, where $x$ is regarded as a vector over the extension $\F_q$ of $\F_p$.
  We define the output of the condenser $f(x,z)$ to be the vector $y$ regarded as a vector of length
  $\ell \log_p q$ over $\F_p$ (by expanding each element of $\F_q$ as a vector of length $\log_p q$ over $\F_p$).
  It can be clearly seen that $f$ is a strong $(\leq k, \eps)$-condenser if $C$ is.

  By Theorem~\ref{thm:GUVcondGeneral}, $C$ is a strong lossless condenser with error upper bounded by
  \[
   \frac{(n-1)(h-1)\ell }{q} \leq \frac{nh \ell}{q_0} = \eps.
  \]
  It remains to analyze the seed length and the output length of the condenser.
  For the output length of the condenser, we have
  \[
   m \log p = \ell \log q \leq (1+k/\log h) \log q \leq d + k (\log q)/(\log h),
  \]
  where the last inequality is due to the fact that we have $d = \lceil \log q \rceil$.
  Thus in order to show the desired upper bound on the output length, it suffices
  to show that $\log q \leq (1+\alpha) \log h_0$. We have
  \[
   \log q \leq \log (p q_0) = \log(pnh \ell/\eps) \leq \log h_0 + \log(p^2 n\ell/\eps)
  \]
  and our task is reduced to showing that $p^2 n \ell/\eps \leq h_0^{\alpha} = 2p^2nk/\eps$. But
  this bound is obviously valid by the choice of $\ell \leq 1+ k/\log h$.

  The seed length is $d = \lceil \log q \rceil$ for which we have
  \begin{eqnarray*}
   d &\leq& \log q + 1 \leq \log q_0 + O(1) \\
   &\leq& \log (nh_0 \ell / \eps) + O(1) \\
   &\leq& \log (nh_0 k / \eps) + O(1) \\
   &\leq& \log(nk/\eps) + \frac{1}{\alpha} \log (2p^2 nk/\eps) \\
   &\leq& \big(1+ \frac{1}{\alpha}\big)(\log(nk/\eps) + O(1))
  \end{eqnarray*}
as desired.

Since $\F_q$ has a fixed characteristic, an efficient deterministic algorithm for representation and
manipulation of the field elements is available \cite{ref:Shoup} which implies that the condenser
is polynomial-time computable and is thus explicit.

Moreover, since $h$ is taken as an integer power of $p$ and $\F_q$ is an extension of $\F_p$,
for any choice of polynomials $F, F', G \in \F_q[X]$, subfield elements $a, b \in \F_p$, and integer $i \geq 0$, we have
\[
 (a F + b F')^{h^i} \equiv a F^{h^i} + b F'^{h^i} \mod G,
\]
meaning that raising a polynomial to power $h^i$ is an $\F_p$-linear operation.
Therefore, the mapping $C$ that defines the condenser (Construction~\ref{constr:GUV})
is $\F_p$-linear for every fixed seed. This in turn implies that the final condenser $f$
is linear, as claimed.
\end{proof}

Guruswami et al.\ used the lossless condenser above as an
intermediate building block for construction of an extractor that is
optimal up to constant factors and extracts almost the entire source
entropy. Namely, they proved the following result that will be useful
for us in later chapters.

\begin{thm} \cite{ref:GUV09} \label{thm:extr} For all positive
  integers $n \geq k$ and all $\eps > 0$, there is an explicit strong
  $(k, \eps)$-extractor $\extr\colon \zo^n \times \zo^d \to \F_2^m$
  with $m = k - 2 \log(1/\eps) - O(1)$ and $d = \log n + O(\log k
  \cdot \log (k/\eps))$. \qed
\end{thm}

\musicBoxExtractor


\Chapter{The Wiretap Channel Problem}
\epigraphhead[70]{\epigraph{\textsl{``Music is meaningless noise unless it touches a receiving mind.''}}{\textit{--- Paul Hindemith}}}
\label{chap:wiretap}

Suppose that Alice wants to send a message to Bob through a
communication channel, and that the message is \emph{partially}
observable by an intruder.  This scenario arises in various practical
situations.  For instance, in a packet network, the sequence
transmitted by Alice through the channel can be fragmented into small
packets at the source and/or along the way and different packets might
be routed through different paths in the network in which an intruder
may have compromised some of the intermediate routers.  An example
that is similar in spirit is furnished by transmission of a piece of
information from multiple senders to one receiver, across different
delivery media, such as satellite, wireless, and/or wired networks.
Due to limited resources, a potential intruder may be able to observe
only a fraction of the lines of transmission, and hence only partially
observe the message.
As another example, one can consider secure storage of data on a
distributed medium that is physically accessible in parts by an
intruder, or a sensitive file on a hard drive that is erased from the
file system but is only partially overwritten with new or random
information, and hence, is partially exposed to a malicious party.

An obvious approach to solve this problem is to use a secret key to
encrypt the information at the source. However, almost all practical
cryptographic techniques are shown to be secure only under unproven
hardness assumptions and the assumption that the intruder possesses
bounded computational power.  This might be undesirable in certain
situations. Moreover, the key agreement problem has its own
challenges.

In the problem that we consider in this chapter, we assume the
intruder to be \emph{information theoretically} limited, and our goal
will be to employ this limitation and construct a protocol that
provides unconditional, infor\-mation-theoretic security, even in the
presence of a computationally unbounded adversary.

The problem described above was first formalized by Wyner
\cite{ref:Wyner1} and subsequently by Ozarow and Wyner
\cite{ref:Wyner2} as an infor\-mation-theoretic problem.  In its most
basic setting, this problem is known as the \emph{wiretap II
  problem}\index{wiretap~II problem} (the description given here
follows from \cite{ref:Wyner2}):

\begin{quote}
Consider a communication system with a {source} which outputs a
sequence $X = (X_1,\ldots, X_m)$ in $\zo^m$ uniformly at random.  A
randomized algorithm, called the encoder,
maps the output of the source to a binary string $Y \in \zo^n$.  The
output of the encoder is then sent through a noiseless channel (called
\emph{the direct channel}\index{direct channel}) and is eventually
delivered to a decoder\footnote{Ozarow and Wyner also consider the
  case in which the decoder errs with negligible probability, but we
  are going to consider only error-free decoders.} $D$ which maps $Y$
back to $X$.
Along the way, an intruder arbitrarily picks a subset $S \subseteq
[n] := \{1,\ldots,n\}$ of size $t \leq n$, and is allowed to observe\footnote{ For a
  vector $x = (x_1, x_2, \ldots, x_n)$ and a subset $S \subseteq [n]$,
  we denote by $x|_S$ the vector of length $|S|$ that is obtained from
  $x$ by removing all the coordinates $x_i$, $i \notin S$.  } $Z :=
Y|_S$ (through a so-called \emph{wiretap channel}\index{wiretap
  channel}), i.e., $Y$ on the coordinate positions corresponding to
the set $S$.
The goal is to make sure that the intruder learns as little as
possible about $X$, regardless of the choice of $S$.
\end{quote}

\begin{figure}
\centerline{\includegraphics[width=\textwidth]{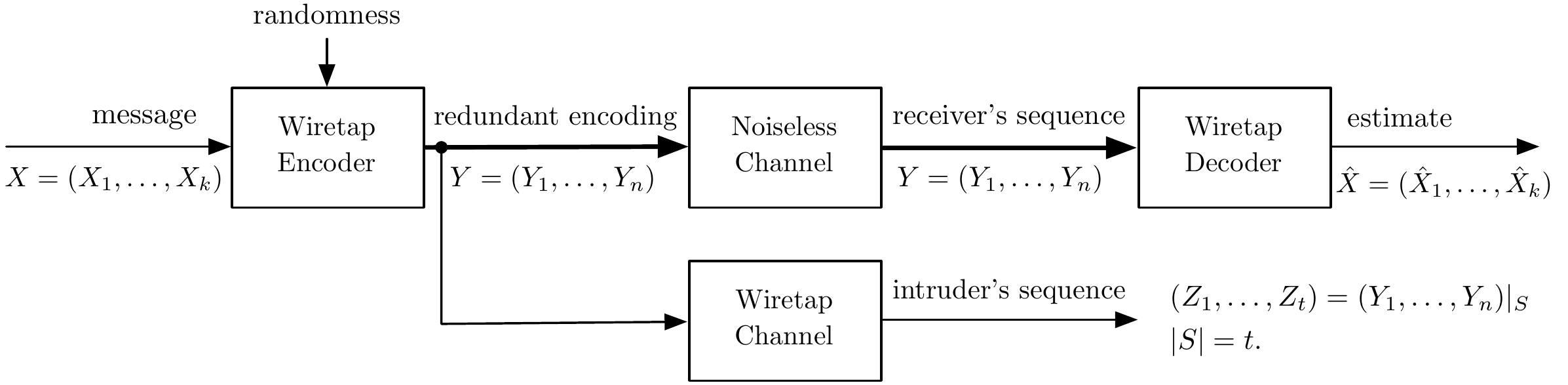}}
\caption[The Wiretap~II Problem]{The Wiretap~II Problem.}
\label{fig:wiretapII}
\end{figure}

The system defined above is illustrated in Figure~\ref{fig:wiretapII}.
The security of the system is defined by the following conditional
entropy, known as ``equivocation'':\index{equivocation} \[ \Delta :=
\min_{S\colon |S| = t}
H(X|Z).\] 
When $\Delta = H(X) = m$, the intruder obtains no information about the
transmitted message and we have \emph{perfect privacy} in the
system. Moreover, when $\Delta \to m$ as $m \to \infty$, we call the
system \emph{asymptotically perfectly} private.\index{perfect
  privacy}\index{asymptotically perfect privacy} These two cases
correspond to what is known in the literature as ``strong
secrecy''\index{strong secrecy}\index{weak secrecy}.  A weaker
requirement (known as ``weak secrecy'') would be to have $m - \Delta =
o(m)$.

\begin{rem} \label{rem:randomMessage}
  The assumption that $X$ is sampled from a uniformly random
  source should not be confused with the fact that Alice is
  transmitting \emph{one particular} message to Bob that is fixed and
  known to her before the transmission. In this case, the randomness
  of $X$ in the model captures the \emph{a priori} uncertainty about
  $X$ for the \emph{outside world}, and in particular the intruder,
  but not the transmitter.

  As an intuitive example, suppose that a random key is agreed upon
  between Alice and a trusted third party, and now Alice wishes to
  securely send her particular key to Bob over a wiretapped
  channel. Or, assume that Alice wishes to send an audio stream to Bob
  that is encoded and compressed using a conventional audio encoding
  method.
%

  Furthermore, the particular choice of the distribution on $X$ as a
  uniformly random sequence will cause no loss of generality. If the
  distribution of $X$ is publicly known to be non-uniform, the
  transmitter can use a suitable source-coding scheme to compress the
  source to its entropy prior to the transmission, and ensure that
  from the intruder's point of view, $X$ is uniformly distributed.
  On the other hand, it is also easy to see that if a protocol
  achieves perfect privacy under uniform message distribution, it
  achieves perfect privacy under any other distribution as well.
\end{rem}

\section{The Formal Model}
\label{subsec:model}
The model that we will be considering in this chapter is motivated by
the original wiretap channel problem but is more stringent in terms of
its security requirements. In particular, instead of using Shannon
entropy as a measure of uncertainty, we will rely on statistical
indistinguishability which is a stronger measure that is more widely
used in cryptography.

\begin{defn} \index{wiretap protocol}
  \label{def:wiretap}
  Let $\Sigma$ be a set of size $q$, $m$ and $n$ be positive integers,
  and $\veps, \gamma>0$.  A $(t,\veps,\gamma)_q$-resilient wiretap
  protocol of block length $n$ and message length $m$ is a pair of
  functions $E\colon\Sigma^m\times\zo^r\to\Sigma^n$ (the encoder) and
  $D\colon\Sigma^n\to\Sigma^m$ (the decoder) that are computable in
  time polynomial in $m$, such that
  \begin{enumerate}
  \item[(a)] (Decodability) For all $x\in\Sigma^m$ and all $z\in\zo^r$
    we have $D(E(x,z))=x$,
  \item[(b)] (Resiliency) Let $\rv{X} \sim \U_{\Sigma^m}$, $\rv{R}\sim
    \U_{r}$, and $\rv{Y}=E(X,R)$. For a set $S \subseteq [n]$ and $w
    \in \Sigma^{|S|}$, let
    $\cX_{S,w}$ denote the distribution of $X$ conditioned on the
    event $Y|_S = w$. Define the set of \emph{bad observations} as \[
    B_S := \{ w \in \Sigma^{|S|} \mid \dist( \cX_{S,w}, \U_{\Sigma^m}
    ) > \eps \},\] where $\dist(\cdot, \cdot)$ denotes the statistical
    distance between two distributions.  Then we require that for
    every $S \subseteq [n]$ of size at most $t$, $\Pr[Y|_S \in B_S]
    \leq \gamma$, where the probability is over the randomness of
    $X$ and $R$.
  \end{enumerate}

  The \emph{encoding} of a vector $x\in\Sigma^k$ is accomplished by
  choosing a vector $Z\in\zo^r$ uniformly at random, and calculating
  $E(x,Z)$.  The quantities $R:=m/n$, $\eps$, and $\gamma$ are called
  the \emph{rate}, the \emph{error}, and the \emph{leakage}
  \index{wiretap protocol!rate}\index{wiretap
    protocol!error}\index{wiretap protocol!leakage}\index{wiretap
    protocol!resilience} of the protocol, respectively.  Moreover, we
  call $\delta:=t/n$ the \emph{(relative) resilience} of the protocol.
\end{defn}

The decodability condition ensures that the functions $E$ and $D$ are
a \emph{matching} encoder/decoder pair, while the resiliency
conditions ensures that the intruder learns almost nothing about the
message from his observation.

In our definition, the imperfection of the protocol is captured by the
two parameters $\eps$ and $\gamma$.  When $\eps = \gamma = 0$, the
above definition coincides with the original wiretap channel problem
for the case of perfect privacy.

When $\gamma = 0$, we will have a \emph{worst-case} guarantee, namely,
that the intruder's views of the message before and after his
observation are statistically close, \emph{regardless} of the outcome
of the observation.

The protocol remains interesting even when $\gamma$ is positive but
sufficiently small.  When $\gamma > 0$, a particular observation might
potentially reveal to the intruder a lot of information about the
message.  However, a negligible $\gamma$ will ensure that such a bad
event (or \emph{leakage}) happens only with negligible probability.

All the constructions that we will study in this chapter achieve zero
leakage (i.e., $\gamma = 0$), except for the general result in
Section~\ref{subsec:arbitrary} for which a nonzero leakage is
inevitable.

The significance of zero-leakage protocols is that they assure
\emph{adaptive} resiliency in the \emph{weak} sense introduced in
\cite{ref:DSS01} for exposure-resilient functions: if the intruder is
given the encoded sequence as an oracle that he can adaptively query
at up to $t$ coordinates (that is, the choice of each query may depend
on the outcome of the previous queries), and is afterwards presented
with a challenge which is either the original message or an
independent uniformly chosen random string, he will not be able to
distinguish between the two cases.

In general, it is straightforward to verify that our model can be used
to solve the original wiretap II problem, with $\Delta \geq
m(1-\eps-\gamma)$:

\begin{lem}
  \label{app:model}
  Suppose that $(E, D)$ is an encoder/decoder pair as in
  Definition~\ref{def:wiretap}.  Then using $E$ and $D$ in the
  wiretap~II problem attains an equivocation \[ \Delta \geq
  m(1-\eps-\gamma).\]
\end{lem}

\begin{proof}
  Let $W := Y|_S$ be the intruder's observation, and denote by $W'$
  the set of \emph{good} observations, namely,
  \[ W' := \{ w \in \Sigma^t \colon \dist( \cX_{S,w}, \U_{\Sigma^m} )
  \leq \eps \}. \] Denote by $H(\cdot)$ the Shannon entropy in $d$-ary
  symbols. Then we will have
  \begin{eqnarray*}
    H(X|W) &=& \sum_{w \in \Sigma^t} \Pr(W = w) H(X|W = w) \\
    &\geq& \sum_{w \in W'} \Pr(W = w) H(X|W = w) \\
    &\stackrel{\mathrm{(a)}}{\geq}& \sum_{w \in W'} \Pr(W = w) (1-\eps) m
    \stackrel{\mathrm{(b)}}{\geq} (1-\gamma)(1-\eps) m \geq (1-\gamma-\eps) m.
  \end{eqnarray*}
  The inequality $\mathrm{(a)}$ follows from the definition of $W'$
  combined with Proposition~\ref{prop:Shannon} in the appendix, and
  $\mathrm{(b)}$ by the definition of leakage parameter.
\end{proof}

Hence, we will achieve asymptotically perfect privacy when
$\eps+\gamma = o(1/m)$.  For all the protocols that we present in this
chapter this quantity will be superpolynomially small; that is,
smaller than $1/m^c$ for every positive constant $c$ (provided that
$m$ is large enough).

\section{Review of the Related Notions in Cryptography}

There are several interrelated notions in the literature on
Cryptography and Theoretical Computer Science that are also closely
related to our definition of the wiretap protocol
(Definition~\ref{def:wiretap}).  These are \emph{resilient functions
  (RF)} and \emph{almost perfect resilient functions (APRF)},
\emph{exposure-resilient functions (ERF)}, and \emph{all-or-nothing
  transforms (AONT)} (cf.\
\cites{ref:tResilient,ref:tResilient2,ref:Riv97,ref:Stinson,ref:FT00,ref:CDH,ref:KJS}
and \cite{ref:Dodis} for a comprehensive account of several important
results in this area).

The notion of resilient functions was introduced in \cite{ref:BBR85}
(and also \cite{ref:Vaz87} as the \emph{bit-extraction problem}). A
deterministic polynomial-time computable function $f\colon \zo^n \to
\zo^m$ is called $t$-resilient \index{resilient function (RF)} if
whenever any $t$ bits of the its input are arbitrarily chosen by an
adversary and the rest of the bits are chosen uniformly at random,
then the output distribution of the function is (close to) uniform.
APRF is a stronger variation where the criterion for uniformity of the
output distribution is defined with respect to the $\ell_\infty$
(i.e., point-wise distance of distributions) rather than
$\ell_1$. This stronger requirement allows for an ``adaptive
security'' of APRFs.


ERFs\index{exposure-resilient function (ERF)}, introduced in
\cite{ref:CDH}, are similar to resilient functions except that the
entire input is chosen uniformly at random, and the view of the
adversary from the output remains (close to) uniform even after
observing any $t$ input bits of his choice.

ERFs and resilient functions are known to be useful in a scenario
similar to the wiretap channel problem where the two parties aim to
agree on \emph{any} random string, for example a session key (Alice
generates $x$ uniformly at random which she sends to Bob, and then
they agree on the string $f(x)$).  Here no control on the content of
the message is required, and the only goal is that at the end of the
protocol the two parties agree on any random string that is uniform
even conditioned on the observations of the intruder.  Hence,
Definition~\ref{def:wiretap} of a wiretap protocol is more stringent
than that of resilient functions, since it requires the existence and
efficient computability of the encoding function $E$ that provides a
control over the content of the
message. 

Another closely related notion is that of all-or-no\-thing transforms,
which was suggested in \cite{ref:Riv97} for protection of block
ciphers.  \index{all-or-nothing transform (AONT)} A randomized
poly\-nomial-time computable function $f\colon \zo^m \to \zo^n$, $(m
\leq n)$, is called a (statistical, non-adaptive, and secret-only)
$t$-AONT with error $\eps$ if it is efficiently invertible and for
every $S \subseteq [n]$ such that $|S| \le t$, and all $x_1, x_2 \in
\zo^m$ we have that the two distributions $f(x_1)|_S$ and $f(x_2)|_S$
are $\eps$-close.

An AONT with $\eps=0$ is called perfect.  It is easy to see that
perfectly private wiretap protocols are equivalent to perfect adaptive
AONTs. It was shown in \cite{ref:DSS01} that such functions can not
exist (with positive, constant rate) when the adversary is allowed to
observe more than half of the encoded bits. A similar result was
obtained in \cite{ref:tResilient} for the case of perfect linear RFs.

As pointed out in \cite{ref:DSS01}, AONTs can be used in the original
scenario of Ozarow and Wyner's wiretap channel problem. However, the
best known constructions of AONTs can achieve rate-resilience
trade-offs that are far from the infor\-mation-theoretic optimum (see
Figure \ref{fig:region}).

While an AONT requires indistinguishability of intruder's view for
every fixed pair $(x_1, x_2)$ of messages, the relaxed notion of
\emph{average-case} AONT \index{all-or-nothing transform
  (AONT)!average case} requires the \emph{expected} distance of
$f(x_1)|_S$ and $f(x_2)|_S$ to be at most $\eps$ for a uniform random
message pair. Hence, for a negligible $\eps$, the distance will be
negligible for all but a negligible fraction of message pairs.
Up to a loss in parameters, wiretap protocols are equivalent to
average case AONTs:

\begin{lem} \label{lem:avgAONT} Let $(E, D)$ be an encoding/decoding
  pair for a $(t, \eps, \gamma)_2$-resilient wiretap protocol. Then
  $E$ is an average-case $t$-AONT with error at most $2(\eps+\gamma)$.

  Conversely, an average-case $t$-AONT with error $\eta^2$ can be used
  as a $(t, \eta, \eta)$-resilient wiretap encoder.
\end{lem}

\begin{proof}
  Consider a $(t, \eps, \gamma)_2$-resilient wiretap protocol as in
  Definition~\ref{def:wiretap}, and accordingly, let the random
  variable $Y=E(X,R)$ denote the encoding of $X$ with a random seed
  $R$. For a set $S \subseteq [n]$ of size at most $t$, denote by $W
  := Y|_S$ the intruder's observation.

  The resiliency condition implies that, the set of bad observations
  $B_S$ has a probability mass of at most $\gamma$ and hence, the
  expected distance $\dist(X|W, X)$ taken over the distribution of $W$
  is at most $\eps+\gamma$.  Now we can apply
  Proposition~\ref{prop:duality} to the jointly distributed pair of
  random variables $(W, X)$, and conclude that the expected distance
  $\dist(W|X, W)$ over the distribution of $X$ (which is uniform) is
  at most $\eps+\gamma$. This implies that the encoder is an
  average-case $t$-AONT with error at most $2(\eps+\gamma)$.

  Conversely, the same argument combined with Markov's bound shows
  that an average-case $t$-AONT with error $\eta^2$ can be seen as
  $(t, \eta, \eta)$-resilient wiretap protocol.
\end{proof}

Note that the converse direction does not guarantee zero leakage,
and hence, zero leakage wiretap protocols are in general stronger than
average-case AONTs.  An average-case to worst-case reduction for AONTs
was shown in \cite{ref:CDH} which, combined with the above lemma, can
be used to show that any wiretap protocol can be used to construct an
AONT (at the cost of a rate loss).

A simple \emph{universal} transformation was proposed in
\cite{ref:CDH} to obtain an AONT from any ERF, by one-time padding the
message with a random string obtained from the ERF.  In particular,
given an ERF $f\colon \zo^n \to \zo^m$, the AONT $g\colon \zo^m \to
\zo^{m+n}$ is defined as $g(x) := (r,x + f(r))$, where $r \in \zo^n$
is chosen uniformly at random.  Hence, the ERF is used to one-time pad
the message with a random secret string.

This construction can also yield a wiretap protocol with zero leakage.
However, it has the drawback of significantly weakening the
rate-resilience trade-off. Namely, even if an information
theoretically optimal ERF is used in this reduction, the resulting
wiretap protocol will only achieve half the optimal rate (see
Figure~\ref{fig:region}).  This is because the one-time padding
strategy necessarily requires a random seed that is at least as long
as the message itself, even if the intruder is restricted to observe
only a small fraction of the transmitted sequence. Hence the rate of
the resulting AONT cannot exceed $1/2$, and it is not clear how to
improve this universal transformation to obtain a worst-case AONT
using a shorter seed.

The main focus of this chapter is on asymptotic trade-offs between the
rate $R$ and the resilience $\delta$ of an asymptotically perfectly
private wiretap protocol.  For applications in cryptography, e.g., the
context of ERFs or AONTs, it is typically assumed that the adversary
learns \emph{all} but a small number of the bits in the encoded
sequence, and the incurred \emph{blow-up} in the encoding is not as
crucially important, as long as it remains within a reasonable
range. On the other hand, as in this chapter we are motivated by the
wiretap channel problem which is a communication problem, optimizing
the transmission \emph{rate} will be the most important concern for
us. We will focus on the case where the fraction $\delta$ of the
symbols observed by the intruder is an arbitrary constant below $1$,
which is the most interesting range in our context. However, some of
our constructions work for sub-constant $1-\delta$ as well.

Following \cite{ref:Wyner2}, it is easy to see that, for resilience
$\delta$, an infor\-mation-theoretic bound $R\le 1-\delta + o(1)$ must
hold.  Lower bounds for $R$ in terms of $\delta$ have been studied by
a number of researchers.

For the case of perfect privacy (where the equivocation $\Delta$ is equal to the
message length $m$), Ozarow and
\index{Ozarow-Wyner's protocol} Wyner~\cite{ref:Wyner2} give a
construction of a wiretap protocol using linear error-correcting
codes, and show that the existence of an $[n,k,d]_q$-code implies the
existence of a perfectly private, $(d-1,0,0)_q$-resilient wiretap
protocol with message length $k$ and block length $n$ (thus, rate $k/n$).

As a result, the so-called Gilbert-Varshamov bound on the
rate-distance trade-offs of linear codes (see
Chapter~\ref{chap:gv}) 
implies that, asymptotically, $R\ge 1-h_q(\delta)$, where $h_q$ is the
$q$-ary entropy function defined as
\[
h_q(x) := x \log_q(q-1) - x \log_q(x) - (1-x) \log_q(1-x).
\]
If $q\ge 49$ is a square, the bound can be further improved to $R\ge
1-\delta-1/(\sqrt{q}-1)$ using Goppa's algebraic-geometric
codes~\cites{gopp:81,tsvz:82}.  In these protocols, the encoder can be
seen as an adaptively secure, perfect AONTs and the decoder is an
adaptive perfect RF.

Moving away from perfect to asymptotically perfect privacy, it was
shown in~\cite{ref:KJS} that for any $\gamma>0$ there exist binary
asymptotically perfectly private wiretap protocols with $R\ge
1-2\delta-\gamma$ and exponentially small error\footnote{Actually,
  what is proved in this paper is the existence of $t$-resilient
  functions which correspond to decoders in our wiretap setting;
  however, it can be shown that these functions also possess efficient
  encoders, so that it is possible to construct wiretap protocols from
  them.}.  This bound strictly improves the coding-theoretic bound of
Ozarow and Wyner for the binary alphabet.

\begin{figure}[t]
  \centering
  \begin{tikzpicture}[xscale=4, yscale=4]
    \foreach \x in {1, ..., 9} { \draw[thin, dotted] (0.1*\x,
      1-0.1*\x) -- (0.1*\x, 0); \draw[thin, dotted] (0, 0.1*\x) -- (1
      - 0.1*\x, 0.1*\x); } \draw[->] (-0.02,0) -- (1.1,0) node[right]
    {$\delta$}; \draw[->] (0,-0.02) -- (0,1.1) node[above]
    {$\mathrm{rate}$}; \draw(0, 1) node[left, scale=1]{${1}$};
    \draw(1, 0) node[below, scale=1]{$1$};
    \draw[very thick] (0,1) -- (1,0); 
    \draw[very thick, densely dotted] (0,0.5) -- (1,0); 
    \draw (0,1) -- (0.5,0) node[below]{$\frac{1}{2}$}; \draw[dashed]
    plot[raw gnuplot, id=GV] function{ h(x) = (- x * log(x) - (1-x) *
      log(1-x) ) / log(2); gv(x) = (x == 0) ? 1 : ( (x==0.5) ? 0: 1 -
      h(x)); plot [x=0:0.5] gv(x); };

    \draw[very thick] (0.7, 1.05) node[left]{{\small $(1)$}} -- (0.9,
    1.05); \draw (0.7, 0.9) node[left]{{\small $(2)$}} -- (0.9, 0.9);
    \draw[dashed](0.7, 0.75) node[left]{{\small $(3)$}} -- (0.9,
    0.75); \draw[very thick, densely dotted](0.7, 0.6)
    node[left]{{\small $(4)$}} -- (0.9, 0.6);
  \end{tikzpicture} \hspace{1cm}
  \begin{tikzpicture}[xscale=4, yscale=4]
    \foreach \x in {1, ..., 9} { \draw[thin, dotted] (0.1*\x,
      1-0.1*\x) -- (0.1*\x, 0); \draw[thin, dotted] (0, 0.1*\x) -- (1
      - 0.1*\x, 0.1*\x); } \draw (0.5, 0) node[below]{$\frac{1}{2}$};
    \draw[->] (-0.02,0) -- (1.1,0) node[right] {$\delta$}; \draw[->]
    (0,-0.02) -- (0,1.1) node[above] {$\mathrm{rate}$}; \draw(0, 1)
    node[left, scale=1]{${1}$}; \draw(1, 0) node[below, scale=1]{$1$};
    \draw[very thick] (0,1) -- (1,0); 
    \draw plot[raw gnuplot, id=plot1] function{ d =
      64; 
      d2 = 62; 
      lambda = (2 * sqrt(d2 - 1) + (d - d2)) / d; alpha = -log(lambda
      ** 2)/log(d); max(a, b) = (a > b) ? a : b; walk(x) = max( alpha
      * (1-x), 1 - x / alpha ); plot [x=0:1] walk(x); }; \draw[dashed]
    plot[raw gnuplot, id=plot2] function{ d = 64; 
      h(q, x) = (x * log(q-1) - x * log(x) - (1-x) * log(1-x) ) /
      log(q); max(a, b) = (a > b) ? a : b; gv(x) = (x == 0) ? 1 : (
      (x==1) ? 0: 1 - h(d, x)); AG(x) = 1 - x - (1 / (sqrt(d) - 1));
      plot [x=0:1] max(gv(x), AG(x)); };
    \draw[very thick] (0.7, 1.05) node[left]{{\small $(1)$}} -- (0.9,
    1.05); \draw (0.7, 0.9) node[left]{{\small $(5)$}} -- (0.9, 0.9);
    \draw[dashed](0.7, 0.75) node[left]{{\small $(6)$}} -- (0.9,
    0.75);
  \end{tikzpicture}
  \caption[Comparison of the rate vs.\ resilience trade-offs achieved
  by various wiretap protocols]{ A comparison of the rate vs.\
    resilience trade-offs achieved by the wiretap protocols for the
    binary alphabet (left) and larger alphabets (right, in this
    example of size $64$).  $(1)$ Infor\-mation-theoretic bound,
    attained by Theorem~\ref{coro:wiretap}; $(2)$ The bound approached
    by \cite{ref:KJS}; $(3)$ Protocol based on best non-explicit binary
    linear codes \cites{gilbert,varshamov}; $(4)$ AONT construction of
    \cite{ref:CDH}, assuming that the underlying ERF is optimal; $(5)$
    Random walk protocol of Corollary~\ref{coro:WTWalk}; $(6)$
    Protocol based on the best known explicit \cite{tsvz:82} and
    non-explicit \cites{gilbert,varshamov} linear codes.  }

  \label{fig:region}
\end{figure}
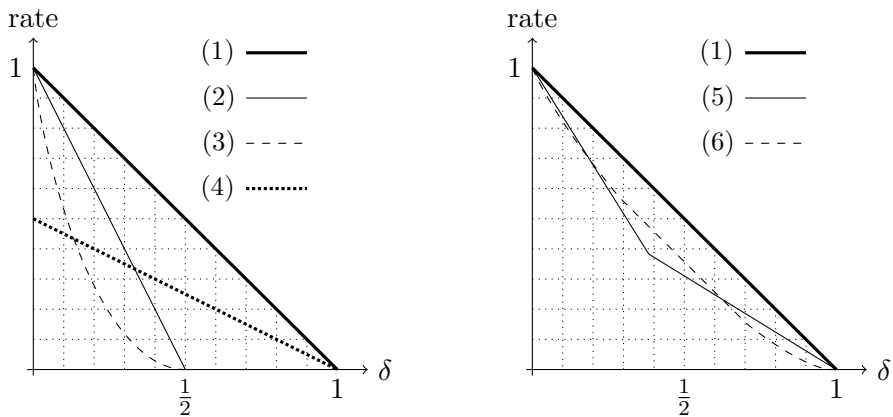

\section{Symbol-Fixing and Affine Extractors}
\label{sec:symb-fix}

Two central notions for our constructions of wiretap protocols in this
chapter are symbol-fixing and affine extractors. In this section, we
introduce these notions, and study some basic constructions.

\begin{defn}
  A $d$-ary \emph{symbol-fixing source} is an imperfect source of
  random symbols from an alphabet of size $d$, that may fix some
  bounded number of the symbols to unknown values. More precisely, an
  $(n, k)_d$ \emph{symbol-fixing} source is the distribution of a
  random variable $\rv{X}=(\rv{X}_1, \rv{X}_2, \ldots, \rv{X}_n) \in
  \Sigma^n$, for some set $\Sigma$ of size $d$, in which at least $k$
  of the coordinates (chosen arbitrarily) are uniformly and
  independently distributed on $\Sigma$ and the rest take
  deterministic values.

  When $d=2$, we will have a \emph{binary} symbol-fixing
  source\index{source!bit-fixing}, or simply a \emph{bit-fixing}
  source\index{source!bit-fixing}. In this case $\Sigma = \{0, 1\}$,
  and the subscript $d$ is dropped from the notation.
\end{defn}

The min-entropy of a $(n, k)_d$ symbol-fixing source is $k \log_2 d$
bits.  For a $d$-ary source with $d \neq 2$, it is more convenient to
talk about the $d$-ary entropy of the source, which is $k$ (in $d$-ary
symbols).

Affine sources are natural generalizations of symbol-fixing sources
when the alphabet size is a prime power.

\begin{defn}
  For a prime power $q$, an $(n,k)_q$ \emph{affine}
  source\index{source!affine} is a distribution on $\F_q^n$ that is
  uniformly supported on an affine translation of some $k$-dimensional
  subspace of $\F_q^n$.
\end{defn}

It is easy to see that the $q$-ary min-entropy of a $k$-dimensional
affine source is $k$.  Due to the restricted structure of
symbol-fixing and affine sources, it is possible to construct seedless
extractors for such sources:

\begin{defn}
  Let $\Sigma$ be a finite alphabet of size $d > 1$.  A function
  $f\colon \Sigma^n \to \Sigma^m$ is a (seedless) $(k,
  \eps)$-extractor\index{extractor!affine}\index{extractor!symbol-fixing}
  for symbol-fixing (resp., affine) sources on $\Sigma^n$ if for every
  $(n, k)_d$ symbol-fixing (resp., affine) source $\cX$, the
  distribution $E(\cX)$ is $\eps$-close to the uniform distribution
  $\U_{\Sigma^m}$.  The extractor is called explicit if it is
  deterministic and polynomial-time computable.
\end{defn}

We will shortly see simple constructions of zero-error, symbol-fixing
and affine extractors using linear functions arising from good
error-correcting codes. These extractors achieve the lowest possible
error, but however are unable to extract the entire source
entropy. Moreover, the affine extractor only works for a
``restricted'' class of affine sources. For unrestricted affine
sources, there are by now various constructions of extractors in the
literature. Here we review some notable examples that are most useful
for the construction of wiretap protocols that we will discuss in this
chapter.

Over large fields, the following affine extractor due to Gabizon and
Raz extract almost the entire source entropy:

\begin{thm}\cite{ref:affine}
  There is a constant $q_0$ such that for any prime power field
  size $q$ and integers $n,k$ such that $q > \max\{q_0, n^{20}\}$,
  there is an explicit affine $(k, \eps)$-extractor $f\colon \F_q^n
  \to \F_q^{k-1}$, where $\eps < q^{-1/21}$. \qed
\end{thm}

In this construction, the field size has to be polynomially large in
$n$. When the field size is small (in particular, constant), the task
becomes much more challenging. The most challenging case thus
corresponds to the binary field $\F_2$, for which an explicit affine
extractor was obtained, when the input entropy is a constant fraction
of the input length, by Bourgain:

\newcommand{\aext}{\mathsf{AExt}} \newcommand{\ibou}{\mathsf{AExt}}

\begin{thm} \cite{ref:Bourgain} \label{thm:Bourgain} For every
  constant $0 < \delta \leq 1$, there is an explicit affine extractor
  $\bou\colon \F_2^n \to \F_2^m$ for min-entropy $\delta n$ with
  output length $m = \Omega(n)$ and error at most
  $2^{-\Omega(m)}$. \qed
\end{thm}
Bourgain's construction was recently simplified, improved, and
extended to work for arbitrary prime fields by Yehudayoff
\cite{ref:Yeh09}.

An ``intermediate'' trade-off is recently obtained by DeVos and
Gabizon \cite{ref:DG10}, albeit with a short output length. This
explicit construction extracts one unbiased bit from any $(n,k)_q$
affine source provided that, for $d := 5n/k$, we have $q > 2d^2$ and
the characteristic of the field is larger than $d$.

\subsection{Symbol-Fixing Extractors from Linear
  Codes} \label{sec:symfix-code}

The simple theorem below states that linear error-correcting codes can
be used to obtain symbol-fixing extractors with zero error.

\begin{thm} \label{thm:symfixExtr} Let $\C$ be an $[n,\Tk,d]_q$ code
  over $\F_q$ and $G$ be a $\Tk \times n$ generator matrix of $\C$.
  Then, the function $E\colon \F_q^n \to \F_q^\Tk$ defined
  as\footnote{We typically consider vectors be represented in row
    form, and use the transpose operator
    ($x^\top$)\index{notation!$x^\top$} to represent column vectors.}
  $E(x) := G x^\top$ is an $(n-d+1, 0)$-extractor for symbol-fixing
  sources over $\F_q$.

  Conversely, if a linear function $E\colon \F_q^n \to \F_q^\Tk$ is an
  $(n-d+1, 0)$-extractor for symbol-fixing sources over $\F_q$, it
  corresponds to a generator matrix of an $[n,\Tk,d]_q$ code.
\end{thm}

\begin{proof}
  Let $\cX$ be a symbol-fixing source with a set $S \subseteq [n]$ of
  fixed coordinates, where\footnote{If the set of fixed symbols if of
    size smaller than $d-1$, the argument still goes through by taking
    $S$ as an arbitrary set of size $d-1$ containing all the fixed
    coordinates.}  $|S| = d-1$, and define $\bar{S} := [n] \setminus
  S$.  Observe that, by the Singleton bound, we must have $|\bar{S}| =
  n-d+1 \geq \Tk$.

  The submatrix of $G$ obtained by removing the columns picked by $S$
  must have rank $\Tk$. Since otherwise, the left kernel of this
  submatrix would be nonzero, meaning that $\C$ has a nonzero codeword
  that consists of entirely zeros at the $d-1$ positions picked by
  $S$, contradicting the assumption that the minimum distance of $\C$
  is $d$.  Therefore, the distribution $E(\cX)$ is supported on a
  $\Tk$-dimensional affine space on $\F_q^\Tk$, meaning that this
  distribution is uniform.

  The converse is straightforward by following the same argument.
\end{proof}

If the field size is large enough; e.g., $q \geq n$, then one can pick
$\C$ in the above theorem to be an MDS code (in particular, a
Reed-Solomon code) to obtain a $(k, 0)$-extractor for all
symbol-fixing sources of entropy $k$ with optimal output length
$k$. However, for a fixed $q$, negative results on the rate-distance
trade-offs of codes (e.g., Hamming, MRRW, and Plotkin bounds) assert
that this construction of extractors must inevitably lose some
fraction of the entropy of the source. Moreover, the construction
would at best be able to extract some constant fraction of the source
entropy only if the entropy of the source (in $q$-ary symbols) is
above $n/q$.

\subsection{Restricted Affine Extractors from Rank-Metric
  Codes} \label{sec:affext-code}

In Section~\ref{sec:Apps}, we will see that affine extractors
can be used to construct wiretap schemes for models that are
more general than the original Wiretap~II problem, e.g., when
the direct channel is noisy. For these applications, the extractor
needs to additionally have a nice structure that is in particular
offered by linear functions.

An obvious observation is that a nontrivial affine extractor cannot be
a linear function.  Indeed, a linear function $f(x) := \langle
\alpha, x \rangle + \beta$, where $\alpha, \beta, x \in \F_q^n$, is
constant on the $(n-1)$-dimensional orthogonal subspace of $\alpha$,
and thus, fails to be an extractor for even $(n-1)$-dimensional affine
spaces. However, in this section we will see that linear affine extractors
can be constructed if the affine source is known to be described by
a set of linear constraints whose coefficients lie on a small \emph{sub-field}
of the underlying field.
Such restricted extractors turn out to be sufficient for some of the
applications that we will consider.

Let $Q$ be a prime power.  Same as linear codes, an affine subspace on
$\F_Q^n$ can be represented by a \emph{generator matrix}, or
\emph{parity-check matrix} and a constant shift. That is, a
$k$-dimensional affine subspace $A \subseteq \F_Q^n$ can be described
as the image of a linear mapping
\[
A := \{ x G + \beta \colon x \in \F_Q^k \},
\]
where $G$ is a $k \times n$ \emph{generator matrix} of rank $k$ over
$\F_Q$, and $\beta \in \F_Q^n$ is a fixed vector. Alternatively, $A$
can be expressed as the translated null-space of a linear mapping
\[
A := \{ x + \beta \in \F_Q^n\colon H x^\top = 0 \},
\]
for an $(n-k) \times n$ \emph{parity check matrix} of rank $n-k$ over
$\F_Q$.

Observe that a symbol-fixing source over $\F_q$ with $q$-ary
min-entropy $k$ can be seen as a $k$-dimensional affine source with a
generator matrix of the form $[I \mid \mathbf{0}]\cdot P$, where $I$
is the $k \times k$ identity matrix, $\mathbf{0}$ denotes the $k
\times (n-k)$ all-zeros matrix, and $P$ is a permutation matrix.
Recall that from Theorem~\ref{thm:symfixExtr} we know that for this
restricted type of affine sources linear extractors exist. In this
section we generalize this idea.

Suppose that $Q = q^m$ for a prime power $q$ so that $\F_Q$ can be
regarded as a degree $m$ extension of $\F_q$ (and isomorphic to
$\F_{q^m}$).  Let $A$ be an affine source over $\F_Q^n$. We will call
the affine source
\emph{$\F_q$-restricted}\index{source!affine!$\F_q$-restricted} if its
support can be represented by a generator matrix (or equivalently, a
parity check matrix) over $\F_q$.

In this section we introduce an affine extractor that is $\F_Q$-linear
and, assuming that $m$ is sufficiently large, extracts from
$\F_q$-restricted affine sources.  The construction of the extractor
is similar to Theorem~\ref{thm:symfixExtr}, except that instead of an
error-correcting code defined over the \emph{Hamming metric}, we will
use \emph{rank-metric} codes.

Consider the function $\rdist\colon \F_q^{m \times n} \times \F_q^{m
  \times n} \to \Z$, where $\F_q^{m \times n}$ denotes the set of $m
\times n$ matrices over $\F_q$, defined as $\rdist(A, B) :=
\rk_q(A-B)$, where $\rk_q$ is the matrix rank over $\F_q$.  It is
straightforward to see that $\rdist$ is a metric\index{rank metric}.

The usual notion of error-correcting codes defined under the Hamming
metric can be naturally extended to the rank metric.  In particular, a
\emph{rank-metric}\index{code!rank-metric} code $\C$ can be defined as
a set of $m \times n$ matrices (known as codewords), whose minimum
distance is the minimum rank distance between pairs of codewords.

For $Q := q^m$, there is a natural correspondence between $m\times n$
matrices over $\F_q$ and vectors of length $n$ over $\F_Q$. Consider
an isomorphism $\varphi\colon \F_Q \to \F_q^m$ between $\F_Q$ and
$\F_q^m$ which maps elements of $\F_Q$ to column vectors of length $m$
over $\F_q$. Then one can define a mapping $\Phi\colon \F_Q^n \to
\F_q^{m \times n}$ defined as
\[
\Phi(x_1, \ldots, x_n) := [\varphi(x_1) \mid \cdots \mid \varphi(x_n)]
\]
to put the elements of $\F_Q^n$ in one-to-one correspondence with $m
\times n$ matrices over $\F_q$.

A particular class of rank-metric codes are linear ones. Suppose that
$\C$ is a linear $[n,\Tk,\tilde{d}]_Q$ code over $\F_Q$. Then, using $\Phi(\cdot)$, $\C$ can
be regarded as a rank-metric code of dimension $\Tk$ over $\F_q^{m
  \times n}$. In symbols, we will denote such a linear $\Tk$-dimensional
rank-metric code as an $[[ n,\Tk,d ]]_{q^m}$ code, where $d$ is the
minimum rank-distance of the code. The rank-distance of a linear
rank-metric code turns out to be equal to the minimum rank of its
nonzero codewords and obviously, one must have $d \leq \tilde{d}$. However,
the Hamming distance of $\C$ might turn out to be much larger than its
rank distance when regarded as a rank-metric code. In particular, $d
\leq m$, and thus, $d$ must be strictly smaller than $\tilde{d}$ when the
degree $m$ of the field extension is less than $\tilde{d}$.

A counterpart of the Singleton bound in the rank-metric states that,
for any $[[n,\Tk,d]]_{q^m}$ code, one must have $d \leq
n-\Tk+1$. Rank-metric codes that attain equality exist and are called
\emph{maximum rank distance (MRD)}\index{code!MRD} codes. A class of
linear rank-metric codes known as \emph{Gabidulin
  codes}\index{code!Gabidulin} \cite{ref:Gab85} are MRD and can be
thought of as the counterpart of Reed-Solomon codes in the rank
metric. In particular, the codewords of a Gabidulin code, seen as
vectors over the extension field, are evaluation vectors of
bounded-degree \emph{linearized} polynomials rather than arbitrary
polynomials as in the case of Reed-Solomon codes.  These codes are
defined for any choice of $n, \Tk, q, m$ as long as $m \geq n$ and $\Tk
\leq n$.

The following is an extension of Theorem~\ref{thm:symfixExtr} to
restricted affine sources.

\begin{thm} \label{thm:resAffExtr} Let $\C$ be an $[[n,\Tk,d]]_{q^m}$
  code defined from a code over $\F_Q$ (where $Q := q^m$) with a
  generator matrix $G \in \F_Q^{\Tk \times n}$.  Then the function
  $E\colon \F_Q^n \to \F_Q^\Tk$ defined as $E(x) := G x^\top$ is an
  $(n-d+1, 0)$-extractor for $\F_q$-restricted affine sources over
  $\F_Q$.

  Conversely, if a linear function $E\colon \F_Q^n \to \F_Q^\Tk$ is an
  $(n-d+1, 0)$-extractor for all $\F_q$-restricted affine sources over
  $\F_Q$, it corresponds to a generator matrix of an
  $[[n,\Tk,d]]_{q^m}$ code.
\end{thm}

\begin{proof}
  Consider a restricted affine source $\cX$ uniformly supported on an
  affine subspace of dimension\footnote{ The argument still holds if
    the dimension of $\cX$ is more than $n-d+1$.  } $n-d+1$
  \[ X := \{ x A+\beta\colon x \in \F_Q^{n-d+1} \}, \] where $A \in
  \F_q^{(n-d+1) \times n}$ has rank $n-d+1$, and $\beta \in \F_Q^n$ is
  a fixed translation.  Note that $\Tk \leq n-d+1$ by the Singleton
  bound for rank-metric codes.

  The output of the extractor is thus uniformly supported on the
  affine subspace
  \[
  B := \{ G A^\top x^\top + G \beta^\top\colon x \in \F_Q^{n-d+1} \}
  \subseteq \F_Q^\Tk.
  \]

  Note that $G A^\top \in \F_Q^{\Tk \times (n-d+1)}$.  Our goal is to
  show that the dimension of $B$ is equal to $\Tk$.  Suppose not, then
  we must have $\rk_Q (GA^\top) < \Tk$. In particular, there is a
  nonzero $y \in \F_Q^\Tk$ such that $y G A^\top = 0$.

  Let $Y := \Phi(y G) \in \F_q^{m \times n}$, where $\Phi(\cdot)$ is
  the isomorphism that maps codewords of $\C$ to their matrix form
  over $\F_q$. By the distance of $\C$, we know that $\rk_q(Y) \geq
  d$. Since $m \geq d$, this means that $Y$ has at least $d$ linearly
  independent rows. On the other hand, we know that the matrix $Y
  A^\top \in \F_q^{\Tk \times (n-d+1)}$ is the zero matrix. Therefore,
  $Y$ has $d$ independent rows (each in $\F_q^n$) that are all
  orthogonal to the $n-d+1$ independent rows of $A$. Since $d +
  (n-d+1) > n$, this is a contradiction.

  Therefore, the dimension of $B$ is exactly $\Tk$, meaning that the
  output distribution of the extractor is indeed uniform.  The
  converse is straightforward by following a similar line of argument.
\end{proof}

Thus, in particular, we see that generator matrices of MRD codes can
be used to construct linear extractors for restricted affine sources
that extract the entire source entropy with zero error. This is
possible provided that the field size is large enough compared to the
field size required to describe the generator matrix of the affine
source. Using Gabidulin's rank metric codes, we immediately obtain the following
corollary of Theorem~\ref{thm:resAffExtr}:

\begin{coro} \label{coro:resAffExtrGabidulin}
  Let $q$ be a prime power. Then for every positive integer $n$,
  $k \leq n$, and $Q := q^n$, there is a linear function $f\colon \F_Q^n \to \F_Q^k$
  that is a $(k, 0)$-extractor for $\F_q$-restricted affine sources
  over $\F_Q$. \qed
\end{coro}

It can be shown using similar proofs that if, in Theorems
\ref{thm:symfixExtr}~and~\ref{thm:resAffExtr}, a parity check matrix
of the code is used instead of a generator matrix, the resulting
linear function would become a lossless $(d-1, 0)$-condenser rather
than an extractor. This is in fact part of a more general
``duality'' phenomenon that is discussed in Section~\ref{sec:dualityAffine}.




\section{Inverting Extractors}
\label{sec:InvExt}

In this section we will introduce the notion of \emph{invertible
  extractors} and its connection with wiretap protocols\footnote{
  Another notion of invertible extractors was introduced in
  \cite{ref:Dod05} and used in \cite{ref:DS05} for a different
  application (entropic security) that should not be confused with the
  one we use.  Their notion applies to seeded extractors with long
  seeds that are efficiently invertible bijections for every fixed
  seed. Such extractors can be seen as a single-step walk on highly
  expanding graphs that mix in one step.  This is in a way similar to
  the multiple-step random walk used in the seedless extractor of
  section \ref{sec:walk}, that can be regarded as a single-step walk
  on the expander graph raised to a certain power.  }.  Later we will
use this connection to construct wiretap protocols with good
rate-resilience trade-offs.

\begin{defn} \index{invertible extractor}
  \label{def:inverter}
  Let $\Sigma$ be a finite alphabet and $f$ be a mapping from
  $\Sigma^n$ to $\Sigma^m$.  For $\gamma \geq 0$, a function $A\colon
  \Sigma^m \times \zo^r \to \Sigma^n$ is called a
  $\gamma$-\emph{inverter} for $f$ if the following conditions hold:
  \begin{enumerate}
  \item[(a)] (Inversion) Given $x \in \Sigma^m$ such that $f^{-1}(x)$
    is nonempty, for every $z \in \zo^r$ we have $f(A(x, z)) =x
    $.
  \item[(b)] (Uniformity) $A(\U_{\Sigma^m}, \U_r) \sim_\gamma
    \U_{\Sigma^n}$.
  \end{enumerate}
  A $\gamma$-inverter is called \emph{efficient} if there is a
  randomized algorithm that runs in worst case polynomial time and,
  given $x \in \Sigma^m$ and $z$ as a random seed, computes $A(x,
  z)$. We call a mapping $\gamma$-\emph{invertible} if it has an
  efficient $\gamma$-inverter, and drop the prefix $\gamma$ from the
  notation when it is zero.
\end{defn}

The parameter $r$ in the above definition captures the amount of random
bits that the inverter (seen as a randomized algorithm) needs to receive.
For our applications, no particular care is needed to optimize this
parameter and, as long as $r$ is polynomially bounded in $n$, it is generally
ignored.

\begin{rem}
  If a function $f$ maps the uniform distribution to a distribution
  that is $\eps$-close to uniform (as is the case for all extractors),
  then any randomized mapping that maps its input $x$ to a
  distribution that is $\gamma$-close to the uniform distribution on
  $f^{-1}(x)$ is easily seen to be an $(\eps+\gamma)$-inverter for
  $f$.  In some situations designing such a function might be easier
  than directly following the above definition.
\end{rem}

The idea of random pre-image sampling was proposed in \cite{ref:DSS01}
for construction of adaptive AONTs from APRFs. However, they ignored
the efficiency of the inversion, as their goal was to show the
existence of (not necessarily efficient) infor\-mation-theoretically
optimal adaptive AONTs.  Moreover, the strong notion of APRF and a
perfectly uniform sampler is necessary for their construction of
AONTs.  As wiretap protocols are weaker than (worst-case) AONTs, they
can be constructed from slightly imperfect inverters as shown by the
following lemma.


\begin{lem} \label{lem:protocol} Let $\Sigma$ be an alphabet of size
  $q > 1$ and $f\colon \Sigma^n \to \Sigma^m$ be a
  $(\gamma^2/2)$-invertible $q$-ary $(k, \eps)$ symbol-fixing
  extractor.  Then, $f$ and its inverter can be seen as a
  decoder/encoder pair for an $(n-k,\veps+\gamma,\gamma)_q$-resilient
  wiretap protocol with block length $n$ and message length $m$.
\end{lem}

\begin{proof}
  Let $E$ and $D$ denote the wiretap encoder and decoder,
  respectively.  Hence, $E$ is the $(\gamma^2/2)$-inverter for $f$,
  and $D$ is the extractor $f$ itself.  From the definition of the
  inverter, for every $x \in \Sigma^m$ and every random seed $r$, we
  have $D(E(x, r)) = x$.  Hence it is sufficient to show that the pair
  satisfies the resiliency condition.

  Let the random variable $X$ be uniformly distributed on $\Sigma^m$
  and the seed $R\in\zo^r$ be chosen uniformly at random.  Denote the
  encoding of $X$ by $Y := E(X, R)$. Fix any $S \subseteq [n]$ of size
  at most $n-k$.

  For every $w \in \Sigma^{|S|}$, let $Y_w$ denote the set $\{ y \in
  \Sigma^n\colon (y|_S) = w \}$.
  Note that the sets $Y_w$ partition the space $\Sigma^n$ into
  $|\Sigma|^{|S|}$ disjoint sets.

  Let $\cY$ and $\cY_S$ denote the distribution of $Y$ and $Y|_S$,
  respectively. The inverter guarantees that $\cY$ is
  $(\gamma^2/2)$-close to uniform. Applying
  Proposition~\ref{prop:conditioning}, we get that
  \[
  \sum_{w \in \Sigma^{|S|}} \Pr[(Y|_S) = w] \cdot \dist( (\cY | Y_w),
  \U_{Y_w} ) \leq \gamma^2.
  \]
  The left hand side is the expectation of $\dist( (\cY | Y_w),
  \U_{Y_w} )$.  Denote by $W$ the set of all \emph{bad outcomes} of
  $Y|_S$, i.e.,
  \[W := \{ w \in \Sigma^{|S|} \mid \dist( (\cY | Y_w), \U_{Y_w} ) >
  \gamma \}.\] By Markov's inequality, we conclude that
  \[
  \Pr[(Y|_S) \in W] \leq \gamma.
  \]
  For every $w \in W$, the distribution of $Y$ conditioned on the
  event $(Y|_S) = w$ is $\gamma$-close to a symbol-fixing source with
  $n-|S| \geq k$ random symbols. The fact that $D$ is a symbol-fixing
  extractor for this entropy and Proposition~\ref{prop:closeFunction}
  imply that, for any such $w$, the conditional distribution of
  $D(Y)|(Y|_S = w)$ is $(\gamma+\eps)$-close to uniform.  Hence with
  probability at least $1-\gamma$ the distribution of $X$ conditioned
  on the outcome of $Y|_S$ is $(\gamma+\eps)$-close to uniform. This
  ensures the resiliency of the protocol.
\end{proof}

By combining Lemma~\ref{lem:protocol} and Theorem~\ref{thm:symfixExtr}
using a Reed-Solomon code, we can obtain a perfectly private,
rate-optimal, wiretap protocol for the Wiretap~II problem over large
alphabets (namely, $q \geq n$), and recover the original result of
Ozarow and Wyner\footnote{In fact, Ozarow and Wyner use a parity check
  matrix of an MDS code in their construction, which is indeed a
  generator matrix for the dual code which is itself
  MDS.}~\cite{ref:Wyner2}:

\begin{coro} \label{coro:wiretapii} For every positive integer $n$,
  prime power $q \geq n$, and $\delta \in [0, 1)$,
  there is a $(\delta n, 0, 0)_q$-resilient wiretap protocol with block length $n$ and rate $1 -
  \delta$ that attains perfect privacy. \qed
\end{coro}

\section{A Wiretap Protocol Based on Random Walks} \label{sec:walk}

In this section we describe a wiretap protocol that achieves a rate
$R$ within a constant fraction of the information theoretically
optimal value $1-\delta$ (the constant depending on the alphabet
size).

To achieve our result, we will modify the symbol-fixing extractor of
Kamp and Zuckerman \cite{ref:KZ}, that is based on random walks on
expander graphs, to make it efficiently invertible without affecting
its extraction properties, and then apply Lemma~\ref{lem:protocol}
above to obtain the desired wiretap protocol.

Before we proceed, let us briefly review some basic notions and facts
related to expander graphs. For a detailed review of the theory of
expander graphs, refer to the excellent survey by Hoory, Linial and
Wigderson \cite{HLW06}, and books \cites{ref:MR,ref:MU}.

We will be working with directed regular expander graphs that are
obtained from undirected graphs by replacing each undirected edge with
two directed edges in opposite directions.  Let $G=(V, E)$ be a
$d$-regular graph.  Then a \emph{labeling} of the edges of $G$ is a
function $L\colon V \times [d] \to V$ such that for every $u \in V$
and $t \in [d]$, the edge $(u, L(u, t))$ is in $E$.  The labeling is
\emph{consistent}\index{consistent labeling} if whenever $L(u, t) =
L(v, t)$, then $u = v$.  Note that the natural labeling of a Cayley
graph (cf.\ \cite{HLW06}) is in fact consistent.

A \emph{family} of $d$-regular graphs 
is an infinite set of $d$-regular graphs such that for every $N \in \N$, the
set contains a graph with at least $N$ vertices.  For a parameter $c
\geq 1$, we will call a family $c$-dense if there is an $N_0 \in \N$
such that, for every $N \geq N_0$, the family has a graph with at
least $N$ and at most $cN$ vertices.  We call a family of graphs
\emph{constructible} if all the graphs in the family have a consistent
labeling that is efficiently computable. That is, there is a uniform,
polynomial-time algorithm that, given $N \in \N$ and $i \in [N], j \in
[d]$, outputs the label of the $j$th neighbor of the $i$th vertex,
under a consistent labeling, in the graph in the family that has $N$
vertices (provided that it exists).

Let $A$ denote the normalized adjacency matrix of a $d$-regular graph
$G$ (that is, the adjacency matrix with all the entries divided by $d$).
We denote by $\lambda_G$ the second largest eigenvalue of $A$ in
absolute value. The \emph{spectral gap}\index{spectral gap} of $G$ is
given by $1-\lambda_G$.  Starting from a probability distribution $p$
on the set of vertices, represented as a real vector with coordinates
index by the vertex set, performing a single-step random walk on $G$
leads to the distribution defined by $pA$. The following is a well
known lemma on the convergence of the distributions resulting from
random walks (see \cite{ref:Lovasz} for a proof):

\begin{lem} \label{lem:walk} Let $G=(V,E)$ be a $d$-regular undirected
  graph, and $A$ be its normalized adjacency matrix. Then for any
  probability vector $p$, we have \[ \| pA - \U_V \|_2 \leq \lambda_G
  \| p - \U_V \|_2,\] where $\| \cdot \|_2$ denotes the $\ell_2$
  norm. \qed
\end{lem}

The extractor of Kamp and Zuckerman~\cite{ref:KZ}\index{extractor!Kamp
  and Zuckerman's} starts with a fixed vertex in a large expander
graph and interprets the input as the description of a walk on the
graph.  Then it outputs the label of the vertex reached at the end of
the walk.  Notice that a direct approach to invert this function will
amount to sampling a path of a particular length between a pair of
vertices in the graph, uniformly among all the possibilities, which
might be a difficult problem for good families of expander
graphs\footnote{In fact intractability of the easier problem of
  finding a loop in certain families of expander graphs forms the
  underlying basis for a class of cryptographic hash functions (cf.\
  \cite{ref:CGL}). Even though this easier problem has been solved in
  \cite{ref:TZ08}, uniform sampling of paths seems to be much more
  difficult.}.  We work around this problem by choosing the starting
point of the walk from the input\footnote{ The idea of choosing the
  starting point of the walk from the input sequence has been used
  before in extractor constructions \cite{ref:Zuc07}, but in the
  context of seeded extractors for general sources with high
  entropy.}. The price that we pay by doing so is a slightly
  larger error compared to the original construction of Kamp and
  Zuckerman that is, asymptotically, of little significance.
  In particular we show the following:

\begin{thm} \label{thm:invWalk} Let $G$ be a constructible $d$-regular
  graph with $d^m$ vertices and second largest eigenvalue
  $\lambda_G\ge 1/\sqrt{d}$.  Then there exists an explicit invertible
  $(k,2^{s/2})_d$ symbol-fixing extractor $\kz\colon [d]^n \to [d]^m$,
  such that
  \[
  s := \left\{ \begin{array}{ll}
      m \log d + k \log \lambda_G^2 & \text{if $k \leq n - m$,} \\
      (n-k) \log d + (n-m) \log \lambda_G^2 & \text{if $k > n - m$.} \\
    \end{array} \right.
  \]
\end{thm}
\newcommand{\inv}{\mathsf{Inv}}
\begin{Proof}
  We first describe the extractor and its inverse. Given an input $(v,
  w) \in [d]^m \times [d]^{n-m},$ the function $\kz$ interprets $v$ as
  a vertex of $G$ and $w$ as the description of a walk starting from
  $v$. The output is the index of the vertex reached at the end of the
  walk.  Figure~\ref{fig:graphWalk} depicts the procedure. The
  $4$-regular graph shown in this toy example has $8$ vertices labeled
  with binary sequences of length $3$. Edges of the graph are
  consistently labeled at both endpoints with the set of labels
  $\{1,2,3,4\}$. The input sequence $(0,1,0\mid2,3,4,2,4$) shown below
  the graph describes a walk starting from the vertex $010$ and
  following the path shown by the solid arrows. The output of the
  extractor is the label of the final vertex $011$.

\begin{figure}
\centerline{\includegraphics[width=7cm]{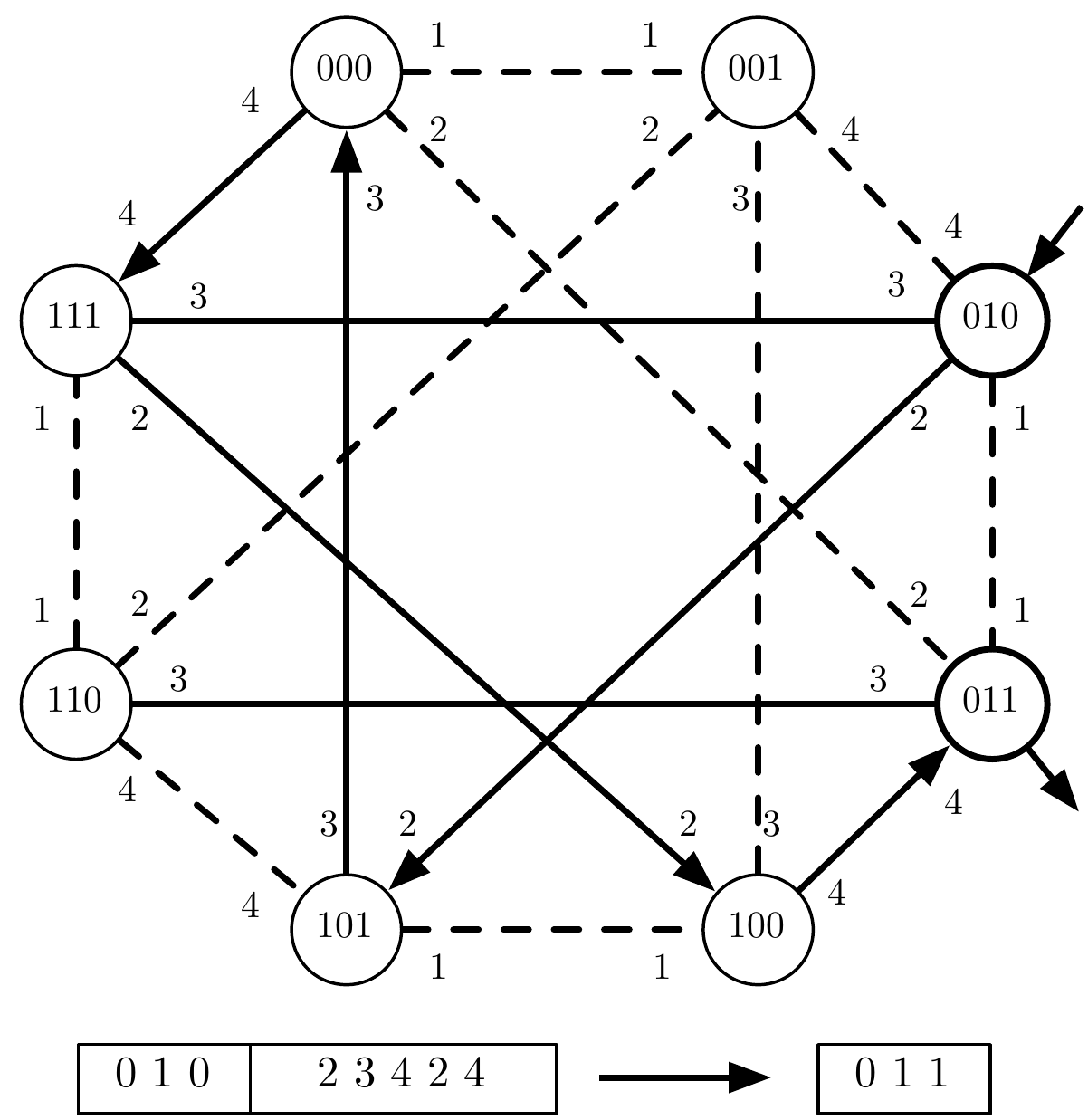}}
\caption[The random-walk symbol-fixing extractor]{The random-walk symbol-fixing extractor.}
\label{fig:graphWalk}
\end{figure}

  The inverter $\inv$ works as follows: Given $x \in [d]^m$,
  $x$ is interpreted as a vertex of $G$. Then $\inv$ picks $W \in
  [d]^{n-m}$ uniformly at random.  Let $V$ be the vertex starting from
  which the walk described by $W$ ends up in $x$. The inverter outputs
  $(V, W)$. It is easy to verify that $\inv$ satisfies the properties
  of a $0$-inverter.

  Now we show that $\kz$ is an extractor with the given parameters. We
  will follow the same line of argument as in the original proof of
  Kamp and Zuckerman.  Let $(x, w) \in [d]^{m} \times [d]^{n-m}$ be a
  vector sampled from an $(n, k)_d$ symbol-fixing source, and let $u
  := \kz(x, w)$.  Recall that $u$ can be seen as the vertex of $G$
  reached at the end of the walk described by $w$ starting from $x$.
  Let $p_i$ denote the probability vector corresponding to the walk
  right after the $i$th step, for $i = 0, \ldots, n-m$, and denote by
  $p$ the uniform probability vector on the vertices of $G$.  Our goal
  is to bound the error $\eps$ of the extractor, which is half the
  $\ell_1$ norm of $p_{n-m} - p$.

  Suppose that $x$ contains $k_1$ random symbols and the remaining
  $k_2 := k - k_1$ random symbols are in $w$. Then $p_0$ has the value
  $d^{-k_1}$ at $d^{k_1}$ of the coordinates and zeros elsewhere,
  hence
  \[
  \| p_0 - p \|_2^2 = d^{k_1} (d^{-k_1} - d^{-m})^2 + (d^m - d^{k_1})
  d^{-2m} = d^{-k_1} - d^{-m} \leq d^{-k_1}.
  \]

  Now for each $i \in [n-m]$, if the $i$th step of the walk
  corresponds to a random symbol in $w$ the $\ell_2$ distance is
  multiplied by $\lambda_G$ by Lemma~\ref{lem:walk}.  Otherwise the
  distance remains the same due to the fact that the labeling of $G$
  is consistent.  Hence we obtain $\| p_{n-m} - p \|_2^2 \leq d^{-k_1}
  \lambda_G^{2 k_2}$.  Translating this into the $\ell_1$ norm by
  using the Cauchy-Schwarz inequality, we obtain $\eps$, namely,
  \begin{equation*} \label{eqn:epsBound} \eps \leq \frac{1}{2}
    d^{(m-k_1)/2} \lambda_G^{k_2} < 2^{((m-k_1) \log d + k_2 \log
      \lambda_G^2)/2}.
  \end{equation*}
  By our assumption, $\lambda_G \geq 1/\sqrt{d}$. Hence, everything
  but $k_1$ and $k_2$ being fixed, the above bound is maximized when
  $k_1$ is minimized.  When $k \leq n-m$, this corresponds to the case
  $k_1 = 0$, and otherwise to the case $k_1 = k-n+m$.  This gives us
  the desired upper bound on $\eps$.
\end{Proof}

Combining this with Lemma~\ref{lem:protocol} and setting up the the
right asymptotic parameters, we obtain our protocol for the wiretap
channel problem.

\begin{coro} \label{coro:WTWalk} Let $\delta \in [0, 1)$ and $\gamma >
  0$ be arbitrary constants, and suppose that there is a constructible
  family
  of $d$-regular expander graphs with spectral gap at least
  $1-\lambda$ that is $c$-dense, for constants $\lambda < 1$ and $c
  \geq 1$.

  Then, for every large enough $n$, there is a $(\delta n,
  2^{-\Omega(n)}, 0)_d$-resilient wiretap protocol with block length
  $n$ and rate \[R = \max\{\alpha (1-\delta), 1 - \delta/\alpha\} -
  \gamma,\] where $\alpha := -\log_d \lambda^2$.
\end{coro}

\begin{Proof} For the case $c = 1$ we use Lemma~\ref{lem:protocol}
  with the extractor $\kz$ of Theorem~\ref{thm:invWalk} and its
  inverse.  Every infinite family of graphs must satisfy $\lambda \geq
  2 \sqrt{d-1}/d$ \cite{ref:nilli}, and in particular we have $\lambda
  \geq 1/\sqrt{d}$, as required by Theorem~\ref{thm:invWalk}.  We
  choose the parameters $k := (1-\delta) n$ and $m := n(\max\{\alpha
  (1-\delta), 1 - \delta/\alpha\} - \gamma)$, which gives $s =
  -\Omega(n)$, and hence, exponentially small error.  The case $c > 1$
  is similar, but involves technicalities for dealing with lack of
  graphs of arbitrary size in the family. We will elaborate on this in
  Appendix~\ref{app:wiretapDetails}.
\end{Proof}

Using explicit constructions of Ramanujan graphs that achieve \[\lambda
\leq 2 \sqrt{d-1}/d\] when $d-1$ is a prime power
\cites{ref:ramanujan1,ref:ramanujan2,ref:pizer}, one can obtain
$\alpha \geq 1 - 2/\log d$, which can be made arbitrarily close to one
(hence, making the protocol arbitrarily close to the optimal bound) by
choosing a suitable alphabet size that does {not} depend on $n$. Namely, we have
the following result:

\begin{coro} \label{coro:WTWalkRamanujan} Let $\delta \in [0, 1)$ and $\gamma >
  0$ be arbitrary constants. Then, there is a positive integer $d$ only depending
  on $\gamma$ such that the following holds:
  For every large enough $n$, there is a $(\delta n,
  2^{-\Omega(n)}, 0)_d$-resilient wiretap protocol with block length
  $n$ and rate at least $1-\delta - \gamma$. \qed
\end{coro}

\section[Invertible Affine Extractors]{Invertible Affine Extractors
  and Asymptotically Optimal Wiretap Protocols}
\label{sec:invAExt}
In this section we will construct a black box transformation for
making certain seedless extractors invertible.  The method is
described in detail for affine extractors, and leads to wiretap
protocols with asymptotically optimal rate-resilience trade-offs.
Being based on affine extractors, these protocols
are only defined for prime power alphabet sizes. On the other hand,
the random-walk based protocol discussed in Section~\ref{sec:walk}
can be potentially instantiated for an arbitrary alphabet size, though
achieving asymptotically sub-optimal parameters (and a positive
rate only for an alphabet of size $3$ or more).

Modulo some minor differences, the construction can be simply
described as follows: A seedless affine extractor is first used to
extract a small number of uniform random bits from the source, and the
resulting sequence is then used as the seed for a seeded extractor
that extracts almost the entire entropy of the source.

Of course, seeded extractors in general are not guaranteed to work if
(as in the above construction) their seed is not independent from the
source. However, as observed by Gabizon and Raz~\cite{ref:affine}, a
\emph{linear} seeded extractor can extract from an affine source if
the seed is the outcome of an affine extractor on the source. This
idea was formalized in a more general setting by
Shaltiel~\cite{ref:mileage}.

Shaltiel's result gives a general framework for transforming any
seedless extractor (for a family of sources satisfying a certain
\emph{closedness} condition) with short output length to one with an
almost optimal output length. The construction uses the imperfect
seedless extractor to extract a small number of uniform random bits
from the source, and will then use the resulting sequence as the seed
for a seeded extractor to extract more random bits from the
source. For a suitable choice of the seeded extractor, one can use
this construction to extract almost all min-entropy of the source.

The closedness condition needed for this result to work for a family
$\cC$ of sources is that, letting $E(x, s)$ denote the seeded
extractor with seed $s$, for every $\cX \in \cC$ and every fixed $s$
and $y$, the distribution $(\cX | E(\cX, s)=y)$ belongs to $\cC$. If
$E$ is a linear function for every fixed $s$, the result will be
available for affine sources (since we are imposing a linear
constraint on an affine source, it remains an affine source).  A more
precise statement of Shaltiel's main result is the following:

\begin{thm} \cite{ref:mileage} \label{thm:mileage} Let $\cC$ be a
  class of distributions on $\F_2^n$ and $F\colon \F_2^n \to \F_2^t$
  be an extractor for $\cC$ with error $\eps$. Let $E\colon \F_2^n
  \times \F_2^t \to \F_2^m$ be a function for which $\cC$ satisfies
  the closedness condition above. Then for every $\cX \in \cC$,
  $E(\cX, F(\cX)) \sim_{\eps \cdot 2^{t+3}} E(\cX, \U_t)$. \qed
\end{thm}

Recall that a seeded extractor is called \emph{linear} if it is a
linear function for every fixed choice of the seed, and that this
condition is satisfied by Trevisan's extractor~\cite{ref:Tre}.  For
our construction, we will use the following theorem implied by the
improvement of this extractor due to Raz, Reingold and Vadhan
(Theorem~\ref{thm:Tre}):

\begin{thm} \cite{ref:RRV} \label{thm:seeded} There is an explicit
  strong linear seeded $(k, \eps)$-extractor $\extr\colon \F_2^n
  \times \F_2^d \to \F_2^m$ with $d = O(\log^3 (n/\eps))$ and $m = k -
  O(d)$. \qed
\end{thm}

\begin{rem}
  We note that our arguments would identically work for any other
  linear seeded extractor as well, for instance those constructed in
  \cites{ref:RMExtractor,ref:SU}.  However, the most crucial parameter
  in our application is the output length of the extractor, being
  closely related to the rate of the wiretap protocols we
  obtain. Among the constructions we are aware of, the result quoted
  in Theorem~\ref{thm:seeded} is the best in this regard. Moreover, an
  affine seeded extractor with better parameters is constructed by
  Gabizon and Raz~\cite{ref:affine}, but it requires a large alphabet
  size to work.
\end{rem}

Now, having the right tools in hand, we are ready to formally describe
our construction of invertible affine extractors with nearly optimal
output length. Broadly speaking, the construction follows the abovementioned idea
of Shaltiel, Gabizon, and Raz \cites{ref:mileage,ref:affine} on enlarging
the output length of affine extractors, with an additional
``twist'' for making the extractor invertible. For concreteness, the description is given over the
binary field $\F_2$:

\begin{thm} \label{thm:invAffine} For every constant $\delta \in (0,
  1]$ and every $\alpha \in (0,1)$, there is an explicit invertible
  affine extractor $D\colon \F_2^n \to \F_2^m$ for min-entropy $\delta
  n$ with output length $m = \delta n - O(n^\alpha)$ and error at most
  $O(2^{-n^{\alpha/3}})$.
\end{thm}

\begin{Proof} 
  Let $\eps := 2^{-n^{\alpha/3}}$, and $t := O(\log^3 (n/\eps)) =
  O(n^\alpha)$ be the seed length required by the extractor $\extr$ in
  Theorem~\ref{thm:seeded} for input length $n$ and error $\eps$, and
  further, let $n' := n - t$. Set up $\extr$ for input length $n'$,
  min-entropy $\delta n - t$, seed length $t$ and error $\eps$.  Also
  set up Bourgain's extractor $\bou$ for input length $n'$ and entropy
  rate $\delta'$, for an arbitrary constant $\delta' < \delta$.  Then
  the function $F$ will view the $n$-bit input sequence as a tuple
  $(s, x)$, $s \in \F_2^t$ and $x \in \F_2^{n'}$, and outputs
  $\extr(x, s + \bou(x)|_{[t]})$.  This is depicted in Figure~\ref{fig:affineExt}.

\begin{figure}
\centerline{\includegraphics[width=12cm]{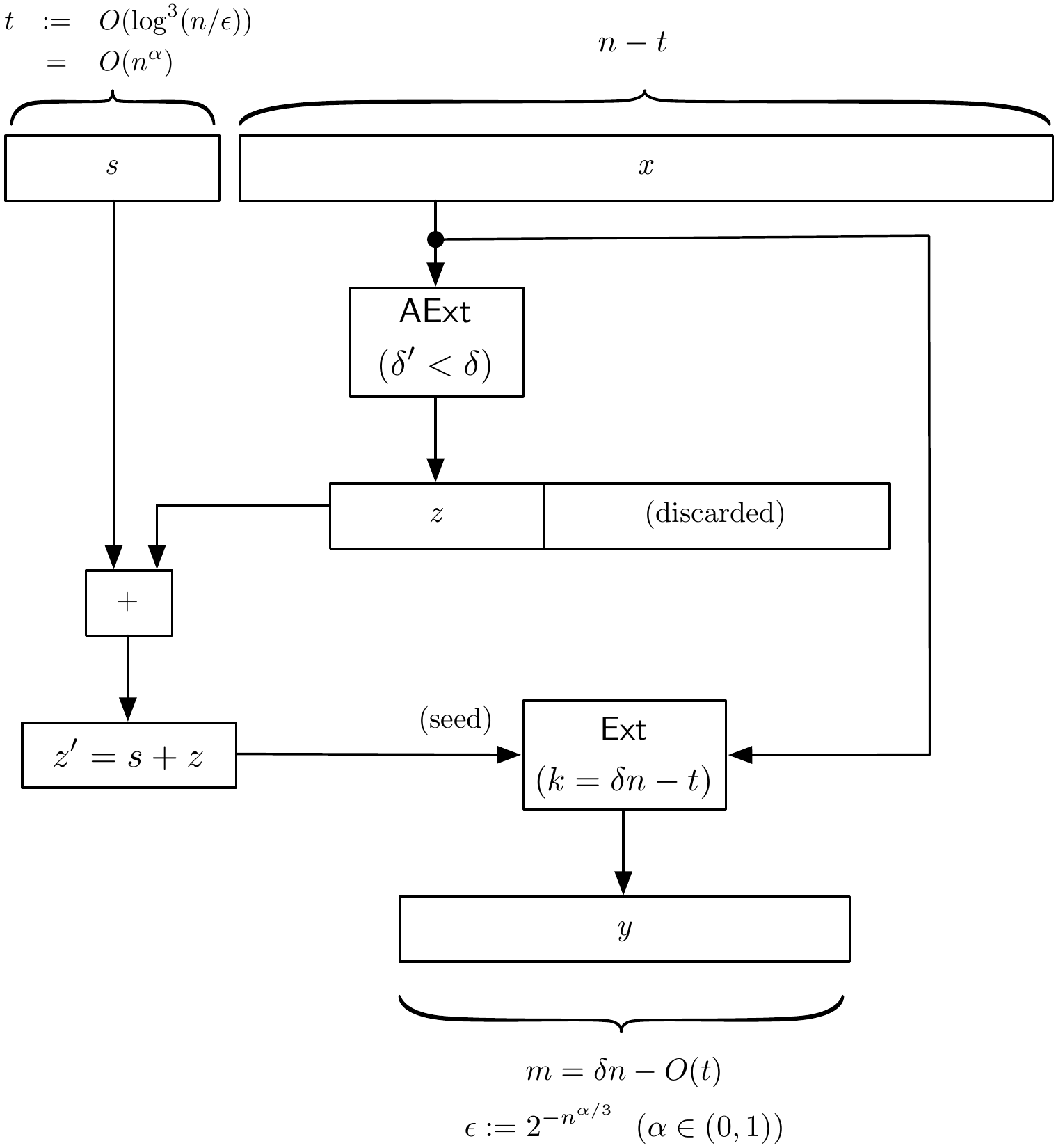}}
\caption[Construction of the invertible affine extractor]{Construction of the invertible affine extractor.}
\label{fig:affineExt}
\end{figure}

  First we show that this is an affine extractor. Suppose that $(S, X)
  \in \F_2^t \times \F_2^{n'}$ is a random variable sampled from an
  affine distribution with min-entropy $\delta n$. The variable $S$
  can have an affine dependency on $X$.  Hence, for every fixed $s \in
  \F_2^t$, the distribution of $X$ conditioned on the event $S=s$ is
  affine with min-entropy at least $\delta n - t$, which is at least
  $\delta' n'$ for large enough $n$. Hence $\bou(X)$ will be
  $2^{-\Omega(n)}$-close to uniform by Theorem~\ref{thm:Bourgain}.
  This implies that $\bou(X)|_{[t]} + S$ can extract $t$ random bits
  from the affine source with error $2^{-\Omega(n)}$.  Combining this
  with Theorem~\ref{thm:mileage}, noticing the fact that the class of
  affine extractors is closed with respect to linear seeded
  extractors, we conclude that $D$ is an affine extractor with error
  at most $\eps + 2^{-\Omega(n)} \cdot 2^{t+3} = O(2^{-n^{\alpha /
      3}})$.

  Now the inverter works as follows: Given $y \in \F_2^m$, first it
  picks $Z \in \F_2^t$ uniformly at random. The seeded extractor
  $\extr$, given the seed $Z$ is a linear function $\extr_Z\colon
  \F_2^{n'} \to \F_2^m$.  Without loss of generality, assume that this
  function is surjective\footnote{Because the seeded extractor is
    strong and linear, for most choices of the seed it is a good
    extractor (by Proposition~\ref{prop:strongDist}),
    and hence necessarily surjective (if not, one of the
    output symbols would linearly depend on the others and obviously
    the output distribution would not be close to uniform).  Hence if $\extr$ is
    not surjective for some seed $z$, one can replace it by a trivial
    surjective linear mapping without affecting its extraction
    properties.}.  Then the inverter picks $X \in \F_2^{n'}$ uniformly
  at random from the affine subspace defined by the linear constraint
  $\extr_Z(X) = y$, and outputs $(Z + \bou(X)|_{[t]}, X)$. It is easy
  to verify that the output is indeed a valid preimage of $y$. To see
  the uniformity of the inverter, note that if $y$ is chosen uniformly
  at random, the distribution of $(Z, X)$ will be uniform on
  $\F_2^n$. Hence $(Z + \bou(X)|_{[t]}, X)$, which is the output of
  the inverter, will be uniform.
\end{Proof}

In the above construction we are using an affine and a linear seeded
extractor as \emph{black boxes}, and hence, they can be replaced by
any other extractors as well (the construction will achieve an optimal
rate provided that the seeded extractor extracts almost the entire
source entropy).  In particular, over large fields one can use the
affine and seeded extractors given by Gabizon and Raz
\cite{ref:affine} that work for sub-constant entropy rates as well.

Moreover, for concreteness we described and instantiated our
construction over the binary field. Observe that Shaltiel's result,
for the special case of affine sources, holds regardless of the
alphabet size. Moreover, Trevisan's linear seeded extractor can be
naturally extended to handle arbitrary alphabets. Hence, in order to
extend our result to non-binary alphabets, it suffices to ensure that
a suitable seedless affine extractor that supports the desired
alphabet size is available. Bourgain's original
result~\cite{ref:Bourgain} is stated and proved for the binary
alphabet; however, it seems that this result can be adapted to work
for larger fields as well~\cite{ref:BourPriv}. Such an extension
(along with some improvements and simplifications) is made explicit by
Yehudayoff \cite{ref:Yeh09}.

An affine extractor is in particular, a symbol-fixing extractor.
Hence Theorem~\ref{thm:invAffine}, combined with
Lemma~\ref{lem:protocol} gives us a wiretap protocol with almost
optimal parameters:

\begin{thm} \label{coro:wiretap} Let $\delta \in [0, 1)$ and $\alpha
  \in (0, 1/3)$ be constants.  Then for a prime power $q > 1$ and
  every large enough $n$ there is a $(\delta n, O(2^{-n^{\alpha}}),
  0)_q$-resilient wiretap protocol with block length $n$ and rate $1 -
  \delta - o(1)$. \qed
\end{thm}

\section{Further Applications}
\label{sec:Apps}
In this section we will sketch some important applications of our
technique to more general wiretap problems.

\subsection{Noisy Channels and Active Intruders}
\label{sec:activeIntruder}
Suppose that Alice wants to transmit a particular sequence to Bob
through a noisy channel. She can use various techniques from coding
theory to encode her information and protect it against noise. Now
what if there is an intruder who can partially observe the transmitted
sequence and even {manipulate} it?  Modification of the sequence by
the intruder can be regarded in the same way as the channel noise;
thus one gets security against active intrusion as a ``{bonus}'' by
constructing a code that is resilient against noise and {passive}
eavesdropping.  There are two natural and {modular} approaches to
construct such a code.

A possible attempt would be to first encode the message using a good
error-correcting code and then applying a wiretap encoder to protect
the encoded sequence against the wiretapper. However, this will not
necessarily keep the information protected against the channel noise,
as the combination of the wiretap encoder and decoder does not have to
be resistant to noise.

Another attempt is to first use a wiretap encoder and then apply an
error-correcting code on the resulting sequence. Here it is not
necessarily the case that the information will be kept secure against
intrusion anymore, as the wiretapper now gets to observe the bits from
the channel-encoded sequence that may reveal information about the
original sequence. However, the wiretap protocol given in
Theorem~\ref{coro:wiretap} is constructed from an invertible {affine}
extractor, and guarantees resiliency even if the intruder is allowed
to observe arbitrary {linear combinations} of the transmitted
sequence (in this case, the
distribution of the encoded sequence subject to the intruder's observation
becomes an affine source and thus, the arguments of the proof of Lemma~\ref{lem:protocol}
remain valid). In particular, Theorem~\ref{coro:wiretap} holds even
if the intruder's observation is allowed to be obtained after applying any
arbitrary linear mapping on the output of the wiretap encoder.
Hence, we can use the wiretap scheme as an outer code and
still ensure privacy against an active intruder and reliability in
presence of a noisy channel, provided that the
error-correcting code being used as the inner code is linear.  This immediately gives us the
following result:

\begin{thm}
  Suppose that there is a $q$-ary linear error-correcting code with
  rate $r$ that is able to correct up to a $\tau$ fraction of errors
  (via unique or list decoding). Then for every constant $\delta \in
  [0, 1)$ and $\alpha \in (0,1/3)$ and large enough $n$, there is a
  $(\delta n, O(2^{-n^{\alpha}}), 0)_q$-resilient wiretap protocol
  with block length $n$ and rate $r - \delta - o(1)$ that can also
  correct up to a $\tau$ fraction of errors. \qed
\end{thm}

\begin{figure}
\centerline{\includegraphics[width=\textwidth]{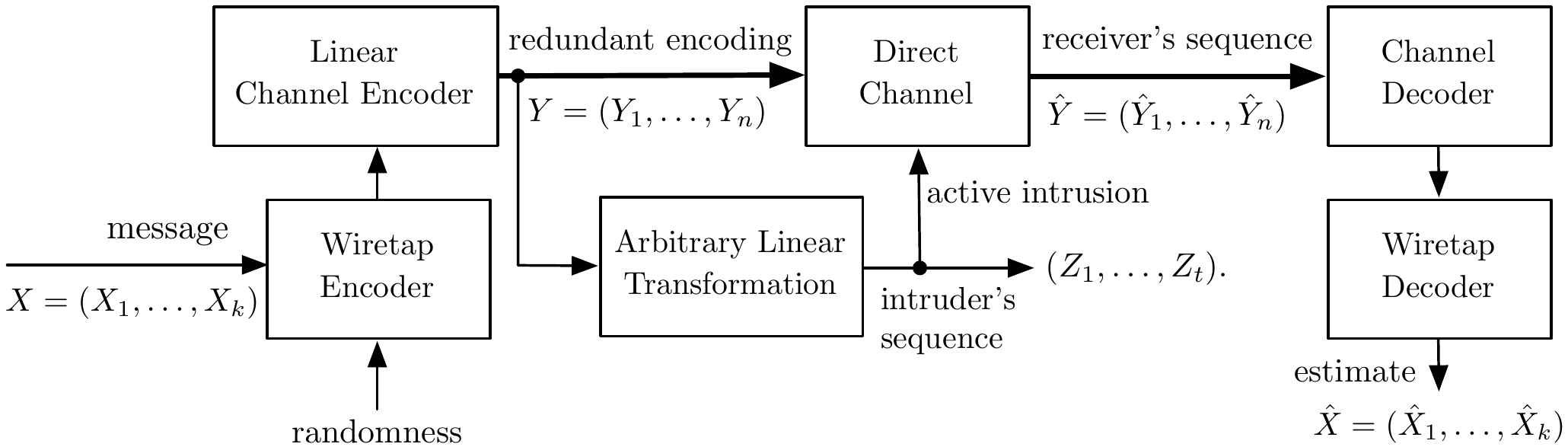}}
\caption[Wiretap scheme composed with channel coding]
{Wiretap scheme composed with channel coding. If the wiretap scheme is
constructed by an invertible affine extractor, it can guarantee secrecy
even in presence of arbitrary linear manipulation of the information.
Active intrusion can be defied using an error-correcting inner code.}
\label{fig:wiretapNoise}
\end{figure}

The setting discussed above is shown in Figure~\ref{fig:wiretapNoise}.
The same idea can be used to protect fountain codes, e.g.,
LT-~\cite{ref:LT} and Raptor Codes \cite{ref:Raptor}, against
wiretappers without affecting the error correction capabilities of the
code.

Obviously this simple composition idea can be used for any type of
channel so long as the inner code is linear, at the cost of reducing
the total rate by almost $\delta$. Hence, if the inner code achieves
the Shannon capacity of the direct channel (in the absence of the
wiretapper), the composed code will achieve the capacity of the
wiretapped channel, which is less than the original capacity by
$\delta$ \cite{ref:CK78}.

\subsection{Network Coding}
\label{sec:NetCod}
Our wiretap protocol from invertible affine extractors is also
applicable in the more general setting of transmission over
{networks}.  A communication network can be modeled as a directed
graph, in which nodes represent the network devices and information is
transmitted along the edges. One particular node is identified as the
\emph{source} and $m$ nodes are identified as \emph{receivers}. The
main problem in network coding is to have the source reliably transmit
information to the receivers at the highest possible rate, while
allowing the intermediate nodes arbitrarily process the information
along the way.

Suppose that, in the graph that defines the topology of the network,
the min-cut between the source to each receiver is $n$. It was shown
in \cite{ref:NetCod1} that the source can transmit information up to
rate $n$ (symbols per transmission) to all receivers (which is
optimal), and in \cites{ref:NetCod2,ref:NetCod3} that {linear} network
coding is in fact sufficient to achieve this rate. That is, the
transmission at rate $n$ is possible when the intermediate nodes are
allowed to forward packets that are (as symbols over a finite field)
linear combinations of the packets that they receive (See
\cite{ref:NetCodBook} for a comprehensive account of these and other
relevant results).

A basic example is shown by the \emph{butterfly} network in
Figure~\ref{fig:netcod}.  This network consists of a source on the top
and two receivers on the bottom, where the min-cut to each receiver is~$2$.
Without processing the incoming data, as in the left figure, one
of the two receivers may receive information at the optimal rate of~$2$
symbols per transmission (namely, receiver~$1$ in the figure).
However, due to the bottleneck existing in the middle (shown by the
thick edge $a\to b$), the other receiver will be forced to receive at
an inferior rate of~$1$ symbol per transmission. However, if linear
processing of the information is allowed, node $a$ may combine its
incoming information by treating packets as symbols over a finite
field and adding them up, as in the right figure. Both receivers may
then solve a full-rank system of linear equations to retrieve the
original source symbols $x_1$ and $x_2$, and thereby achieve the optimal
min-cut rate.

\begin{figure}
\centerline{\includegraphics[width=\textwidth]{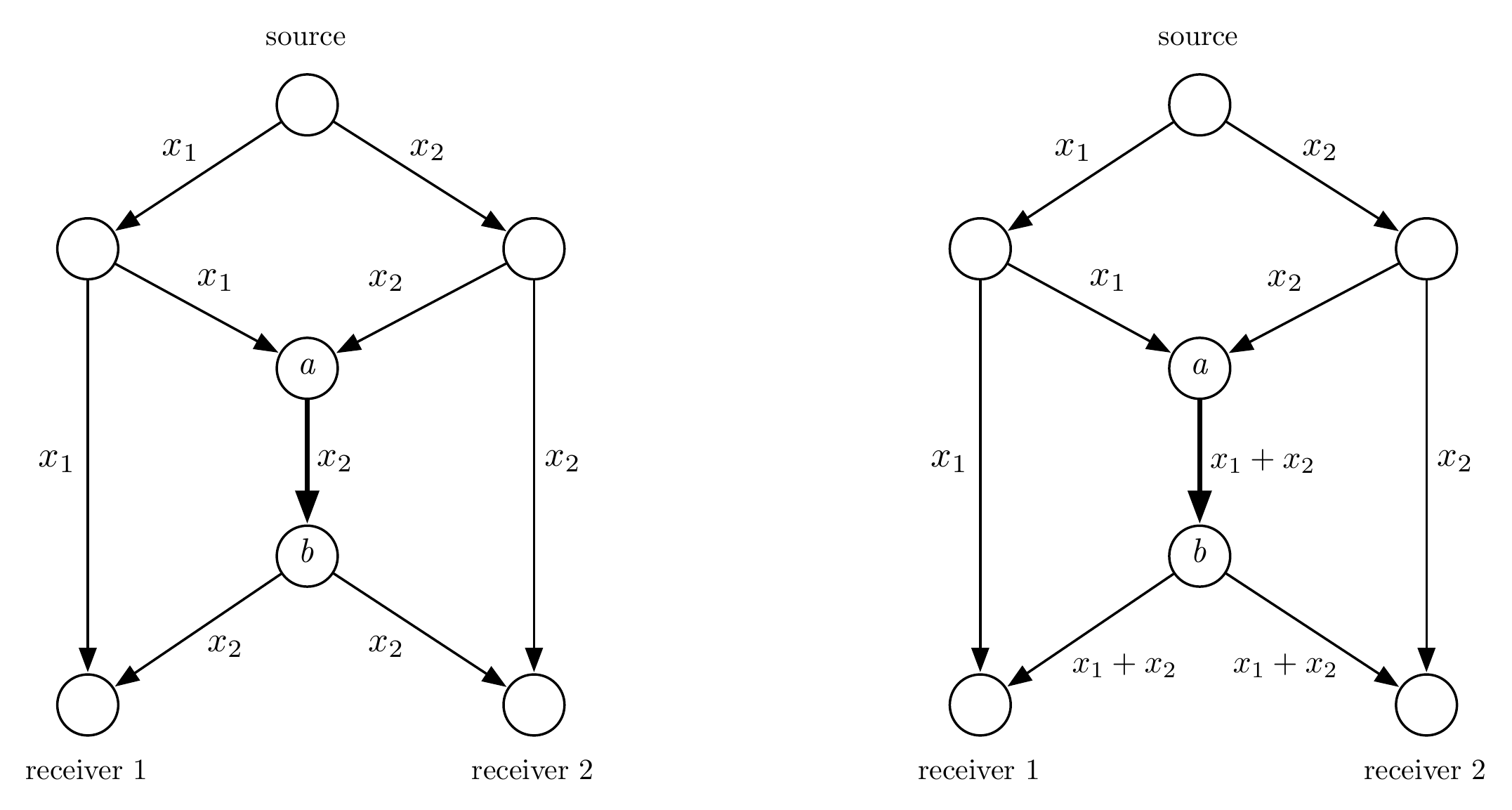}}
\caption[Network coding versus unprocessed forwarding]
{Network coding (right), versus unprocessed forwarding (left).}
\label{fig:netcod}
\end{figure}

Designing wiretap protocols for networks is an important question in
network coding, which was first posed by Cai and Yeung
\cite{ref:NetWT}. In this problem, an intruder can choose a bounded
number, say $t$, of the edges and eavesdrop all the packets going
through those edges. They designed a network code that could provide
the optimal multicast rate of $n - t$ with perfect privacy. However
this code requires an alphabet size of order $\binom{|E|}{t}$, where
$E$ is the set of edges.  Their result was later improved in
\cite{ref:FMSS04} who showed that a random linear coding scheme can
provide privacy with a much smaller alphabet size if one is willing to
achieve a slightly sub-optimal rate. Namely, they obtain rate
$n-t(1+\eps)$ with an alphabet of size roughly
$\mathrm{\Theta}(|E|^{1/\eps})$, and show that achieving the exact
optimal rate is not possible with small alphabet size.

El~Rouayheb and Soljanin \cite{ref:Emina} suggested to use the
original code of Ozarow and Wyner \cite{ref:Wyner2} as an {outer code}
at the source and showed that a careful choice of the network code can
provide optimal rate with perfect privacy. However, their code
eventually needs an alphabet of size at least $\binom{|E|-1}{t-1} +
m$. Building upon this work, Silva and Kschischang
\cite{ref:RankMetric} constructed an outer code that provides similar
results while leaving the underlying network code unchanged.  However,
their result comes at the cost of increasing the packet size by a
multiplicative factor of at least the min-cut bound, $n$ (or in
mathematical terms, the original alphabet size $q$ of the network is
enlarged to at least $q^n$).  For practical purposes, this is an
acceptable solution provided that an estimate on the min-cut size of
the network is available at the wiretap encoder.

\begin{figure}
\centerline{\includegraphics[width=8cm]{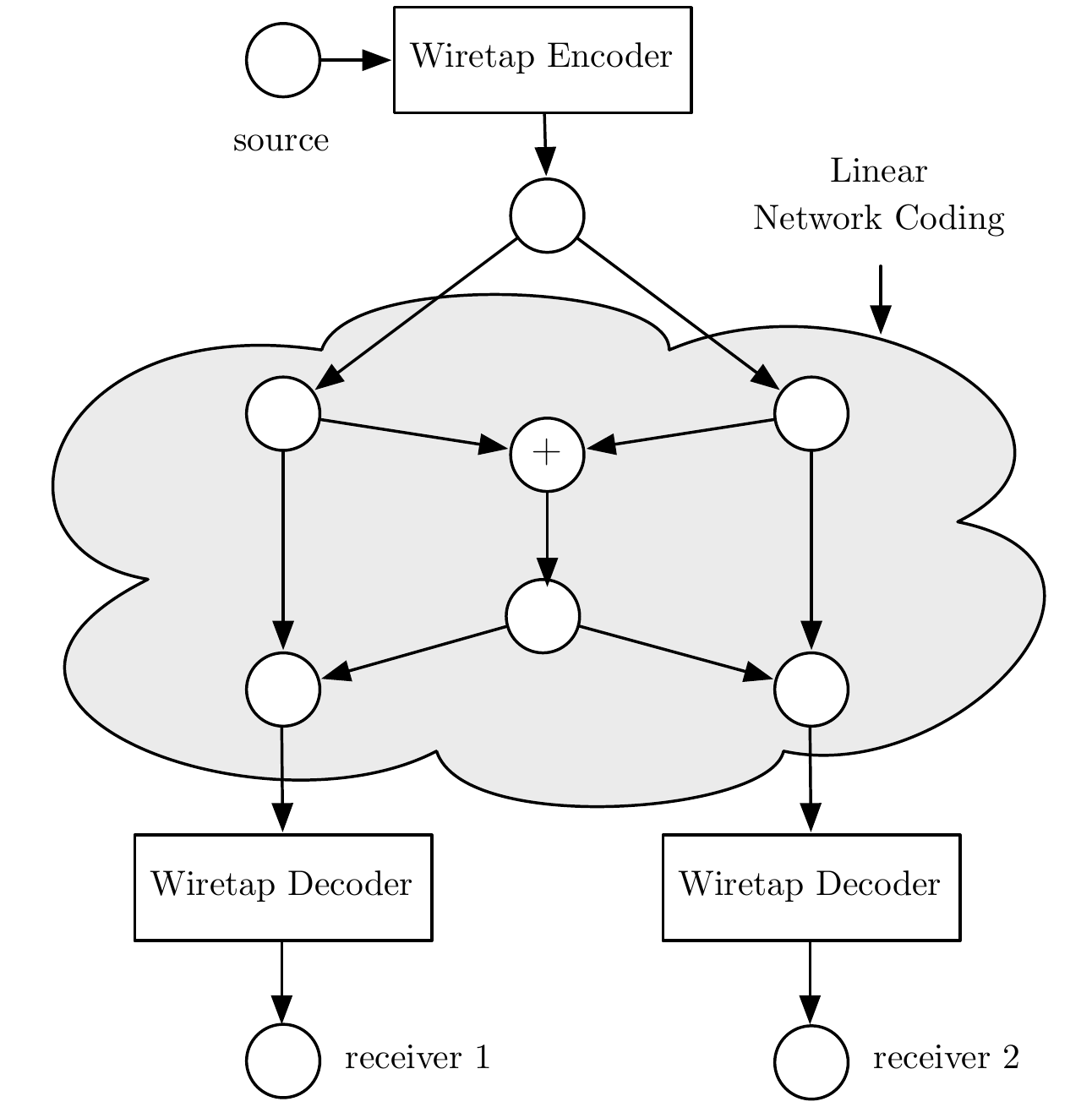}}
\caption[Linear network coding with an outer layer of wiretap
  encoding added for providing secrecy]{Linear network coding with an outer layer of wiretap
  encoding added for providing secrecy.}
\label{fig:netCodWiretap}
\end{figure}

By the discussion presented in Section~\ref{sec:activeIntruder},
the rate-optimal wiretap protocol given in Theorem~\ref{coro:wiretap}
stays resilient even in presence of any linear post-processing of the
encoded information. Thus, using the wiretap encoder given by this
result as an outer-code in the source node, one can construct an
asymptotically optimal wiretap protocol for networks that is
completely unaware of the network and eliminates all the restrictions
in the above results. This is schematically shown in Figure~\ref{fig:netCodWiretap}.
Hence, extending our notion of $(t, \eps,
\gamma)_q$-resilient wiretap protocols naturally to communication
networks, we obtain the following:

\begin{thm}
  Let $\delta \in [0, 1)$ and $\alpha \in (0, 1/3)$ be constants, and
  consider a network that uses a linear coding scheme over a finite
  field $\F_{q}$ for reliably transmitting information at rate $R$.
Suppose that, at each transmission, an intruder can arbitrarily observe
up to $\delta R$ intermediate links in the network.
Then the source
  and the receiver nodes can use an outer code of rate $1 - \delta -
  o(1)$ (obtaining a total rate of $R(1-\delta) - o(1)$) which is completely
  independent of the network, leaves the
  network code unchanged, and provides almost perfect privacy with error
  $O(2^{-{R}^{\alpha}})$ and zero leakage over a $q$-ary
  alphabet. \qed
\end{thm}

In addition to the above result that uses the invertible affine extractor
of Theorem~\ref{thm:invAffine}, it is possible to use other rate-optimal
invertiable affine extractors. In particular,
observe that the restricted affine extractor of
Theorem~\ref{thm:resAffExtr} (and in particular, Corollary~\ref{coro:resAffExtrGabidulin})
is a linear function (over the extension
field) and is thus, obviously has an efficient $0$-inverter (since
inverting the extractor amounts to solving a system of linear
equations).  By using this extractor (instantiated with Gabidulin's
MRD codes as in Corollary~\ref{coro:resAffExtrGabidulin}),
 we may recover the result of Silva and Kschischang
\cite{ref:RankMetric} in our framework. More precisely, we have the following result:

\begin{coro}
  Let $q$ be any prime power, and
  consider a network with minimum cut of size $n$ that uses a linear coding scheme over
  $\F_{q}$ for reliably transmitting information at rate $R$.
Suppose that, at each transmission, an intruder can arbitrarily observe up to
$\delta R$ intermediate links in the network, for some $\delta \in [0, 1)$.
Then the source
  and the receiver nodes can use an outer code of rate $1 - \delta$
  over $\F_{q^n}$ (obtaining a total rate of $R(1-\delta)$)
  that provides perfect privacy over a $q^n$-ary
  alphabet. \qed
\end{coro}

\subsection{Arbitrary Processing}
\label{subsec:arbitrary}

In this section we consider the erasure wiretap problem in its most
general setting, which is still of practical importance.  Suppose that
the information emitted by the source goes through an arbitrary
communication medium and is arbitrarily processed on the way to
provide protection against noise, to obtain better throughput, or for
other reasons. Now consider an intruder who is able to eavesdrop a
bounded amount of information at various points of the channel. One
can model this scenario in the same way as the original point-to-point
wiretap channel problem, with the difference that instead of observing
$t$ arbitrarily chosen bits, the intruder now gets to choose an
arbitrary Boolean circuit $\cC$ with $t$ output bits (which captures the
accumulation of all the intermediate processing) and observes the
output of the circuit when applied to the transmitted
sequence\footnote{ In fact this models a ``harder'' problem, as in our
  problem the circuit $\cC$ is given by the communication scheme and
  not the intruder. Nevertheless, we consider the harder problem.}.

Obviously there is no way to guarantee resiliency in this setting,
since the intruder can simply choose $\cC$ to compute $t$ output bits
of the wiretap decoder.  However, suppose that in addition there is an
auxiliary communication channel between the source and the receiver
(that we call the \emph{side channel}) that is separated from the main
channel, and hence, the information passed through the two channel do
not \emph{blend} together by the intermediate processing.

We call this scenario the \emph{general wiretap problem}, and extend
our notion of $(t, \eps, \gamma)$-resilient protocol to this problem,
with the slight modification that now the output of the encoder (and
the input of the decoder) is a pair of strings $(y_1, y_2) \in \F_2^n
\times \F_2^d$, where $y_1$ (resp., $y_2$) is sent through the main
(resp., side) channel.  Now we call $n+d$ the block length and let the
intruder choose an arbitrary pair of circuits $(\cC_1, \cC_2)$, one
for each channel, that output a total of $t$ bits, and observe
$(\cC_1(y_1), \cC_2(y_2))$.

The infor\-mation-theoretic upper bounds for the achievable rates in the
original wiretap problem obviously extend to the general wiretap
problem as well.  Below we show that for the general problem, secure
transmission is possible at asymptotically optimal rates even if the
intruder intercepts the \emph{entire communication} passing through
the side channel (as shown in Figure~\ref{fig:arbitraryWiretap}).

\begin{figure}
\centerline{\includegraphics[width=8cm]{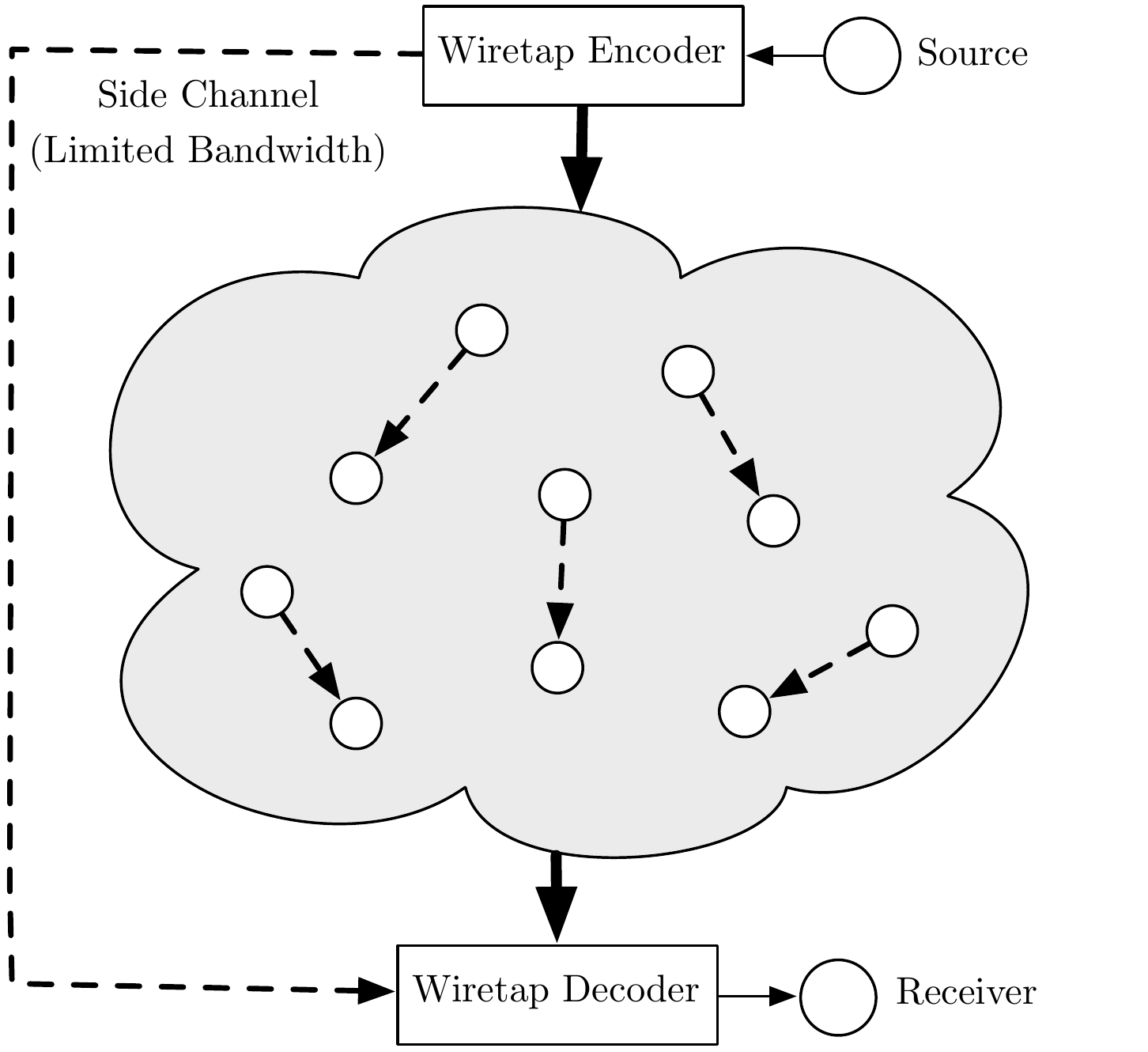}}
\caption[The wiretap channel problem in presence of arbitrary
intermediate processing]{The wiretap channel problem in presence of
  arbitrary intermediate processing. In this example, data is
  transmitted over a packet network (shown as a cloud) in which some
  intermediate links (showed by the dashed arrows) are accessible to
  an intruder.}
\label{fig:arbitraryWiretap}
\end{figure}

%

Similar as before, our idea is to use invertible extractors to
construct general wiretap protocols, but this time we use invertible
strong seeded extractors. Strong seeded extractors were used in
\cite{ref:CDH} to construct ERFs, and this is exactly what we use as
the decoder in our protocol.  As the encoder we will use the
corresponding inverter, which outputs a pair of strings, one for the
extractor's input which is sent through the main channel and another
as the seed which is sent through the side channel.  Hence we will
obtain the following result:

\begin{thm} \label{thm:generalWiretap} Let $\delta \in [0, 1)$ be a
  constant. Then for every $\alpha, \eps > 0$, there is a $(\delta n,
  \eps, 2^{-\alpha n}+\eps)$-resilient wiretap protocol for the
  general wiretap channel problem that sends $n$ bits through
  the main channel and $d := O(\log^3(n/\eps^2))$ bits through the
  side channel and achieves rate $1 - \delta - \alpha -
  O(d/(n+d))$. The protocol is secure even when the entire
  communication through the side channel is observable by the
  intruder.
\end{thm}

\begin{proof} 
  We will need the following claim in our proof, which is easy to
  verify using an averaging argument:
  \begin{claims} \label{prop:preimage} Let $f\colon \zo^n \to
    \zo^{\delta n}$ be a Boolean function. Then for every $\alpha >
    0$, and $X \sim \U_n$, the probability that $f(X)$ has fewer than
    $2^{n(1-\delta-\alpha)}$ preimages is at most $2^{-\alpha n}$.
  \end{claims}

  Now, let $\extr$ be the linear seeded extractor of
  Theorem~\ref{thm:seeded}, set up for input length $n$, seed length
  $d = O(\log^3(n/\eps^2))$, min-entropy $n(1-\delta-\alpha)$, and
  output length $m = n(1-\delta-\alpha) - O(d)$, and error $\eps^2$.
  Then the encoder chooses a seed $Z$ for the extractor uniformly at
  random and sends it through the side channel.

  For the chosen value of $Z$, the extractor is a linear function, and
  as before, given a message $x \in \zo^m$, the encoder picks a random
  vector in the affine subspace that is mapped by this linear function
  to $x$ and sends it through the public channel.

  The decoder, in turn, applies the extractor to the seed received
  from the secure channel and the transmitted string. The resiliency
  of the protocol can be shown in a similar manner as in
  Lemma~\ref{lem:protocol}. Specifically, note that by
  the above claim, with probability at least $1-2^{-\alpha n}$, the
  string transmitted through the main channel, conditioned on the
  observation of the intruder from the main channel, has a
  distribution $\cY$ with min-entropy at least
  $n(1-\delta-\alpha)$. Now in addition suppose that the seed $z$ is
  entirely revealed to the intruder. As the extractor is strong, with
  probability at least $1-\eps$, $z$ is a \emph{good} seed for $\cY$,
  meaning that the output of the extractor applied to $\cY$ and seed
  $z$ is $\eps$-close to uniform (by Proposition~\ref{prop:strongDist}), and hence the view of the intruder
  on the original message remains $\eps$-close to uniform.
\end{proof}



We observe that it is not possible to guarantee zero leakage for the
general wiretap problem above.  Specifically, suppose that $(\cC_1,
\cC_2)$ are chosen in a way that they have a single preimage for a
particular output $(w_1, w_2)$.  With nonzero probability the
observation of the intruder may turn out to be $(w_1, w_2)$, in which
case the entire message is revealed.  Nevertheless, it is possible to
guarantee negligible leakage as the above theorem does.  Moreover,
when the general protocol above is used for the original wiretap II
problem (where there is no intermediate processing involved), there is
no need for a separate side channel and the entire encoding can be
transmitted through a single channel.  Contrary to
Theorem~\ref{coro:wiretap} however, the general protocol will not
guarantee zero leakage even for this special case.


\begin{subappendices}

  \section{Some Technical Details} \label{app:wiretapDetails}

  This appendix is devoted to some technical details that are
  omitted in the main text of the chapter.

  The following proposition quantifies the Shannon entropy of a
  distribution that is close to uniform:

  \begin{prop} \label{prop:Shannon} Let $\cX$ be a probability
    distribution on a finite set $S$, $|S| > 4$, that is $\eps$-close
    to the uniform distribution on $S$, for some $\eps \leq 1/4$. Then
    $H(\cX) \geq \log_2 |S| (1 - \eps)
    $ 
  \end{prop}

\begin{Proof}
  Let $n := |S|$, and let $f(x) := - x \log_2 x$. The function $f(x)$
  is concave, passes through the origin and is strictly increasing in
  the range $[0, 1/\mathrm{e}]$.  From the definition, we have $H(\cX)
  = \sum_{s \in S} f(\Pr_\cX(s))$.  For each term $s$ in this
  summation, the probability that $\cX$ assigns to $s$ is either at
  least $1/n$, which makes the corresponding term at least $\log_2
  n/n$ (due to the particular range of $|S|$ and $\eps$), or is equal
  to $1/n - \eps_s$, for some $\eps_s > 0$, in which case the term
  corresponding to $s$ is less than $\log_2 n/n$ by at most $\eps_s
  \log_2 n$ (this follows by observing that the slope of the line
  connecting the origin to the point $(1/n, f(1/n))$ is $\log_2
  n$). The bound on the statistical distance implies that the
  differences $\eps_s$ add up to at most $\eps$.  Hence, the Shannon
  entropy of $\cX$ can be less than $\log_2 n$ by at most $\eps \log_2
  n $.
\end{Proof}

\begin{prop} \label{prop:duality} Let $(X, Y)$ be a pair of random
  variables jointly distributed on a finite set $\Omega \times
  \Gamma$. Then\footnote{Here we are abusing the notation and denote
    by $Y$ the marginal distribution of the random variable $Y$, and
    by $Y|(X=a)$ the distribution of the random variable $Y$
    conditioned on the event $X=a$.}  $ \Exp_Y[\dist(X|Y, X)] =
  \Exp_X[\dist(Y|X, Y)].  $
\end{prop}

\begin{Proof}
  For $x \in \Omega$ and $y \in \Gamma$, we will use shorthands $p_x,
  p_y, p_{xy}$ to denote $\Pr[X=x], \Pr[Y=y], \Pr[X=x,Y=y]$,
  respectively. Then we have
  \begin{eqnarray*}
    \Exp_Y[\dist(X|Y, X)] &=& \sum_{y \in \Gamma} p_y \dist(X|(Y=y), X)
    = \frac{1}{2} \sum_{y \in \Gamma} p_y \sum_{x \in \Omega} |p_{xy}/p_y - p_x| \\
    &=& \frac{1}{2} \sum_{y \in \Gamma} \sum_{x \in \Omega} |p_{xy} - p_x p_y|
    = \frac{1}{2} \sum_{x \in \Omega} p_x \sum_{y \in \Gamma} |p_{xy}/p_x - p_y| \\
    &=& \sum_{x \in \Omega} p_x \dist(Y|(X=x), Y) = \Exp_X[\dist(Y|X, Y)].
  \end{eqnarray*}
\end{Proof}

\begin{prop} \label{prop:conditioning} Let $\Omega$ be a finite set
  that is partitioned into subsets $S_1, \ldots, S_k$ and suppose that
  $\cX$ is a distribution on $\Omega$ that is $\gamma$-close to
  uniform. Denote by $p_i$, $i = 1, \ldots k$, the probability
  assigned to the event $S_i$ by $\cX$. Then
  \[
  \sum_{i \in [k]} p_i \cdot \dist( \cX | S_i, \U_{S_i} ) \leq 2
  \gamma.
  \]
\end{prop}

\begin{Proof}
  Let $N := |\Omega|$, and define for each $i$, $ \gamma_i := \sum_{s
    \in S_i} \left| \Pr_\cX(s) - \frac{1}{N} \right|,$ so that
  $\gamma_1 + \cdots + \gamma_k \leq 2 \gamma$.  Observe that by
  triangle's inequality, for every $i$ we must have $|p_i - |S_i|/N |
  \leq \gamma_i$.  To conclude the claim, it is enough to show that
  for every $i$, we have $\dist( \cX | S_i, \U_{S_i} ) \leq \gamma_i /
  p_i$.  This is shown in the following.
  {\allowdisplaybreaks \begin{eqnarray*}
    p_i \cdot \dist( \cX | S_i, \U_{S_i} ) &=& \frac{p_i}{2} \sum_{s \in S_i} \left| \frac{\Pr_\cX(s)}{p_i} - \frac{1}{|S_i|} \right| \\
    &=& \frac{1}{2} \sum_{s \in S_i} \left| \Pr_\cX(s) - \frac{p_i}{|S_i|} \right| \\
    &=& \frac{1}{2} \sum_{s \in S_i} \left| \left(\Pr_\cX(s) - \frac{1}{N} \right) + \frac{1}{|S_i|} \left( \frac{|S_i|}{N} - p_i \right) \right| \\
    &\leq& \frac{1}{2} \sum_{s \in S_i} \left| \Pr_\cX(s) - \frac{1}{N} \right| + \frac{1}{2|S_i|} \sum_{s \in S_i} \left| \frac{|S_i|}{N} - p_i  \right| \\
    &\leq& \frac{\gamma_i}{2} + \frac{1}{2|S_i|} \cdot |S_i| \gamma_i = \gamma_i.
  \end{eqnarray*}}
\end{Proof}

The following proposition shows that any function
maps close distributions to close distributions:

\begin{prop} \label{prop:closeFunction} Let $\Omega$ and $\Gamma$ be
  finite sets and $f$ be a function from $\Omega$ to $\Gamma$. Suppose
  that $\mathcal{X}$ and $\mathcal{Y}$ are probability distributions
  on $\Omega$ and $\Gamma$, respectively, and let $\mathcal{X'}$ be a
  probability distribution on $\Omega$ which is $\delta$-close to
  $\mathcal{X}$. Then if $f(\mathcal{X}) \sim_\eps \mathcal{Y}$, then
  $f(\mathcal{X'}) \sim_{\eps+\delta} \mathcal{Y}$.
\end{prop}

\begin{proof}
Let $X$, $X'$ and $Y$ be random variables distributed
according to $\mathcal{X}$, $\mathcal{X'}$, and $\mathcal{Y}$,
respectively. We want to upperbound
\[
 \left| \Pr[ f(X') \in T ] - \Pr[ Y \in T ] \right|
\]
for every $T \subseteq \Gamma$.
By the triangle inequality, this is no more than
\[
 \left| \Pr[ f(X') \in T ] - \Pr[ f(X) \in T ] \right| + \left| \Pr[ f(X) \in T ] - \Pr[ Y \in T ] \right|.
\]
Here the summand on the right hand side is upperbounded by the distance
of $f(\mathcal{X})$ and $\mathcal{Y}$, that is assumed to be at most $\eps$.
Let $T' \eqdef \{x \in \Omega \mid f(x) \in T\}$.
Then the summand on the left can be written as
\[
 \left| \Pr[X' \in T'] - \Pr[X \in T'] \right|
\]
which is at most $\delta$ by the assumption that $\mathcal{X} \sim_\delta \mathcal{X'}$.
\end{proof}

\subsection*{Omitted Details of the Proof of
  Corollary~\ref{coro:WTWalk}}

\newcommand{\MOD}{\mathrm{mod}\ } \newcommand{\modf}{\mathsf{Mod}}
\newcommand{\imod}{\mathsf{Inv}}

Here we prove Corollary~\ref{coro:WTWalk} for the case $c > 1$.  The
construction is similar to the case $c = 1$, and in particular the
choice of $m$ and $k$ will remain the same. However, a subtle
complication is that the expander family may not have a graph with
$d^m$ vertices and we need to adapt the extractor of
Theorem~\ref{thm:invWalk} to support our parameters, still with
exponentially small error. To do so, we pick a graph $G$ in the family
with $N$ vertices, such that \[c^{\eta m} d^m \leq N \leq c^{\eta m +
  1} d^m,\] for a small absolute constant $\eta > 0$ that we are free
to choose.  The assumption on the expander family guarantees that such
a graph exists. Let $m'$ be the smallest integer such that $d^{m'}
\geq c^{\eta m} N$.  Index the vertices of $G$ by integers in $[N]$.
Note that $m'$ will be larger than $m$ by a constant multiplicative
factor that approaches $1$ as $\eta \to 0$.

For positive integers $q$ and $p \leq q$, define the function
$\modf_{q, p}\colon [q] \to [p]$ by \[\modf_{q, p}(x) := 1 + (x\ \MOD
p).\] The extractor $\kz$ interprets the first $m'$ symbols of the
input as an integer $u$, $0\le u< d^{m'}$ and performs a walk on $G$
starting from the vertex $\modf_{d^{m'}, N}(u+1)$, the walk being
defined by the remaining input symbols. If the walk reaches a vertex
$v$ at the end, the extractor outputs $\modf_{N, d^{m}}(v)-1$, encoded
as a $d$-ary string of length $m$.  A similar argument as in
Theorem~\ref{thm:invWalk} can show that with our choice of the
parameters, the extractor has an exponentially small error, where the
error exponent is now inferior to that of Theorem~\ref{thm:invWalk} by
$O(m)$, but the constant behind $O(\cdot)$ can be made arbitrarily
small by choosing a sufficiently small $\eta$.

The real difficulty lies with the inverter because $\modf$ is not a
balanced function (that is, all images do not have the same number of
preimages), thus we will not be able to obtain a perfect
inverter. Nevertheless, it is possible to construct an inverter with a
close-to-uniform output in $\ell_\infty$ norm. This turns out to be as
good as having a perfect inverter, and thanks to the following lemma,
we will still be able to use it to construct a wiretap protocol with
zero leakage:

\begin{lem}
  \label{lem:infty}
  Suppose that $f\colon [d]^n \rightarrow [d]^m $ is a
  $(k,2^{-\Omega(m)})_d$ symbol-fixing extractor and that $\cX$ is a
  distribution on $[d]^n$ such that $\|\mathcal{X}-\U_{[d]^n}\|_\infty
  \le 2^{-\Omega(m)}/d^n$. Denote by $\cX'$ the distribution $\cX$
  conditioned on any fixing of at most $n-k$ coordinates. Then
  $f(\cX') \sim_{2^{-\Omega(m)}} \U_{[d]^m}$.
\end{lem}
\begin{Proof}
  By Proposition~\ref{prop:closeFunction}, it suffices to show that
  $\cX'$ is $2^{-\Omega(m)}$-close to an $(n, k)_d$ symbol-fixing
  source.  Let $S \subseteq [d]^m$ denote the support of $\cX'$, and
  let $\eps/d^n$ be the $\ell_\infty$ distance between $\mathcal{X}$
  and $\U_{[d]^n}$, so that by our assumption, $\eps =
  2^{-\Omega(m)}$.  By the bound on the $\ell_\infty$ distance, we
  know that $\Pr_\mathcal{X}(S)$ is between $\frac{|S|}{d^n}(1 -
  \eps)$ and $\frac{|S|}{d^n}(1 + \eps)$. Hence for any $x \in S$,
  $\Pr_{\cX'}(x)$, which is $\Pr_{\cX}(x)/\Pr_{\cX}(S)$, is between
  $\frac{1}{|S|}\cdot \frac{1 - \eps}{1 + \eps}$ and
  $\frac{1}{|S|}\cdot \frac{1 + \eps}{1 - \eps}$.  This differs from
  $1/|S|$ by at most $O(\eps)/|S|$. Hence, $\cX'$ is
  $2^{-\Omega(m)}$-close to $\U_S$.
\end{Proof}

In order to invert our new construction, we will need to construct an
inverter $\imod_{q, p}$ for the function $\modf_{q, p}$.  For that,
given $x \in[p]$ we will just sample uniformly in its preimages.  This
is where the non-balancedness of $\modf$ causes problems, since if $p$
does not divide $q$ the distribution $\imod_{q, p}(\U_{[p]})$ is not
uniform on $[q]$.
\begin{lem}
  Suppose that $q>p$.  Given a distribution $\mathcal{X}$ on $[p]$
  such that $\| \mathcal{X} - \U_{[p]} \|_\infty \le \frac{\eps}{p}$,
  we have $\| \imod_{q, p}(\mathcal{X}) - \U_{[q]} \|_\infty \le
  \frac{1}{q}\cdot \frac{p+\epsilon q}{q-p}$.
\end{lem}
\begin{Proof}
  Let $X \sim \mathcal{X}$ and $Y \sim \imod_{q,p}(\mathcal{X})$.
  Since we invert the modulo function by taking for a given output a
  random preimage uniformly, $\Pr[Y=y]$ is equal to $\Pr[X =
  \modf_{q,p}(y)]$ divided by the number of $y$ with the same value
  for $\modf_{q,p}(y)$.  The latter number is either $\lfloor q/p
  \rfloor$ or $\lceil q/p \rceil$, so
  \begin{equation*}
    \frac{1-\eps}{p \lceil q/p \rceil} \le \Pr(Y = y) \le
    \frac{1+\eps}{p \lfloor q/p \rfloor}
  \end{equation*}
  Bounding the floor and ceiling functions by $q/p \pm 1$, we obtain
  \begin{equation*}
    \frac{1-\eps}{q+p}
    \le \Pr(Y = y) \le
    \frac{1+\eps}{q-p}
  \end{equation*}
  That is
  \begin{equation*}
    \frac{-p-\epsilon q}{q(q+p)}
    \le \Pr(Y = y)-\frac{1}{q} \le
    \frac{p+\epsilon q}{q(q-p)}\ ,
  \end{equation*}
  which concludes the proof since this is true for all $y$.
\end{Proof}

Now we describe the inverter $\inv(x)$ for the extractor, again
abusing the notation.  First the inverter calls $\imod_{N, d^{m}}(x)$
to obtain $x_1 \in [N]$.  Then it performs a random walk on the graph,
starting from $x_1$, to reach a vertex $x_2$ at the end which is
inverted to obtain $x_3=\imod_{d^{m'}, N}(x_2)$ as a $d$-ary string of
length $m'$.  Finally, the inverter outputs $y=(x_3,w)$, where $w$
corresponds the \emph{inverse} of the random walk of length $n-m'$.
It is obvious that this procedure yields a valid preimage of $x$.

Using the previous lemma, if $x$ is chosen uniformly, $x_1$ will be at
$\ell_\infty$-distance \[\epsilon_1 := \frac{1}{N}\cdot
\frac{d^m}{N-d^m} = \frac{1}{N}O(c^{-\eta m}).\] For a given walk, the
distribution of $x_2$ will just be a permutation of the distribution
of $x_1$ and applying the lemma again, we see that the
$\ell_\infty$-distance of $x_3$ from the uniform distribution is
\[ \epsilon_2 := \frac{1}{d^{m'}} \cdot \frac{N+\epsilon_1
  d^{m'}}{d^{m'}-N} = \frac{1}{d^{m'}} O(c^{-\eta m}).\] This is true
for all the $d^{n-m'}$ possible walks so the $\ell_\infty$-distance of
the distribution of $y$ from uniform is bounded by $\frac{1}{d^n}
O(c^{-\eta m})$.  Applying Lemma~\ref{lem:infty} in an argument
similar to Lemma~\ref{lem:protocol} concludes the proof.

\end{subappendices}

\musicBoxWiretap


\Chapter{Group Testing}
\epigraphhead[70]{\epigraph{\textsl{``War does not determine who is right---only who is left.''}}{\textit{--- Bertrand Russell}}}
\label{chap:testing}

The history of group testing is believed to date back to the second
World War.  During the war, millions of blood samples taken from
draftees had to be subjected to a certain test, and be analyzed in
order to identify a few thousand cases of syphilis. The tests were
identical for all the samples.  Here the idea of group testing came to
a statistician called Robert Dorfman (and perhaps, a few other
researchers working together with him, among them David
Rosenblatt). He made a very intuitive observation, that, the samples
are constantly subjected to the same test, which is extremely
sensitive and remains reliable even if the sample is
diluted. Therefore, it makes sense to, instead of analyzing each
sample individually, pool every few samples in a group, and apply the
test on the mixture of the samples.  If the test outcome is negative,
we will be sure that none of the samples participating in the pool are
positive. On the other hand, if the outcome is positive, we know that
one or more of the samples are positive, and will have to proceed with
more refined, or individual, tests in order to identify the individual
positives within the group.

Since the number of positives in the entire population was suspected
to be in order of a few thousands---a small fraction of the
population---Dorfman's idea would save a great deal of time and
resources. Whether or not the idea had been eventually implemented at
the time, Dorfman went on to publish a paper on the topic
\cite{ref:Dor43}, which triggered an extensive line of research in
combinatorics known today as \emph{combinatorial group testing}.

The main challenge in group testing is to design the pools in such a
way to minimize the number of tests required in order to identify the
exact set of positives.  Larger groups would save a lot of tests if
their outcome is negative, and are rather wasteful otherwise (since in
the latter case they convey a relatively small amount of information).

Of course the applications of group testing are not limited to blood
sampling. To mention another early example, consider a production line
of electric items such as light bulbs (or resistors, capacitors, etc).
As a part of the quality assurance, defective items have to be
identified and discarded. Group testing can be used to aid this
process. Suppose that a group of light bulbs are connected in series,
and an electric current is passed through the circuit. If all the
bulbs are illuminated, we can be sure than none is defective, and
otherwise, we know that at least one is defective.

Since its emergence decades ago, group testing has found a large
number of surprising applications that are too numerous to be
extensively treated here. We particularly refer to applications in
molecular biology and DNA library screening (cf.\
\cites{ref:BKBB95,ref:KKM97,ref:Mac99,ref:ND00,ref:STR03,ref:WHL06,ref:WLHD08}
and the references therein), multiaccess communication
\cite{ref:Wol85}, data compression \cite{ref:HL00}, pattern matching
\cite{ref:CEPR07}, streaming algorithms \cite{ref:CM05}, software
testing \cite{ref:BG02}, compressed sensing \cite{ref:CM06}, and
secure key distribution \cite{ref:CDH07}, among others. Moreover,
entire books are specifically targeted to combinatorial group testing
\cites{ref:groupTesting,ref:DH06}.

In formal terms, the classical group testing\index{group
  testing!classical formulation} problem can be described as
follows. Suppose that we wish to ``learn'' a Boolean vector of length
$n$, namely $x=(x_1, \ldots, x_n) \in \zo^n$ using as few questions as
possible. Each question can ask for a single bit $x_i$, or more
generally, specify a group of coordinates $\cI \subseteq [n]$
($\cI \neq \emptyset$) and ask for the bit-wise ``or'' of the entries
at the specified coordinates; i.e., $\bigvee_{i \in \cI} x_i$. We will
refer to this type of questions as \emph{disjunctive
  queries}\index{disjunctive query}. Obviously, in order to be able to
uniquely identify $x$, there is in general no better way than asking
for individual bits $x_1, \ldots, x_n$ (and thus, $n$ questions),
since the number of Boolean vectors of length $n$ is $2^n$ and thus,
information theoretically, $n$ bits of information is required to
describe an arbitrary $n$-bit vector. Therefore, without imposing
further restrictions on the possible realizations of the unknown
vector, the problem becomes trivial.

Motivated by the blood sampling application that we just described,
natural restriction that is always assumed in group testing on the
unknown vector $x$ is that it is \emph{sparse}. Namely, for an integer
parameter $d > 0$, we will assume that the number of nonzero entries
of $x$ is at most $d$. We will refer to such a vector as
$d$-sparse\index{sprase vector}. The number of $d$-sparse Boolean
vectors is
\[
\sum_{i=0}^n \binom{n}{i} = 2^{\Theta(d \log(n/d))},
\]
and therefore, in principle, any $d$-sparse Boolean vector can be
described using only $O(d \log(n/d))$ bits of information, a number
that can be substantially smaller than $n$ if $d \ll n$. The precise
interpretation of the assumption ``$d \ll n$'' varies from a setting
to another. For a substantial part of this chapter, one can think of
$d = O(\sqrt{n})$. The important question in group testing that we
will address in this chapter is that, whether the
infor\-mation-theoretic limit $\Omega(d \log(n/d))$ on the number of
questions can be achieved using disjunctive queries as well.

\vspace{2mm} \noindent \textbf{Notation for this chapter: }
In this chapter we will be constantly working with Boolean vectors and
their support. The \emph{support}\index{support} of a vector $x=(x_1,
\ldots, x_n) \in \zo^n$, denoted by $\supp(x)$, is a subset of $[n]$
such that $i \in \supp(x)$ if and only if $x_i=1$. Thus the Hamming
weight\index{Hamming weight} of $x$, that we will denote by $\wgt(x)$
can be defined as $\wgt(x)=|\supp(x)|$, and a $d$-sparse vector has
the property that $\wgt(x) \leq d$.  \index{notation!{$\supp(x)$}}
\index{notation!{$\wgt(x)$}} \index{notation!{$\cM[i, j] \text{ for
      matrix $\cM$}$}} \index{notation!{$x(i) \text{ for vector
      $x$}$}} \index{notation!{$\cM|_S \text{ for matrix $\cM$}$}}

For a matrix $\cM$, we denote by $\cM[i,j]$ the entry of $\cM$ at the
$i$th row and $j$th column.  Moreover, we denote the $i$th entry of a
vector $x$ by $x(i)$ (assuming a one-to-one correspondence between the
coordinate positions of $x$ and natural numbers).  For an $m \times n$
Boolean matrix $\cM$ and $S \subseteq [n]$, we denote by $\cM|_S$ the
$m \times |S|$ submatrix of $\cM$ formed by restricting $\cM$ to the
columns picked by $S$.

For non-negative integers $e_0$ and $e_1$, we say that an ordered pair
of binary vectors $(x, y)$, each in $\zo^n$, are $(e_0, e_1)$-close
\index{notation!{$(e_0,e_1)$-close}} (or $x$ is $(e_0, e_1)$-close to
$y$) if $y$ can be obtained from $x$ by flipping at most $e_0$ bits
from $0$ to $1$ and at most $e_1$ bits from $1$ to $0$.  Hence, such
$x$ and $y$ will be $(e_0+e_1)$-close in Hamming-distance. Further,
$(x, y)$ are called $(e_0, e_1)$-far if they are not $(e_0,
e_1)$-close.  Note that if $x$ and $y$ are seen as characteristic
vectors of subsets $X$ and $Y$ of $[n]$, respectively, they are $(|Y
\sm X|,|X \sm Y|)$-close.  Furthermore, $(x,y)$ are $(e_0, e_1)$-close
if and only if $(y,x)$ are $(e_1, e_0)$-close.

\section{Measurement Designs and Disjunct Matrices}

Suppose that we wish to correctly identify a $d$-sparse vector $x \in
\zo^n$ using a reasonable amount of disjunctive queries (that we will
simply refer to as ``measurements''). In order to do so, consider
first the following simple scheme:

\begin{enumerate}
\item If $n \leq 2d$, trivially measure the vector by querying $x_1,
  \ldots, x_n$ individually.

\item Otherwise, partition the coordinates of $x$ into $\lfloor 2d
  \rfloor$ blocks of length either $\lfloor n/(2d) \rfloor$ or $\lceil
  n/(2d) \rceil$ each, and query the bitwise ``or'' of the positions
  within each block.

\item At least half of the measurement outcomes must be negative,
  since the vector $x$ is $d$-sparse. Recursively run the measurements
  over the union of those blocks that have returned positive.
\end{enumerate}

In the above procedure, each recursive call reduces the length of the
vector to half or less, which implies that the depth of the recursion
is $\log (n/2d)$.  Moreover, since $2d$ measurements are made at each
level, altogether we will have $O(d \log(n/d))$
measurements. Therefore, the simple scheme above is optimal in the
sense that it attains the infor\-mation-theoretic limit $\Omega(d
\log(n/d))$ on the number of measurements, up to constant factors.

The main problem with this scheme is that, the measurements are
\emph{adaptive} in nature. That is, the choice of the coordinate
positions defining each measurement may depend on the outcomes of the
previous measurements. However, the scheme can be seen as having
$O(\log(n/d))$ adaptive \emph{stages}. Namely, each level of the
recursion consists of $2d$ queries whose choices depend on the query
outcomes of the previous levels, but otherwise do not depend on the
outcomes of one another and can be asked in parallel.

Besides being of theoretical interest, for certain application such as
those in molecular biology, adaptive schemes can be infeasible or too
costly, and the ``amortized'' cost per test can be substantially
lowered when all queries are specified and fixed before any
measurements are performed.  Thus, a basic goal would be to design a
measurement scheme that is fully non-adaptive so that all measurements
can be performed in parallel. The trivial scheme, of course, is an
example of a non-adaptive scheme that achieves $n$ measurements. The
question is that, how close can one get to the infor\-mation-theoretic
limit $\Omega(\log(n/d))$ using a fully non-adaptive scheme? In order
to answer this question, we must study the \emph{combinatorial
  structure} of non-adaptive group testing schemes.

Non-adaptive measurements can be conveniently thought of in a matrix
form, known as the \emph{measurement matrix}\index{measurement
  matrix}, that is simply the incidence matrix of the set of
queries. Each query can be represented by a Boolean row vector of
length $n$ that is the characteristic vector of the set of indices
that participate in the query. In particular, for a query that takes a
subset $\cI \subseteq [n]$ of the coordinate positions, the
corresponding vector representation would the Boolean vector of length
$n$ that is supported on the positions picked by $\cI$.  Then the
measurement matrix is obtained by arranging the vector encodings on
the individual queries as its rows.  In particular, the measurement
matrix corresponding to a set of $m$ non-adaptive queries will be the
$m \times n$ Boolean matrix that has a $1$ at each position $(i,j)$ if
and only if the $j$th coordinate participates in the $i$th query.
Under this notation, the measurement outcomes corresponding to a
Boolean vector $x \in \zo^n$ and an $m \times n$ measurement matrix
$\cM$ is nothing but the Boolean vector of length $m$ that is equal to
the bit-wise ``or'' of those columns of $\cM$ picked by the support of
$x$. We will denote the vector of measurement outcomes by
$\cM[x]$\index{notation!$\cM[x]$}.
For example, for the measurement matrix
\newcommand{\Y}[1]{\mathbf{#1}}
\[
\cM :=
\begin{pmatrix}
\Y{0}&\Y{0}&1&\Y{1}&0&1&1&0 \\
\Y{1}&\Y{0}&1&\Y{0}&0&1&0&1 \\
\Y{0}&\Y{1}&0&\Y{1}&0&1&0&0 \\
\Y{0}&\Y{0}&0&\Y{0}&1&0&1&1 \\
\Y{1}&\Y{0}&1&\Y{0}&1&1&1&0
\end{pmatrix}
\]
and Boolean vector $x := (1,1,0,1,0,0,0,0)$, we have $\cM[x] = (1,1,1,0,1)$,
which is the bit-wise ``or'' of the columns shown in boldface.

Now suppose that the measurement matrix $\cM$ is chosen so that it can
be used to distinguish between any two $d$-sparse vectors. In
particular, for every set $S \subseteq [n]$
of indices such that $|S| \leq d-1$, $d$ being the sparsity parameter,
the $(d-1)$-sparse vector $x \in \zo^n$ supported on $S$ must be
distinguishable from the $d$-sparse vector $x' \in \zo^n$ supported on
$S \cup \{ i \}$, for any arbitrary index $i \in [n] \setminus S$.
Now observe that the Boolean function ``or'' is
\emph{monotone}. Namely, for a Boolean vector $(a_1, \ldots, a_n) \in
\zo^n$ that is monotonically less than or equal to another vector
$(b_1, \ldots, b_n) \in \zo^n$ (i.e., for every $j \in [n]$, $a_j \leq
b_j$), it must be that
\[ \bigvee_{j \in [n]} a_j \leq \bigvee_{j \in [n]} b_j. \] Therefore,
since we have chosen $x$ and $x'$ so that $\supp(x) \subseteq
\supp(x')$, we must have $\supp(\cM[x]) \subseteq
\supp(\cM[x'])$. Since by assumption, $\cM[x]$ and $\cM[x']$ must
differ in at least one position, at least one of the rows of $\cM$
must have an entry~$1$ at the $i$th row but all zeros at those
corresponding to the set $S$. This is the idea behind the classical
notion of \emph{disjunct} matrices, formally defined below (in a
slightly generalized form).

\begin{defn} \label{def:classicDisjunct} \index{disjunct!$d$-disjunct}
  \index{disjunct!$(d,e)$-disjunct} For integer parameters $d, e \geq
  0$ (respectively called the \emph{sparsity} parameter and
  \emph{noise tolerance}), a Boolean matrix is $(d,e)$-disjunct if for
  every choice of $d+1$ distinct columns $C_0, C_1,\ldots, C_d$ of the
  matrix we have
  \[
  |\supp(C_0) \setminus \cup_{i=1}^{d} \supp(C_i)| > e.
  \]
  A $(d, 0)$-disjunct matrix is simply called $d$-disjunct.
\end{defn}

In the discussion preceding the above definition we saw that the
notion of $(d-1)$-disjunct matrices is \emph{necessary} for
non-adaptive group testing, in that any non-adaptive measurement
scheme must correspond to a $(d-1)$-disjunct matrix. It turns out that
this notion is also sufficient, and thus precisely captures the
combinatorial structure needed for non-adaptive group testing.

\begin{thm} \label{thm:classicDisjunct} Suppose that $\cM$ is an $m
  \times n$ matrix that is $(d,e)$-disjunct. Then for every pair of
  distinct $d$-sparse vectors $x, x' \in \zo^n$ such that $\supp(x)
  \nsubseteq \supp(x')$, we have
  \begin{equation} \label{thm:classicDisjunct:eqn} |\supp(M[x])
    \setminus \supp(M[x'])| > e.
  \end{equation}
  Conversely, if $\cM$ is such that \eqref{thm:classicDisjunct:eqn}
  holds for every choice of $x, x'$ as above, then it must be
  $(d-1,e)$-disjunct.
\end{thm}

\begin{proof}
  For the forward direction, let $S := \supp(x')$ and $i \in \supp(x)
  \setminus \supp(x')$.  Then Definition~\ref{def:classicDisjunct}
  implies that there is a set $E \subseteq [m]$ of rows of $\cM$ such
  that $|E| > e$ and for every $j \in E$, we have $\cM[i,j]=1$ and the
  $j$th row of $\cM$ restricted to the columns in $S$ (i.e., the
  support of $x'$) entirely consists of zeros.  Thus, the measurement
  outcomes for $x'$ at positions in $E$ must be zeros while those
  measurements have a positive outcome for $x$ (since they include at
  least one coordinate, namely $i$, on the support of $x$). Therefore,
  \eqref{thm:classicDisjunct:eqn} holds.

  For the converse, consider any set $S \subseteq [n]$ of size at most
  $d-1$ and $i \in [n] \setminus S$. Consider $d$-sparse vectors $x,
  x' \in \zo^n$ such that $\supp(x') := S$ and $\supp(x) := S \cup
  \{i\}$. By assumption, there must be a set $E \subseteq [m]$ of size
  larger than $e$ such that, for every $j \in E$, we have $M[x](j) =
  1$ but $M[x'](j) = 0$.  This implies that on those rows of $\cM$
  that are picked by $E$, the $i$th entry must be one while those
  corresponding to $S$ must be zeros. Therefore, $\cM$ is
  $(d,e)$-disjunct.
\end{proof}

From the above theorem we know that the measurement outcomes
corresponding to distinct $d$-sparse vectors differ from one another
in more than $e$ positions provided that the measurement matrix is
$(d,e)$-disjunct. When $e > 0$, this would allow for
distinguishability of sparse vectors even in presence of \emph{noise}.
Namely, even if up to $\lfloor e/2 \rfloor$ of the measurements are
allowed to be incorrect, it would still possible to uniquely
reconstruct the vector being measured.  For this reason, we have
called the parameter $e$ the ``noise tolerance''.

\subsection{Reconstruction}

So far we have focused on \emph{combinatorial} distinguishability of
sparse vectors.  However, for applications unique distinguishability
is by itself not sufficient and it is important to have efficient
``decoding'' algorithms to reconstruct the vector being measured.

Fortunately, monotonicity of the ``or'' function substantially
simplifies the decoding problem. In particular, if two Boolean vectors
$x, x'$ such that the support of $x$ is not entirely contained in that
of $x'$ are distinguishable by a measurement matrix, adding new
elements to the support of $x$ will never make it ``less
disginguishable'' from $x'$. Moreover, observe that the proof of
Theorem~\ref{thm:classicDisjunct} never uses sparsity of the vector
$x$.  Therefore we see that, $(d,e)$-disjunct matrices are not only
able to distinguish between $d$-sparse vectors, but moreover, the only
Boolean vector (be it sparse or not) that may reproduce the
measurement outcomes resulting from a $d$-sparse vector $x \in \zo^n$
is $x$ itself. Thus, given a vector of measurement outcomes, in order
to reconstruct the sparse vector being measured it suffices to produce
\emph{any} vector that is consistent with the measurement outcomes.
This observation leads us to the following simple decoding algorithm,
that we will call the \emph{distance decoder}: \index{distance
  decoder}

\begin{enumerate}
\item Given a measurement outcome $\tilde{y} \in \zo^m$, identify the
  set $S_\ty \subseteq [n]$ of the column indices of the measurement
  matrix $\cM$ such that each $i \in [n]$ is in $S_\ty$ if and only if
  the $i$th column of $\cM$, denoted by $c_i$, satisfies
  \[
  |\supp(c_i) \setminus \supp(\tilde{y})| \leq \floor{e/2}.
  \]

\item The reconstruction outcome $\tilde{x} \in \zo^n$ is the Boolean
  vector supported on $S_\ty$.
\end{enumerate}

\begin{lem}
  Let $x \in \zo^n$ be $d$-sparse and $y := \cM[x]$, where the
  measurement matrix $\cM$ is $(d,e)$-disjunct. Suppose that a
  measurement outcome $\ty$ that has Hamming distance at most
  $\floor{e/2}$ with $y$ is given to the distance decoder.  Then the
  outcome $\tx$ of the decoder is equal to $x$.
\end{lem}

\begin{proof}
  Since the distance decoder allows for a ``mismatch'' of size up to
  $e$ for the columns picked by the set $S_\ty$, we surely know that
  $\supp(x) \subseteq S_\ty = \supp(\tx)$.  Now suppose that there is
  an index $i \in [n]$ such that $i \in S_\ty$ but $i \notin
  \supp(x)$.  Since $\cM$ is $(d,e)$-disjunct, we know that for the
  $i$th column $c_i$ we have
  \[
  |\supp(c_i) \setminus \supp(y)| > e.
  \]
  On the other hand, since $i \in S_\ty$, it must be that
  \[
  |\supp(c_i) \setminus \supp(\ty)| \leq \floor{e/2},
  \]
  and moreover, by assumption we have that
  \[
  |\supp(\ty) \setminus \supp(y)| \leq \floor{e/2}.
  \]
  This is a contradiction. Therefore we must have $S_\ty \subseteq
  \supp(x)$, implying that $x = \tx$.
\end{proof}

\subsection{Bounds on Disjunct Matrices}

So far we have seen that the notion of disjunct matrices is all we
need for non-adaptive group testing. But how small can the number of
rows of such matrices be? Equivalently, what is the smallest number of
measurements required by a non-adaptive group testing scheme that can
correctly identify the support of $d$-sparse vectors?

\subsubsection{Upper and Lower Bounds}

In the following, we use the probabilistic method to show that, a
randomly constructed matrix is with overwhelming probability disjunct,
and thus obtain an upperbound on the number of the rows of disjunct
matrices.

\begin{thm} \label{thm:classicDisjunct:random} \index{disjunct!upper
    bound} Let $p \in [0,1)$ be an arbitrary real parameter, and $d,
  n$ be integer parameters such that $d < n$. Consider a random $m
  \times n$ Boolean matrix $\cM$ such that each entry of $\cM$ is,
  independently, chosen to be $1$ with probability $q := 1/d$. Then
  there is an $m_0 = O(d^2 \log(n/d)/(1-p)^2)$ and $e = \Omega(p m/d)$
  such that $\cM$ is $(d,e)$-disjunct with probability $1-o(1)$
  provided that $m \geq m_0$.
\end{thm}

\begin{proof}
  Consider any set $S$ of $d$ columns of $\cM$, and any column outside
  those, say the $i$th column where $i \notin S$. First we upper bound
  the probability of a \emph{failure} for this choice of $S$ and $i$,
  i.e., the probability that the number of the positions at the $i$th
  column corresponding to which all the columns in $S$ have zeros is
  at most $e$.  Clearly if this event happens the $(d, e)$-disjunct
  property of $\cM$ would be violated.  On the other hand, if for no
  choice of $S$ and $i$ a failure happens the matrix would be indeed
  $(d, e)$-disjunct.

  Now we compute the failure probability $p_f$ for a fixed $S$ and
  $i$.  A row is \emph{good} if at that row the $i$th column has a $1$
  but all the columns in $S$ have zeros. For a particular row, the
  probability that the row is good is $q(1-q)^d$. Then failure
  corresponds to the event that the number of good rows is at most
  $e$. The distribution of the number of good rows is binomial with
  mean $\mu = q(1-q)^d m$. Choose $e := pmq(1-q)^d = \Omega(pm/d)$.
  By a Chernoff bound, the failure probability is at most
  \begin{eqnarray*}
    p_f &\leq& \exp( -(\mu-e)^2 / (2\mu)) \\
    &\leq& \exp(-mq (1-p)^2/6)
  \end{eqnarray*}
  where the second inequality is due to the fact that $(1-q)^d =
  (1-1/d)^d$ is always between $1/3$ and $1/2$.

  Now if we apply a union bound over all possible choices of $S$ and
  $i$, the probability of coming up with a bad choice of $\cM$ would
  be at most \[ n \binom{n}{d} \exp(-mq (1-p)^2/6). \] This
  probability vanishes so long as $m \geq m_0$ for some $m_0 = O(d^2
  \log (n/d) / (1-p)^2)$.
\end{proof}

The above result shows, in particular, that $d$-disjunct matrices with
$n$ columns and $O(d^2 \log (n/d))$ rows exist. This is by off from
the infor\-mation-theoretic barrier $O(d \log(n/d))$ by a multiplicative
factor $O(d)$, which raises the question, whether better disjunct
matrices can be found.  In the literature of group testing,
combinatorial lower bounds on the number of rows of disjunct matrices
exist, which show that the above upper bound is almost the best one
can hope for. In particular, D'yachkov and Rykov \cite{ref:DR83} have
shown that the number of rows of any $d$-disjunct matrices has to be
$\Omega(d^2 \log_d n)$. \index{disjunct!lower bound} Several other
concrete lower bounds on the size of disjunct matrices is known, which
are all asymptotically equivalent (e.g.,
\cites{ref:Rus94,ref:Fue96}). Moreover, for a nonzero noise tolerance
$e$, the lower bounds can be extended to $\Omega(d^2 \log_d n + ed)$.

\subsubsection{The Fixed-Input Case}

The probabilistic construction of disjunct matrices presented in
Theorem~\ref{thm:classicDisjunct:random} almost surely produces a
disjunct matrices using $O(d^2 \log(n/d))$ measurements.  Obviously,
due to almost-matching lower bounds, by lowering the number of the
measurement the disjunctness property cannot be assured
anymore. However, the randomized nature of the designs can be used to
our benefit to show that, using merely $O(d \log n)$ measurements
(almost matching the infor\-mation-theoretic lower bound) it is possible
(with overwhelming probability) to distinguish a ``fixed'' $d$-sparse
vector from any other (not necessarily sparse) vector. More precisely
we have the following result, whose proof is quite similar to that of
Theorem~\ref{thm:classicDisjunct:random}.

\begin{thm}
  \label{thm:classicDisjunct:fixedInput}
  Let $p \in [0,1)$ be an arbitrary real parameter, $d, n$ be integer
  parameters such that $d < n$, and $x \in \zo^n$ be a fixed
  $d$-sparse vector. Consider a random $m \times n$ Boolean matrix
  $\cM$ such that each entry of $\cM$ is, independently, chosen to be
  $1$ with probability $q := 1/d$. Then there is an $m_0 = O(d (\log
  n)/(1-p)^2)$ and $e = \Omega(p m/d)$ such that, provided that $m
  \geq m_0$, with probability $1-o(1)$ the following holds: For every
  $y \in \zo^n$, $y \neq x$, the Hamming distance between the outcomes
  $M[y]$ and $M[x]$ is greater than $e$.
\end{thm}

\begin{proof}
  We follow essentially the same argument as the proof of
  Theorem~\ref{thm:classicDisjunct:random}, but will need a weaker
  union bound at the end. Call a column $i$ of $\cM$ \emph{good} if
  there are more than $e$ rows of $\cM$ at which the $i$th column has
  a $1$ but those on the support of $x$ (excluding the $i$th column)
  have zeros. Now we can follow the argument in the proof of
  Theorem~\ref{thm:classicDisjunct:random} to show that under the
  conditions of the theorem, with probability $1-o(1)$, all columns of
  $\cM$ are good (the only difference is that, the last union bound
  will enumerate a set of $n$ possibilities rather than
  $(n-1)\binom{n}{d}$).

  Now suppose that for the particular outcome of $\cM$ all columns are
  good, and take any $y \in \zo^n$, $y \neq x$. One of the following
  cases must be true, and in either case, we show that $\cM[x]$ and
  $\cM[y]$ are different at more than $e$ positions:
  \begin{enumerate}
  \item There is an $i \in \supp(y) \setminus \supp(x)$: Since the
    $i$th column is good, we know that for more than $e$ rows of
    $\cM$, the entry at the $i$th column is $1$ while those at
    $\supp(x)$ are all zeros. This implies that at positions
    corresponding to such rows, $\cM[y]$ must be $1$ but $\cM[x]$ must
    be zero.

  \item We have $\supp(y) \subseteq \supp(x)$: In this case, take any
    $i \in \supp(x) \setminus \supp(y)$, and again use the fact that
    the $i$th column is good to conclude that at more than $e$
    positions the outcome $\cM[y]$ must be zero but $\cM[x]$ must
    be~$1$.
  \end{enumerate}
\end{proof}

As a corollary, the above theorem shows that, with overwhelming
probability, once we fix the outcome of the random matrix $\cM$
constructed by the theorem, the matrix $\cM$ will be able to
distinguish between \emph{most} $d$-sparse vectors even in presence of
any up to $\floor{e/2}$ incorrect measurement outcomes. In particular,
we get an average-case result, that there is a \emph{fixed}
measurement scheme with only $O(d \log n)$ measurements using which it
is possible to uniquely reconstruct a randomly chosen $d$-sparse
vector (e.g., under the uniform distribution) with overwhelming
probability over the distribution from which the sparse vector is
drawn.

\subsubsection{Sparsity of the Measurements}

The probabilistic construction of
Theorem~\ref{thm:classicDisjunct:random} results in a rather sparse
matrix, namely, one with density $q = 1/d$ that decays with the
sparsity parameter $d$. Below we show that sparsity is a necessary
condition for the probabilistic construction to work at an optimal level on the
number of measurements:

\begin{lem} \label{lem:sparsityMeasurements} Let $\cM$ be an $m \times
  n$ Boolean random matrix, where $m = O(d^2 \log n)$ for an integer
  $d > 0$, which is constructed by setting each entry independently to
  $1$ with probability $q$. Then either $q = O(\log d/d)$ or otherwise
  the probability that $\cM$ is $(d,e)$-disjunct (for any $e \geq 0$)
  approaches to zero as $n$ grows.
\end{lem}

\begin{proof}
  Suppose that $\cM$ is an $m \times n$ matrix that is
  $(d,e)$-disjunct.  Observe that, for any integer $t \in (0,d)$, if
  we remove any $t$ columns of $\cM$ and all the rows on the support
  of those columns, the matrix must remain $(d-t, e)$-disjunct. This
  is because any counterexample for the modified matrix being $(d-t,
  e)$-disjunct can be extended to a counterexample for $\cM$ being
  $(d,e)$-disjunct by adding the removed columns to its support.

  Now consider any $t$ columns of $\cM$, and denote by $m_0$ the
  number of rows of $\cM$ at which the entries corresponding to the
  chosen columns are all zeros. The expected value of $m_0$ is
  $(1-q)^t m$. Moreover, for any constant $\delta > 0$ we have
  \begin{equation} \label{eqn:chernoffDisj} \Pr[m_0 > (1+\delta)
    (1-q)^t m] \leq \exp( -\delta^2 (1-q)^t m/4 )
  \end{equation}
  by a Chernoff bound.

  Let $t_0$ be the largest integer for which \[(1+\delta) (1-q)^{t_0} m
  \geq \log n.\]  If $t_0 < d-1$, we let $t := 1+t_0$ above, and this
  makes the right hand side of \eqref{eqn:chernoffDisj} upper bounded
  by $o(1)$.  So with probability $1-o(1)$, the chosen $t$ columns of
  $\cM$ will keep $m_0$ at most $(1+\delta)(1-q)^t m$, and removing
  those columns and $m_0$ rows on their union leaves the matrix
  $(d-t_0-1, e)$-disjunct, which obviously requires at least $\log n$
  rows (as even a $(1, 0)$-disjunct matrix needs so many
  rows). Therefore, we must have
  \[
  (1+\delta)(1-q)^t m \geq \log n
  \]
  or otherwise (with overwhelming probability) $\cM$ will not be
  $(d,e)$-disjunct.  But the latter inequality is not satisfied by the
  assumption on $t_0$.  So if $t_0 < d-1$, little chance remains for
  $\cM$ to be $(d,e)$-disjunct.

  Now consider the case $t_0 \geq
  d-1$. Thus, by the choice of $t_0$, we must have
  \[
  (1+\delta)(1-q)^{d-1} m \geq \log n.
  \]
  The above inequality implies that we must
  have
  \[
  q \leq \frac{\log(m(1+\delta)/\log n)}{{d-1}},
  \]
  which, for $m = O(d^2 \log n)$ gives $q = O(\log d/d)$.
\end{proof}

\section[Noise resilient schemes]{Noise resilient schemes and
  approximate reconstruction}

So far, we have introduced the notion of $(d, e)$-disjunct matrices
that can be used in non-adaptive group testing schemes to identify
$d$-sparse vectors up to a number of measurement errors depending on
the parameter $e$.  However, as the existing lower bounds suggest, the
number of rows of such matrices cannot reach to the infor\-mation-theoretic
optimum $O(d \log(n/d))$ and moreover, the noise tolerance
$e$ can be at most a factor $1/d$ of the number of measurements. This
motivates two natural questions:

\begin{enumerate}
\item Can the number of measurements be lowered at the cost of causing
  a slight amount of ``confusion''? We know, by
  Theorem~\ref{thm:classicDisjunct:fixedInput} that, it is possible to
  identify sparse vectors on average using only $O(d \log n)$
  measurements. But can something be said in the \emph{worst case} ?

\item What can be said if the amount of possible errors can be
  substantially high; e.g., when a constant fraction of the
  measurements can produce false outcomes?
\end{enumerate}

In order to answer the above questions, in this section we introduce a
notion of measurement schemes that can be ``more flexible'' than that
of disjunct matrices, and aims to study the trade-off between the
amount of errors expected on the measurements versus the ambiguity of
the reconstruction. More formally we define the following notion.

\begin{defn} \label{def:matrix} \index{resilient matrix} Let $m, n, d,
  e_0, e_1, e'_0, e'_1$ be integers.  An $m \times n$ measurement
  matrix $A$ is called $(e_0, e_1, e'_0, e'_1)$-resilient for
  $d$-sparse vectors if, for every $y \in \zo^m$ there exists $z \in
  \zo^n$ (called a \emph{valid decoding of $y$}) such that for every
  $x \in \zo^n$, whenever $(x, z)$ are $(e'_0, e'_1)$-far, $(A[x], y)$
  are $(e_0, e_1)$-far\footnote{ In particular this means that for
    every $x, x' \in \zo^n$, if $(A[x], A[x'])$ are $(e_0,
    e_1)$-close, then $x$ and $x'$ must be $(e'_0+e'_1,
    e'_0+e'_1)$-close. }.

  The matrix $A$ is called explicit if it can be computed in
  polynomial time in its size, and \emph{fully explicit} if each entry
  of the matrix can be computed in time $\poly(m, \log n)$.
\end{defn}

Intuitively, the definition states that two measurements are allowed
to be confused only if they are produced from close vectors.  The
parameters $e_0$ and $e'_0$ correspond to amount of tolerable
\emph{false positives} on the measurement outcomes and reconstructed
vector, where by false positive\index{false positive} we mean an error
caused by mistaking a $0$ for $1$.  Similarly, $e_1$ and $e'_1$ define
the amount of tolerable \emph{false negatives}\index{false negative}
on both sides, where a false negative occurs when a bit that actually
must be $1$ is flipped to $0$.

In particular, an $(e_0, e_1, e'_0, e'_1)$-resilient matrix gives a
group testing scheme that reconstructs the sparse vector up to $e'_0$
false positives and $e'_1$ false negatives even in the presence of
$e_0$ false positives and $e_1$ false negatives in the measurement
outcome.  Under this notation, unique (exact) decoding would be
possible using an $(e_0, e_1, 0, 0)$-resilient matrix if the amount of
measurement errors is bounded by at most $e_0$ false positives and
$e_1$ false negatives.  However, when $e'_0+e'_1$ is positive,
decoding may require a bounded amount of ambiguity, namely, up to
$e'_0$ false positives and $e'_1$ false negatives in the decoded
sequence.

Observe that the special case of $(0,0,0,0)$-resilient matrices
corresponds to the classical notion of $d$-disjunct matrices, while a
$(d,e)$-disjunct matrix would give a $(\floor{e/2}, \floor{e/2}, 0,
0)$-resilient matrix for $d$-sparse vectors.

Definition~\ref{def:matrix} is in fact reminiscent of
\emph{list-decoding} in error-correcting codes, but with the stronger
requirement that the list of decoding possibilities must consist of
vectors that are close to one another.

\subsection{Negative Results}

In coding theory, it is possible to construct codes that can tolerate
up to a constant fraction of adversarially chosen errors and still
guarantee unique decoding.  Hence it is natural to wonder whether a
similar possibility exists in group testing, namely, whether there is
a measurement matrix that is robust against a constant fraction of
adversarial errors and still recovers the measured vector exactly. We
already have mentioned that this is in general not possible, since any
$(d,e)$-disjunct matrix (a notion that is necessary for this task)
requires at least $de$ rows, and thus the fraction of tolerable errors
by disjunct matrices cannot be above $1/d$.  Below we extend this
result to the more ``asymmetric'' notion of resilient matrices, and
show that the fraction of tolerable false positives and false
negatives must be both below $1/d$.

\begin{lem} \label{lem:distance} Suppose that an $m \times n$
  measurement matrix $\cM$ is $(e_0, e_1, e'_0, e'_1)$-resilient for
  $d$-sparse vectors.  Then $(\max\{e_0,e_1\}+1)/(e'_0+e'_1+1) \leq
  m/d$.
\end{lem}

\begin{proof}
  We use similar arguments as those used in \cites{ref:BKS06,ref:GG08}
  in the context of black-box hardness amplification in $\mathsf{NP}$:
  Define a partial ordering $\prec$ between binary vectors using
  bit-wise comparisons (with $0 < 1$).  Let $t := d/(e'_0+e'_1+1)$ be
  an integer\footnote{For the sake of simplicity in this presentation
    we ignore the fact that certain fractions might in general give
    non-integer values. However, it should be clear that this will
    cause no loss of generality.}, and consider any monotonically
  increasing sequence of vectors $x_0 \prec \cdots \prec x_t$ in
  $\zo^n$ where $x_i$ has weight $i(e'_0+e'_1+1)$. Thus, $x_0$ and
  $x_t$ will have weights zero and $d$, respectively. Note that we
  must also have $\cM[x_0] \prec \cdots \prec \cM[x_t]$ due to
  monotonicity of the ``or'' function.

  A fact that is directly deduced from Definition~\ref{def:matrix} is
  that, for every $x, x' \in \zo^n$, if $(\cM[x], \cM[x'])$ are $(e_0,
  e_1)$-close, then $x$ and $x'$ must be $(e'_0+e'_1,
  e'_0+e'_1)$-close. This can be seen by setting $y := \cM[x']$ in the
  definition, for which there exists a valid decoding $z \in
  \zo^n$. As $(\cM[x], y)$ are $(e_0, e_1)$-close, the definition
  implies that $(x, z)$ must be $(e'_0, e'_1)$-close.  Moreover,
  $(\cM[x'], y)$ are $(0, 0)$-close and thus, $(e_0, e_1)$-close,
  which implies that $(z, x')$ must be $(e'_1, e'_0)$-close. Thus by
  the triangle inequality, $(x, x')$ must be $(e'_0+e'_1,
  e'_0+e'_1)$-close.

  Now, observe that for all $i$, $(x_i, x_{i+1})$ are $(e'_0+e'_1,
  e'_0+e'_1)$-far, and hence, their encodings must be $(e_0,
  e_1)$-far, by the fact we just mentioned.  In particular this
  implies that $\cM[x_t]$ must have weight at least $t(e_0+1)$, which
  must be trivially upper bounded by $m$. Hence it follows that
  $(e_0+1)/(e'_0+e'_1+1) \leq m/d$. Similarly we can also show that
  $(e_1+1)/(e'_0+e'_1+1) \leq m/d$.
\end{proof}

As shown by the lemma above, tolerance of a measurement matrix against
a constant fraction of errors would make an ambiguity of order
$\Omega(d)$ in the decoding inevitable, irrespective of the number of
measurements.  For most applications this might be an unsatisfactory
situation, as even a close estimate of the set of positives might not
reveal whether any particular individual is defective or not, and in
certain scenarios (such as an epidemic disease or industrial quality
assurance) it is unacceptable to miss any defective individuals.  This
motivates us to focus on approximate reconstructions with
\emph{one-sided} error.  Namely, we will require the support of the
reconstruction $\hat{x}$ to always contain the support of the original
vector $x$ being measured, and be possibly larger by up to $O(d)$
positions. It can be argued that, for most applications, such a scheme
is as good as exact reconstruction, as it allows one to significantly
narrow-down the set of defectives to up to $O(d)$ \emph{candidate
  positives}. In particular, as observed in \cite{ref:Kni95}, one can
use a \emph{second stage} if necessary and individually test the
resulting set of candidates, using more reliable measurements, to
identify the exact set of positives.  In the literature, such schemes
are known as \emph{trivial two-stage} schemes\index{group
  testing!two-stage schemes}.

The trade-off given by the following lemma only focuses on false
negatives and is thus useful for trivial two-stage schemes:

\begin{lem} \label{lem:falseNeg} Suppose that an $m \times n$
  measurement matrix $M$ is $(e_0, e_1, e'_0, e'_1)$-resilient for
  $d$-sparse vectors.  Then for every $\eps > 0$, either \[ e_1 <
  \frac{(e'_1+1)m}{\eps d}\] or \[e'_0 \geq
  \frac{(1-\eps)(n-d+1)}{(e'_1+1)^2}.\]
\end{lem}

\begin{proof}
  Let $x \in \zo^n$ be chosen uniformly at random among vectors of
  weight $d$.  Randomly flip $e'_1+1$ of the bits on the support of
  $x$ to $0$, and denote the resulting vector by $x'$. Using the
  partial ordering $\prec$ in the proof of the last lemma, it is
  obvious that $x' \prec x$, and hence, $\cM[x'] \prec \cM[x]$.  Let
  $b$ denote any disjunction of a number of coordinates in $x$ and
  $b'$ the same disjunction in $x'$. We must have
  \[ \Pr[b' = 0 | b = 1] \leq \frac{e'_1+1}{d}, \] as for $b$ to be
  $1$ at least one of the variables on the support of $x$ must be
  present in the disjunction and one particular such variable must
  necessarily be flipped to bring the value of $b'$ down to zero.
  Using this, the expected Hamming distance between $\cM[x]$ and
  $\cM[x']$ can be bounded as follows:
  \begin{eqnarray*}
    \Ex[ \dist(\cM[x], \cM[x']) ] = \sum_{i \in [m]} \mathds{1}( \cM[x]_i = 1 \land \cM[x']_i = 0 ) \leq \frac{e'_1+1}{d} \cdot m,
  \end{eqnarray*}
  where the expectation is over the randomness of $x$ and the bit
  flips, $\dist(\cdot, \cdot)$ denotes the Hamming distance between
  two vectors, and $\mathds{1}(\cdot)$ denotes an indicator predicate.

  Fix a particular choice of $x'$ that keeps the expectation at most
  $(e'_1+1)m/d$. Now the randomness is over the possibilities of $x$,
  that is, flipping up to $e'_1+1$ zero coordinates of $x'$ randomly.
  Denote by $\cX$ the set of possibilities of $x$ for which $\cM[x]$
  and $\cM[x']$ are $\frac{(e'_1+1)m}{\eps d}$-close, and by $\cS$ the
  set of all vectors that are monotonically larger than $x'$ and are
  $(e'_1+1)$-close to it.  Obviously, $\cX \subseteq \cS$, and, by
  Markov's inequality, we know that $|\cX| \geq (1-\eps) |\cS|$.

  Let $z$ be any valid decoding of $\cM[x']$, Thus, $(x', z)$ must be
  $(e'_0, e'_1)$-close.  Now assume that $e_1 \geq
  \frac{(e'_1+1)m}{\eps d}$ and consider any $x \in \cX$.  Hence,
  $(\cM[x], \cM[x'])$ are $(e_0, e_1)$-close and $(x, z)$ must be
  $(e'_0, e'_1)$-close by Definition~\ref{def:matrix}.  Regard $x, x',
  z$ as the characteristic vectors of sets $X, X', Z \subseteq [n]$,
  respectively, where $X' \subseteq X$. We know that $|X \sm Z| \leq
  e'_1$ and $|X \sm X'| = e'_1+1$.  Therefore,
  \begin{equation} \label{eqn:sets} |(X \sm X') \cap Z| = |X \sm X'| -
    |X\sm Z| + |X' \sm Z| > 0,
  \end{equation}
  and $z$ must take at least one nonzero coordinate from $\supp(x) \sm
  \supp(x')$.

  Now we construct an $(e'_1+1)$-hypergraph\footnote{See
    Appendix~\ref{app:proofs} for definitions.} $H$ as follows: The
  vertex set is $[n] \sm \supp(x')$, and for every $x \in \cX$, we put
  a hyperedge containing $\supp(x) \sm \supp(x')$. The density of this
  hypergraph is at least $1-\eps$, by the fact that $|\cX| \geq
  (1-\eps) \cS$.  Now Lemma~\ref{lem:denseGraph} implies that $H$ has
  a matching of size at least
  \[
  t := \frac{(1-\eps)(n-d+1)}{(e'_1+1)^2}.
  \]
  As by \eqref{eqn:sets}, $\supp(z)$ must contain at least one element
  from the vertices in each hyperedge of this matching, we conclude
  that $|\supp(z) \sm \supp(x')| \geq t$, and that $e'_0 \geq t$.
\end{proof}

The lemma above shows that if one is willing to keep the number $e'_1$
of false negatives in the reconstruction at the zero level (or bounded
by a constant), only an up to $O(1/d)$ fraction of false negatives in
the measurements can be tolerated (regardless of the number of
measurements), unless the number $e'_0$ of false positives in the
reconstruction grows to an enormous amount (namely, $\Omega(n)$ when
$n-d = \Omega(n)$) which is certainly undesirable.

Recall that exact reconstruction of $d$-sparse vectors of length $n$,
even in a noise-free setting, requires at least $\Omega(d^2 \log_d n)$
non-adaptive measurements.  However, it turns out that there is no
such restriction when an approximate reconstruction is sought for,
except for the following bound which can be shown using simple
counting and holds for adaptive noiseless schemes as well:

\begin{lem} \label{lem:lowerbound} Let $\cM$ be an $m \times n$
  measurement matrix that is $(0, 0, e'_0, e'_1)$-resilient for
  $d$-sparse vectors. Then \[m \geq d \log (n/d) - d - e'_0 - O(e'_1
  \log ((n-d-e'_0)/e'_1)),\] where the last term is defined to be zero
  for $e'_1 = 0$.
\end{lem}

\begin{proof}
  The proof is a simple counting argument.  For integers $a > b > 0$,
  we use the notation $V(a, b)$ for the volume of a Hamming ball of
  radius $b$ in $\zo^a$. It is given by
  \[
  V(a, b) = \sum_{i = 0}^b \binom{a}{i} \leq 2^{a h(b/a)},
  \]
  where $h(\cdot)$ is the binary entropy function defined as
  \[
  h(x) := -x\log_2(x) -(1-x)\log_2(1-x),
  \]
  and thus
  \[
  \log V(a, b) \leq b \log \frac{a}{b} + (a-b) \log \frac{a}{a-b} =
  \mathrm{\Theta}(b \log (a/b)).
  \]
  Also, denote by $V'(a, b, e_0, e_1)$ the number of vectors in
  $\zo^a$ that are $(e_0, e_1)$-close to a fixed $b$-sparse
  vector. Obviously, $V'(a, b, e_0, e_1) \leq V(b, e_0) V(a-b, e_1)$.
  Now consider any (without loss of generality, deterministic) reconstruction algorithm $D$
  and let $X$ denote the set of all vectors in $\zo^n$ that it returns
  for some noiseless encoding; that is,
  \[
  X := \{ x \in \zo^n \mid \exists y \in \cB, x = D(A[y]) \},
  \]
  where $\cB$ is the set of $d$-sparse vectors in $\zo^n$.  Notice
  that all vectors in $X$ must be $(d+e'_0)$-sparse, as they have to
  be close to the corresponding ``correct'' decoding.  For each vector
  $x \in X$ and $y \in \cB$, we say that $x$ is \emph{matching} to $y$
  if $(y, x)$ are $(e'_0, e'_1)$-close. A vector $x \in X$ can be
  matching to at most $v := V'(n, d+e'_0, e'_0, e'_1)$ vectors in
  $\cB$, and we upper bound $\log v$ as follows:
  \[
  \log v \leq \log V(n-d-e'_0, e'_1) + \log V(d+e'_0, e'_0) = O(e'_1
  \log ((n-d-e'_0)/e'_1)) + d + e'_0,
  \]
  where the term inside $O(\cdot)$ is interpreted as zero when $e'_1 =
  0$.  Moreover, every $y \in \cB$ must have at least one matching
  vector in $X$, namely, $D(\cM[y])$. This means that $|X| \geq
  |\cB|/v$, and that
  \[
  \log |X| \geq \log |\cB| - \log v \geq d \log (n/d) - d - e'_0 -
  O(e'_1 \log ((n-d-e'_0)/e'_1)).
  \]
  Finally, we observe that the number of measurements has to be at
  least $|X|$ to enable $D$ to output all the vectors in $X$.
\end{proof}

According to the lemma, even in the noiseless scenario, any
reconstruction method that returns an approximation of the sparse
vector up to $e'_0=O(d)$ false positives and without false negatives
will require $\Omega(d \log (n/d))$ measurements. As we will show in
the next section, an upper bound of $O(d \log n)$ is in fact
attainable even in a highly noisy setting using only non-adaptive
measurements.  This in particular implies an asymptotically optimal
trivial two-stage group testing scheme.

\subsection{A Noise-Resilient Construction} \label{sec:nrConstr}

In this section we introduce our general construction and design
measurement matrices for testing $d$-sparse vectors in $\zo^n$.  The
matrices can be seen as adjacency matrices of certain unbalanced
bipartite graphs constructed from good randomness condensers or
extractors. The main technique that we use to show the desired
properties is the \emph{list-decoding view} of randomness condensers,
extractors, and expanders, developed over the recent years starting
from the work of Ta-Shma and Zuckerman on \emph{extractor codes}
\cite{ref:TZ04} and followed by Guruswami, Umans, Vadhan \cite{ref:GUV09}
and Vadhan \cite{ref:Vad10}.

\subsubsection{Construction from Condensers}

We start by introducing the terms and tools that we will use in our
construction and its analysis.

\begin{defn} (mixtures, agreement, and agreement list) \label{def:mix}
  \index{mixture} \index{agreement list} Let $\Sigma$ be a finite
  set. A \emph{mixture} over $\Sigma^n$ is an $n$-tuple $S := (S_1,
  \ldots, S_n)$ such that every $S_i$, $i \in [n]$, is a nonempty
  subset of $\Sigma$.

  The \emph{agreement} of $w := (w_1, \ldots w_n) \in \Sigma^n$ with
  $S$, denoted by $\agr(w, S)$, is the quantity \[ \frac{1}{n} |\{i
  \in [n]\colon w_i \in S_i\}|.  \] Moreover, we define the quantity
  \[ \wgt(S) := \sum_{i \in [n]} |S_i| \] and \[\rho(S) := \wgt(S)/(n
  |\Sigma|),\] where the latter is the expected agreement of a random
  vector with $S$.

  For example, consider a mixture $S := (S_1,\ldots, S_8)$ over $[4]^8$
  where $S_1 :=\emptyset, S_2 :=\{1,3\}, S_3 :=\{1,2\}, S_4 :=\{1,4\},
  S_5 :=\{1\}, S_6 :=\{3\}, S_7 :=\{4\}, S_8 :=\{1,2,3,4\}$. For this
  example, we have
  \[
   \agr((1,3,2,3,4,3,4,4),S) = 5/8,
  \]
  and $\rho(S) = 13/32$.

  For a code $\cC \subseteq \Sigma^n$ and $\alpha \in (0, 1]$, the
  \emph{$\alpha$-agreement list} of $\cC$ with respect to $S$, denoted
  by $\List_\cC(S, \alpha)$, is defined as the set\footnote{When
    $\alpha=1$, we consider codewords with full agreement with the
    mixture.}  \[ \List_\cC(S, \alpha) := \{ c \in \cC\colon \agr(c,
  S) > \alpha \}.  \]
\end{defn}

\begin{defn} (induced code) \index{induced code} Let $f\colon \Gamma
  \times \Omega \to \Sigma$ be a function mapping a finite set $\Gamma
  \times \Omega$ to a finite set $\Sigma$. For $x \in \Gamma$, we use
  the shorthand $f(x)$ to denote the vector $y := (y_i)_{i \in
    \Omega}$, $y_i := f(x, i)$, whose coordinates are indexed by the
  elements of $\Omega$ in a fixed order.  The \emph{code induced by
    $f$}, denoted by $\cC(f)$ is the set \[ \{ f(x)\colon x \in \Gamma
  \}. \] The induced code has a natural encoding function given by $x
  \mapsto f(x)$.
\end{defn}

\begin{defn} (codeword graph) \index{codeword graph} Let $\cC
  \subseteq \Sigma^n$, $|\Sigma| = q$, be a $q$-ary code. The
  \emph{codeword graph} of $\cC$ is a bipartite graph with left vertex
  set $\cC$ and right vertex set $n \times \Sigma$, such that for
  every $x = (x_1, \ldots, x_n) \in \cC$, there is an edge between $x$
  on the left and $(1, x_1), \ldots, (n, x_n)$ on the right. The
  \emph{adjacency matrix} of the codeword graph is an $n |\Sigma|
  \times |\cC|$ binary matrix whose $(i, j)$th entry is $1$ if and only if there
  is an edge between the $i$th right vertex and the $j$th left vertex.
\end{defn}

\begin{figure}
\newcommand{\ro}[1]{\mbox{\begin{rotate}{90}$#1$\end{rotate}}\text{\hspace{0.5em}}} 
\begin{centering}
\begin{tabular}{rl}
\vspace{3em} &
\multirow{4}{*}{\vspace{-4cm}\hspace{-2mm}\mbox{\includegraphics[width=5cm]{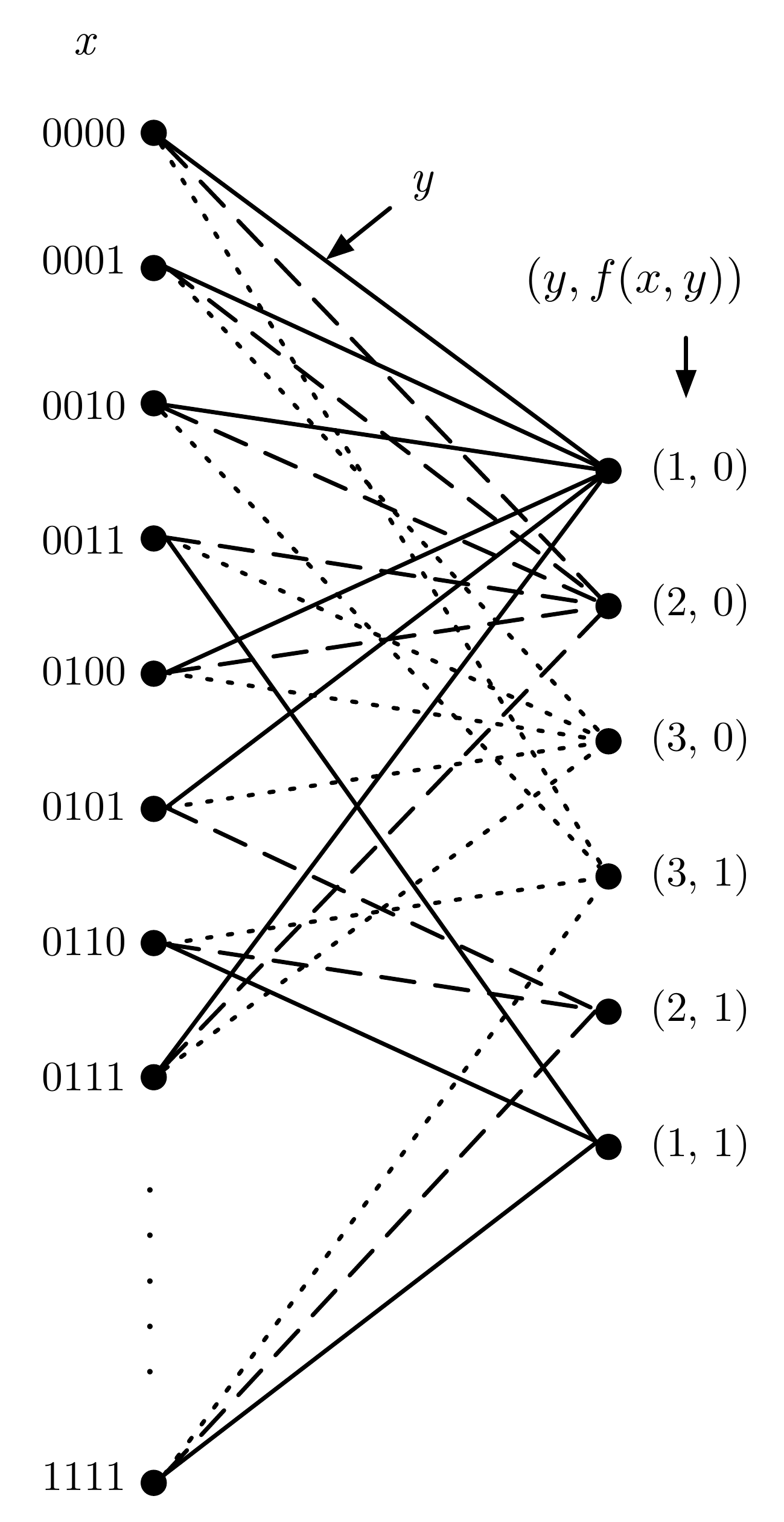}}}
\\
\hspace{-5mm} \mbox{$
\begin{array}{c}
\begin{array}{ll}
\begin{matrix} \hspace{1.4em}
\ro{x=0000}&\ro{x=0001}&\ro{x=0010}&\ro{x=0011}&\ro{x=0100}&\ro{x=0101}&\ro{x=0110}&\ro{x=0111}&\text{\hspace{1em}}&\ro{x=1111}
\end{matrix}\\
\begin{pmatrix}
0&0&0&1&0&0&1&0&\ldots&1 \\
0&0&0&0&0&1&1&0&\ldots&1 \\
1&0&1&0&0&0&1&0&\ldots&1 \\
\end{pmatrix}
\begin{matrix}
y=1\\ y=2\\ y=3
\end{matrix}
\end{array} \\
f(x,y)
\end{array}
$} &
\\
\vspace{15mm}
\\ \hspace{-5mm}
\mbox{$
\begin{array}{rr}
\begin{matrix}
(y,f(x,y)) \hspace{2em}
\ro{x=0000}&\ro{x=0001}&\ro{x=0010}&\ro{x=0011}&\ro{x=0100}&\ro{x=0101}&\ro{x=0110}&\ro{x=0111}&\text{\hspace{1em}}&\ro{x=1111}\hspace{0.5em}
\end{matrix}\\
\begin{array}{r}
(1,0)\\(1,1)\\(2,0)\\(2,1)\\(3,0)\\(3,1)\\
\end{array}
\begin{pmatrix}
1&1&1&0&1&1&0&1&\ldots&0 \\
0&0&0&1&0&0&1&0&\ldots&1 \\ \hline
1&1&1&1&1&0&0&1&\ldots&0 \\
0&0&0&0&0&1&1&0&\ldots&1 \\ \hline
0&1&0&1&1&1&0&1&\ldots&0 \\
1&0&1&0&0&0&1&0&\ldots&1 \\
\end{pmatrix}
\end{array}
$} &
\end{tabular}
\end{centering}
\caption[A function with its truth table, codeword graph of the induced
code, and the adjacency matrix of the graph]{A function $f\colon \zo^4 \times [3] \to \zo$ with its truth table (top left), codeword graph of the induced
code (right), and the adjacency matrix of the graph (bottom left). Solid, dashed and dotted edges
in the graph respectively correspond to the choices $y=1$, $y=2$, and $y=3$ of the second argument.}
\label{fig:codewordGraph}
\end{figure}

A simple example of a function with its truth table, codeword graph of the induced code along with its
adjacency matrix is given in Figure~\ref{fig:codewordGraph}.

The following theorem is a straightforward generalization of the result in
\cite{ref:TZ04} that is also shown in \cite{ref:GUV09} (we have
included a proof for completeness):

\begin{thm} \label{thm:list} Let $f\colon \zo^\tn \times \zo^\tee \to
  \zo^\tl$ be a strong $k \to_\eps k'$ condenser, and $\cC\subseteq
  \Sigma^{2^\tee}$ be its induced code, where $\Sigma :=
  \zo^\tl$. Then for any mixture $S$ over $\Sigma^{2^\tee}$ we have \[
  |\List_\cC(S, \rho(S) 2^{\tl-k'} + \eps)| < 2^k. \]
\end{thm}

\begin{proof}
  Index the coordinates of $S$ by the elements of $\zo^t$ and denote
  the $i$th coordinate by $S_i$.  Let $Y$ be any random variable with
  min-entropy at least $\tee+k'$ distributed on $\F_2^{\tee+k'}$.
  Define an infor\-mation-theoretic test $T\colon \zo^\tl \times
  \zo^\tee \to \zo$ as follows: $T(x, i) = 1$ if and only if $x \in
  S_i$.  Observe that
  \[ \Pr[T(Y) = 1] \leq \wgt(S) 2^{-(\tee+k')} = \rho(S) 2^{\tl-k'},\]
  and that for every vector $w \in (\zo^\ell)^{2^t}$, \[ \Pr_{i \sim
    \U_\tee}[T(w_i, i) = 1] = \agr(w, S).\] Now, let the random
  variable $X=(X_1, \ldots, X_{2^\tee})$ be uniformly distributed on
  the codewords in $\List_\cC(S, \rho(S) 2^{\tl-k'} + \eps)$ and $Z
  \sim \U_\tee$.  Thus, from Definition~\ref{def:mix} we know that \[
  \Pr_{X, Z}[T(X_Z, Z) = 1] > \rho(S) 2^{\tl-k'} + \eps.\] As the
  choice of $Y$ was arbitrary, this implies that $T$ is able to
  distinguish between the distribution of $(Z, X)$ and any
  distribution on $\zo^{\tee+\tl}$ with min-entropy at least
  $\tee+k'$, with bias greater than $\eps$, which by the definition of
  condensers implies that the min-entropy of $X$ must be less than
  $k$, or \[|\List_\cC(S, \rho(S) 2^{\tl-k'} + \eps)| < 2^k.\]
\end{proof}

Now using the above tools, we are ready to describe and analyze our
construction of error-resilient measurement matrices. We first state a
general result without specifying the parameters of the condenser, and
then instantiate the construction with various choices of the
condenser, resulting in matrices with different properties.

\begin{thm} \label{thm:main} Let $f\colon \zo^\tn \times \zo^\tee \to
  \zo^\tl$ be a strong $k \to_\eps k'$ condenser, and $\cC$ be its
  induced code. Suppose that the parameters $\PI, \nu, \gamma > 0$ are
  chosen so that \[(\PI + \gamma)2^{\tl-k'} + \nu/\gamma < 1 -
  \eps, \] and $d := \gamma 2^\tl$. Then the adjacency matrix of the
  codeword graph of $\cC$ (which has $m := 2^{\tee+\tl}$ rows and $n
  := 2^\tn$ columns) is a $(\PI m , (\nu/d) m, 2^k-d, 0)$-resilient
  measurement matrix for $d$-sparse vectors. Moreover, it allows for a
  reconstruction algorithm with running time $O(mn)$. 
\end{thm}

\begin{proof}
  Define $L := 2^\tl$ and $T := 2^\tee$.  Let $\cM$ be the adjacency
  matrix of the codeword graph of $\cC$.  It immediately follows from
  the construction that the number of rows of $\cM$ (denoted by $m$)
  is equal to $TL$. Moreover, notice that the Hamming weight of each
  column of $\cM$ is exactly $T$.

  Let $x \in \zo^n$ and denote by $y \in \zo^m$ its encoding, i.e., $y
  := \cM[x]$, and by $\hat{y} \in \zo^m$ a \emph{received word}, or a
  \emph{noisy} version of $y$.

  The encoding of $x$ can be schematically viewed as follows: The
  coefficients of $x$ are assigned to the left vertices of the
  codeword graph and the encoded bit on each right vertex is the
  bitwise ``or'' of the values of its neighbors.

  The coordinates of $x$ can be seen in one-to-one correspondence with
  the codewords of $\cC$. Let $X \subseteq \cC$ be the set of
  codewords corresponding to the support of $x$.  The coordinates of
  the noisy encoding $\hat{y}$ are indexed by the elements of $[T]
  \times [L]$ and thus, $\hat{y}$ naturally defines a mixture $S =
  (S_1, \ldots, S_{T})$ over $[L]^T$, where $S_i$ contains $j$ iff
  $\hat{y}$ at position $(i, j)$ is $1$.

  Observe that $\rho(S)$ is the relative Hamming weight (denoted below
  by $\delta(\cdot)$) of $\hat{y}$; thus, we have
  \[
  \rho(S) = \delta(\hat{y}) \leq \delta(y) + \PI \leq d/L + \PI =
  \gamma + \PI,
  \]
  where the last inequality comes from the fact that the relative
  weight of each column of $\cM$ is exactly $1/L$ and that $x$ is
  $d$-sparse.

  Furthermore, from the assumption we know that the number of false
  negatives in the measurement is at most $\nu TL/d = \nu
  T/\gamma$. Therefore, any codeword in $X$ must have agreement at
  least $1- \nu/\gamma$ with $S$.  This is because $S$ is indeed
  constructed from a mixture of the elements in $X$, modulo false
  positives (that do not decrease the agreement) and at most $\nu
  T/\gamma$ false negatives each of which can reduce the agreement by
  at most $1/T$.

  Accordingly, we consider a decoder which, similar to the distance
  decoder that we have introduced before, simply outputs a binary
  vector $\hat{x}$ supported on the coordinates corresponding to those
  codewords of $\cC$ that have agreement larger than $1 - \nu/\gamma$
  with $S$. Clearly, the running time of the decoder is linear in the
  size of the measurement matrix.

  By the discussion above, $\hat{x}$ must include the support of $x$.
  Moreover, Theorem~\ref{thm:list} applies for our choice of
  parameters, implying that $\hat{x}$ must have weight less than
  $2^k$.
\end{proof}

\subsubsection{Instantiations}

Now we instantiate the general result given by Theorem~\ref{thm:main}
with various choices of the underlying condenser, among the results
discussed in Section~\ref{sec:extrConstr}, and compare the obtained
parameters.  First, we consider two extreme cases, namely, a
non-explicit optimal condenser with zero overhead (i.e., extractor)
and then a non-explicit optimal condenser with zero loss (i.e.,
lossless condenser) and then consider how known explicit constructions
can approach the obtained bounds. A summary of the obtained results is
given in Table~\ref{tab:results}.

\begin{table} 
  \caption[Summary of the noise-resilient group testing schemes]{A summary of constructions in Section~\ref{sec:nrConstr}. The parameters $\alpha \in [0,1)$ and
    $\delta \in (0, 1]$ are arbitrary constants, $m$ is the number of measurements,
    $e_0$ (resp., $e_1$) the number of tolerable false positives (resp., negatives) in the measurements,
    and $e'_0$ is the number of false positives in the reconstruction. The fifth column shows whether
    the construction is explicit (Exp) or randomized (Rnd), and the last column shows
    the running time of the reconstruction algorithm.}
  \begin{center}
    \begin{tabular}{|c|c|c|c|l|l|}
      \hline
      &&&& Exp/ & Rec. \\
      $m$ & $e_0$ & $e_1$ & $e'_0$ & Rnd & Time \\
      \hline
      $O(d \log n)$ & $\alpha m$ & $\Omega(m/d)$ & $O(d)$ & Rnd & $O(mn)$ \\
      $O(d \log n)$ & $\Omega(m)$ & $\Omega(m/d)$ & $\delta d$ & Rnd & $O(mn)$ \\
      $O(d^{1+o(1)} \log n)$ & $\alpha m$ & $\Omega(m/d)$ & $O(d)$ & Exp & $O(mn)$ \\
      $d \cdot \qpoly(\log n)$ & $\Omega(m)$ & $\Omega(m/d)$ & $\delta d$ & Exp & $O(mn)$ \\
      $d \cdot \qpoly(\log n)$ & $\alpha m$ & $\Omega(m/d)$ & $O(d)$ & Exp & $\poly(m)$ \\
      $\poly(d) \poly(\log n)$ & $\poly(d) \poly(\log n)$ & $\Omega(e_0/d)$ & $\delta d$ & Exp & $\poly(m)$ \\
      \hline
    \end{tabular}
  \end{center}
  \label{tab:results}
\end{table}

\subsubsection*{Optimal Extractors}

Recall Radhakrishan and Ta-Shma's non-constructive bound that for
every choice of the parameters $k, \tn, \eps$, there is a strong $(k,
\eps)$-extractor with input length $\tn$, seed length $t = \log
(\tn-k) + 2 \log (1/\eps) + O(1)$ and output length $\tl = k - 2\log
(1/\eps) - O(1)$, and that the bound is achieved by a random function.
Plugging this result in Theorem~\ref{thm:main}, we obtain a
non-explicit measurement matrix from a simple, randomized construction
that achieves the desired trade-off with high probability:

\begin{coro} \label{coro:optext} For every choice of constants $\PI
  \in [0, 1)$ and $\nu \in [0, \nu_0)$, $\nu_0 := (\sqrt{5-4\PI} -
  1)^3/8$, and positive integers $d$ and $n \geq d$, there is an $m
  \times n$ measurement matrix, where $m = O(d \log n)$, that is $(\PI
  m, (\nu/d) m, O(d), 0)$-resilient for $d$-sparse vectors of length
  $n$ and allows for a reconstruction algorithm with running time
  $O(mn)$.
\end{coro}

\begin{proof}
  For simplicity we assume that $n = 2^\tn$ and $d = 2^\td$ for
  positive integers $\tn$ and $\td$. However, it should be clear that
  this restriction will cause no loss of generality and can be
  eliminated with a slight change in the constants behind the
  asymptotic notations.

  We instantiate the parameters of Theorem~\ref{thm:main} using an
  optimal strong extractor.  If $\nu = 0$, we choose $\gamma, \eps$
  small constants such that $\gamma+\eps < 1-\PI$. Otherwise, we
  choose $\gamma := \sqrt[3]{\nu}$, which makes $\nu/\gamma =
  \sqrt[3]{\nu^2}$, and $\eps < 1-\PI - \sqrt[3]{\nu} -
  \sqrt[3]{\nu^2}$.  (One can easily see that the right hand side of
  the latter inequality is positive for $\nu < \nu_0$). Hence, the
  condition $\PI + \nu/\gamma < 1 - \eps - \gamma$ required by
  Theorem~\ref{thm:main} is satisfied.

  Let $r = 2 \log (1/\eps) + O(1) = O(1)$ be the entropy loss of the
  extractor for error $\eps$, and set up the extractor for min-entropy
  $k = \log d + \log(1/\gamma) + r$, which means that $K := 2^k =
  O(d)$ and $L := 2^\tl = d/\gamma = O(d)$. Now we can apply
  Theorem~\ref{thm:main} and conclude that the measurement matrix is
  $(\PI m, (\nu/d) m, O(d), 0)$-resilient.  The seed length required
  by the extractor is $\tee \leq \log \tn + 2\log(1/\eps) + O(1)$,
  which gives $T := 2^t = O(\log n)$.  Therefore, the number of
  measurements will be $m = TL = O(d \log n)$.
\end{proof}

\subsubsection*{Optimal Lossless Condensers}

Now we instantiate Theorem~\ref{thm:main} with an optimal strong
lossless condenser with input length $\tn$, entropy requirement $k$,
seed length $\tee = \log \tn + \log(1/\eps) + O(1)$ and output length
$\tl = k+\log(1/\eps)+ O(1)$. Thus we get the following corollary.

\begin{coro} \label{coro:optcond} For positive integers $n \geq d$ and
  every constant $\delta > 0$ there is an $m \times n$ measurement
  matrix, where $m = O(d \log n)$, that is $(\Omega(m), \Omega(1/d) m,
  \nextLine \delta d, 0)$-resilient for $d$-sparse vectors of length
  $n$ and allows for a reconstruction algorithm with running time
  $O(mn)$.
\end{coro}

\begin{proof}
  We will use the notation of Theorem~\ref{thm:main} and apply it
  using an optimal strong lossless condenser.  This time, we set up the condenser
  with error $\eps := \frac{1}{2} \delta/(1+\delta)$ and min-entropy
  $k$ such that $K := 2^k = d/(1-2\eps)$. As the error is a constant,
  the overhead and hence $2^{\tl-k}$ will also be a constant. The seed
  length is $\tee = \log (\tn/\eps) + O(1)$, which makes $T := 2^\tee
  = O(\log n)$. As $L := 2^\tl = O(d)$, the number of measurements
  becomes $m = TL = O(d \log n)$, as desired.

  Moreover, note that our choice of $K$ implies that $K-d = \delta
  d$.  Thus we only need to choose $\PI$ and $\nu$ appropriately to
  satisfy the condition \begin{equation} \label{eqn:pncondition} (\PI + \gamma)L/K + \nu/\gamma < 1 -
  \eps, \end{equation} where $\gamma = d/L = K/(L(1+\delta))$ is a constant, as
  required by the lemma.  Substituting for $\gamma$ in \eqref{eqn:pncondition}
  and after simple manipulations, we get the
  condition \[ \PI L/K + \nu (L/K) (1+\delta) < \frac{\delta}{2(1+\delta)}, \]
  which can be satisfied by choosing $\PI$ and $\nu$ to be appropriate
  positive constants.
\end{proof}

Both results obtained in Corollaries
\ref{coro:optext}~and~\ref{coro:optcond} almost match the lower bound
of Lemma~\ref{lem:lowerbound} for the number of measurements. However,
we note the following distinction between the two results:
Instantiating the general construction of Theorem~\ref{thm:main} with
an extractor gives us a sharp control over the fraction of tolerable
errors, and in particular, we can obtain a measurement matrix that is
robust against \emph{any} constant fraction (bounded from $1$) of
false positives. However, the number of potential false positives in
the reconstruction will be bounded by some constant fraction of the
sparsity of the vector that cannot be made arbitrarily close to zero.

On the other hand, using a lossless condenser enables us to bring down
the number of false positives in the reconstruction to an arbitrarily
small fraction of $d$ (which is, in light of Lemma~\ref{lem:distance},
the best we can hope for), though it does not give as
good a control on the fraction of tolerable errors as in the extractor
case, though we still obtain resilience against the same order of
errors.

Recall that the simple divide-and-conquer adaptive construction given
in beginning the chapter consists of $O(\log (n/d))$
\emph{non-adaptive stages}, where within each stage $O(d)$
non-adaptive measurements are made, but the choice of the measurements
for each stage fully depends on all the previous outcomes. By the
lower bounds on the size of disjunct matrices, we know that the number
of non-adaptive rounds cannot be reduced to $1$ without affecting the
total number of measurements by a multiplicative factor of
$\tilde{\Omega}(d)$. However, our non-adaptive upper bounds
(Corollaries \ref{coro:optext}~and~\ref{coro:optcond}) show that the
number of rounds can be reduced to $2$, while preserving the total
number of measurements at $O(d \log n)$. In particular, in a two-stage
scheme, the first non-adaptive round would output an approximation of
the $d$-sparse vector up to $O(d)$ false positive (even if the
measurements are highly unreliable) and the second round simply
examines the $O(d)$ possible positions using trivial singleton
measurements to pinpoint the exact support of the vector.

\subsubsection*{Applying the Guruswami-Umans-Vadhan's Extractor}

While Corollaries \ref{coro:optext}~and~\ref{coro:optcond} give
probabilistic constructions of noise-resi\-lient measurement matrices,
certain applications require a fully explicit matrix that is
guaranteed to work. To that end, we need to instantiate
Theorem~\ref{thm:main} with an explicit condenser. First, we use the
nearly-optimal explicit extractor of Guruswami, Umans and Vadhan
(Theorem~\ref{thm:extr}), that currently gives the best trade-off for
the range of parameters needed for our application.  Using this
extractor, we obtain a similar trade-off as in
Corollary~\ref{coro:optext}, except for a higher number of
measurements which would be bounded by $O(2^{O(\log^2 \log d)} d \log
n) = O(d^{1+o(1)} \log n)$.

\begin{coro} \label{coro:guv} For every choice of constants $\PI \in
  [0, 1)$ and $\nu \in [0, \nu_0)$, $\nu_0 := (\sqrt{5-4\PI} -
  1)^3/8$, and positive integers $d$ and $n \geq d$, there is a fully
  explicit $m \times n$ measurement matrix, where \[ m = O(2^{O(\log^2
    \log d)} d \log n) = O(d^{1+o(1)} \log n), \] that is $(\PI m,
  (\nu/d) m, O(d), 0)$-resilient for $d$-sparse vectors of length $n$
  and allows for a reconstruction algorithm with running time
  $O(mn)$. \qed
\end{coro}

\subsubsection*{Applying ``Zig-Zag'' Lossless Condenser}

An important explicit construction of lossless condensers that has an
almost optimal output length is due to Capalbo et
al.~\cite{ref:CRVW02}.  This construction borrows the notion of
``zig-zag products'' that is a combinatorial tool for construction of
expander graphs as a major ingredient of the condenser.  The following
theorem quotes a setting of this construction that is most useful for
our application:

\begin{thm} \cite{ref:CRVW02} \label{thm:CRVW} For every $k \leq n \in
  \N$, $\eps > 0$ there is an explicit $k \to_\eps k$
  condenser\footnote{Though not explicitly mentioned in
    \cite{ref:CRVW02}, these condensers can be considered to be
    strong.} with seed length $d=O(\log^3 (n/\eps))$ and output length
  $m=k+\log(1/\eps)+O(1)$. \qed
\end{thm}

Combining Theorem~\ref{thm:main} with the above condenser, we obtain a
similar result as in Corollary~\ref{coro:optcond}, except that the
number of measurements would be $d 2^{\log^3(\log n)}= d \cdot
\qpoly(\log n)$.

\begin{coro} \label{coro:crvw} For positive integers $n \geq d$ and
  every constant $\delta > 0$ there is a fully explicit $m \times n$
  measurement matrix, where \[ m = d 2^{\log^3(\log n)}= d \cdot
  \qpoly(\log n),\] that is $(\Omega(m), \Omega(1/d) m, \delta d,
  0)$-resilient for $d$-sparse vectors of length $n$ and allows for a
  reconstruction algorithm with running time $O(mn)$. \qed
\end{coro}

\subsubsection{Measurements Allowing Sublinear Time Reconstruction}
\index{sublinear time reconstruction}

The naive reconstruction algorithm given by Theorem~\ref{thm:main}
works efficiently in linear time in the size of the measurement
matrix. However, for very sparse vectors (i.e., $d \ll n$), it might
be of practical importance to have a reconstruction algorithm that
runs in \emph{sublinear} time in $n$, the length of the vector, and
ideally, polynomial in the number of measurements, which is merely
$\poly(\log n, d)$ if the number of measurements is optimal.

As shown in \cite{ref:TZ04}, if the code $\cC$ in
Theorem~\ref{thm:list} is obtained from a strong extractor constructed
from a \emph{black-box pseudorandom generator (PRG)}, it is possible
to compute the agreement list (which is guaranteed by the theorem to
be small) more efficiently than a simple exhaustive search over all
possible codewords.  In particular, in this case they show that
$\List_\cC(S, \rho(S)+\eps)$ can be computed in time $\poly(2^\tee,
2^\tl, 2^k, 1/\eps)$ (where $\tee, \tl, k, \eps$ are respectively the
seed length, output length, entropy requirement, and error of the
extractor), which can be much smaller than $2^\tn$ ($\tn$ being the
input length of the extractor).

Currently two constructions of extractors from black-box PRGs are
known: Trevisan's extractor \cite{ref:Tre} (as well as its improvement
in \cite{ref:RRV}) and Shaltiel-Umans' extractor \cite{ref:SU}.
However, the latter can only extract a sub-constant fraction of the
min-entropy and is not suitable for our needs, albeit it requires a
considerably shorter seed than Trevisan's extractor.  Thus, here we
only consider Raz's improvement of Trevisan's extractor given in
Theorem~\ref{thm:Tre}.  Using this extractor in
Theorem~\ref{thm:main}, we obtain a measurement matrix for which the
reconstruction is possible in polynomial time in the number of
measurements; however, as the seed length required by this extractor
is larger than Theorem~\ref{thm:extr}, we will now require a higher
number of measurements than before. Specifically, using Trevisan's
extractor, we get the following.

\begin{coro} \label{coro:efficient} For every choice of constants $\PI
  \in [0, 1)$ and $\nu \in [0, \nu_0)$, $\nu_0 := (\sqrt{5-4\PI} -
  1)^3/8$, and positive integers $d$ and $n \geq d$, there is a fully
  explicit $m \times n$ measurement matrix $\cM$ that is $(\PI m,
  (\nu/d) m, O(d), 0)$-resilient for $d$-sparse vectors of length $n$,
  where \[ m = O(d 2^{\log^3 \log n}) = d \cdot \qpoly(\log n). \]
  Furthermore, $\cM$ allows for a reconstruction algorithm with
  running time $\poly(m)$, which would be sublinear in $n$ for $d =
  O(n^c)$ and a suitably small constant $c > 0$.  \qed \end{coro}

On the condenser side, we observe that the strong lossless (and lossy)
condensers due to Guruswami et al.\ (given in
Theorem~\ref{thm:GUVcond}) also allow efficient list-recovery.  The
code induced by this condenser is precisely a list-decodable code due
to Parvaresh and Vardy~\cite{ref:PV05}. Thus, the efficient list
recovery algorithm of the condenser is merely the list-decoding
algorithm for this code\footnote{ For similar reasons, any
  construction of measurement matrices based on codeword graphs of
  algebraic codes that are equipped efficient soft-decision decoding
  (including the original Reed-Solomon based construction of Kautz and
  Singleton \cite{ref:KS64}) allow sublinear time
  reconstruction.}. Combined with Theorem~\ref{thm:main}, we can show
that codeword graphs of Parvaresh-Vardy codes correspond to good
measurement matrices that allow sublinear time recovery, but with
incomparable parameters to what we obtained from Trevisan's extractor
(the proof is similar to Corollary~\ref{coro:optcond}):

\begin{coro} \label{coro:guvcond} For positive integers $n \geq d$ and
  any constants $\delta, \alpha > 0$ there is an $m \times n$
  measurement matrix, where \[ m = O(d^{3+\alpha+2/\alpha} (\log
  n)^{2+2/\alpha}), \] that is $(\Omega(e), \Omega(e /d), \delta d,
  0)$-resilient for $d$-sparse vectors of length $n$, where \[ e :=
  (\log n)^{1+1/\alpha} d^{2+1/\alpha}.\] Moreover, the matrix allows
  for a reconstruction algorithm with running time $\poly(m)$. \qed
\end{coro}

We remark that we could also use a lossless condenser due to Ta-Shma
et al.\ \cite{ref:TUZ01} which is based on Trevisan's extractor
and also allows efficient list recovery, but it achieves inferior
parameters compared to Corollary~\ref{coro:guvcond}.

\subsubsection{Connection with
  List-Recoverability} \label{sec:listrec}

Extractor codes that we used in Theorem~\ref{thm:main} are instances
of \emph{soft-decision decodable} codes\footnote{To be precise, here
  we are dealing with a special case of soft-decision decoding with
  binary weights.}  that provide high list-decodability in ``extremely
noisy'' scenarios. In fact it is not hard to see that good extractors
or condensers are required for our construction to carry through, as
Theorem~\ref{thm:list} can be shown to hold, up to some loss in parameters, in the reverse direction
as well (as already shown by Ta-Shma and Zuckerman \cite{ref:TZ04}*{Theorem~1} for the case of extractors).

However, for designing measurement matrices for the noiseless (or
low-noise) case, it is possible to resort to the slightly weaker
notion of \emph{list recoverable codes}. Formally, a code $\cC$ of
block length $\tn$ over an alphabet $\Sigma$ is called \emph{$(\alpha,
  d, \tl)$-list recoverable} if for every mixture $S$ over
$\Sigma^\tn$ consisting of sets of size at most $d$ each, we have
$|\List_\cC(S, \alpha)| \leq \tl$.  A simple argument similar to
Theorem~\ref{thm:main} shows that the adjacency matrix of the codeword
graph of such a code with rate $R$ gives a $(\log n)|\Sigma|/R \times
n$ measurement matrix\footnote{For codes over large alphabets, the
  factor $|\Sigma|$ in the number of rows can be improved using
  \emph{concatenation} with a suitable \emph{inner} measurement
  matrix.}  for $d$-sparse vectors in the noiseless case with at most
$\tl-d$ false positives in the reconstruction.

Ideally, a list-recoverable code with $\alpha=1$, alphabet size
$O(d)$, positive constant rate, and list size $\tl=O(d)$ would give an
$O(d \log n) \times n$ matrix for $d$-sparse vectors, which is almost
optimal (furthermore, the recovery would be possible in sublinear time
if $\cC$ is equipped with efficient list recovery). However, no
explicit construction of such a code is so far known.

Two natural choices of codes with good list-recoverability properties
are Reed-Solomon and Algebraic-Geometric codes, which in fact provide
soft-decision decoding with short list size (cf.\ \cite{ref:Venkat}).
However, while the list size is polynomially bounded by $\tn$ and $d$,
it can be much larger than $O(d)$ that we need for our application
even if the rate is polynomially small in $d$.

On the other hand, it is shown in \cite{ref:GR08} that \emph{folded
  Reed-Solomon Codes} are list-recoverable with constant rate, but
again they suffer from large alphabet and list size\footnote{As shown
  in \cite{ref:GUV09}, folded Reed-Solomon codes can be used to
  construct lossless condensers, which eliminates the list size
  problem.  However, they give inferior parameters compared to
  Parvaresh-Vardy codes used in Corollary~\ref{coro:guvcond}.}.

We also point out a construction of $(\alpha, d, d)$ list-recoverable
codes (allowing list recovery in time $O(\tn d)$) in \cite{ref:GR08}
with rate polynomially small but alphabet size exponentially large in
$d$, from which they obtain superimposed codes.

\subsubsection{Connection with the Bit-Probe Model and
  Designs} \label{sec:bitprobe}

An important problem in data structures is the static set membership
problem in bit-probe model, which is the following:
Given a set $S$ of at most $d$ elements from a universe of size $n$,
store the set as a string of length $m$ such that any query of the
type ``is $x$ in $S$?''  can be reliably answered by reading few bits
of the encoding.  The query algorithm might be probabilistic, and be
allowed to err with a small one or two-sided error. Information
theoretically, it is easy to see that $m=\Omega(d \log (n/d))$
regardless of the bit-probe complexity and even if a small constant
error is allowed.

Remarkably, it was shown in \cite{ref:BMRV02} that the lower bound on
$m$ can be (non-explicitly) achieved using only one
bit-probe. Moreover, a part of their work shows that any one-probe
scheme with negative one-sided error $\eps$ (where the scheme only
errs in case $x \notin S$) gives a $\lfloor d/\eps
\rfloor$-superimposed code (and hence, requires $m=\Omega(d^2 \log n)$
by \cite{ref:DR83}). It follows that from any such scheme one can
obtain a measurement matrix for exact reconstruction of sparse
vectors, which, by Lemma~\ref{lem:distance}, cannot provide high
resiliency against noise. The converse direction, i.e., using
superimposed codes to design bit-probe schemes does not necessarily
hold unless the error is allowed to be very close to $1$. However, in
\cite{ref:BMRV02} \emph{combinatorial designs}\footnote{ A design is a
  collection of subsets of a universe, each of the same size, such
  that the pairwise intersection of any two subset is upper bounded by
  a prespecified parameter.}  based on low-degree polynomials are used
to construct one bit-probe schemes with $m=O(d^2 \log^2 n)$ and small
one-sided error.

On the other hand, Kautz and Singleton \cite{ref:KS64} observed that
the encoding of a combinatorial design as a binary matrix corresponds
to a superimposed code (which is in fact slightly
error-resilient). Moreover, they used Reed-Solomon codes to construct
a design, which in particular gives a $d$-superimposed code. This is
in fact the same design that is used in \cite{ref:BMRV02}, and in our
terminology, can be regarded as the adjacency matrix of the codeword
graph of a Reed-Solomon code.

It is interesting to observe the intimate similarity between our
framework given by Theorem~\ref{thm:main} and classical constructions
of superimposed codes. However, some key differences are worth
mentioning.  Indeed, both constructions are based on codeword graphs
of error-correcting codes. However, classical superimposed codes owe
their properties to the large distance of the underlying code. On the
other hand, our construction uses extractor and condenser codes and
does not give a superimposed code simply because of the substantially
low number of measurements. However, as shown in
Theorem~\ref{thm:main}, they are good enough for a slight relaxation
of the notion of superimposed codes because of their soft-decision
list decodability properties, which additionally enables us to attain
high noise resilience and a considerably smaller number of
measurements.

Interestingly, Buhrman et al.\ \cite{ref:BMRV02} use randomly
chosen bipartite graphs to construct storage schemes with two-sided
error requiring nearly optimal space $O(d \log n)$, and Ta-Shma
\cite{ref:Ta02} later shows that expander graphs from lossless
condensers would be sufficient for this purpose. However, unlike
schemes with negative one-sided error, these schemes use encoders that
cannot be implemented by the ``or'' function and thus do not translate
to group testing schemes.

\section{The Threshold Model} \label{sec:threshold}

A natural generalization of classical group testing, introduced by
Damaschke \cite{ref:thresh1}, considers the case where the measurement
outcomes are determined by a \emph{threshold predicate} instead of
logical ``or''.

In particular, the threshold model\index{group testing!threshold
  model} is characterized by two integer parameters $\ell, u$ such
that $0 < \ell \leq u$ (that are considered to be fixed constants),
and each measurement outputs positive if the number of positives
within the corresponding pool is at least $u$. On the other hand, if
the number of positives is less than $\ell$, the test returns
negative, and otherwise the outcome can be arbitrary.  In this view,
classical group testing corresponds to the special case where $\ell =
u = 1$.  In addition to being of theoretical interest, the threshold
model is interesting for applications, in particular in biology, where
the measurements have reduced or unpredictable sensitivity or may
depend on various factors that must be simultaneously present in the
sample.

The difference $g := u - \ell$ between the thresholds is known as the
\emph{gap}\index{group testing!threshold model!gap} parameter.  As
shown by Damaschke \cite{ref:thresh1}, in threshold group testing
identification of the set of positives is only possible when the
number of positives is at least $u$.  Moreover, regardless of the
number of measurements, in general the set of positives can only be
identified within up to $g$ false positives and $g$ false negatives
(thus, unique identification can be guaranteed only when $\ell = u$).

Additionally, Damaschke constructed a scheme for identification of the
positives in the threshold model. For the gap-free case where $g=0$,
the number of measurements in this scheme is $O((d+u^2) \log n)$,
which is nearly optimal (within constant factors).  However, when $g >
0$, the number of measurements becomes $O(dn^b + d^u)$, for an
arbitrary constant $b > 0$, if up to $g+(u-1)/b$ misclassifications
are allowed. Moreover, Chang et al.\ \cite{ref:CCFS10} have proposed a
different scheme for the gap-free case that achieves $O(d \log n)$
measurements.

A drawback of the scheme presented by Damaschke (as well as the one by
Chang et al.) is that the measurements are adaptive. As mentioned
before, for numerous applications (in particular, molecular biology),
adaptive measurements are infeasible and must be avoided.

In this section, we consider the non-adaptive threshold testing
problem in a possibly noisy setting, and develop measurement matrices
that can be used in the threshold model. Similar to the classical
model of group testing, non-adaptive measurements in the threshold
model can be represented as a Boolean matrix, where the $i$th row is
the characteristic vector of the set of items that participate in the
$i$th measurement.

\subsection{Strongly Disjunct Matrices} \label{sec:stronglyDisjunct}
\index{disjunct!strongly disjunct}

Non-adaptive threshold testing has been considered by Chen and
Fu~\cite{ref:thresh2}. They observe that, a generalization of the
standard notion of disjunct matrices (the latter being extensively
used in the literature of classical group testing) is suitable for the
threshold model. In this section, we refer to this generalized notion
as \emph{strongly disjunct} matrices and to the standard notion as
\emph{classical} disjunct matrices. Strongly disjunct matrices can be
defined as follows.

\begin{defn} \label{def:strongDisjunct} A Boolean matrix (with at
  least $d+u$ columns) is said to be strongly $(d,e;u)$-disjunct if
  for every choice of $d+u$ distinct columns \[ C_1,\ldots, C_u, C'_1,
  \ldots, C'_d,\] all distinct, we have
  \[
  |\cap_{i=1}^u \supp(C_i) \setminus \cup_{i=1}^{d} \supp(C'_i) | > e.
  \]
\end{defn}
Observe that, $(d,e;u)$-disjunct matrices are, in particular,
$(d',e';u')$-dis\-junct for any $d' \leq d$, $e' \leq e$, and $u' \leq
u$.  Moreover, \emph{classical} $(d,e)$-disjunct matrices correspond
to the special case $u=1$.

An important motivation for the study of this notion is the following
\emph{hidden hypergraph learning problem} (cf.\
\cite{ref:DH06}*{Chapter~6} and \cite{ref:groupTesting}*{Chapter~12}),
itself being motivated by the so-called \emph{complex model} in
computational biology \cite{ref:CDH07}.  A $(\le
u)$-hypergraph\index{hypergraph}\index{group testing!hypergraph
  learning} is a tuple $(V,E)$ where $V$ and $E$ are known as the set
of vertices and hyper-edges, respectively.  Each hyperedge $e \in E$
is a non-empty subset of $V$ of size at most $u$. The classical notion
of undirected graphs (with self-loops) corresponds to $(\le
2)$-hypergraphs.

Now, suppose that $G$ is a $(\leq u)$-hypergraph on a vertex set $V$
of size $n$, and denote by $\mathcal{V}(G)$ the set of vertices
induced by the hyper-edge set of $G$; i.e., $v \in \mathcal{V}(G)$ if
and only if $G$ has a hyper-edge incident to $v$. Then assuming that
$|\mathcal{V}(G)| \leq d$ for a \emph{sparsity parameter} $d$, the aim
in the hypergraph-learning problem is to identify $G$ using as few
(non-adaptive) queries of the following type as possible: Each query
specifies a set $Q \subseteq V$, and its corresponding answer is a
Boolean value which is $1$ if and only if $G$ has a hyperedge
contained in $Q$.

It is known that \cites{ref:GHTWZ06,ref:CDH07}, in the hypergraph
learning problem, any suitable grouping strategy defines a strongly
disjunct matrix (whose rows are characteristic vectors of individual
queries $Q$), and conversely, any strongly disjunct matrix can be used
as the incidence matrix of the set of queries. Below we recollect a
simple proof of this fact.

\begin{lem}
  Let $\cM$ be a strongly $(d, e; u)$-disjunct matrix with columns
  indexed by the elements of a vertex set $V$, and $G$ and $G'$ be any
  two distinct $(\le u)$-hypergraphs on $V$ such that $\mathcal{V}(G)
  \leq d$ and $\mathcal{V}(G') \leq d$. Then the vector of the
  outcomes corresponding to the queries defined by $\cM$ on $G$ and
  $G'$ differ in more than $e$ positions.  Conversely, if $\cM$ is
  such that the query outcomes differ in more than $e$ positions for
  every choice of the hypergraphs $G$ and $G'$ as above, then it must
  be strongly $(d-u, e; u)$-disjunct.
\end{lem}

\begin{proof}
  Suppose that $\cM$ is an $m \times |V|$ strongly $(d, e;
  u)$-disjunct matrix, and consider distinct $(\le u)$-hypergraphs
  $G=(V,E)$ and $G'=(V,E')$ with $\mathcal{V}(G) \leq d$ and
  $\mathcal{V}(G') \leq d$.  Denote by $y, y' \in \zo^m$ the vector of
  query outcomes for the two graphs $G$ and $G'$, respectively.
  \Wlog, let $S \in E$ be chosen such that no hyper-edge of $G'$ is
  contained in it.  Let $V' := \mathcal{V}(G')\setminus S$, and denote
  by $C_1, \ldots, C_{|S|}$ (resp., $C'_1, \ldots, C_{|V'|}$) the
  columns of $\cM$ corresponding to the vertices in $S$ (resp.,
  $V'$). By Definition~\ref{def:strongDisjunct}, there is a set $T
  \subseteq [m]$ of more than $e$ indices such that for every $i \in
  [|S|]$ (resp., $i \in [|V'|]$) and every $t \in T$, $C_i(t) = 1$
  (resp., $C'_i(t) = 0$).  This means that, for each such $t$, the
  answer to the $t$th query must be $1$ for $G$ (as the query includes
  the vertex set of $S$) but $0$ for $G'$ (considering the assumption
  that no edge of $G'$ is contained in $S$).

  For the converse, let $S, Z \subseteq [V]$ be disjoint sets of
  vertices such that $|S| = u$ and $|Z| = d-u$, and denote by $\{C_1,
  \ldots, C_u\}$ and $\{C'_1, \ldots, C'_{d-u}\}$ the set of columns
  of $\cM$ picked by $S$ and $T$, respectively.  Take any $v \in S$,
  let the $u$-hypergraph $G=(V,E)$ be a $u$-clique on $Z \cup S
  \setminus \{v\}$, and $G'=(V,E')$ be such that $E' := E \cup \{S\}$.
  Denote by $y, y' \in \zo^m$ the vector of query outcomes for the two
  graphs $G$ and $G'$, respectively. Since $G'$ is a subgraph of $G$,
  it must be that $\supp(y') \subseteq \supp(y)$.

  Let $T := \supp(y) \setminus \supp(y)$. By the distinguishing
  property of $\cM$, the set $T$ must have more than $e$ elements.
  Take any $t \in T$. We know that the $t$th query defined by $\cM$
  returns positive for $G$ but negative for $G'$. Thus this query must
  contain the vertex set of $S$, but not any of the elements in $Z$
  (since otherwise, it would include some $z \in Z$ and subsequently,
  $\{z\} \cup S \setminus \{v\}$, which is a hyperedge of $G'$).  It
  follows that for each $i \in [u]$ (resp., $i \in [d-u]$), we must
  have $C_i(t) = 1$ (resp., $C'_i(t) = 0$) and the disjunctness
  property as required by Definition~\ref{def:strongDisjunct} holds.
\end{proof}

The parameter $e$ determines ``noise tolerance'' of the measurement
scheme.  Namely, a strongly $(d,e;u)$-disjunct matrix can uniquely
distinguish between $d$-sparse hypergraphs even in presence of up to
$\lfloor e/2 \rfloor$ erroneous query outcomes.

The key observation made by Chen and Fu~\cite{ref:thresh2} is that
threshold group testing corresponds to the special case of the
hypergraph learning problem where the hidden graph $G$ is known to be
a $u$-clique\footnote{A $u$-clique on the vertex set $V$ is a $(\leq
  u)$-hypergraph $(V,E)$ such that, for some $V' \subseteq V$, $E$ is
  the set of all subsets of $V'$ of size $u$.}. In this case, the
unknown Boolean vector in the corresponding threshold testing problem
would be the characteristic vector of $\mathcal{V}(G)$.  It follows
that strongly disjunct matrices are suitable choices for the
measurement matrices in threshold group testing.

More precisely, the result by Chen and Fu states that, for threshold
parameters $\ell$ and $u$, a strongly $(d-\ell-1, 2e; u)$-disjunct
matrix suffices to distinguish between $d$-sparse vectors in the
threshold model\footnote{Considering unavoidable assumptions that up
  to $g := u - \ell$ false positives and $g$ false negatives are
  allowed in the reconstruction, and that the vector being measured
  has weight at least $u$.}, even if up to $e$ erroneous measurements
are allowed.

Much of the known results for classical disjunct matrices can be
extended to strongly disjunct matrices by following similar ideas.  In
particular, the probabilistic result of
Theorem~\ref{thm:classicDisjunct:random} can be generalized to show
that strongly $(d,e;u)$-disjunct matrices exist with \[ m = O(d^{u+1}
(\log (n/d))/(1-p)^2)\] rows and error tolerance \[ e = \Omega(p d
\log (n/d)/(1-p)^2),\] for any noise parameter $p \in [0,1)$. On the
negative side, however, several concrete lower bounds are known for
the number of rows of such matrices
\cites{ref:SW00,ref:DVMT02,ref:SW04}.  In asymptotic terms, these
results show that one must have \[ m = \Omega(d^{u+1} \log_d n + e
d^u),\] and thus, the probabilistic upper bound is essentially
optimal.

\subsection{Strongly Disjunct Matrices from Codes} \label{sec:KSext}
\index{Kautz-Singleton construction}

For the underlying strongly disjunct matrix, Chen and
Fu~\cite{ref:thresh2} use a greedy construction \cite{ref:CFH08} that
achieves, for any $e \geq 0$, $O((e+1) d^{u+1} \log(n/d))$ rows, but
may take exponential time in the size of the resulting matrix.

Nevertheless, as observed by several researchers
\cites{ref:DVMT02,ref:KL04,ref:GHTWZ06,ref:CDH07}, a classical
explicit construction of combinatorial designs due to Kautz and
Singleton~\cite{ref:KS64} can be extended to construct strongly
disjunct matrices. This concatenation-based construction transforms
any error-correcting code having large distance into a disjunct
matrix.

While the original construction of Kautz and Singleton uses
Reed-Solomon codes and achieves nice bounds, it is possible to use
other families of codes. In particular, as was shown by Porat and
Rothschild \cite{ref:PR08}, codes on the Gilbert-Varshamov bound (see
Appendix~\ref{app:coding}) would result in nearly optimal disjunct
matrices. Moreover, for a suitable range of parameters, they give a
\emph{deterministic} construction of such codes that runs in
polynomial time in the size of the resulting disjunct matrix (albeit
exponential in code's dimension\footnote{In this regard, this
  construction of disjunct matrices can be considered \emph{weakly
    explicit} in that, contrary to fully explicit constructions, it is
  not clear if each individual entry of the matrix can be computed in
  time $\poly(d, \log n)$.  }).

In this section, we will elaborate on details of this (known) class of
constructions, and in addition to Reed-Solomon codes and codes on the
Gilbert-Varshamov bound (that, as mentioned above, were used by Kautz, Singleton,
Porat and Rothschild), will consider a family of algebraic-geometric
codes and Hermitian codes which give nice bounds as well.
Construction~\ref{constr:strong} describes the general idea, which in
analyzed in the following lemma.

\newConstruction{Extension of Kautz-Singleton's method \cite{ref:KS64}.}%
{An $(\tn, \tk, \td)_q$ error-correcting code $\C
    \subseteq [q]^{\tn}$, and integer parameter $u > 0$.}%
    {An $m \times n$ Boolean matrix $\cM$, where $n
    = q^\tk$, and $m = \tn q^u$.}%
    {First, consider the mapping
    $\varphi\colon [q] \to \zo^{q^u}$ from $q$-ary symbols to column
    vectors of length $q^u$ defined as follows.  Index the coordinates
    of the output vector by the $u$-tuples from the set $[q]^u$. Then
    $\varphi(x)$ has a $1$ at position $(a_1, \ldots, a_u)$ if and only if there
    is an $i \in [u]$ such that $a_i = x$. Arrange all codewords of
    $\C$ as columns of an $\tn \times q^\tk$ matrix $\cM'$ with
    entries from $[q]$. Then replace each entry $x$ of $\cM'$ with
    $\varphi(x)$ to obtain the output $m \times n$ matrix $\cM$.}%
    {constr:strong}

\begin{lem} \label{lem:ks} Construction~\ref{constr:strong} outputs a
  strongly $(d,e;u)$-disjunct matrix for every $d < (\tn -
  e)/((\tn-\td)u)$.
\end{lem}

\begin{Proof}
  Let $C := \{ c_1, \ldots, c_u \} \subseteq [n]$ and $C' := \{ c'_1,
  \ldots, c'_d \} \subseteq [n]$ be disjoint subsets of column
  indices. We wish to show that, for more than $e$ rows of $\cM$, the
  entries at positions picked by $C$ are all-ones while those picked
  by $C'$ are all-zeros. For each $j \in [n]$, denote the $j$th column
  of $\cM'$ by $\cM'(j)$, and let $\cM'(C) := \{ \cM'(c_j)\colon j \in
  [u] \}$, and $\cM'(C') := \{ \cM'(c'_j)\colon j \in [d] \}$.

  From the minimum distance of $\C$, we know that every two distinct
  columns of $\cM'$ agree in at most $\tn - \td$ positions.  By a
  union bound, for each $i \in [d]$, the number of positions where
  $\cM'(c'_i)$ agrees with one or more of the codewords in $\cM'(C)$
  is at most $u(\tn - \td)$, and the number of positions where some
  vector in $\cM'(C')$ agrees with one or more of those in $\cM'(C)$
  is at most $du(\tn - \td)$. By assumption, we have $\tn - du(\tn -
  \td) > e$, and thus, for a set $E \subseteq [\tn]$ of size greater
  than $e$, at positions picked by $E$ none of the codewords in
  $\cM'(C')$ agree with any of the codewords in $\cM'(C)$.

  Now let $w \in [q]^n$ be any of the rows of $\cM'$ picked by $E$,
  and consider the $q^u \times n$ Boolean matrix $W$ formed by
  applying the mapping $\varphi(\cdot)$ on each entry of $w$. We know
  that $\{ w(c_j)\colon j \in [u]\} \cap \{ w(c'_j)\colon j \in [d]\}
  = \emptyset$.  Thus we observe that the particular row of $W$
  indexed by $(w(c_1), \ldots, w(c_u))$ (and in fact, any of its
  permutations) must have all-ones at positions picked by $C$ and
  all-zeros at those picked by $C'$. As any such row is a distinct row
  of $\cM$, it follows that $\cM$ is strongly $(d,e;u)$-disjunct.
\end{Proof}

Now we mention a few specific instantiations of the above
construction.  We will first consider the family of Reed-Solomon
codes, that are also used in the original work of Kautz and
Singleton~\cite{ref:KS64}, and then move on to the family of algebraic
geometric (AG) codes on the Tsfasman-Vl{\u a}du{\c t}-Zink (TVZ)
bound, and Hermitian codes, and finally, codes on the
Gilbert-Varshamov (GV) bound. A quick review of the necessary
background on coding-theoretic terms is given in
Appendix~\ref{app:coding}.

\subsubsection*{Reed-Solomon Codes} Let $p \in [0,1)$ be an arbitrary
``noise'' parameter.  If we take $\C$ to be an $[\tn, \tk, \td]_{\tn}$
Reed-Solomon code over an alphabet of size $\tn$ (more precisely, the
smallest prime power that is no less than $\tn$), where $\td =
\tn-\tk+1$, we get a strongly disjunct $(d,e;u)$-matrix with \[ m =
O(du \log n / (1-p))^{u+1} \] rows and \[ e = p \tn = \Omega(p d u
(\log n)/(1-p)).\]

\subsubsection*{AG Codes on the TVZ Bound} \index{TVZ bound} Another
interesting family for the code $\C$ is the family of algebraic
geometric codes that attain the Tsfasman-Vl{\u a}du{\c t}-Zink bound
(cf.\ \cites{tsvz:82,gast:95}).  This family is defined over any
alphabet size $q \geq 49$ that is a square prime power, and achieves a
minimum distance $\td \geq \tn-\tk-\tn/(\sqrt{q}-1)$. Let $e := pn$,
for a noise parameter $p \in [0,1)$. By Lemma~\ref{lem:ks}, the
underlying code $\C$ needs to have minimum distance at least
$\tn(1-(1-p)/(du))$. Thus in order to be able to use the
above-mentioned family of AG codes, we need to have $q \gg
(du/(1-p))^2 =: q_0$. Let us take an appropriate $q \in [2q_0, 8q_0]$,
and following Lemma~\ref{lem:ks}, $\tn - \td = \lceil \tn(1-p)/(du)
\rceil$.  Thus the dimension of $\C$ becomes at least
\[
\tk \geq \tn - \td - \frac{\tn}{\sqrt{q}-1} = \Omega \left(
  \frac{\tn(1-p)}{du} \right) = \Omega(\tn / \sqrt{q_0}),
\]
and subsequently\footnote{Note that, given the parameters $p, d, n$,
  the choice of $q$ depends on $p, d$, as explained above, and then
  one can choose the code length $\tn$ to be the smallest integer for
  which we have $q^\tk \geq n$. But for the sake of clarity we have
  assumed that $q^\tk = n$.} we get that $ \log n = \tk \log q \geq
\tk = \Omega(\tn / \sqrt{q_0}).  $ Now, noting that $m = q^u \tn$, we
conclude that
\[
m = q^u \tn = O(q_0^{u+1/2} \log n) = O\left(\frac{du}{1-p}
\right)^{2u+1} \log n,
\]
and $e = \Omega(p d u (\log n)/(1-p))$.

We see that the dependence of the number of measurements on the
sparsity parameter $d$ is worse for AG codes than Reed-Solomon codes
by a factor $d^u$, but the construction from AG codes benefits from a
linear dependence on $\log n$, compared to $\log^{u+1} n$ for
Reed-Solomon codes.  Thus, AG codes become more favorable only when
the sparsity is substantially low; namely, when $d \ll \log n$.

\subsubsection*{Hermitian Codes} \index{Hermitian codes} A
particularly nice family of AG codes arises from the Hermitian
function field\footnote{See \cite{stich} for an extensive
treatment of the notions in algebraic geometry.}.  Let $q'$ be a prime power and $q := q'^2$. Then the
Hermitian function field over $\F_q$ is a finite extension of the
rational function field $\F_q(x)$, denoted by $\F_q(x,y)$, where we
have $y^{q'} + y = x^{q'+1}$.  The structure of this function field is
relatively well understood and the family of Goppa codes defined over
the rational points of the Hermitian function field is known as
Hermitian codes. This family is recently used by Ben-Aroya and Ta-Shma
\cite{ref:BT09} for construction of small-bias sets. Below we quote
some parameters of Hermitian codes from their work.

The number of rational points of the Hermitian function field is equal
to ${q'}^3+1$, which includes a common pole $Q_\infty$ of $x$ and
$y$. The genus of the function field is $g = q'(q'-1)/2$. For some
integer parameter $r$, we take $G := rQ_\infty$ as the divisor
defining the Riemann-Roch space $\cL(G)$ of the code $\C$, and the set
of rational points except $Q_\infty$ as the evaluation points of the
code. Thus the length of $\C$ becomes $\tn = {q'}^3$. Moreover, the
minimum distance of the code is $\td = n-\deg(G) = n-r$. When $r \geq
2g - 1$, the dimension of the code is given by the Riemann-Roch
theorem, which is equal to $r-g+1$. For the low-degree regime where $r
< 2g-1$, the dimension $\tk$ of the code is the size of the
Wirestrauss semigroup of $G$, which turns out to be the set $W=\{
(i,j) \in \N^2\colon j \leq q'-1 \land iq'+j(q'+1) \leq r\}$.

Now, given parameters $d, p$ of the disjunct matrix, define $\rho :=
(1-p)/((d+1)u)$,
take the alphabet size $q$ as a square prime power,
and set $r := \rho q^{3/2}$. First we consider the case where $r < 2g
- 1 = 2q - 2\sqrt{q} - 1$.  In this case, the dimension of the
Hermitian code becomes $k = |W| = \Omega(r^2/q) = \Omega(\rho^2 q^2).$
The distance $\td$ of the code satisfies $\td = \tn - r \geq \tn
(1-\rho)$ and thus, for $e := p \tn$, conditions of Lemma~\ref{lem:ks}
are satisfied.  The number of the rows of the resulting measurement
matrix becomes $m = q^{u+3/2}$, and we have $n = q^\tk$. Therefore,
\begin{gather*}
  \log n = k \log q \geq k = \Omega(\rho^2 q^2) \\ \Rightarrow q =
  O(\sqrt{\log n}/\rho) \Rightarrow m = O\left(\big(\frac{d\sqrt{\log
        n}}{1-p}\big)^{u+3/2}\right),
\end{gather*}
and in order to ensure that $r < 2g-1$, we need to have $du/(1-p) \gg
\sqrt{\log n}$.  On the other hand, when $du/(1-p) \ll \sqrt{\log n}$,
we are in the high-degree regime, in which case the dimension of the
code becomes $k = r - g+1 = \Omega(r) = \Omega(\rho q^{3/2})$, and we
will thus have
\[
q = O((\log n / \rho)^{2/3}) \Rightarrow m = O\left(\big(\frac{d \log
    n}{1-p}\big)^{1+2u/3}\right)
\]
Altogether, we conclude that Construction~\ref{constr:strong} with
Hermitian codes results in a strongly $(d,e;u)$-disjunct matrix with
\[
m = O\left(\big( \frac{d \sqrt{\log n}}{1-p} + \big(\frac{d \log
    n}{1-p}\big)^{2/3} \big)^{u+3/2}\right)
\]
rows, where $e = p \cdot \Omega\left( d (\log n)/(1-p) + (d \sqrt{\log
    n} / (1-p))^{3/2} \right)$.  Compared to the Reed-Solomon codes,
the number of measurements has a slightly worse dependence on $d$, but
a much better dependence on $n$. Compared to AG codes on the TVZ
bound, the dependence on $d$ is better while the dependence on $n$ is
inferior.

\subsubsection*{Codes on the GV Bound}
A $q$-ary $(\tn,\tk,\td)$-code (of sufficiently large length) is said
to be on the Gilbert-Varshamov bound if it satisfies $\tk \geq
\tn(1-h_q(\td/\tn))$, where $h_q(\cdot)$ is the $q$-ary entropy
function defined as
\[
h_q(x) := x \log_q(q-1) - x \log_q(x) - (1-x) \log_q(1-x).
\]
It is well known that a random linear code achieves the bound with
overwhelming probability (cf.\ \cite{ref:MS}).  Now we apply
Lemma~\ref{lem:ks} on a code on the GV bound, and calculate the
resulting parameters.  Let $\rho := (1-p)/(4du)$, choose any alphabet
size $q \in [1/\rho, 2/\rho]$, and let $\C$ be any $q$-ary code of
length $\tn$ on the GV bound, with minimum distance $\td \geq \tn
(1-2/q)$. By the Taylor expansion of the function $h_q(x)$ around $x =
1-1/q$, we see that the dimension of $\C$ asymptotically behaves as $
\tk = \Theta(\tn / (q \log q)).  $ Thus the number of columns of the
resulting measurement matrix becomes $n = q^\tk = 2^{\Omega(\tn /
  q)}$, and therefore, the number $m$ of its rows becomes
\[
m = q^u \tn = O(q^{u+1} \log n) = O((d/(1-p))^{u+1} \log n),
\]
and the matrix would be strongly $(d,e;u)$-disjunct for \[ e = p \tn =
\Omega(p d (\log n)/(1-p)).\]

We remark that for the range of parameters that we are interested in,
Porat and Rothschild \cite{ref:PR08} have recently come up with a
deterministic construction of linear codes on the GV bound that runs
in time $\poly(q^\tk)$ (and thus, polynomial in the size of the
resulting measurement matrix). Their construction is based on a
derandomization of the probabilistic argument for random linear codes
using the method of conditional expectations, and as such, can be
considered \emph{weakly explicit} (in the sense that, the entire
measurement matrix can be computed in polynomial time in its length;
but for a fully explicit construction one must be ideally able to
deterministically compute any single entry of the measurement matrix
in time $\poly(d, \log n)$, which is not the case for this
construction).

\begin{table}
  \caption[Bounds obtained by constructions of strongly $(d,e;u)$-disjunct matrices]{Bounds obtained by strongly $(d,e;u)$-disjunct matrices. The
    noise parameter $p \in [0,1)$ is arbitrary. The first four rows correspond to
    the explicit coding-theoretic construction described in Section~\ref{sec:KSext},
    with the underlying code indicated as a remark.
  }

  \begin{center}
    \begin{tabular}{|l|l|p{50mm}|}
      \hline
      Number of rows & Noise tolerance & Remark \\ \hline \hline
      $O((\frac{d}{1-p})^{u+1} \log n)$ & $\Omega(p d \frac{\log n}{1-p})$ & Using codes on the GV bound. \\
      $O((\frac{d \log n}{1-p})^{u+1})$ & $\Omega(p d \frac{\log n}{1-p})$ & Using Reed-Solomon codes. \\
      $O((\frac{d}{1-p})^{2u+1} \log n)$ & $\Omega(p d \frac{\log n}{1-p})$ & Using Algebraic Geometric codes. \\
      $O((\frac{d \sqrt{\log n}}{1-p})^{u+3/2})$ & $\Omega(p (\frac{d \sqrt{\log n}}{1-p})^{3/2})$ & Using Hermitian codes ($d \gg \sqrt{\log n}$). \\
      \hline
      $O(d^{u+1} \frac{\log (n/d)}{(1-p)^2})$ & $\Omega(p d \frac{\log (n/d)}{(1-p)^2})$ & Probabilistic construction. \\
      $\Omega(d^{u+1} \log_d n + e d^{u})$ & $e$ & Lower bound (Section~\ref{sec:stronglyDisjunct}). \\
      \hline
    \end{tabular}
  \end{center}
  \label{tab:params:strongly}
\end{table}

We see that, for a fixed $p$, Construction~\ref{constr:strong} when
using codes on the GV bound achieves almost optimal
parameters. Moreover, the explicit construction based on the
Reed-Solomon codes possesses the ``right'' dependence on the sparsity
$d$, AG codes on the TVZ bound have a matching dependence on the
vector length $n$ with random measurement matrices, and finally, the
trade-off offered by the construction based on Hermitian codes lies in
between the one for Reed-Solomon codes and AG codes. These parameters
are summarized in Table~\ref{tab:params:strongly}. Note that the
special case $u=1$ would give classical $(d,e)$-disjunct matrices as
in Definition \ref{def:classicDisjunct}.

\subsection{Disjunct Matrices for Threshold Testing}

Even though, as discussed above, the general notion of strongly
$(d,e;u)$-disjunct matrices is sufficient for threshold group testing
with upper threshold $u$, in this section we show that a weaker notion
of disjunct matrices (which turns out to be \emph{strictly} weaker
when the lower threshold $\ell$ is greater than $1$), would also
suffice.  We proceed by showing how such measurement matrices can be
constructed.

Before introducing our variation of disjunct matrices, let us fix some
notation that will be useful for the threshold model.  Consider the
threshold model with thresholds $\ell$ and $u$, and an $m \times n$
measurement matrix $\cM$ that defines the set of measurements.  For a
vector $x \in \zo^n$, denote by $\cM[x]_{\ell,u}$ the set of vectors
in $\zo^m$ that correctly encode the measurement outcomes
corresponding to the vector $x$.  In particular, for any $y \in
\cM[x]_{\ell,u}$ we have $y(i) = 1$ if $|\supp(\cM_j) \cap \supp(x)|
\geq u$, and $y(i) = 0$ if $|\supp(\cM_j) \cap \supp(x)| < \ell$,
where $\cM_j$ indicates the $j$th row of $\cM$.  In the gap-free case,
the set $\cM[x]_{\ell,u}$ may only have a single element that we
denote by $\cM[x]_u$. Note that the gap-free case with $u=1$ reduces
to ordinary group testing, and thus we have $\cM[x]_1 = \cM[x]$.

To make the main ideas more transparent, until Section~\ref{sec:gap}
we will focus on the gap-free case where $\ell = u$. The extension to
nonzero gaps is straightforward and will be discussed in
Section~\ref{sec:gap}. Moreover, often we will implicitly assume that
the Hamming weight of the Boolean vector that is to be identified is
at least $u$ (since otherwise, any $(u-1)$-sparse vector would be
confused with the all-zeros vector). Moreover, we will take the
thresholds $\ell, u$ as fixed constants while the parameters $d$ and
$n$ are allowed to grow.

\subsubsection{The Definition and Properties}

Our variation of disjunct matrices along with an ``auxiliary'' notion
of \emph{regular} matrices is defined in the following.

\begin{defn} \label{def:regularMatrix} \index{regular
    matrix}\index{disjunct!$(d,e;u)$-disjunct} A Boolean matrix $\cM$
  with $n$ columns is called $(d,e;u)$-regular if for every subset of
  columns $S \subseteq [n]$ (called the \emph{critical set}) and every
  $Z \subseteq [n]$ (called the \emph{zero set}) such that $u \leq |S|
  \leq d$, $|Z| \leq |S|$, $S \cap Z = \emptyset$, there are more than
  $e$ rows of $\cM$ at which $\cM|_S$ has weight exactly $u$ and (at
  the same rows) $\cM|_Z$ has weight zero. Any such row is said to
  \emph{$u$-satisfy} $S$ and $Z$.

  If, in addition, for every \emph{distinguished column} $i \in S$,
  more than $e$ rows of $\cM$ both $u$-satisfy $S$ and $Z$ and have a
  $1$ at the $i$th column, the matrix is called $(d,e;u)$-disjunct
  (and the corresponding ``good'' rows are said to $u$-satisfy $i$,
  $S$, and $Z$).
\end{defn}

It is easy to verify that (assuming $2d \leq n$) the classical notion
of $(2d-1,e)$-disjunct matrices is equivalent to strongly
$(2d-1,e;1)$-disjunct and $(d,e;1)$-disjunct.  Moreover, any
$(d,e;u)$-disjunct matrix is $(d,e;u)$-regular, $(d-1,e;u-1)$-regular,
and $(d,e)$-disjunct (but the reverse implications do not in general
hold). Therefore, the lower bound \[m = \Omega(d^{2} \log_d n + ed)\]
that applies for $(d,e)$-disjunct matrices holds for
$(d,e;u)$-disjunct matrices as well.

Below we show that our notion of disjunct matrices is necessary and
sufficient for the purpose of threshold group testing:

\begin{lem} \label{lem:disjunct} Let $\cM$ be an $m \times n$ Boolean
  matrix that is $(d,e;u)$-disjunct. Then for every distinct
  $d$-sparse vectors $x, x' \in \zo^n$ such that\footnote{ Note that
    at least one of the two possible orderings of any two distinct
    $d$-sparse vectors, at least one having weight $u$ or more,
    satisfies this condition.} $\supp(x) \nsubseteq \supp(x')$,
  $\wgt(x) \geq |\supp(x') \setminus \supp(x)|$ and $\wgt(x) \geq u$,
  we have
  \begin{equation} \label{eqn:distin} |\supp(\cM[x]_u) \setminus
    \supp(\cM[x']_u)| > e.
  \end{equation}
  Conversely, assuming $d \geq 2u$, if $\cM$ satisfies
  \eqref{eqn:distin} for every choice of $x$ and $x'$ as above, it
  must be $(\lfloor d/2 \rfloor,e;u)$-disjunct.
\end{lem}

  \begin{Proof}
    First, suppose that $\cM$ is $(d,e;u)$-disjunct, and let $y :=
    \cM[x]_u$ and $y' := \cM[x']_u$.  Take any $i \in \supp(x)
    \setminus \supp(x')$, and let $S := \supp(x)$ and $Z := \supp(x')
    \setminus \supp(x)$.  Note that $|S| \leq d$ and by assumption, we
    have $|Z| \leq |S|$.
    Now, Definition~\ref{def:regularMatrix} implies that there is a
    set $E$ of more than $e$ rows of $M$ that $u$-satisfy $i$ as the
    distinguished column, $S$ as the critical set and $Z$ as the zero
    set. Thus for every $j \in E$, the $j$th row of $\cM$ restricted
    to the columns chosen by $\supp(x)$ must have weight exactly $u$,
    while its weight on $\supp(x')$ is less than $u$. Therefore, $y(j)
    = 1$ and $y'(j) = 0$ for more than $e$ choices of $j$.

    For the converse, consider any choice of a distinguished column $i
    \in [n]$, a critical set $S \subseteq [n]$ containing $i$ (such
    that $|S| \geq u$), and a zero set $Z \subseteq [n]$ where $|Z|
    \leq |S|$.  Define $d$-sparse Boolean vectors $x, x' \in \zo^n$ so
    that $\supp(x) := S$ and $\supp(x') := S \cup Z \setminus \{ i\}$.
    Let $y := \cM[x]_u$ and $y' := \cM[x']_u$ and $E := \supp(y)
    \setminus \supp(y')$. By assumption we know that $|E| > e$. Take
    any $j \in E$. Since $y(j) = 1$ and $y'(j) = 0$, we get that the
    $j$th row of $\cM$ restricted to the columns picked by $S \cup Z
    \setminus \{i\}$ must have weight at most $u-1$, whereas it must
    have weight at least $u$ when restricted to $S$.  As the sets
    $\{i\}, S \setminus \{i\}$, and $Z$ are disjoint, this can hold
    only if $\cM[j, i] = 1$, and moreover, the $j$th row of $\cM$
    restricted to the columns picked by $S$ (resp., $Z$) has weight
    exactly $u$ (resp., zero).  Hence, this row (as well as all the
    rows of $\cM$ picked by $E$) must $u$-satisfy $i, S$, and $Z$,
    confirming that $\cM$ is $(\lfloor d/2 \rfloor,e;u)$-disjunct.
  \end{Proof}

  \newConstruction[b]%
  {Direct product of measurement matrices.}%
  {Boolean matrices $\cM_1$ and $\cM_2$ that are
      $m_1 \times n$ and $m_2 \times n$, respectively.}%
      {An $m \times n$ Boolean matrix $\cM_1 \rep
      \cM_2$, where $m := m_1 m_2$. \index{notation!$\cM_1 \rep
        \cM_2$}}%
{Let the rows of $\cM := \cM_1 \rep
      \cM_2$ be indexed by the set $[m_1] \times [m_2]$. Then the row
      corresponding to $(i,j)$ is defined as the bit-wise or of the
      $i$th row of $\cM_1$ and the $j$th row of $\cM_2$.}%
{constr:replace}

  We will use regular matrices as intermediate building blocks in our
  constructions of disjunct matrices to follow. The connection with
  disjunct matrices is made apparent through a direct product of
  matrices defined in Construction~\ref{constr:replace}.
  Intuitively, using this product, regular matrices can be used to
  transform any measurement matrix suitable for the standard group
  testing model to one with comparable properties in the threshold
  model. The following lemma formalizes this idea.

  \begin{lem} \label{lem:rep} Let $\cM_1$ and $\cM_2$ be Boolean
    matrices with $n$ columns, such that $\cM_1$ is
    $(d-1,e_1;u-1)$-regular. Let $\cM := \cM_1 \rep \cM_2$, and
    suppose that for $d$-sparse Boolean vectors $x, x' \in \zo^n$ such
    that $\wgt(x) \geq \wgt(x')$, we have
    \[
    | \supp(\cM_2[x]_1) \setminus \supp(\cM_2[x']_1)| \geq e_2.
    \]
    Then, $ |\supp(\cM[x]_u) \setminus \supp(\cM[x']_u)| \geq (e_1+1)
    e_2.  $
  \end{lem}

  \begin{Proof}
    First we consider the case where $u > 1$.  Let $y := \cM_2[x]_1
    \in \zo^{m_2}$, $y' := \cM_2[x']_1 \in \zo^{m_2}$, where $m_2$ is
    the number of rows of $\cM_2$, and let $E := \supp(y) \setminus
    \supp(y')$.  By assumption, $|E| \geq e_2$. Fix any $i \in E$ so
    that $y(i) = 1$ and $y'(i) = 0$.
    Therefore, the $i$th row of $\cM_2$ must have all zeros at
    positions corresponding to $\supp(x')$ and there is a $j \in
    \supp(x) \setminus \supp(x')$ such that $\cM_2[i, j] = 1$.  Define
    $S := \supp(x) \setminus \{j\}$, $Z := \supp(x') \setminus
    \supp(x)$, $z := \cM[x]_u$ and $z' := \cM[x']_u$.

    As $\wgt(x) \geq \wgt(x')$, we know that $|Z| \leq |S|+1$.  The
    extreme case $|Z| = |S|+1$ only happens when $x$ and $x'$ have
    disjoint supports, in which case one can remove an arbitrary
    element of $Z$ to ensure that $|Z| \leq |S|$ and the following
    argument (considering the assumption $u>1$) still goes through.
    By the definition of regularity, there is a set $E_1$ consisting
    of at least $e_1+1$ rows of $\cM_1$ that $(u-1)$-satisfy the
    critical set $S$ and the zero set $Z$. Pick any $k \in E_1$, and
    observe that $z$ must have a $1$ at position $(k,i)$. This is
    because the row of $\cM$ indexed by $(k,i)$ has a $1$ at the $j$th
    position (since the $i$th row of $\cM_2$ does), and at least $u-1$
    more $1$'s at positions corresponding to $\supp(x) \setminus
    \{j\}$ (due to regularity of $\cM_1$).  On the other hand, note
    that the $k$th row of $\cM_1$ has at most $u-1$ ones at positions
    corresponding to $\supp(x')$ (because $\supp(x') \subseteq S \cup
    Z$), and the $i$th row of $\cM_2$ has all zeros at those positions
    (because $y'(i)=0$).  This means that the row of $\cM$ indexed by
    $(k,i)$ (which is the bit-wise or of the $k$th row of $\cM_1$ and
    the $i$th row of $\cM_2$) must have less than $u$ ones at
    positions corresponding to $\supp(x')$, and thus, $z'$ must be $0$
    at position $(k,i)$.  Therefore, $z$ and $z'$ differ at position
    $(k,i)$.

    Since there are at least $e_2$ choices for $i$, and for each
    choice of $i$, at least $e_1+1$ choices for $k$, we conclude that
    in at least $(e_1+1) e_2$ positions, $z$ has a one while $z'$ has
    a zero.

    The argument for $u=1$ is similar, in which case it suffices to
    take $S := \supp(x)$ and $Z := \supp(x') \setminus \supp(x)$.
  \end{Proof}

  As a corollary it follows that, when $\cM_1$ is a
  $(d-1,e_1;u-1)$-regular and $\cM_2$ is a $(d,e_2)$-disjunct matrix,
  the product $\cM := \cM_1 \rep \cM_2$ will distinguish between any
  two distinct $d$-sparse vectors (of weight at least $u$) in at least
  $(e_1+1)(e_2+1)$ positions of the measurement outcomes.  This
  combined with Lemma~\ref{lem:disjunct} would imply that $\cM$ is, in
  particular, $(\lfloor d/2 \rfloor, (e_1+1) (e_2+1) -1; u)$-disjunct.
  However, using a direct argument similar to the above lemma it is
  possible to obtain a slightly better result, given by
  Lemma~\ref{lem:repdisj} (the proof follows the same line of argument
  as that of Lemma~\ref{lem:rep} and is thus omitted).

  \begin{lem} \label{lem:repdisj} Suppose that $\cM_1$ is a
    $(d,e_1;u-1)$-regular and $\cM_2$ is a $(2d,e_2)$-disjunct matrix.
    Then $\cM_1 \rep \cM_2$ is a $(d, (e_1+1) (e_2+1) -1; u)$-disjunct
    matrix. \qed
  \end{lem}

  As another particular example, we remark that the resilient
  measurement matrices that we constructed in
  Section~\ref{sec:nrConstr} for the ordinary group testing model
  can be combined with regular matrices to offer the same qualities
  (i.e., approximation of sparse vectors in highly noisy settings) in
  the threshold model. In the same way, numerous existing results in
  group testing can be ported to the threshold model by using
  Lemma~\ref{lem:rep} (e.g., constructions of measurement matrices
  suitable for trivial two-stage schemes; cf.\ \cite{ref:CD08}).

  \subsubsection{Constructions} \label{sec:constr}

  In this section, we obtain several constructions of regular and
  disjunct matrices. Our first construction, described in
  Construction~\ref{constr:prob}, is a randomness-efficient
  probabilistic construction that can be analyzed using standard
  techniques from the probabilistic method. The bounds obtained by
  this construction are given by Lemma~\ref{lem:probDisjunct} below.
  The amount of random bits required by this construction is
  polynomially bounded in $d$ and $\log n$, which is significantly
  smaller than it would be had we picked the entries of $\cM$ fully
  independently.

\newConstruction%
{Probabilistic construction of regular and disjunct matrices.}%
{Integer parameters $n, m', d, u$.}%
{An $m \times n$ Boolean matrix $\cM$, where
      $m := m' \lceil \log(d/u) \rceil$.}%
{Let $r := \lceil \log(d/u)
      \rceil$. Index the rows of $\cM$ by $[r] \times [m']$.
%
      Sample the $(i, j)$th row of $\cM$ independently from a
      $(u+1)$-wise independent distribution on $n$ bit vectors, where
      each individual bit has probability $1/(2^{i+2} u)$ of being
      $1$.}%
{constr:prob}

  \begin{lem} \label{lem:probDisjunct} For every $p \in [0,1)$ and
    integer parameter $u > 0$, Construction~\ref{constr:prob}
    with\footnote{The subscript in $O_u(\cdot)$ and $\Omega_u(\cdot)$
      implies that the hidden constant in the asymptotic notation is
      allowed to depend on $u$.}
    $m' = O_{u}(d \log (n/d)/(1-p)^2)$ (resp., $m' = O_{u}(d^2 \log
    (n/d)/(1-p)^2)$) outputs a $(d, \Omega_{u}(pm');u)$-regular
    (resp., $(d,\Omega_{u}(pm'/d);u)$-disjunct) matrix with
    probability $1-o(1)$.
  \end{lem}

  \begin{Proof}
    We show the claim for regular matrices, the proof for disjunct
    matrices is similar.  Consider any particular choice of a critical
    set $S \subseteq [n]$ and a zero set $Z \subseteq [n]$ such that
    $u \leq |S| \leq d$ and $|Z| \leq |S|$. Choose an integer $i$ so
    that $2^{i-1} u \leq |S| \leq 2^i u$, and take any $j \in
    [m']$. Denote the $(i,j)$th row of $\cM$ by the random variable
    $\rv{w} \in \zo^n$, and by $q$ the ``success'' probability that
    $\rv{w}|_S$ has weight exactly $u$ and $\rv{w}|_Z$ is all zeros.
    For an integer $\ell > 0$, we will use the shorthand $1^\ell$
    (resp., $0^\ell$) for the all-ones (resp., all-zeros) vector of
    length $\ell$.  We have

\begin{align*}
  q &= \sum_{\substack{R \subseteq [S] \\ |R| = u}} \Pr[(\rv{w}|_R) = 1^u \land (\rv{w}|_{Z \cup (S \setminus R)}) = 0^{|S|+|Z|-u}] \\
  &= \sum_{R} \Pr[(\rv{w}|_R) = 1^u] \cdot \Pr[(\rv{w}|_{Z \cup (S \setminus R)}) = 0^{|S|+|Z|-u} \mid (\rv{w}|_R) = 1^u] \displaybreak[0]\\
  &\stackrel{\mathrm{(a)}}{=} \sum_{R} (1/(2^{i+2} u))^u \cdot (1-\Pr[(\rv{w}|_{Z \cup (S \setminus R)}) \neq 0^{|S|+|Z|-u} \mid (\rv{w}|_R) = 1^u]) \displaybreak[0]\\
  &\stackrel{\mathrm{(b)}}{\geq} \sum_{R} (1/(2^{i+2} u))^u \cdot (1-(|S|+|Z|-u)/(2^{i+2} u)) \\
  &\stackrel{\mathrm{(c)}}{\geq} \frac{1}{2} \binom{|S|}{u} (1/(2^{i+2} u))^u \geq \frac{1}{2} \left(\frac{|S|}{u}\right)^u \cdot (1/(2^{i+2} u))^u \geq \frac{1}{2^{3u+1} \cdot u^u} =: c, \\
\end{align*}
where $\mathrm{(a)}$ and $\mathrm{(b)}$ use the fact that the entries
of $\rv{w}$ are $(u+1)$-wise independent, and $\mathrm{(b)}$ uses an
additional union bound. Moreover, in $\mathrm{(c)}$ the binomial term
counts the number of possibilities for the set $R$. Note that the
lower bound $c > 0$ obtained at the end is a constant that only
depends on $u$.  Now, let $e := m'p q$, and observe that the expected
number of ``successful'' rows is $m' q$. Using Chernoff bounds, and
independence of the rows, the probability that there are at most $e$
rows (among $(i,1), \ldots, (i, m')$) whose restriction to $S$ and $Z$
has weights $u$ and $0$, respectively, becomes upper bounded by
\[
\exp( -(m'q-e)^2/(2m'q) ) = \exp( -(1-p)^2 m' q/2 ) \leq \exp(
-(1-p)^2 m' c/2 ).
\]
Now take a union bound on all the choices of $S$ and $Z$ to conclude
that the probability that the resulting matrix is not $(d,e;
u)$-regular is at most
\begin{gather*}
  \left(\sum_{s=u}^{d} \binom{n}{s} \sum_{z=0}^{s}
    \binom{n-s}{z}\right)
  \exp( -(1-p)^2 m' c/2 ) \\
  \leq d^2 \binom{n}{d}^2 \exp( -(1-p)^2 m' c/2 ),
\end{gather*}
which can be made $o(1)$ by choosing $m' = O_{u}(d
\log(n/d)/(1-p)^2)$.
\end{Proof}

\newConstruction%
{A building block for construction of regular matrices.}%
{A strong lossless $(\tk, \eps)$-condenser
    $f\colon \zo^\tn \times \zo^\tee \to \zo^\tl$, integer parameter
    $u \geq 1$ and real parameter $p \in [0,1)$ such that $\eps <
    (1-p)/16$,}%
{An $m \times n$ Boolean matrix $\cM$, where $n
    := 2^\tn$ and $ m = 2^{\tee+\tk} O_u(2^{u(\tl-\tk)}) $.}%
    {Let $G_1=(\zo^\tl, \zo^\tk, E_1)$ be any bipartite bi-regular
    graph with left vertex set $\zo^\tl$, right vertex set $\zo^\tk$,
    left degree $d_\ell := 8u$, and right degree $d_r := 8u
    2^{\tl-\tk}$.  Replace each right vertex $v$ of $G_1$ with
    $\binom{d_r}{u}$ vertices, one for each subset of size $u$ of the
    vertices on the neighborhood of $v$, and connect them to the
    corresponding subsets. Denote the resulting graph by $G_2 =
    (\zo^\tl, V_2, E_2)$, where $|V_2| = 2^\tk \binom{d_r}{u}$.
    Define the bipartite graph $G_3=(\zo^n, V_3, E_3)$, where $V_3 :=
    \zo^t \times V_2$, as follows: Each left vertex $x \in \zo^n$ is
    connected to $(y, \Gamma_2(f(x,y))$, for each $y \in \zo^t$, where
    $\Gamma_2(\cdot)$ denotes the neighborhood function of $G_2$
    (i.e., $\Gamma_2(v)$ denotes the set of vertices adjacent to $v$
    in $G_2$).  The output matrix $\cM$ is the bipartite adjacency
    matrix of $G_3$.}%
{constr:main}

\newConstruction%
{Regular matrices from strong lossless condensers.}%
{Integer parameters $d \geq u \geq 1$, real
    parameter $p \in [0,1)$, and a family $f_0, \ldots, f_r$ of strong
    lossless condensers, where $r := \lceil \log(d/u') \rceil$ and
    $u'$ is the smallest power of two such that $u' \geq u$.  Each
    $f_i\colon \zo^\tn \times \zo^\tee \to \zo^{\tl(i)}$ is assumed to
    be a strong lossless $(\tk(i), \eps)$-condenser, where $\tk(i) :=
    \log u'+i+1$ and $\eps < (1-p)/16$.}%
{An $m \times n$ Boolean matrix $\cM$, where $n
    := 2^\tn$ and $
    m = 2^{\tee} d \sum_{i=0}^r O_u(2^{u (\tl(i)-\tk(i))}) $.}%
{For each $i \in \{0, \ldots, r\}$, denote
    by $\cM_i$ the output matrix of Construction~\ref{constr:main}
    when instantiated with $f_i$ as the underlying condenser, and by
    $m_i$ its number of rows.  Define $r_i := 2^{r-i}$ and let
    $\cM'_i$ denote the matrix obtained from $\cM_i$ by repeating each
    row $r_i$ times.  Construct the output matrix $\cM$ by stacking
    $\cM'_0, \ldots, \cM'_r$ on top of one another.}%
{constr:reg}

Now we turn to a construction of regular matrices using strong
lossless condensers. Details of the construction are described in
Construction~\ref{constr:reg} that assumes a family of lossless
condensers with different entropy requirements\footnote{We have
  assumed that all the functions in the family have the same seed
  length $t$. If this is not the case, one can trivially set $t$ to be
  the largest seed length in the family.}, and in turn, uses
Construction~\ref{constr:main} as a building block.

The following theorem analyzes the obtained parameters without
specifying any particular choice for the underlying family of
condensers.

\begin{thm} \label{thm:regular} The $m \times n$ matrix $\cM$ output
  by Construction~\ref{constr:reg} is $(d,p \gamma 2^t;u)$-regular,
  where $\gamma = \max\{1, \Omega_u(d \cdot \min\{2^{\tk(i)-\tl(i)}
  \colon i=0,\ldots,r \}) \}$.
\end{thm}

\begin{Proof}
  As a first step, we verify the upper bound on the number of
  measurements $m$. Each matrix $\cM_i$ has $m_i = 2^{t+\tk(i)}
  O_u(2^{u(\tl(i)-\tk(i))})$ rows, and $M'_i$ has $m_i r_i$ rows,
  where $r_i = 2^{r-i}$. Therefore, the number of rows of $M$ is
  \[
  \sum_{i=0}^r r_i m_i = \sum_{i=0}^r 2^{t+\log u' + r+1} m_i = 2^t d
  \sum_{i=0}^r O_u(2^{u(\tl(i)-\tk(i))}).
  \]

  Let $S, Z \subseteq \zo^\tn$ respectively denote any choice of a
  critical set and zero set of size at most $d$, where $|Z| \leq |S|$,
  and choose an integer $i \geq 0$ so that $2^{i-1} u' \leq |S| \leq
  2^i u'$.  Arbitrarily grow the two sets $S$ and $Z$ to possibly
  larger, and disjoint, sets $S' \supseteq S$ and $Z' \supseteq Z$
  such that $|S'| = |Z'| = 2^i u'$ (for simplicity we have assumed
  that $d \leq n/2$).  Our goal is to show that there are ``many''
  rows of the matrix $\cM_i$ (in Construction~\ref{constr:reg}) that
  $u$-satisfy $S$ and $Z$.

  Let $\tk := \tk(i) = \log u' + i+1$, $\tl := \tl(i)$, and denote by
  $G_1, G_2, G_3$ the bipartite graphs used by the instantiation of
  Construction~\ref{constr:main} that outputs $\cM_i$. Thus we need to
  show that ``many'' right vertices of $G_3$ are each connected to
  exactly $u$ of the vertices in $S$ and none of those in $Z$.

  Consider the uniform distribution $\cX$ on the set $S' \cup Z'$,
  which has min-entropy $\log u' + i+1$.  By an averaging argument,
  since the condenser $f_i$ is strong, for more than a $p$ fraction of
  the choices of the seed $y \in \zo^\tee$ (call them \emph{good
    seeds}), the distribution $\cZ_y := f_i(\cX, y)$ is
  $\eps/(1-p)$-close (in particular, $1/16$-close) to a distribution
  with min-entropy $\log u' + i+1$.

  Fix any good seed $y \in \zo^\tee$.  Let $G=(\zo^\tn, \zo^\tl, E)$
  denote a bipartite graph representation of $f_i$, where each left
  vertex $x \in \zo^\tn$ is connected to $f_i(x,y)$ on the
  right. Denote by $\Gamma_y(S' \cup Z')$ the right vertices of $G$
  corresponding to the neighborhood of the set of left vertices picked
  by $S' \cup Z'$. Note that $\Gamma_y(S' \cup Z') = \supp(\cZ_y)$.
  Using Proposition~\ref{prop:flatmap} in the appendix, we see that
  since $\cZ_y$ is $1/16$-close to having min-entropy $\log(|S' \cup
  Z'|)$, there are at least $(7/8) |S' \cup Z'|$ vertices in
  $\Gamma(S' \cup Z')$ that are each connected to exactly one left
  vertex in $S' \cup Z'$.  Since $|S| \geq |S' \cup Z'| / 4$, this
  implies that at least $|S' \cup Z'|/8$ vertices in $\Gamma(S' \cup
  Z')$ (call them $\Gamma'_y$) are connected to exactly one left
  vertex in $S$ and no other vertex in $S' \cup Z'$.  In particular we
  get that $|\Gamma'_y| \geq 2^{k-3}$.

  Now, in $G_1$, let $T_y$ be the set of left vertices corresponding
  to $\Gamma'_y$ (regarding the left vertices of $G_1$ in one-to-one
  correspondence with the right vertices of $G$).  The number of edges
  going out of $T_y$ in $G_1$ is $d_\ell |T_y| \geq u 2^k$. Therefore,
  as the number of the right vertices of $G_1$ is $2^k$, there must be
  at least one right vertex that is connected to at least $u$ vertices
  in $T_y$.  Moreover, a counting argument shows that the number of
  right vertices connected to at least $u$ vertices in $T_y$ is also
  at least $2^{k-\tl} 2^k/(10u)$.

  Observe that in construction of $G_2$ from $G_1$, any right vertex
  of $G_1$ is replicated $\binom{d_r}{u}$ times, one for each
  $u$-subset of its neighbors. Therefore, for a right vertex of $G_1$
  that is connected to \emph{at least} $u$ left vertices in $T_y$, one
  or more of its copies in $G_2$ must be connected to \emph{exactly}
  $u$ vertex in $T_y$ (among the left vertices of $G_2$) and no other
  vertex (since the right degree of $G_2$ is equal to $u$).

  Define $\gamma' := \max\{1, 2^{k-\tl} 2^k/(10u)\}$. From the
  previous argument we know that, looking at $T_y$ as a set of left
  vertices of $G_2$, there are at least $\gamma'$ right vertices on
  the neighborhood of $T_y$ in $G_2$ that are connected to exactly $u$
  of the vertices in $T_y$ and none of the left vertices outside
  $T_y$. Letting $v_y$ be any such vertex, this implies that the
  vertex $(y,v_y) \in V_3$ on the right part of $G_3$ is connected to
  exactly $u$ of the vertices in $S$, and none of the vertices in $Z$.
  Since the argument holds for every good seed $y$, the number of such
  vertices is at least the number of good seeds, which is more than $p
  \gamma' 2^t$. Since the rows of the matrix $m_i$ are repeated $r_i =
  2^{r-i}$ times in $M$, we conclude that $M$ has at least $p \gamma'
  2^{t+r-i} \geq p \gamma 2^t$ rows that $u$-satisfy $S$ and $Z$, and
  the claim follows.
\end{Proof}

\subsubsection*{Instantiations} \label{sec:instan}

We now instantiate the result obtained in Theorem~\ref{thm:regular} by
various choices of the family of lossless condensers. The crucial
factors that influence the number of measurements are the seed length
and the output length of the condenser. In particular, we will
consider optimal lossless condensers (with parameters achieved by
random functions), zig-zag based construction of
Theorem~\ref{thm:CRVW}, and the coding-theoretic construction of
Guruswami et al., quoted in Theorem~\ref{thm:GUVcond}. The results are
summarized in the following theorem.


%

\begin{thm} \label{thm:instan} Let $u > 0$ be fixed, and $p \in [0,1)$
  be a real parameter. Then for integer parameters $d, n \in \N$ where
  $u \leq d \leq n$,

  \begin{enumerate}
  \item Using an optimal lossless condenser in
    Construction~\ref{constr:reg} results in an $m_1 \times n$ matrix
    $\cM_1$ that is $(d,e_1; u)$-regular, where \[ m_1 = O(d (\log n)
    (\log d) / (1-p)^{u+1}) \] and $e_1 = \Omega(p d \log n)$,

  \item Using the lossless condenser of Theorem~\ref{thm:CRVW} in
    Construction~\ref{constr:reg} results in an $m_2 \times n$ matrix
    $\cM_2$ that is $(d,e_2; u)$-regular, where \[ m_2 = O(T_2 d (\log
    d) / (1-p)^u)\] for some \[ T_2 = \exp(O(\log^3((\log n)/(1-p))))
    = \qpoly(\log n),\] and $e_2 = \Omega(pd T_2 (1-p))$.

  \item Let $\beta > 0$ be any fixed constant. Then
    Construction~\ref{constr:reg} can be instantiated using the
    lossless condenser of Theorem~\ref{thm:GUVcond} so that we obtain
    an $m_3 \times n$ matrix $\cM_3$ that is $(d, e_3; u)$-regular,
    where \[ m_3 = O(T_3^{1+u} d^{1+\beta} (\log d))\] for \[ T_3 :=
    ((\log n) (\log d)/(1-p))^{1+u/\beta} = \poly(\log n, \log d),\]
    and $e_3 = \Omega(p \max\{ T_3, d^{1-\beta/u}\})$.
  \end{enumerate}
\end{thm}

\begin{Proof}
  First we show the claim for $M_1$. In this case, we take each $f_i$
  in Construction~\ref{constr:reg} to be an optimal lossless condenser
  satisfying the bounds obtained in\footnote{This result is similar in
    spirit to the probabilistic argument used in
    \cite{ref:lowerbounds} for showing the existence of good
    extractors.} \cite{ref:CRVW02}.  Thus we have that $2^\tee =
  O(\tn/\eps)=O(\log n/\eps)$, and for every $i = 0, \ldots, r$, we
  have $2^{\tl(i)-\tk(i)} = O(1/\eps)$, where $\eps = O(1-p)$. Now we
  apply Theorem~\ref{thm:regular} to obtain the desired bounds (and in
  particular, $\gamma = \Omega(\eps d)$).

  Similarly, for the construction of $\cM_2$ we set up each $f_i$
  using the explicit construction of condensers in
  Theorem~\ref{thm:CRVW} for min-entropy $\tk(i)$. In this case, the
  maximum required seed length is $t = O(\log^3(\tn/\eps))$, and we
  let
  \[ T_2 := 2^t = \exp(O(\log^3((\log n)/(1-p)))).\] Moreover, for
  every $i = 0, \ldots, r$, we have $2^{\tl(i)-\tk(i)} =
  O(1/\eps)$. Plugging these parameters in Theorem~\ref{thm:regular}
  gives $\gamma = \Omega(\eps d)$ and the bounds on $m_2$ and $e_2$
  follow.

  Finally, for $\cM_3$ we use Theorem~\ref{thm:GUVcond} with $\alpha
  := \beta/u$.  Thus the maximum seed length becomes \[t = (1+u/\beta)
  \log(\tn (\log d)/(1-p)) + O(1),\] and for every $i = 0, \ldots, r$,
  we have $\tl(i)-\tk(i) = O(t+\beta (\log d)/u)$.  Clearly, $T_3 =
  \Theta(2^t)$, and thus (using Theorem~\ref{thm:regular}) the number
  of measurements becomes $m_3 = T^{1+u} d^{1+\beta} (\log
  d)$. Moreover, we get \[\gamma = \max\{1, \Omega(d^{1-\beta/u} /T
  )\},\] which gives \[ e_3 = \Omega(pT\gamma) = p \max\{T,
  d^{1-\beta/u}\},\] as claimed.
\end{Proof}

By combining this result with Lemma~\ref{lem:repdisj} using any
explicit construction of classical disjunct matrices, we will obtain
$(d,e;u)$-disjunct matrices that can be used in the threshold model
with any fixed threshold, sparsity $d$, and error tolerance $\lfloor
e/2 \rfloor$.

In particular, using the coding-theoretic explicit construction of
nearly optimal classical disjunct matrices (see
Table~\ref{tab:params:strongly}), we obtain $(d,e;u)$-disjunct
matrices with \[ m=O(m' d^2 (\log n)/(1-p)^2) \] rows and error
tolerance \[ e = \Omega(e' p d (\log n)/(1-p)),\] where $m'$ and $e'$
are respectively the number of rows and error tolerance of any of the
regular matrices obtained in Theorem~\ref{thm:instan}.

\begin{table}
  \caption[Summary of the parameters achieved by various threshold testing schemes]{Summary of the parameters achieved by various threshold testing schemes. The
    noise parameter $p \in [0,1)$ is arbitrary, and thresholds $\ell, u = \ell+g$ are fixed constants.
    ``Exp'' and ``Rnd'' respectively indicate explicit and randomized constructions.
  }

  \begin{center}
    \begin{tabular}{|l|l|p{60mm}|}
      \hline
      Number of rows & Tolerable & Remarks \\
      & errors & \\
      \hline \hline
      $O(d^{g+2} \frac{(\log d) \log (n/d)}{(1-p)^2})$ & $\Omega(p d \frac{\log (n/d)}{(1-p)^2})$ & Rnd: Construction~\ref{constr:prob}. \\
      $O(d^{g+3} \frac{(\log d) \log^2 n}{(1-p)^{2}})$ & $\Omega(p d^2 \frac{\log^2 n}{(1-p)^2})$ &
      Constructions \ref{constr:reg}~and~\ref{constr:replace} combined, assuming optimal condensers and strongly disjunct matrices. \\
      $O(d^{g+3} \frac{(\log d) T_2 \log n}{(1-p)^{g+2}})$ & $\Omega(p d^2 \frac{T_2 \log n}{1-p})$ &
      Exp $(\star)$ \\
%

      $O(d^{g+3+\beta} \frac{T_3^{\ell} \log n}{(1-p)^{g+2}})$ &
      $\Omega(p d^{2-\beta} \frac{ \log n}{1-p} )$ &
      Exp $(\star\star)$ \\
      \hline \hline
      $\Omega(d^{g+2} \log_d n + e d^{g+1})$ & $e$ & Lower bound (see Section~\ref{sec:gap}). \\
      \hline
    \end{tabular}
  \end{center}
  \begin{description}
  \item [$(\star)$] Constructions
    \ref{constr:reg}~and~\ref{constr:replace} combined using
    Theorem~\ref{thm:CRVW} and \cite{ref:PR08}, where $T_2 =
    \exp(O(\log^3 \log n )) = \qpoly(\log n)$.
  \item [$(\star\star)$] Constructions
    \ref{constr:reg}~and~\ref{constr:replace} combined using
    Theorem~\ref{thm:GUVcond} and \cite{ref:PR08}, where $\beta>0$ is
    any arbitrary constant and $T_3 = ((\log n)(\log d))^{1+u/\beta} =
    \poly(\log n, \log d)$.
  \end{description}

   \label{tab:params}
   \hrule
 \end{table}

 We note that in all cases, the final dependence on the sparsity
 parameter $d$ is, roughly, $O(d^3)$ which has an exponent independent
 of the threshold $u$.  Table~\ref{tab:params} summarizes the obtained
 parameters for the general case (with arbitrary gaps).  We see that,
 when $d$ is not negligibly small (e.g., $d=n^{1/10}$), the bounds
 obtained by our explicit constructions are significantly better than
 those offered by strongly disjunct matrices (as in
 Table~\ref{tab:params:strongly}).

 \subsubsection{The Case with Positive Gaps} \label{sec:gap}

 In preceding sections we have focused on the case where $g=0$.
 However, we observe that all the techniques that we have developed so
 far can be extended to the positive-gap case in a straightforward
 way. The main observations are as follows.

 \begin{enumerate}
 \item Definition~\ref{def:regularMatrix} can be adapted to allow more
   than a single distinguished column in disjunct matrices. In
   particular, in general we may require the matrix $M$ to have more
   than $e$ rows that $u$-satisfy every choice of a critical set $S$,
   a zero set $Z$, and any $g+1$ designated columns $D \subseteq S$
   (at which all entries of the corresponding rows must be $1$).
   Denote this generalized notion by $(d,e;u,g)$-disjunct matrices.
   It is straightforward to extend the arguments of
   Lemma~\ref{lem:disjunct} to show that the generalized notion of
   $(d,e;u,g)$-disjunct matrices is necessary and sufficient to
   capture non-adaptive threshold group testing with upper threshold
   $u$ and gap $g$.

 \item Lemma~\ref{lem:probDisjunct} can be generalized to show that
   Construction~\ref{constr:prob} (with probability $1-o(1)$) results
   in a $(d, \Omega_u(pd \log(n/d)/(1-p)^2); u,g)$-disjunct matrix if
   the number of measurements is increased by a factor $O(d^g)$.

 \item Lemma \ref{lem:rep} can be extended to positive gaps, by taking
   $\cM_1$ as a $(d-1, e_1; \ell-1)$-regular matrix, provided that,
   for every $y \in \cM_2[x]_{1,g+1}$ and $y' \in \cM_2[x']_{1,g+1}$,
   we have $ |\supp(y) \setminus \supp(y')| \geq e_2.  $ In particular
   this is the case if $\cM_2$ is strongly
   $(d,e_2-1;g+1)$-disjunct\footnote{Here we are also considering the unavoidable assumption that \\
     $\max\{|\supp(x) \setminus \supp(x')|, |\supp(x') \setminus
     \supp(x)|\} > g$.}. Similarly for Lemma~\ref{lem:repdisj},
   $\cM_2$ must be taken as a strongly $(2d,e_2;g+1)$-disjunct
   matrix. Consequently, using the coding-theoretic construction of
   strongly disjunct matrices described in Section~\ref{sec:KSext},
   our explicit constructions of $(d,e;u)$-disjunct matrices can be
   extended to the gap model at the cost of a factor $O(d^g)$ increase
   in the number of measurements (as summarized in
   Table~\ref{tab:params}).

 \item Observe that a $(d,e;u,g)$-disjunct matrix is in particular,
   strongly $(d-g,e;g+1)$-disjunct and thus, the lower bound
   $\Omega(d^{g+2} \log_d n + e d^{g+1})$ on the number of rows of
   strongly disjunct matrices applies to them as well.
 \end{enumerate}

 \section{Notes}

 The notion of $d$-disjunct matrices is also known in certain
 equivalent forms; e.g., \emph{$d$-superimposed codes},
 \emph{$d$-separable matrices}, or \emph{$d$-cover-free families}
 (cf. \cite{ref:groupTesting}).  The special case of
 Definition~\ref{def:matrix} corresponding to $(0,0,e'_0,0)$-resilient
 matrices is related to the notion of \emph{selectors} in
 \cite{ref:DBGV05} and \emph{resolvable matrices} in \cite{ref:EGH07}.
 Lemma~\ref{lem:lowerbound} is similar in spirit to the lower bound
 obtained in \cite{ref:DBGV05} for the size of selectors.

 The notion of strongly disjunct matrices, in its general form, has
 been studied in the literature under different names and equivalent
 formulations, e.g., superimposed $(u,d)$-designs/codes and $(u,d)$
 cover-free families (see
 \cites{ref:SW00,ref:DVMT02,ref:KL04,ref:SW04,ref:CDH07,ref:CFH08} and
 the references therein).

 \begin{subappendices}

   \section{Some Technical Details} \label{app:proofs}

   \newcommand{\Deg}{\mathsf{deg}}

   For a positive integer $c > 1$, define a $c$-hypergraph as a tuple
   $(V, E)$, where $V$ is the set of vertices and $E$ is the set of
   hyperedges such that every $e \in E$ is a subset of $V$ of size
   $c$. The degree of a vertex $v \in V$, denoted by $\Deg(v)$, is the
   size of the set $\{e \in E\colon v \in E\}$.  Note that $|E| \leq
   \binom{|V|}{c}$ and $\Deg(v) \leq \binom{|V|}{c-1}$.  The
   \emph{density} of the hypergraph is given by $|E| /
   \binom{|V|}{c}$.  A \emph{vertex cover} on the hypergraph is a
   subset of vertices that contains at least one vertex from every
   hyperedge. A \emph{matching} is a set of pairwise disjoint
   hyperedges. It is well known that any dense hypergraph must have a
   large matching. Below we reconstruct a proof of this claim.

   \begin{prop} \label{prop:matching} Let $H$ be a $c$-hypergraph such
     that every vertex cover of $H$ has size at least $k$. Then $H$
     has a matching of size at least $k/c$.
   \end{prop}

\begin{proof}
  Let $M$ be a maximal matching of $H$, i.e., a matching that cannot
  be extended by adding further hyperedges. Let $C$ be the set of all
  vertices that participate in hyperedges of $M$. Then $C$ has to be a
  vertex cover, as otherwise one could add an uncovered hyperedge to
  $M$ and violate maximality of $M$. Hence, $c|M| = |C| \geq k$, and
  the claim follows.
\end{proof}

\begin{lem} \label{lem:denseGraph} Let $H=(V, E)$ be a $c$-hypergraph
  with density at least $\eps > 0$. Then $H$ has a matching of size at
  least $\frac{\eps}{c^2} (|V|-c+1)$.
\end{lem}

\begin{proof}
  For every subset $S \subseteq V$ of size $c$, denote by
  $\mathds{1}(S)$ the indicator value of $S$ being in $E$.  Let $C$ be
  any vertex cover of $H$. Denote by $\cS$ the set of all subsets of
  $V$ of size $c$. Then we have
  \[
  \eps \binom{|V|}{c} \leq \sum_{S \in \cS} \mathds{1}(S) \leq \sum_{v
    \in C} \Deg(v) \leq |C| \binom{|V|}{c-1}.
  \]
  Hence, $|C| \geq \eps(n-c+1)/c$, and the claim follows using
  Proposition~\ref{prop:matching}.
\end{proof}

\end{subappendices}

\musicBoxTesting


\Chapter{Capacity Achieving Codes}
\epigraphhead[70]{\epigraph{\textsl{``How is an error possible in
      mathematics?''}}{\textit{--- Henri Poincar\'e}}}
\label{chap:capacity}
\newcommand{\bsc}{\mathsf{BSC}} \newcommand{\bec}{\mathsf{BEC}}
\newcommand{\cF}{\mathcal{F}}

One of the basic goals of coding theory is coming up with efficient
constructions of error-correcting codes that allow reliable
transmission of information over discrete communication
channels. Already in the seminal work of Shannon~\cite{ref:Shannon},
the notion of \emph{channel capacity} was introduced which is a
characteristic of the communication channel that determines the
maximum rate at which reliable transmission of information (i.e., with
vanishing error probability) is possible. However, Shannon's result
did not focus on the \emph{feasibility} of the underlying code and
mainly concerned with the existence of reliable, albeit possibly
complex, coding schemes.  Here feasibility can refer to a combination
of several criteria, including: succinct description of the code and
its efficient computability, the existence of an efficient encoder and
an efficient decoder, the error probability, and the set of message
lengths for which the code is defined.


Besides heuristic attempts, there is a large body of rigorous work in
the literature on coding theory with the aim of designing feasible
capacity approaching codes for various discrete channels, most
notably, the natural and fundamental cases of the binary erasure
channel (BEC) and binary symmetric channel (BSC).  Some notable
examples in ``modern coding'' include Turbo codes and sparse graph
codes (e.g., LDPC codes and Fountain codes, cf.\ \cites{ref:modern,
  ref:blahut, ref:Raptor}).  These classes of codes are either known
or strongly believed to contain capacity achieving ensembles for the
erasure and symmetric channels.

While such codes are very appealing both theoretically and
practically, and are in particular designed with efficient decoding in
mind, in this area there still is a considerable gap between what we
can prove and what is evidenced by practical results, mainly due to
complex combinatorial structure of the code constructions.  Moreover,
almost all known code constructions in this area involve a
considerable amount of randomness, which makes them prone to a
possibility of design failure (e.g., choosing an ``unfortunate''
degree sequence for an LDPC code). While the chance of such
possibilities is typically small, in general there is no known
efficient way to certify whether a particular outcome of the code
construction is satisfactory.  Thus, it is desirable to come up with
constructions of provably capacity achieving code families that are
explicit, i.e., are efficient and do not involve any randomness.

Explicit construction of capacity achieving codes was considered as
early as the classic work of Forney \cite{ref:forney}, who showed that
concatenated codes can achieve the capacity of various memoryless
channels. In this construction, an outer MDS code is concatenated with
an inner code with small block length that can be found in reasonable
time by brute force search. An important subsequent work by Justesen
\cite{ref:Justesen} (that was originally aimed for explicit
construction of asymptotically good codes) shows that it is possible
to eliminate the brute force search by varying the inner code used for
encoding different symbols of the outer encoding, provided that the
ensemble of inner codes contains a large fraction of capacity
achieving codes.

Recently, Arikan~\cite{ref:Arikan09} gave a framework for
deterministic construction of capacity achieving codes for discrete
memoryless channels (DMCs) with binary input that are equipped with
efficient encoders and decoders and attain slightly worse than
exponentially small error probability. These codes are defined for
every block length that is a power of two, which might be considered a
restrictive requirement.  Moreover, the construction is currently
explicit (in the sense of polynomial-time computability of the code
description) only for the special case of BEC and requires exponential
time otherwise.

In this chapter, we revisit the concatenation scheme of Justesen and
give new constructions of the underlying ensemble of the inner codes.
The code ensemble used in Justesen's original construction is
attributed to Wozencraft.
Other ensembles that are known to be useful in this scheme include the
ensemble of Goppa codes and shortened cyclic codes
(see~\cite{ref:Roth}, Chapter~12).  The number of codes in these
ensembles is exponential in the block length and they achieve
exponentially small error probability.  These ensembles are also known
to achieve the Gilbert-Varshamov bound, and owe their capacity
achieving properties to the property that each nonzero vector belongs
to a small number of the codes in the ensemble.

Here, we will use extractors and lossless condensers to construct much
smaller ensembles with similar, random-like, properties.  The quality
of the underlying extractor or condenser determines the quality of the
resulting code ensemble. In particular, the size of the code ensemble,
the decoding error and proximity to the channel capacity are
determined by the \emph{seed length}, the \emph{error}, and the
\emph{output length} of the extractor or condenser being used.



As a concrete example, we will instantiate our construction with
appropriate choices of the underlying condenser (or extractor) and
obtain, for every block length $n$, a capacity achieving ensemble of
size $2^n$ that attains exponentially small error probability for both
erasure and symmetric channels (as well as the broader range of
channels described above), and an ensemble of
quasipolynomial\footnote{A quantity $f(n)$ is said to be
  quasipolynomial\index{quasipolynomial} in $n$ (denoted by $f(n) =
  \qpoly(n)$) if $f(n) = 2^{(\log n)^{O(1)}}$.} size $2^{O(\log^3 n)}$
that attains the capacity of BEC.  Using nearly optimal extractors and
condensers that require logarithmic seed lengths, it is possible to
obtain polynomially small capacity achieving ensembles for any block
length.

Finally, we apply our constructions to
Justesen's concatenation scheme to obtain an explicit construction of
capacity-achieving codes for both BEC and BSC that attain
exponentially small error, as in the original construction of Forney.
Moreover, the running time of the encoder is almost linear in the
block length, and decoding takes almost linear time for BEC and almost
quadratic time for BSC. Using our quasipolynomial-sized ensemble as
the inner code, we are able to construct a fully explicit code for BEC
that is defined and capacity achieving for every choice of the message
length.

\section{Discrete Communication Channels} \label{sec:dmc}

\newcommand{\SC}{\mathsf{SC}} \newcommand{\CAP}{\mathsf{Cap}}
\newcommand{\Ch}{\mathscr{C}}

A \emph{discrete communication channel}\index{channel} is a randomized
process that takes a potentially infinite stream of symbols $X_0, X_1,
\ldots$ from an \emph{input alphabet} $\Sigma$ and outputs an infinite
stream $Y_0, Y_1, \ldots$ from an \emph{output alphabet} $\Gamma$.
The indices intuitively represent the \emph{time}, and each output
symbol is only determined from what channel has observed in the
past. More precisely, given $X_0, \ldots, X_t$, the output symbol
$Y_t$ must be independent of $X_{t+1}, X_{t+2}, \ldots$. Here we will
concentrate on finite input and finite output channels, that is, the
alphabets $\Sigma$ and $\Gamma$ are finite. In this case, the
conditional distribution $p(Y_t|X_t)$ of each output symbol $Y_t$
given the input symbol $X_t$ can be written as a stochastic $|\Sigma|
\times |\Gamma|$ \emph{transition matrix}, where each row is a
probability distribution.

Of particular interest is a \emph{memoryless}
channel\index{channel!memoryless}, which is intuitively ``oblivious''
of the past. In this case, the transition matrix is independent of the
time instance. That is, we have $p(Y_t | X_t)=p(Y_0 | X_0)$ for every
$t$. When the rows of the transition matrix are permutations of one
another and so is the case for the columns, the channel is called
\emph{symmetric}\index{channel!symmetric}. For example, the channel
defined by
\[
p(Y|X) = \begin{pmatrix} 0.4 & 0.1 & 0.5 \\ 0.5 & 0.4 & 0.1 \\ 0.1 &
  0.5 & 0.4
\end{pmatrix}
\]
is symmetric. Intuitively, a symmetric channel does not ``read'' the
input sequence.  An important class of symmetric channels is defined
by \emph{additive noise}.  In an additive noise channel, the input and
output alphabets are the same finite field $\F_q$ and each output
symbol $Y_t$ is obtained from $X_t$ using
\[
Y_t = X_t + Z_t,
\]
where the addition is over $\F_q$ and the channel noise $Z_t \in \F_q$
is chosen independently of the input sequence\footnote{ In fact, since
  we are only using the additive structure of $\F_q$, it can be
  replaced by any additive group, and in particular, the ring $\Z/q\Z$
  for an arbitrary integer $q>1$. This way, $q$ does not need to be
  restricted to a prime power.}.  Typically $Z_t$ is also independent
of time $t$, in which case we get a memoryless additive noise channel.
For a noise distribution $\cZ$, we denote the memoryless additive
noise channel over the input (as well as output) alphabet $\Sigma$ by
$\SC(\Sigma, \cZ)$.

Note that the notion of additive noise channels can be extended to the
case where the input and alphabet sets are vector spaces $\F_q^n$, and
the noise distribution is a probability distribution over $\F_q^n$. By
considering an isomorphism between $\F_q^n$ and the field extension
$\F_{q^n}$, such a channel is essentially an additive noise channel
$\SC(\F_{q^n}, \cZ)$, where $\cZ$ is a noise distribution over
$\F_{q^n}$. On the other hand, the channel $\SC(\F_{q^n}, \cZ)$ can be
regarded as a ``block-wise memoryless'' channel over the alphabet
$\F_q$. Namely, in a natural way, each channel use over the alphabet
$\F_{q^n}$ can be regarded as $n$ subsequent uses of a channel over
the alphabet $\F_q$. When regarding the channel over $\F_q$, it does
not necessarily remain memoryless since the additive noise
distribution $\cZ$ can be an arbitrary distribution over $\F_{q^n}$
and is not necessarily expressible as a product distribution over
$\F_q$. However, the noise distribution of blocks of $n$ subsequent
channel uses are independent from one another and form a product
distribution (since the original channel $\SC(\F_{q^n}, \cZ)$ is
memoryless over $\F_{q^n}$).  Often by choosing larger and larger
values of $n$ and letting $n$ grow to infinity, it is possible to
obtain good approximations of a non-memoryless additive noise channel
using memoryless additive noise channels over large alphabets.

An important additive noise channel is the \emph{$q$-ary symmetric
  channel}\index{channel!$q$-ary symmetric}, which is defined by a
(typically small) noise parameter $p \in [0, 1)$. For this channel,
the noise distribution $\cZ$ has a probability mass $1-p$ on zero, and
$p/(q-1)$ on every nonzero alphabet letter.  A fundamental special
case is the \emph{binary symmetric channel}\index{channel!binary
  symmetric (BSC)} (BSC), which corresponds to the case $q=2$ and is
denoted by $\bsc(p)$.

Another fundamentally important channel is the \emph{binary erasure
  channel}\index{channel!binary erasure (BEC)}.  The input alphabet
for this channel is $\{0,1\}$ and the output alphabet is the set
$\{0,1,?\}$.  The transition is characterized by an \emph{erasure
  probability} $p \in [0,1)$. A transmitted symbol is output intact by
the channel with probability $1-p$. However, with probability $p$, a
special \emph{erasure symbol} ``$?$'' is delivered by the channel. The
behavior of the binary symmetric channel $\bsc(p)$ and binary erasure
channel $\bec(p)$ is schematically described by Figure~\ref{fig:dmc}.

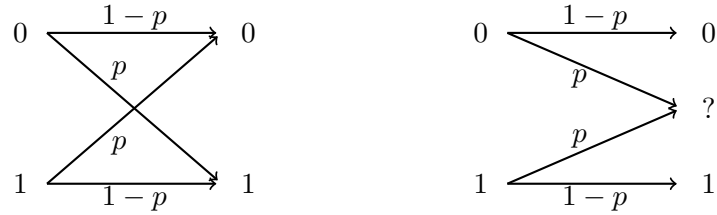
\begin{figure}[t]
  \begin{center}
    \mbox{
      \begin{tikzpicture}[rounded corners, thick]
        \path (0,2) node[circle] (left_0) {$0$} (0,0) node[circle]
        (left_1) {$1$} (3,0) node[circle] (right_1) {$1$} (3,2)
        node[circle] (right_0) {$0$};

        \draw [->,shorten >= 2pt] (left_0.east) -- (right_0.west);
        \draw [->,shorten >= 2pt] (left_0.east) -- (right_1.west);
        \draw [->,shorten >= 2pt] (left_1.east) -- (right_1.west);
        \draw [->,shorten >= 2pt] (left_1.east) -- (right_0.west);

        \draw (1.5,2.2) node {$1-p$}; \draw (1.5,-0.2) node {$1-p$};
        \draw (1.3,1.5) node {$p$}; \draw (1.3,0.5) node {$p$};
      \end{tikzpicture}} \hspace{2cm} \mbox{
      \begin{tikzpicture}[rounded corners, thick]
        \path (0,2) node[circle] (left_0) {$0$} (0,0) node[circle]
        (left_1) {$1$} (3,0) node[circle] (right_1) {$1$} (3,2)
        node[circle] (right_0) {$0$} (3,1) node[circle] (right_E)
        {$?$};

        \draw [->,shorten >= 2pt] (left_0.east) -- (right_0.west);
        \draw [->,shorten >= 2pt] (left_0.east) -- (right_E.west);
        \draw [->,shorten >= 2pt] (left_1.east) -- (right_1.west);
        \draw [->,shorten >= 2pt] (left_1.east) -- (right_E.west);

        \draw (1.5,2.2) node {$1-p$}; \draw (1.5,-0.2) node {$1-p$};
        \draw (1.3,1.4) node {$p$}; \draw (1.3,0.6) node {$p$};
      \end{tikzpicture}}
  \end{center}
  \caption[The binary symmetric and binary erasure channels]{The
    binary symmetric channel (left) and binary erasure channel
    (right). On each graph, the left part corresponds to the input
    alphabet and the right part to the output alphabet. Conditional
    probability of each output symbol given an input symbol is shown
    by the labels on the corresponding arrows.}
  \label{fig:dmc}
\end{figure}

A \emph{channel encoder} $\mathcal{E}$ for a channel $\mathscr{C}$
with input alphabet $\Sigma$ and output alphabet $\Gamma$ is a mapping
$\C\colon \zo^k \to \Sigma^n$. A \emph{channel decoder}, on the other
hand, is a mapping $\mathcal{D}\colon \Gamma^n \to \zo^k$. A channel
encoder and a channel decoder collectively describe a \emph{channel
  code}. Note that the image of the encoder mapping defines a block
code of length $n$ over the alphabet $\Sigma$. The parameter $n$
defines the \emph{block length} of the code.  For a sequence $Y \in
\Sigma^n$, denote by the random variable $\mathscr{C}(Y)$ a sequence
$\hat{Y} \in \Gamma^n$ that is output by the channel, given the input
$Y$.

Intuitively, a channel encoder adds sufficient redundancy to a given
``message'' $X \in \zo^k$ (that is without loss of generality modeled
as a binary string of length $k$), resulting in an encoded sequence $Y
\in \Sigma^n$ that can be fed into the channel. The channel
manipulates the encoded sequence and delivers a sequence $\hat{Y} \in
\Gamma^n$ to a recipient whose aim is to recover $X$.  The recovery
process is done by applying the channel decoder on the received
sequence $\hat{Y}$.  The transmission is successful when
$\mathcal{D}(\hat{Y})=X$. Since the channel behavior is not
deterministic, there might be a nonzero probability, known as the
\emph{error probability}, that the transmission is unsuccessful. More
precisely, the error probability of a channel code is defined as
\[
p_e := \sup_{X \in \zo^k} \Pr[\mathcal{D}(\mathscr{C}(\mathcal{E}(X)))
\neq X],
\]
where the probability is taken over the randomness of $\mathscr{C}$.
A schematic diagram of a simple communication system consisting of an
encoder, point-to-point channel, and decoder is shown in
Figure~\ref{fig:channel}.

\begin{figure}[b]
  \begin{center} \vspace{5mm}
    \includegraphics[width=\textwidth]{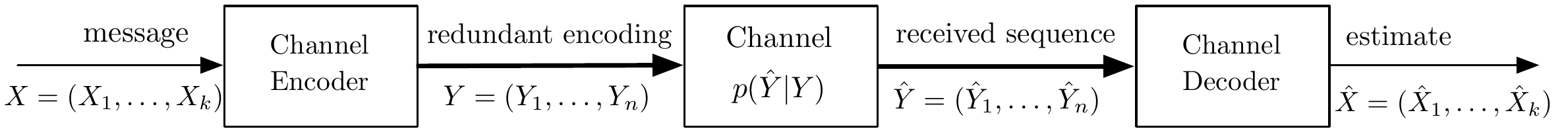}
  \end{center}
  \caption[The schematic diagram of a point-to-point communication
  system]{The schematic diagram of a point-to-point communication
    system. The stochastic behavior of the channel is captured by the
    conditional probability distribution $p(\hat{Y}|Y)$.}
  \label{fig:channel}
\end{figure}

For linear codes over additive noise channels, it is often convenient
to work with \emph{syndrome decoders}. Consider a linear code with
generator and parity check matrices $G$ and $H$, respectively. The
encoding of a message $x$ (considered as a row vector) can thus be
written as $xG$.  Suppose that the encoded sequence is transmitted
over an additive noise channel, which produces a noisy sequence $y :=
xG+z$, for a randomly chosen $z$ according to the channel
distribution.  The receiver receives the sequence $y$ and, without
loss of generality, the decoder's task is to obtain an estimate of the
noise realization $z$ from $y$. Now, observe that
\[
Hy^\top = HG^\top x^\top + H z^\top = H z^\top,
\]
where the last equality is due to the orthogonality of the generator
and parity check matrices. Therefore, $H z^\top$ is available to the
decoder and thus, in order to decode the received sequence, it
suffices to obtain an estimate of the noise sequence $z$ from the
\emph{syndrome} $H z^\top$. A syndrome decoder is a function that,
given the syndrome, outputs an estimate of the noise sequence (note
that this is independent of the codeword being sent). The error
probability of a syndrome decoder can be simply defined as the
probability (over the noise randomness) that it obtains an incorrect
estimate of the noise sequence. Obviously, the error probability of a
syndrome decoder upper bounds the error probability of the channel
code.

The \emph{rate} of a channel code (in bits per channel use) is defined
as the quantity $k/n$.  We call a rate $r \geq 0$ \emph{feasible} if
for every $\eps > 0$, there is a channel code with rate $r$ and error
probability at most $\eps$. The rate of a channel code describes its
efficiency; the larger the rate, the more information can be
transmitted through the channel in a given ``time frame''.  A
fundamental question is, given a channel $\mathscr{C}$, to find the
largest possible rate at which reliable transmission is possible. In
his fundamental work, Shannon \cite{ref:Shannon} introduced the notion
of \emph{channel capacity} that answers this question.  Shannon
capacity can be defined using purely infor\-mation-theoretic
terminology. However, for the purposes of this chapter, it is more
convenient to use the following, more ``computational'', definition
which turns out to be equivalent to the original notion of Shannon
capacity:
\[
\CAP(\mathscr{C}) := \sup \{ r \mid \text{$r$ is a feasible rate for
  the channel $\mathscr{C}$} \}.
\]

Capacity of memoryless symmetric channels has a particularly nice
form. Let $\cZ$ denote the probability distribution defined by any of
the rows of the transition matrix of a memoryless symmetric channel
$\mathscr{C}$ with output alphabet $\Gamma$. Then, capacity of
$\mathscr{C}$ is given by
\[
\CAP(\mathscr{C}) = \log_2 |\Gamma| - H(\cZ),
\]
where $H(\cdot)$ denotes the Shannon entropy
\cite{ref:cover}*{Section~7.2}. In particular, capacity of the binary
symmetric channel $\bsc(p)$ (in bits per channel use) is equal to
\[
1-h(p) = 1+p \log_2 p +(1-p) \log_2 (1-p).
\]
Capacity of the binary erasure channel $\bec(p)$ is moreover known to
be $1-p$ \cite{ref:cover}*{Section~7.1}.

A \emph{family}\index{code!family} of channel codes of rate $r$ is an
infinite set of channel codes, such that for every (typically small)
\emph{rate loss} $\delta \in (0, r)$ and block length $n$, the family
contains a code $\C(n, \delta)$ of length at least $n$ and rate at
least $r-\delta$. The family is called \emph{explicit} if there is a
deterministic algorithm that, given $n$ and $\delta$ as parameters,
computes the encoder function of the code $\C(n, \delta)$ in
polynomial time in $n$.  For linear channel codes, this is equivalent
to computing a generator or parity check matrix of the code in
polynomial time.  If, additionally, the algorithm receives an
auxiliary index $i \in [s]$, for a \emph{size parameter} $s$ depending
on $n$ and $\delta$, we instead get an
\emph{ensemble}\index{code!ensemble} of size $s$ of codes. An ensemble
can be interpreted as a set of codes of length $n$ and rate at least
$r-\delta$ each, that contains a code for each possibility of the
index $i$.

We call a family of codes \emph{capacity
  achieving}\index{code!capacity achieving} for a channel
$\mathscr{C}$ if the family is of rate $\CAP(\mathscr{C})$ and
moreover, the code $\C(n, \delta)$ as described above can be chosen to
have an arbitrarily small error probability for the channel
$\mathscr{C}$.  If the error probability decays exponentially with the
block length $n$; i.e., $p_e = O(2^{-\gamma n})$, for a constant
$\gamma > 0$ (possibly depending on the rate loss), then the family is
said to achieve an \emph{error exponent}\index{code!error exponent}
$\gamma$.  We call the family \emph{capacity achieving for all
  lengths} if it is capacity achieving and moreover, there is an
integer constant $n_0$ (depending only on the rate loss $\delta$) such
that for every $n \geq n_0$, the code $\C(n, \delta)$ can be chosen to
have length exactly $n$.

\section{Codes for the Binary Erasure Channel} \label{sec:bec}

Any code with minimum distance $d$ can tolerate up to $d-1$ erasures
in the \emph{worst case}\footnote{See Appendix~\ref{app:coding} for a
  quick review of the basic notions in coding theory.}.
Thus one way to ensure reliable communication over $\bec(p)$ is to use
binary codes with relative minimum distance of about $p$. However,
known negative bounds on the rate-distance trade-off (e.g., the sphere
packing and MRRW bounds) do not allow the rate of such codes to
approach the capacity $1-p$.  However, by imposing the weaker
requirement that \emph{most} of the erasure patterns should be
recoverable, it is possible to attain the capacity with a positive,
but arbitrarily small, error probability (as guaranteed by the
definition of capacity).

In this section, we consider a different relaxation that preserves the
worst-case guarantee on the erasure patterns; namely we consider
\emph{ensembles} of linear codes with the property that \emph{any}
pattern of up to $p$ erasures must be tolerable by all but a
negligible fraction of the codes in the ensemble. This in particular
allows us to construct ensembles in which all but a negligible
fraction of the codes are capacity achieving for BEC.  Note that as we
are only considering linear codes, recoverability from a particular
erasure pattern $S \subseteq [n]$ (where $n$ is the block length)
is a property of the code and independent of the encoded sequence.

Now we introduce two constructions, which employ strong, linear
extractors and lossless condensers as their main
ingredients. Throughout this section we denote by $f\colon \F_2^n
\times \F_2^d \to \F_2^r$ a strong, linear, lossless condenser for
min-entropy $m$ and error $\eps$ and by $g\colon \F_2^n \times
\F_2^{d'} \to \F_2^k$ a strong, linear extractor for min-entropy $n-m$
and error $\eps'$. We assume that the errors $\eps$ and $\eps'$ are
substantially small.  Using this notation, we define the ensembles
$\cF$ and $\cG$ as in Construction~\ref{constr:ensembles}.

\begin{constr}[b!] 
  \begin{framed}
    \begin{description}
    \item[Ensemble $\cF$:] Define a code $\C_u$ for each seed $u \in
      \F_2^d$ as follows: Let $H_u$ denote the $r \times n$ matrix
      that defines the linear function $f(\cdot, u)$, i.e., for each
      $x \in \F_2^n$, $H_u \cdot x = f(x,u)$. Then $H_u$ is a parity
      check matrix for $\C_u$.

    \item[Ensemble $\cG$:] Define a code $\C'_u$ for each seed $u \in
      \F_2^{d'}$ as follows: Let $G_u$ denote the $k \times n$ matrix
      that defines the linear function $g(\cdot, u)$.  Then $G_u$ is a
      generator matrix for $\C'_u$.
    \end{description}
  \end{framed}
  \caption{Ensembles $\cF$ and $\cG$ of error-correcting codes.}
  \label{constr:ensembles}
\end{constr}

Obviously, the rate of each code in $\cF$ is at least
$1-r/n$. Moreover, as $g$ is a strong extractor we can assume without
loss of generality that the rank of each $G_u$ is exactly\footnote{
  This causes no loss of generality since, if the rank of some $G_u$
  is not maximal, one of the $k$ symbols output by the linear function
  $g(\cdot, u)$ would linearly depend on the others and thus, the
  function would fail to be an extractor for any source (so one can
  arbitrarily modify $g(\cdot, u)$ to have rank $k$ without negatively
  affecting the parameters of the extractor $g$).  } $k$. Thus, each
code in $\cG$ has rate $k/n$. Lemma \ref{lem:becCodes} below is our
main tool in quantifying the erasure decoding capabilities of the two
ensembles. Before stating the lemma, we mention a proposition showing
that linear condensers applied on \emph{affine sources} achieve either
zero or large errors:

\begin{prop} \label{prop:zeroError} Suppose that a distribution $\cX$
  is uniformly supported on an affine $k$-dimensional subspace over
  $\F_q^n$. Consider a linear function $f\colon \F_q^n \to \F_q^m$,
  and define the distribution $\cY$ as $\cY := f(\cX)$. Suppose that,
  for some integer $k$ and $\eps < 1/2$, $\cY$ is $\eps$-close to
  having either min-entropy $m \log q$ or at least $k \log q$.  Then,
  $\eps=0$.
\end{prop}

\begin{proof}
  By linearity, $\cY$ is uniformly supported on an affine subspace $A$
  of $\F_q^m$. Let $k' \leq m$ be the dimension of this subspace, and observe
  that $k' \leq k$.

  First, suppose that $\cY$ is $\eps$-close to a distribution with
  min-entropy $m \log q$; i.e., the uniform distribution on $\F_q^m$.
  Now, the statistical distance between $\cY$ and the uniform
  distribution is, by definition,
  \[
  \sum_{x \in A} ( q^{-k'} - q^{-m} ) = 1-q^{k'-m}-1.
  \]
  Since $\eps < 1/4$, $q \geq 2$, and $k'$ and $m$ are integers, this
  implies that the distance is greater than $1/2$ (a contradiction)
  unless $k'=m$, in which case it becomes zero. Therefore, the output
  distribution is exactly uniform over $\F_q^m$.

  Now consider the case where $\cY$ is $\eps$-close to having
  min-entropy at least $k \log q$. Considering that $k' \leq k$, the
  definition of statistical distance implies that $\eps$ is at least
  \[
  \sum_{x \in A} (q^{-k'} - q^{-k}) = 1-q^{k'-k}.
  \]
  Similarly as before, we get that $k' = k$, meaning that $\cY$ is
  precisely a distribution with min-entropy $k \log q$.
\end{proof}

\begin{lem} \label{lem:becCodes} Let $S \subseteq [n]$ be a set of
  size at most $m$. Then all but a $5 \eps$ fraction of the codes in
  $\cF$ and all but a $5 \eps'$ fraction of those in $\cG$ can
  tolerate the erasure pattern defined by
  $S$. 
\end{lem}

\begin{proof}
  We prove the result for the ensemble $\cG$. The argument for $\cF$
  is similar.  Consider a probability distribution $\cS$ on $\F_2^n$
  that is uniform on the coordinates specified by $\bar{S} := [n]
  \setminus S$ and fixed to zeros elsewhere. Thus the min-entropy of
  $\cS$ is at least $n-m$, and the distribution $(U, g(\cS, U))$,
  where $U \sim \U_{d'}$, is $\eps'$-close to $\U_{d'+k}$.

  By Corollary~\ref{coro:strongExt}, for all but a $5 \eps'$ fraction
  of the choices of $u \in \F_2^{d'}$, the distribution of $g(\cS, u)$
  is $(1/5)$-close to $\U_k$. Fix such a $u$. By
  Proposition~\ref{prop:zeroError}, the distribution of $g(\cS, u)$
  must in fact be exactly uniform. Thus, the $k \times m$ submatrix of
  $G_u$ consisting of the columns picked by $\bar{S}$ must have rank
  $k$, which implies that for every $x \in \F_2^k$, the projection of
  the encoding $x \cdot G_u$ to the coordinates chosen by $\bar{S}$
  uniquely identifies $x$.
\end{proof}

The lemma combined with a counting argument implies the following
corollary:

\begin{coro} \label{coro:becCodes} Let $\cS$ be any distribution on
  the subsets of $[n]$ of size at most $m$.  Then all but a
  $\sqrt{5\eps}$ (resp., $\sqrt{5\eps'}$) fraction of the codes in
  $\cF$ (resp., $\cG$) can tolerate erasure patterns sampled from
  $\cS$ with probability at least $1-\sqrt{5\eps}$ (resp.,
  $1-\sqrt{5\eps'}$). \qed
\end{coro}

Note that the result holds irrespective of the distribution $\cS$,
contrary to the familiar case of $\bec(p)$ for which the erasure
pattern is an i.i.d.\ (i.e., independent and identically-distributed)
sequence.  For the case of $\bec(p)$, the erasure pattern (regarded as
its binary characteristic vector in $\F_2^n$) is given by $S := (S_1,
\ldots, S_n)$, where the random variables $S_1, \ldots, S_n \in \F_2$
are i.i.d.\ and $\Pr[S_i = 1] = p$. We denote this particular
distribution by $\cB_{n,p}$, which assigns a nonzero probability to
every vector in $\F_2^n$.  Thus in this case we cannot directly apply
Corollary~\ref{coro:becCodes}. However, note that $\cB_{n,p}$ can be
written as a convex combination
\begin{equation} \label{eqn:convex} \cB_{n,p} = (1-\gamma) \U_{n,\leq
    p'} + \gamma \cD,
\end{equation}
for $p' := p + \Omega(1)$ that is arbitrarily close to $p$, where
$\cD$ is an ``error distribution'' whose contribution $\gamma$ is
exponentially small. The distribution $\U_{n, \leq p'}$ is the
distribution $\cB_{n,p}$ conditioned on vectors of weight at most
$np'$. Corollary~\ref{coro:becCodes} applies to $\U_{n,\leq p'}$ by
setting $m = np'$.  Moreover, by the convex combination above, the
erasure decoding error probability of any code for erasure pattern
distributions $\cB_{n,p}$ and $\U_{n,\leq p'}$ differ by no more than
$\gamma$. Therefore, the above result applied to the erasure
distribution $\U_{n,\leq p'}$ handles the particular case of $\bec(p)$
with essentially no change in the error probability.

In light of Corollary~\ref{coro:becCodes}, in order to obtain rates
arbitrarily close to the channel capacity, the output lengths of $f$
and $g$ must be sufficiently close to the entropy requirement
$m$. More precisely, it suffices to have $r \leq (1+\alpha) m$ and $k
\geq (1-\alpha) m$ for arbitrarily small constant $\alpha > 0$. The
seed length of $f$ and $g$ determine the size of the code ensemble.
Moreover, the error of the extractor and condenser determine the
erasure error probability of the resulting code ensemble. As achieving
the channel capacity is the most important concern for us, we will
need to instantiate $f$ (resp., $g$) with a linear, strong, lossless
condenser (resp., extractor) whose output length is close to $m$. We
mention one such instantiation for each function.

For both functions $f$ and $g$, we can use the explicit extractor and
lossless condenser obtained from the Leftover Hash Lemma
(Lemma~\ref{lem:leftover}), which is optimal in the output length, but
requires a large seed, namely, $d=n$. The ensemble resulting this way
will thus have size $2^n$, but attains a positive error exponent
$\delta/2$ for an arbitrary rate loss $\delta > 0$.  Using an optimal
lossless condenser or extractor with seed length $d=\log(n) +
O(\log(1/\eps))$ and output length close to $m$, it is possible to
obtain a polynomially small capacity-achieving ensemble.  However, in
order to obtain an explicit ensemble of codes, the condenser of
extractor being used must be explicit as well.

In the world of linear extractors, we can use Trevisan's extractor
(Theorem~\ref{thm:Tre}) to improve the size of the ensemble compared
to what obtained from the Leftover Hash Lemma. In particular,
Trevisan's extractor combined with Corollary~\ref{coro:becCodes}
(using ensemble $\cG$) immediately gives the following result:

\begin{coro} \label{coro:bec} Let $p, c > 0$ be arbitrary
  constants. Then for every integer $n > 0$, there is an explicit
  ensemble $\cG$ of linear codes of rate $1-p-o(1)$ such that, the
  size of $\cG$ is quasipolynomial, i.e., $|\cG| = 2^{O(c^3 \log^3
    n)}$, and, all but an $n^{-c}=o(1)$ fraction of the codes in the
  ensemble have error probability at most $n^{-c}$ when used over
  $\bec(p)$. \qed
\end{coro}

For the ensemble $\cF$, on the other hand, we can use the linear
lossless condenser of Guruswami et al.\ that only requires a
logarithmic seed (Corollary \ref{coro:GUVcondLinear}).  Using this
condenser combined with Corollary~\ref{coro:becCodes}, we can
strengthen the above result as follows:

\begin{coro}
  \label{coro:becGUV} Let $p, c, \alpha > 0$ be arbitrary
  constants. Then for every integer $n > 0$, there is an explicit
  ensemble $\cF$ of linear codes of rate $1-p-\alpha$ such that $|\cG|
  = O(n^{c'})$ for a constant $c'$ only depending on $c, \alpha$.
  Moreover, all but an $n^{-c}=o(1)$ fraction of the codes in the
  ensemble have error probability at most $n^{-c}$ when used over
  $\bec(p)$. \qed
\end{coro}

\section{Codes for the Binary Symmetric Channel} \label{sec:bsc}

The goal of this section is to design capacity achieving code
ensembles for the binary symmetric channel $\bsc(p)$. In order to do
so, we obtain codes for the general (and not necessarily memoryless)
class $\SC(\F_q, \cZ)$ of symmetric channels, where $\cZ$ is any flat
distribution or sufficiently close to one.  For concreteness, we will
focus on the binary case where $q=2$.

Recall that the capacity of $\bsc(\cZ)$, seen as a binary channel, is
$1-h(\cZ)$ where $h(\cZ)$ is the entropy rate of $\cZ$.  The special
case $\bsc(p)$ is obtained by setting $\cZ = \cB_{n,p}$; i.e., the
product distribution of $n$ Bernoulli random variables with
probability $p$ of being equal to $1$.

The code ensemble that we use for the symmetric channel is the
ensemble $\cF$, obtained from linear lossless condensers, that we
introduced in the preceding section. Thus, we adopt the notation (and
parameters) that we used before for defining the ensemble $\cF$.
Recall that each code in the ensemble has rate at least $1-r/n$.  In
order to show that the ensemble is capacity achieving, we consider the
following brute-force decoder for each code:

\begin{quote}
  \emph{Brute-force decoder for code $\C_u$: } Given a received word
  $\hat{y} \in \F_2^n$, find a codeword $y \in \F_2^n$ of $\C_u$ used
  and a vector $z \in \supp(\cZ)$ such that $\hat{y} = y+z$. Output
  $y$, or an arbitrary codeword if no such pair is found. If there is
  more than one choice for the codeword $y$, arbitrarily choose one of
  them.
\end{quote}

For each $u \in \F_2^d$, denote by $\cE(\C_u, \cZ)$ the error
probability of the above decoder for code $\C_u$ over $\bsc(\F_2,
\cZ)$. The following lemma quantifies this probability:

\begin{lem} \label{lem:bscCodes} Let $\cZ$ be a flat distribution with
  entropy $m$.  Then for at least a $1-2\sqrt{\eps}$ fraction of the
  choices of $u \in \F_2^d$, we have $\cE(\C_u, \cZ) \leq
  \sqrt{\eps}$.
\end{lem}

\begin{proof}
  The proof is straightforward from the almost-injectivity property
  of lossless condensers discussed in Section~\ref{sec:inject}.  We
  will use this property to construct a syndrome decoder for the code
  ensemble that achieves a sufficiently small error probability.

  By Corollary~\ref{coro:strongExt}, for a $1-2\sqrt{\eps}$ fraction
  of the choices of $u \in \zo^d$, the distribution $\cY := f(\cZ, u)$
  is $(\sqrt{\eps}/2)$-close to having min-entropy at least $m$. Fix
  any such $u$.  We show that the error probability $\cE(\C_u, \cZ)$
  is bounded by $\sqrt{\eps}$.

  For each $y \in \F_2^r$, define
  \[ \mathcal{N}(y) := | \{ x \in \supp(\cZ)\colon f(x, u) = y \}| \]
  and recall that $f(x, u) = H_u \cdot x$. Now suppose that a message
  is encoded using the code $\C_u$ to an encoding $x \in \C_u$, and
  that $x$ is transmitted through the channel. The error probability
  $\cE(\C_u, \cZ)$ can be written as
  \begin{eqnarray}
    \cE(\C_u, \cZ) &=& \Pr_{z \sim \cZ}[ \exists x' \in \C_u, \exists z' \in \supp(\cZ)\setminus z \colon x+z = x'+z' ] \nonumber \\
    &\leq& \Pr_{z \sim \cZ}[ \exists x' \in \C_u, \exists z' \in \supp(\cZ)\setminus z \colon H_u\cdot (x+z) = H_u \cdot(x'+z') ]\nonumber \\
    &=& \Pr_{z \sim \cZ}[ \exists z' \in \supp(\cZ) \setminus z \colon H_u\cdot z = H_u \cdot z' ] \label{eqn:errorProbCollision} \\
    &=& \Pr_{z \sim \cZ}[ \mathcal{N}(H \cdot z) > 1 ]\nonumber \\
    &=& \Pr_{z \sim \cZ}[ \mathcal{N}(f(x,u)) > 1 ], \label{eqn:errorProbLast}
  \end{eqnarray}
  where $\eqref{eqn:errorProbCollision}$ uses the fact that any
  codeword of $\C_u$ is in the right kernel of $H_u$.

  By the first part of Proposition~\ref{prop:flatmap}, there is a set
  $T \subseteq \F_2^r$ of size at least $(1-\sqrt{\eps})|\supp(\cZ)|$
  such that, $\mathcal{N}(y) = 1$ for every $y \in T$. Since $\cZ$ is
  uniformly distributed on its support, this combined with
  \eqref{eqn:errorProbLast} immediately implies that $\cE(\C_u, \cZ)
  \leq \sqrt{\eps}$.
\end{proof}

The lemma implies that any linear lossless condenser with entropy
requirement $m$ can be used to construct an ensemble of codes such
that all but a small fraction of the codes are good for reliable
transmission over $\bsc(\cZ)$, where $\cZ$ is an arbitrary flat
distribution with entropy at most $m$. Similar to the case of BEC, the
seed length determines the size of the ensemble, the error of the
condenser bounds the error probability of the decoder, and the output
length determines the proximity of the rate to the capacity of the
channel.
Again, using the condenser given by the Leftover Hash Lemma
(Lemma~\ref{lem:leftover}), we can obtain a capacity achieving
ensemble of size $2^n$. Moreover, using the linear lossless condenser
of Guruswami et al.\ (Corollary \ref{coro:GUVcondLinear}) the ensemble
can be made polynomially small (similar to the result given by
Corollary~\ref{coro:becGUV}).

It is not hard to see that the converse of the above result is also
true; namely, that any ensemble of linear codes that is universally
capacity achieving with respect to any choice of the noise
distribution $\cZ$ defines a strong linear, lossless, condenser. This
is spelled out in the lemma below.

\begin{lem} \label{lem:bscCodesConverse} Let $\{\C_1, \ldots, \C_T\}$
  be a binary code ensemble of length $n$ and dimension $n-r$ such
  that for every flat distribution $\cZ$ with min-entropy at most $m$
  on $\F_2^n$, all but a $\gamma$ fraction of the codes in the
  ensemble (for some $\gamma \in [0,1)$) achieve error probability at
  most $\eps$ (under syndrome decoding) when used over $\SC(\F_{q^n},
  \cZ)$. Then the function $f\colon \F_2^n \times [T] \to \F_2^{r}$
  defined as
  \[
  f(x,u) := H_u \cdot x,
  \]
  where $H_u$ is a parity check matrix of $\C_u$, is a strong,
  lossless, $(m, 2\eps+\gamma)$-condenser.
\end{lem}

\begin{proof}
  The proof is straightforward using similar arguments as in
  Lemma~\ref{lem:bscCodes}.  Without loss of generality (by
  Proposition~\ref{prop:convex}), let $\cZ$ be a flat distribution
  with min-entropy $m$, and denote by $D\colon \F_2^r \to \F_2^n$ the
  corresponding syndrome decoder. Moreover, without loss of generality
  we have taken the decoder to be a deterministic function. For a
  randomized decoder, one can fix the internal coin flips so as to
  preserve the upper bound on its error probability.  Now let $u$ be
  chosen such that $\C_u$ achieves an error probability at most $\eps$
  (we know this is the case for at least $\gamma T$ of the choices of
  $u$).

  Denote by $T \subseteq \supp(\cZ)$ the set of noise realizations
  that can potentially confuse the syndrome decoder. Namely,
  \[
  T := \{ z \in \supp(\cZ)\colon \exists z' \in \supp(\cZ), z' \neq z,
  H_u \cdot z = H_u \cdot z'\}.
  \]
  Note that, for a random $Z \sim \cZ$, conditioned on the event that
  $Z \in T$, the probability that the syndrome decoder errs on $Z$ is
  at least $1/2$, since we know that $Z$ can be confused by at least
  one different noise realization. We can write this more precisely as
  \[
  \Pr_{Z\sim \cZ}[ D(Z) \neq Z \mid Z \in T] \geq 1/2.
  \]
  Since the error probability of the decoder is upper bounded by
  $\eps$, we conclude that
  \[
  \Pr_{Z\sim \cZ}[ Z \in T] \leq 2\eps.
  \]
  Therefore, the fraction of the elements on support of $\cZ$ that
  collide with some other element under the mapping defined by $H_u$
  is at most $2\eps$. Namely,
  \[
  |\{ H_u \cdot z\colon z \in \supp(\cZ) \}| \geq 2^m(1-2\eps),
  \]
  and this is true for at least $1-\gamma$ fraction of the choices of
  $u$.  Thus, for a uniformly random $U \in [T]$ and $Z \sim \cZ$, the
  distribution of $(U, H_U \cdot Z)$ has a support of size at least
  \[
  (1-\gamma)(1-2\eps) T 2^m \geq (1-\gamma-2\eps) T 2^m.
  \]
  By the second part of Proposition~\ref{prop:flatmap}, we conclude
  that this distribution is $(2\eps+\gamma)$-close to having entropy
  $m+\log T$ and thus, the function $f$ defined in the statement is a
  strong lossless $(m, 2\eps+\gamma)$-condenser.
\end{proof}

By this lemma, the known lower bounds on the seed length and the
output length of lossless condensers that we discussed in
Chapter~\ref{chap:extractor} translate into lower bounds on the size
of the code ensemble and proximity to the capacity that can be
obtained from our framework.  In particular, in order to get a
positive error exponent (i.e., exponentially small error in the block
length), the size of the ensemble must be exponentially large.

It is worthwhile to point out that the code ensembles $\cF$ and $\cG$
discussed in this and the preceding section preserve their erasure and
error correcting properties under any change of basis in the ambient
space $\F_2^n$, due to the fact that a change of basis applied on any
linear condenser results in a linear condenser with the same
parameters.  This is a property achieved by the trivial, but large,
ensemble of codes defined by the set of all $r \times n$ parity check
matrices.  Observe that no single code can be universal in this sense,
and it is inevitable to have a sufficiently large ensemble to attain
this property.

%
%

\subsection*{The Case $\bsc(p)$}

For the special case of $\bsc(p)$, the noise distribution $\cB_{n, p}$
is not a flat distribution. Fortunately, similar to the BEC case, we
can again use convex combinations to show that the result obtained in
Lemma~\ref{lem:bscCodes} can be extended to this important noise
distribution. The main tool that we need is an extension of
Lemma~\ref{lem:bscCodes} to convex combinations with a small number of
components.

Suppose that the noise distribution $\cZ$ is not a flat distribution
but can be written as a convex combination
\begin{equation} \label{eqn:noiseConvex} \cZ = \alpha_1 \cZ_1 + \cdots
  + \alpha_t \cZ_t.
\end{equation}
of $t$ flat distributions, where the number $t$ of summands is not too
large, and
\[
|\supp(\cZ_1)| \geq |\supp(\cZ_2)| \geq \cdots \geq |\supp(\cZ_t)|.
\]
For this more general case, we need to slightly tune our brute-force
decoder in the way it handles ties. In particular, we now require the
decoder to find a codeword $y \in \C_u$ and a potential noise vector
$z \in \supp(\cZ)$ that add up to the received word, as before.
However, in case more than one matching pair is found, we will require
the decoder to choose the one whose noise vector $z$ belongs to the
component $\cZ_1, \ldots, \cZ_t$ with smallest support (i.e., largest
index). If the noise vector $z \in \supp(\cZ_i)$ that maximizes the
index $i$ is still not unique, the decoder can arbitrarily choose one.
Under these conventions, we can now prove the following:

\begin{lem} \label{lem:bscConvex} Suppose that a noise distribution
  $\cZ$ is as in $\eqref{eqn:noiseConvex}$, where each component
  $\cZ_i$ has entropy at most $m$, and the function $f$ defining the
  ensemble $\cF$ is a strong lossless $(\leq m+1, \eps)$-condenser.
  Then for at least a $1-t(t+1)\sqrt{\eps}$ fraction of the choices of
  $u \in \F_2^d$, the brute-force decoder satisfies $\cE(\C_u, \cZ)
  \leq 2t \sqrt{\eps}$.
\end{lem}

\begin{proof}
  For each $1 \leq i \leq j \leq t$, we define a flat distribution
  $\cZ_{ij}$ that is uniformly supported on $\supp(\cZ_i) \cup
  \supp(\cZ_j)$. Observe that each $\cZ_{ij}$ has min-entropy at most
  $m+1$ and thus the function $f$ is a lossless condenser with error
  at most $\eps$ for this source. By Corollary~\ref{coro:strongExt}
  and a union bound, for a $1-t(t+1)\sqrt{\eps}$ fraction of the
  choices of $u \in \zo^d$, all $t(t+1)/2$ distributions
  \[
  f(\cZ_{ij}, u)\colon 1 \leq i \leq j \leq t
  \]
  are simultaneously $(\sqrt{\eps}/2)$-close to having min-entropy at
  least $m$. Fix any such $u$.

  Consider a random variable $Z$, representing the channel noise, that
  is sampled from $\cZ$ as follows: First choose an index $I \in [t]$
  randomly according to the distribution induced by $(\alpha_1,
  \ldots, \alpha_t)$ over the indices, and then sample a random noise
  $Z \sim \cZ_I$.  Using the same line of reasoning leading to
  $\eqref{eqn:errorProbCollision}$ in the proof of
  Lemma~\ref{lem:bscCodes}, the error probability with respect to the
  code $\C_u$ (i.e., the probability that the tuned distance decoder
  gives a wrong estimate on the noise realization $Z$) can now be
  bounded as
  \[
  \cE(\C_u, \cZ) \leq \Pr_{I, Z}[ \exists i \in \{ I, \ldots, t\},
  \exists z' \in \supp(\cZ_i)\setminus Z\colon f(Z, u) = f(z', u) ].
  \]
  For $i=1, \ldots, t$, denote by $\cE_i$ the right hand side
  probability in the above bound conditioned on the event that
  $I=i$. Fix any choice of the index $i$. Now it suffices to obtain an
  upper bound on $\cE_i$ irrespective of the choice of $i$, since
  \[
  \cE(\C_u, \cZ) \leq \sum_{i\in [t]} \alpha_i \cE_i.
  \]
  We call a noise realization $z \in \supp(\cZ_i)$ \emph{confusable}
  if
  \[
  \exists j \geq i, \exists z' \in \supp(\cZ_j) \setminus z\colon f(z,
  u) = f(z', u).
  \]
  That is, a noise realization is confusable if it can potentially
  cause the brute-force decoder to compute a wrong noise estimate. Our
  goal is to obtain an upper bound on the fraction of vectors on
  $\supp(\cZ_i)$ that are confusable.

  For each $j \geq i$, we know that $f(\cZ_{ij}, u)$ is
  $(\sqrt{\eps}/2)$-close to having min-entropy at least $m$.
  Therefore, by the first part of Proposition~\ref{prop:flatmap}, the
  set of confusable elements
  \[
  \{ z \in \supp(\cZ_i)\colon \exists z' \in \supp(\cZ_j)\setminus z
  \text{ such that } f(z, u)=f(z', u) \}
  \]
  has size at most $\sqrt{\eps} |\supp(\cZ_{ij})| \leq 2\sqrt{\eps}
  |\supp(\cZ_i)|$ (using the fact that, since $j \geq i$, the support
  of $\cZ_j$ is no larger than that of $\cZ_i$).  By a union bound on
  the choices of $j$, we see that the fraction of confusable elements
  on $\supp(\cZ_i)$ is at most $2t\sqrt{\eps}$. Therefore, $\cE_i \leq
  2t\sqrt{\eps}$ and we get the desired upper bound on the error
  probability of the brute-force decoder.
\end{proof}

The result obtained by Lemma~\ref{lem:bscConvex} can be applied to the
channel $\bsc(p)$ by observing that the noise distribution $\cB_{n,p}$
can be written as a convex combination
\[
\cB_{n,p} = \sum_{i=n(p-\eta)}^{n(p+\eta)} \alpha_i \U_{n,i} + \gamma
\cD,
\]
where $\U_{n,i}$ denotes the flat distribution supported on binary
vectors of length $n$ and Hamming weight exactly $i$, and $\cD$ is the
distribution $\cB_{n,p}$ conditioned on the vectors whose Hamming
weights lie outside the range $[n(p-\eta), n(p+\eta)]$.  The parameter
$\eta > 0$ can be chosen as an arbitrarily small real number, so that
the min-entropies of the distributions $\U_{n,i}$ become arbitrarily
close to the Shannon entropy of $\cB_{n,p}$; namely, $n h(p)$. This
can be seen by the estimate
\[
\binom{n}{w} = 2^{n h(w/n) \pm o(n)},
\]
$h(\cdot)$ being the binary entropy function, that is easily derived
from Stirling's formula.  By Chernoff bounds, the error $\gamma$ can
be upper bounded as
\[
\gamma = \Pr_{Z \sim \cB_{n,p}}[|\wgt(Z) - np| > \eta n] \leq 2
e^{-c_\eta np} = 2^{-\Omega(n)},
\]
where $c_\eta > 0$ is a constant only depending on $\eta$, and is thus
exponentially small.  Thus the error probability attained by any code
under noise distributions $\cB_{n,p}$ and $\cZ :=
\sum_{i=n(p-\eta)}^{n(p+\eta)} \alpha_i \U_{n,i}$ differ by the
exponentially small quantity $\gamma$. We may now apply
Lemma~\ref{lem:bscConvex} on the noise distribution $\cZ$ to attain
code ensembles for the binary symmetric channel $\bsc(p)$.  The error
probability of the ensemble is at most $2n\sqrt{\eps}$, and this bound
is satisfied by at least a $1-n^2 \sqrt{\eps}$ fraction of the codes.
Finally, the code ensemble is capacity achieving for $\bsc(p)$
provided that the condenser $f$ attains an output length $r \leq
(1+\alpha)(p+\eta)n$ for arbitrarily small constant $\alpha$, and
$\eps = o(n^{-4})$.

\section{Explicit Capacity Achieving Codes} \label{sec:explicit}

In the preceding sections, we showed how to obtain small ensembles of
explicit capacity achieving codes for various discrete channels,
including the important special cases $\bec(p)$ and $\bsc(p)$.  Two
drawbacks related to these constructions are:
\begin{enumerate}
\item While an overwhelming fraction of the codes in the ensemble are
  capacity achieving, in general it is not clear how to pin down a
  single, capacity achieving code in the ensemble.

\item For the symmetric additive noise channels, the brute-force
  decoder is extremely inefficient and is of interest only for proving
  that the constructed ensembles are capacity achieving.
\end{enumerate}

In a classic work, Justesen \cite{ref:Justesen} showed that the idea
of code concatenation\footnote{ A quick review of code concatenation
  and its basic properties appears in Appendix~\ref{app:coding}.}
first introduced by Forney \cite{ref:forney} can be used to transform
any ensemble of capacity achieving codes, for a memoryless channel,
into an explicit, efficiently decodable code with improved error
probability over the same channel. In this section we revisit this
idea and apply it to our ensembles. For concreteness, we focus on the
binary case and consider a memoryless channel $\mathscr{C}$ that is
either $\bec(p)$ or $\bsc(p)$.

Throughout this section, we consider an ensemble $\cS$ of linear codes
with block length $n$ and rate $R$, for which it is guaranteed that
all but a $\gamma = o(1)$ fraction of the codes are capacity achieving
(for a particular DMSC, in our case either $\bec(p)$ or $\bsc(p)$)
with some vanishing error probability $\eta = o(1)$ (the asymptotics
are considered with respect to the block length $n$).

\newcommand{\cout}{\C_{\mathrm{out}}} Justesen's concatenated codes
take an outer code $\cout$ of block length $s := |\cS|$, alphabet
$\F_{2^k}$, rate $R'$ as the \emph{outer code}. The particular choice
of the outer code in the original construction is Reed-Solomon codes.
However, we point out that any outer code that allows unique decoding
of some constant fraction of errors at rates arbitrarily close to one
would suffice for the purpose of constructing capacity achieving
codes. In particular, in this section we will use an expander-based
construction of asymptotically good codes due to Spielman
\cite{ref:Spielman}, from which the following theorem can be easily
derived\footnote{There are alternative choices of the outer code that
  lead to a similar result, e.g., expander-based codes due to
  Guruswami and Indyk \cite{ref:GI05}.}:

\begin{thm} \label{thm:expander} For every integer $k > 0$ and every
  absolute constant $R' < 1$, there is an explicit family of
  $\F_2$-linear codes over $\F_{2^k}$ for every block length and rate
  $R'$ that is error-correcting for an $\Omega(1)$ fraction of
  errors. The running time of the encoder and the decoder is linear in
  the bit-length of the codewords.
\end{thm}

\subsection{Justesen's Concatenation Scheme} \index{Justesen's
  concatenation}

The concatenation scheme of Justesen differs from traditional
concatenation in that the outer code is concatenated with an ensemble
of codes rather than a single inner code.

In this construction, size of the ensemble is taken to be matching
with the block length of the outer code, and each symbol of the outer
code is encoded with one of the inner codes in the ensemble. We use
the notation $\C := \cout \diamond \cS$ to denote concatenation of an
outer code $\cout$ with the ensemble $\cS$ of inner codes.  Suppose
that the alphabet size of the outer code is taken as $2^{\lfloor Rn
  \rfloor}$, where we recall that $n$ and $R$ denote the block length
and rate of the inner codes in $\cS$.

The encoding of a message with the concatenated code can be obtained
as follows: First, the message is encoded using $\cout$ to obtain an
encoding $(c_1, \ldots, c_{s}) \in \F_{2^k}^s$, where $k = \lfloor Rn
\rfloor$ denotes the dimension of the inner codes. Then, for each $i
\in [s]$, the $i$th symbol of the encoding $c_i$ is further encoded by
the $i$th code in the ensemble $\cS$ (under some arbitrary ordering of
the codes in the ensemble), resulting in a binary sequence $c'_i$ of
length $n$. The $ns$-bit long binary sequence $(c'_1, \ldots, c'_s)$
defines the encoding of the message under $\cout \diamond \cS$.  The
concatenation is scheme is depicted in Figure~\ref{fig:justesen}.

\begin{figure}
  \centerline{\includegraphics[width=\textwidth]{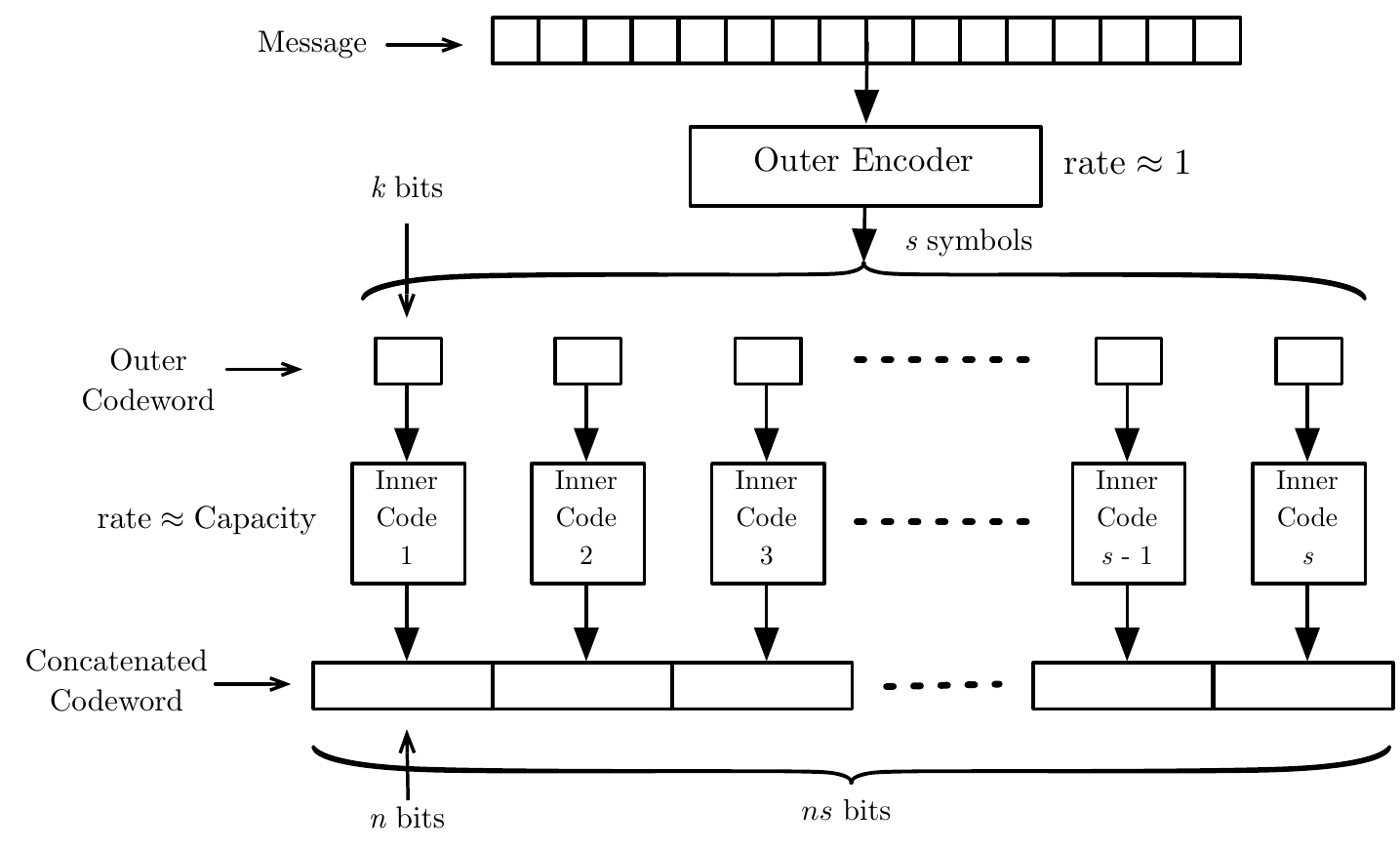}}
  \caption[Justesen's concatenation scheme]{Justesen's concatenation
    scheme.}
  \label{fig:justesen}
\end{figure}

Similar to classical concatenated codes, the resulting binary code
$\C$ has block length $N := n s$ and dimension $K := kk'$, where $k'$
is the dimension of the outer code $\cout$.  However, the neat idea in
Justesen's concatenation is that it eliminates the need for a
brute-force search for finding a good inner code, as long as almost
all inner codes are guaranteed to be good.

\subsection{The Analysis}

In order to analyze the error probability attained by the concatenated
code $\cout \diamond \cS$, we consider the following naive
decoder\footnote{Alternatively, one could use methods such as Forney's
  Generalized Minimum Distance (GMD) decoder for Reed-Solomon codes
  \cite{ref:forney}.  However, the naive decoder suffices for our
  purposes and works for any asymptotically good choice of the outer
  code.}:

\begin{enumerate}
\item Given a received sequence $(y_1, \ldots, y_s) \in (\F_2^n)^{s}$,
  apply an appropriate decoder for the inner codes (e.g., the
  brute-force decoder for BSC, or Gaussian elimination for BEC) to
  decode each $y_i$ to a codeword $c'_i$ of the $i$th code in the
  ensemble.

\item Apply the outer code decoder on $(c'_1, \ldots, c'_{s})$ that is
  guaranteed to correct some constant fraction of errors, to obtain a
  codeword $(c_1, \ldots, c_{s})$ of the outer code $\cout$.

\item Recover the decoded sequence from the corrected encoding $(c_1,
  \ldots, c_{s})$.
\end{enumerate}

Since the channel is assumed to be memoryless, the noise distributions
on inner codes are independent. Let $\cG \subseteq [s]$ denote the set
of coordinate positions corresponding to ``good'' inner codes in $\cS$
that achieve an error probability bounded by $\eta$.  By assumption,
we have $\cG \geq (1-\gamma)|\cS|$.

Suppose that the outer code $\cout$ corrects some $\gamma + \alpha$
fraction of adversarial errors, for a constant $\alpha > \eta$.  Then
an error might occur only if more than $\alpha N$ of the codes in
$\cG$ fail to obtain a correct decoding. We expect the number of
failures within the good inner codes to be $\eta |\cG|$. Due to the
noise independence, it is possible to show that the fraction of
failures may deviate from the expectation $\eta$ only with a
negligible probability. In particular, a direct application of
Chernoff bound implies that the probability that more than an $\alpha$
fraction of the good inner codes err is at most
\begin{equation} \label{eqn:justesenError} \eta^{\alpha' |\cG|} =
  2^{-\Omega_\alpha(\log(1/\eta) s)},
\end{equation}
where $\alpha' > 0$ is a constant that only depends on $\alpha$.  This
also upper bounds the error probability of the concatenated code.  In
particular, we see that if the error probability $\eta$ of the inner
codes is exponentially small in their block length $n$, the
concatenated code also achieves an exponentially small error in its
block length $N$.

Now we analyze the encoding and decoding complexity of the
concatenated code, assuming that Spielman's expander codes
(Theorem~\ref{thm:expander}) are used for the outer code. With this
choice, the outer code becomes equipped with a linear-time encoder and
decoder.  Since any linear code can be encoded in quadratic time (in
its block length), the concatenated code can be encoded in $O(n^2 s)$,
which for $s \gg n$ can be considered ``almost linear'' in the block
length $N = ns$ of $\C$. The decoding time of each inner code is cubic
in $n$ for the erasure channel, since decoding reduces to Gaussian
elimination, and thus for this case the naive decoder runs in time
$O(n^3 s)$.  For the symmetric channel, however, the brute-force
decoder used for the inner codes takes exponential time in the block
length, namely, $2^{Rn} \poly(n)$.  Therefore, the running time of the
decoder for concatenated code becomes bounded by $O(2^{Rn} s
\poly(n))$. When the inner ensemble is exponentially large; i.e., $s =
2^n$ (which is the case for our ensembles if we use the Leftover Hash
Lemma), the decoding complexity becomes $O(s^{1+R} \poly(\log s))$
which is at most quadratic in the block length of $\C$.

Since the rate $R'$ of the outer code can be made arbitrarily close to
$1$, rate of the concatenated code $\C$ can be made arbitrarily close
to the rate $R$ of the inner codes. Thus, if the ensemble of inner
codes is capacity-achieving, so would be the concatenated code.

\subsection{Density of the Explicit Family}

In the preceding section we saw how to obtain explicit capacity
achieving codes from capacity achieving code ensembles using
concatenation.  One of the important properties of the resulting
family of codes that is influenced by the size of the inner code
ensemble is the set of block lengths $N$ for which the concatenated
code is defined.  Recall that $N = ns$, where $n$ and $s$ respectively
denote the block length of the inner codes and the size of the code
ensemble, and the parameter $s$ is a function of $n$. For instance,
for all classical examples of capacity achieving code ensembles
(namely, Wozencraft's ensemble, Goppa codes and shortened cyclic
codes) we have $s(n) = 2^n$. In this case, the resulting explicit
family of codes would be defined for integer lengths of the form $N(i)
= i 2^i$.

A trivial approach for obtaining capacity achieving codes for all
lengths is to use a \emph{padding trick}. Suppose that we wish to
transmit a particular bit sequence of length $K$ through the channel
using the concatenated code family of rate $\rho$ that is taken to be
sufficiently close to the channel capacity. The sequence might
originate from a source that does not produce a constant stream of
bits (e.g., consider a terminal emulator that produces data only when
a user input is available).

Ideally, one requires the length of the encoded sequence to be $N =
\lceil K/\rho \rceil$. However, since the family might not be defined
for the block length $N$, we might be forced to take a code $\C$ in
the family with smallest length $N' \geq N$ that is of the form $N' =
n s(n)$, for some integer $n$, and pad the original message with
redundant symbols.  This way we have encoded a sequence of length $K$
to one of length $N'$, implying an effective rate $K/N'$.  The rate
loss incurred by padding is thus equal to $\rho - K/N' = K(1/N -
1/N')$.  Thus, if $N' \geq N(1+\delta)$ for some positive constant
$\delta > 0$, the rate loss becomes lower bounded by a constant and
thus, even if the original concatenated family is capacity achieving,
it no longer remains capacity achieving when extended to arbitrarily
chosen lengths using the padding trick.

Therefore, if we require the explicit family obtained from
concatenation to remain capacity achieving for all lengths, the set of
block lengths $\{i s(i)\}_{i \in \N}$ for which it is defined must be
sufficiently dense. This is the case provided that we have
\[
\frac{s(n)}{s(n+1)} = 1-o(1),
\]
which in turn, requires the capacity achieving code ensemble to have a
sub-exponential size (by which we mean $s(n) = 2^{o(n)}$).

Using the framework introduced in this chapter, linear extractors and
lossless condensers that achieve nearly optimal parameters would
result in code ensembles of polynomial size in $n$.  The explicit
erasure code ensemble obtained from Trevisan's extractor
(Corollary~\ref{coro:bec}) or Guruswami-Umans-Vadhan's lossless
condenser (Corollary~\ref{coro:becGUV}) combined with Justesen's
concatenation scheme results in an explicit sequence of capacity
achieving codes for the binary erasure channel that is defined for
every block length, and allows almost linear-time (i.e., $N^{1+o(1)}$)
encoding and decoding. Moreover, the latter sequence of codes that is
obtained from a lossless condenser is capacity achieving for the
binary symmetric channel (with a matching bit-flip probability) as
well.

\section{Duality of Linear Affine
  Condensers} \label{sec:dualityAffine}

In Section~\ref{sec:bec} we saw that linear extractors for bit-fixing
sources can be used to define generator matrices of a family of
erasure-decodable codes. On the other hand, we showed that linear
lossless condensers for bit-fixing sources define parity check
matrices of erasure-decodable codes.

Recall that generator and parity check matrices are dual notions, and
in our construction we have considered matrices in one-to-one
correspondence with linear mappings.  Indeed, we have used linear
mappings defined by extractors and lossless condensers to obtain
generator and parity check matrices of our codes (where the $i$th row
of the matrix defines the coefficient vector of the linear form
corresponding to the $i$th output of the mapping).  Thus, we get a
natural duality between linear functions: If two linear functions
represent generator and parity check matrices of the same code, they
can be considered dual\footnote{Note that, under this notion of
  duality, the dual of a linear function need not be unique even
  though its linear-algebraic properties (e.g., kernel) would be
  independent of its choice.}.  Just in the same way that the number
of rows of a generator matrix and the corresponding parity check
matrix add up to their number of columns (provided that there is no
linear dependence between the rows), the dual of a linear function
mapping $\F_q^n$ to $\F_q^m$ (where $m \leq n$) that has no linear
dependencies among its $n-m$ outputs can be taken to be a linear
function mapping $\F_q^n$ to $\F_q^{n-m}$.

In fact, a duality between linear extractors and lossless condensers
for affine sources is implicit in the analysis leading to
Corollary~\ref{coro:becCodes}. Namely, it turns out that if a linear
function is an extractor for an affine source, the dual function
becomes a \emph{lossless condenser} for the \emph{dual distribution},
and vice versa.  This is made precise (and slightly more general) in
the following theorem.

\begin{thm} \label{thm:conDuality} Suppose that the linear mapping
  defined by a matrix $G \in \F_q^{m\times n}$ of rank $m \leq n$ is a
  $(k \log q) \to_0 (k' \log q)$ condenser for a $k$-dimensional
  affine source $\cX$ over $\F_q^n$ so that for $X \sim \cX$, the
  distribution of $G \cdot X^\top$ has entropy at least $k' \log q$.
  Let $H \in \F_q^{(n-m) \times n}$ be a dual matrix for $G$ (i.e., $G
  H^\top = 0$) of rank $n-m$ and $\cY$ be an $(n-k)$-dimensional
  affine space over $\F_q^n$ supported on a translation of the dual
  subspace corresponding to the support of $\cX$. Then for $Y \sim
  \cY$, the distribution of $H \cdot Y^\top$ has entropy at least
  $(n-k+k'-m) \log q$.
\end{thm}

\begin{proof}
  \newcommand{\rker}{\mathsf{rker}} \newcommand{\lker}{\mathsf{lker}}
  Suppose that $\cX$ is supported on a set
  \[
  \{ x \cdot A_G + a\colon x \in \F_q^k \},
  \]
  where $A_G \in \F_q^{k \times n}$ has rank $k$ and $a \in \F_q^n$ is
  a fixed row vector.  Moreover we denote the dual distribution $\cY$
  by the set
  \[
  \{ y \cdot A_H + b\colon y \in \F_q^{n-k} \},
  \]
  where $b \in \F_q^n$ is fixed and $A_H \in \F_q^{(n-k) \times n}$ is
  of rank $n-k$, and we have the orthogonality relationship $A_H \cdot
  A_G^\top = 0$.

  The assumption that $G$ is a $(k \log q) \to_0 (k' \log
  q)$-condenser implies that the distribution
  \[
  G \cdot (A_G^\top \cdot \U_{\F_q^k} + a^\top),
  \]
  where $\U_{\F_q^k}$ stands for a uniformly random row vector in
  $\F_q^k$, is an affine source of dimension at least $k'$, equivalent
  to saying that the matrix $G \cdot A_G^\top \in \F_q^{m \times k}$
  has rank at least $k'$ (since rank is equal to the dimension of the
  image), or in symbols,
  \begin{equation} \label{eqn:rankGAg} \rk(G \cdot A_G^\top) \geq k'.
  \end{equation}
  Observe that since we have assumed $\rk(G) = m$, its right kernel is
  $(n-m)$-dimensional, and thus the linear mapping defined by $G$
  cannot reduce more than $n-m$ dimensions of the affine source
  $\cX$. Thus, the quantity $n-k+k'-m$ is non-negative.

  By a similar argument as above, in order to show the claim we need
  to show that
  \[
  \rk(H \cdot A_H^\top) \geq n-k+k'-m.
  \]
  Suppose not. Then the right kernel of $H\cdot A_H^\top \in
  \F_q^{(n-m)\times(n-k)}$ must have dimension larger than
  $(n-k)-(n-k+k'-m) = m-k'$. Denote this right kernel by $\mathcal{R}
  \subseteq \F_q^{n-k}$.  Since the matrix $A_H$ is assumed to have
  maximal rank $n-k$, and $n-k \geq m-k'$, for each nonzero $y \in
  \mathcal{R}$, the vector $y \cdot A_H \in \F_q^n$ is nonzero and
  since $H \cdot (A_H^\top y^\top) = 0$ (by the definition of right
  kernel), the duality of $G$ and $H$ implies that there is a nonzero
  $x \in \F_q^m$ where
  \[
  x \cdot G = y \cdot A_H,
  \]
  and the choice of $y$ uniquely specifies $x$. In other words, there
  is a subspace $\mathcal{R'} \subseteq \F_q^m$ such that \[
  \dim(\mathcal{R'}) = \dim(\mathcal{R}), \] and
  \[
  \{ x \cdot G\colon x \in \mathcal{R'} \} = \{ y \cdot A_H\colon y
  \in \mathcal{R} \}.
  \]
  But observe that, by orthogonality of $A_G$ and $A_H$, every $y$
  satisfies $y \cdot A_H A_G^\top = 0$, meaning that for every $x \in
  \mathcal{R'}$, we must have $x \cdot G A_G^\top = 0$.  Thus the left
  kernel of $G A_G^\top$ has dimension larger than $m-k'$ (since
  $\mathcal{R'}$ does), and we conclude that the matrix $G A_G^\top$
  has rank less than $k'$, a contradiction for \eqref{eqn:rankGAg}.
\end{proof}

Since every $k$-dimensional affine space over $\F_q^n$ has an
$(n-k)$-dimensional dual vector space, the above result combined with
Proposition~\ref{prop:zeroError} directly implies the following
corollary:

\begin{coro} \label{coro:conDuality} Suppose that the linear mapping
  defined by a matrix $G \in \F_q^{m\times n}$ of rank $m \leq n$ is a
  $(k \log q) \to_\eps (k' \log q)$ condenser, for some $\eps <
  1/2$. Let $H \in \F_q^{(n-m) \times n}$ of rank $n-m$ be so that
  $GH^\top = 0$.  Then, the linear mapping defined by $H$ is an $(n-k)
  \log q \to_0 (n-k+k'-m) \log q$ condenser. \qed
\end{coro}

Similarly, linear \emph{seeded} condensers for affine sources define
linear seeded \emph{dual condensers} for affine sources with
complementary entropy (this is done by taking the dual linear function
for every fixing of the seed).

Two important special cases of the above results are related to affine
extractors and lossless condensers. When the linear mapping $G$ is an
affine extractor for $k$-dimensional distributions, the dual mapping
$H$ becomes a lossless condenser for $(n-k)$-dimensional spaces, and
vice versa.

\musicBoxCapacity


\Chapter{Codes on the Gilbert-Varshamov Bound}
\epigraphhead[70]{\epigraph{\textsl{``I confess that Fermat's Theorem as an isolated
proposition has very little interest for me, because I could easily lay down a
multitude of such propositions, which one could neither prove
nor dispose of.''}}{\textit{--- Carl Friedrich Gauss}}}
\label{chap:gv}

\newcommand{\D}{\mathcal{D}} \newcommand{\sE}{\mathsf{E}}
\newcommand{\DSPACE}{\mathsf{DSPACE}} \newcommand{\DTIME}{\mathsf{DTIME}}
\newcommand{\PSPACE}{\mathsf{PSPACE}} \newcommand{\CP}{\mathsf{P}}
\newcommand{\NP}{\mathsf{NP}}
\newenvironment{proof_sketch}[1] {\noindent\hspace{2em}{\itshape Proof
    Sketch #1: }} {\hspace*{\fill}~\QED\par\endtrivlist\unskip}
\newenvironment{proofIdea}[1] {\noindent\hspace{2em}{\itshape Proof
    Idea #1: }} {\hspace*{\fill}~\QED\par\endtrivlist\unskip}
\newcommand{\fE}{f_{\mathsf{E}}}
\newcommand{\LE}{\mathcal{L}_{\mathsf{E}}} \newcommand{\logq}{\log_q}

One of the central problems in coding theory is the construction of
codes with extremal parameters.  Typically, one fixes an alphabet size
$q$, and two among the three fundamental parameters of the code
(block-length, number of codewords, and minimum distance), and asks
about extremal values of the remaining parameter such that there is a
code over the given alphabet with the given parameters.  For example,
fixing the minimum distance $d$ and the block-length $n$, one may ask
for the largest number of codewords $M$ such that there exists a code
over an alphabet with $q$ elements having $n,M,d$ as its parameters,
or in short, an $(n, M,d)_q$-code.

Answering this question in its full generality is extremely difficult,
especially when the parameters are large.  For this reason,
researchers have concentrated on asymptotic assertions: to any
$[n,\log M,d]_q$-code $C$ we associate a point $(\delta(C),R(C)) \in
[0,1]^2$, where $\delta(C) = d/n$ and $R(C)=\logq M/n$ are
respecitvely the relative distance and rate of the code.  A particular
point $(\delta,R)$ is called {\em asymptotically achievable} (over a
$q$-ary alphabet) if there exists a sequence $(C_1,C_2,\ldots)$ of
codes of increasing block-length such that $\delta(C_i)\to \delta$ and
$R(C_i)\to R$ as $i\to\infty$.

Even with this asymptotic relaxation the problem of determining the
shape of the set of asymptotically achievable points remains
difficult.  Let $\alpha_q(\delta)$ be defined as the supremum of all
$R$ such that $(\delta,R)$ is asymptotically achievable over a $q$-ary
alphabet.  It is known that $\alpha_q$ is a continuous function of
$\delta$~\cite{mani:81}, that $\alpha_q(0)=1$ (trivial), and
$\alpha_q(\delta)=0$ for $\delta\ge (q-1)/q$ (by the Plotkin bound).
However, for no $\delta\in(0,(q-1)/q)$ and for no $q$ is the value of
$\alpha_q(\delta)$ known.

What is known are lower and upper bounds for $\alpha_q$.  The best
lower bound known is due to Gilbert and
Varshamov\cites{gilbert,varshamov} which states that $\alpha_q(\delta)
\ge 1 - h_q(\delta)$, where the $q$-ary entropy function $h_q$ is
defined as
\[
h_q(\delta) \eqdef -\delta \log_q \delta
-(1-\delta)\log_q(1-\delta)+\delta \log_q(q-1).
\]
Up until 1982, years of research had made it plausible to think that
this bound is tight, i.e., that $\alpha_q(\delta) = 1-h_q(\delta)$.
Goppa's invention of algebraic-geometric codes~\cite{gopp:81}, and the
subsequent construction of Tsfasman, Vl{\u a}du{\c t}, and Zink
\cite{tsvz:82} using curves with many points over a finite field and
small genus showed however that the bound is not tight when the
alphabet size is large enough.  Moreover, Tsfasman et~al.~also gave a
polynomial time construction of such codes (which has been greatly
simplified since, see, e.g.,~\cite{gast:95}).

The fate of the binary alphabet is still open.  Many researchers still
believe that $\alpha_2(\delta) = 1-h_2(\delta)$.  In fact, for a
randomly chosen linear code $C$ (one in which the entries of a
generator matrix are chosen independently and uniformly over the
alphabet) and for any positive $\varepsilon$ we have $R(C)\ge 1 -
h_q(\delta(C))-\varepsilon$ with high probability (with probability at
least $1-2^{-nc_\varepsilon}$ where $n$ is the block-length and
$c_\varepsilon$ is a constant depending on $\varepsilon$).  However,
even though this shows that most randomly chosen codes are arbitrarily
close to the Gilbert-Varshamov bound, no explicit polynomial time
construction of such codes is known when the alphabet size is small
(e.g., for binary alphabets).

In this chapter, we use the technology of pseudorandom generators
which has played a prominent role in the theoretical computer science
research in recent years to (conditionally) produce, for any
block-length $n$ and any rate $R<1$, a list of $\poly(n)$ many codes
of block length $n$ and designed rate $R$ (over an arbitrary alphabet)
such that a very large fraction of these codes has parameters
arbitrarily close to the Gilbert-Varshamov bound.  Here, $\poly(n)$
denotes a polynomial in $n$.

In a nutshell, our construction is based on the pseudorandom generator
of Nisan and Wigderson \cite{NW}. In particular, we will first
identify a Boolean function $f$ of which we assume that it satisfies a
certain com\-plex\-ity-theoretic assumption.  More precisely, we assume
that the function cannot be computed by algorithms that require
sub-exponential amount of memory.
A natural candidate for such a function is given later in the chapter.
This function is then extended to produce $nk$ bits from $O(\log n)$
bits.  The extended function is called a \emph{pseudorandom
  generator.}  The main point about this extended function is that the
$nk$ bits produced cannot be distinguished from random bits by a
Turing machine with restricted resources.
In our case, the output cannot be distinguished from a random sequence
when a Turing machine is used which uses only an amount of space that
is polynomially bounded in the length of its input.

The new $nk$ bits are regarded as the entries of a generator matrix of
a code.  Varying the base $O(\log n)$ bits in all possible ways gives
us a polynomially long list of codes of which we can show that a
majority lies asymptotically on the Glibert-Varshamov bound, provided
the hardness assumption is satisfied\footnote{ We remark that the
  method used in this chapter can be regarded as a ``relativized''
  variation of the original Nisan-Wigderson generator and, apart from
  construction of error-correcting codes, can be applied to a vast
  range of probabilistic constructions of combinatorial objects (e.g.,
  Ramsey graphs, combinatorial designs, etc). Even though this
  derandomization technique seems to be ``folklore'' among the
  theoretical computer science community, it is included in the thesis
  mainly since there appears to be no elaborate and specifically
  focused writeup of it in the literature.}.

\section{Basic Notation}

We begin with the definitions of the terms we will use throughout the
chapter. For simplicity, we restrict ourselves to the particular cases
of our interest and will avoid presenting the definitions in full
generality.  See Appendix~\ref{app:coding} for a quick review of the
basic notions in coding theory and \cites{ref:sipser,papa} for
com\-plex\-ity-theoretic notions.

Our main tool in this chapter is a hardness-based pseudorandom
generator.  Informally, this is an efficient algorithm that receives a
sequence of truly random bits at input and outputs a much longer
sequence \emph{looking} random to any distinguisher with bounded
computational power.  This property of the pseudorandom generator can
be guaranteed to hold by assuming the existence of functions that are
hard to compute for certain computational devices. This is indeed a
broad sketch; Depending on what we precisely mean by the quantitative
measures just mentioned, we come to different definitions of
pseudorandom generators. Here we will be mainly interested in
computational hardness against algorithms with \emph{bounded} space
complexity.

Hereafter, we will use the shorthand $\DSPACE[s(n)]$ to denote the
class of problems solvable with $O(s(n))$ bits of working memory and
$\sE$ for the class of problems solvable in time $2^{O(n)}$ (i.e.,
$\sE=\bigcup_{c\in\N}\DTIME[2^{cn}]$, where $\DTIME[t(n)]$ stands for
the class of problems deterministically solvable in time $O(t(n))$).

Certain arguments that we use in this chapter require
\emph{non-uniform} computational models.  Hence, we will occasionally
refer to algorithms that receive \emph{advice strings} to help them
carry out their computation. Namely, in addition to the input string,
the algorithm receives an \emph{advice} string whose content only
depends on the \emph{length} of the input and not the input itself. It
is assumed that, for every $n$, there is an advice string that makes
the algorithm work correctly on \emph{all} inputs of length $n$. We
will use the notation $\DSPACE[f(n)]/g(n)$ for the class of problems
solvable by algorithms that receive $g(n)$ bits of advice and use
$O(f(n))$ bits of working memory.

\begin{defn}
  Let $S\colon \N \to \N$ be a (constructible) function.  A Boolean
  function $f \colon \{0,1\}^{\ast} \to \{0,1\}$ is said to have
  \emph{hardness} $S$ if for every algorithm $A$ in
  $\DSPACE[S(n)]/O(S(n))$ and infinitely many $n$ (and no matter how
  the advice string is chosen) it holds that \[ |\Pr_x[A(x)=f(x)] -
  1/2| < 1/S(n), \] where $x$ is uniformly sampled from $\{0,1\}^n$.
\end{defn}

Obviously, any Boolean function can be trivially computed correctly on
at least half of the inputs by an algorithm that always outputs a
constant value (either $0$ or $1$). Intuitively, for a hard function
no \emph{efficient} algorithm can do much better. For the purpose of
this chapter, the central hardness assumption that we use is the
following:

\begin{assum}
  \label{assum:main}
  There is a Boolean function in $\sE$ with hardness at least
  $2^{\epsilon n}$, for some constant $\epsilon > 0$.
\end{assum}

The term \emph{pseudorandom generator} emphasizes the fact that it is
infor\-mation-theoretically impossible to transform a sequence of truly
random bits into a longer sequence of truly random bits, hence the
best a transformation with a nontrivial stretch can do is to generate
bits that \emph{look} random to a \emph{particular family} of
observers. To make this more precise, we need to define
\emph{computational indistinguishability} first.

\begin{defn}
  Let $p=\{p_n\}$ and $q=\{q_n\}$ be families of probability
  distributions, where $p_n$ and $q_n$ are distributed over
  $\{0,1\}^n$. Then $p$ and $q$ are $(S, \ell,
  \epsilon)$-\emph{indistinguishable} (for some $S, \ell \colon \N \to
  \N$ and $\epsilon \colon \N \to (0,1)$) if for every algorithm $A$
  in $\DSPACE(S(n))/O(\ell(n))$ and infinitely many $n$ (and no matter
  how the advice string is chosen) we have that
  \[ |\Pr_x[A(x)=1]-\Pr_y[A(y)=1]| < \epsilon(n), \] where $x$ and $y$
  are sampled from $p_n$ and $q_n$, respectively.
\end{defn}

This is in a way similar to computational hardness. Here the
\emph{hard task} is \emph{telling the difference} between the
sequences generated by different sources.  In other words, two
probability distributions are indistinguishable if any
resource-bounded observer is \emph{fooled} when given inputs sampled
from one distribution rather than the other. Note that this may even
hold if the two distributions are not statistically close to each
other.

Now we are ready to define pseudorandom generators we will later need.

\begin{defn}
  \label{def:PRG}
  A deterministic algorithm that computes a function \[ G\colon
  \{0,1\}^{c\log n} \to \{0,1\}^n\] (for some constant $c > 0$) is
  called a (high-end) \emph{pseudorandom generator} if the following
  conditions hold:
  \begin{enumerate}
  \item It runs in polynomial time with respect to $n$.
  \item Let the probability distribution $G_n$ be defined uniformly
    over the range of $G$ restricted to outputs of length $n$. Then
    the family of distributions $\{G_n\}$ is $(n, n,
    1/n)$-indistinguishable from the uniform distribution.
  \end{enumerate}
  An input to the pseudorandom generator is referred to as a
  \emph{random seed}. Here the length of the output as a function of
  the seed length $s$, known as the \emph{stretch} of the pseudorandom
  generator, is required to be the exponential function $2^{s/c}$.
\end{defn}

\section{The Pseudorandom Generator}
\label{se:prg}

A pseudorandom generator, as we just defined, extends a truly random
sequence of bits into an exponentially long sequence that \emph{looks}
random to any efficient distinguisher. From the definition it is not
at all clear whether such an object could exist. In fact the existence
of pseudorandom generators (even much weaker than our definition) is
not yet known. However, there are various constructions of
pseudorandom generators based on unproven (but seemingly plausible)
assumptions. The presumed assumption is typically chosen in line with
the same guideline, namely, a computational task being
\emph{intractable}. For instance, the early constructions of
\cite{shamir} and \cite{BM} are based on the intractability of certain
num\-ber-theoretic problems, namely, integer factorization and the
discrete logarithm function.  Yao \cite{yao} extends these ideas to
obtain pseudorandomness from one-way permutations. This is further
generalized by \cite{HILL} who show that the existence of \emph{any}
one-way function is sufficient. However, these ideas are mainly
motivated by cryptographic applications and often require strong
assumptions.

The prototypical pseudorandom generator for the applications in
derandomization, which is of our interest, is due to Nisan and
Wigderson\cite{NW}. They provide a broad range of pseudorandom
generators with different strengths based on a variety of hardness
assumptions. In rough terms, their generator works by taking a hard
function for a certain complexity class, evaluating it in carefully
chosen points (related to the choice of the random seed), and
outputting the resulting sequence. Then one can argue that an
efficient distinguisher can be used to efficiently compute the hard
function, contradicting the assumption. Note that for certain
complexity classes, hard functions are \emph{provably} known. However,
they typically give generators too weak to be applied in typical
derandomizations.  Here we simply apply the Nisan-Wigderson
construction to obtain a pseudorandom generator which is \emph{robust}
against space-efficient computations. This is shown in the following
theorem:

\begin{thm}
  Assumption~\ref{assum:main} implies the existence of a pseudorandom
  generator as in Definition~\ref{def:PRG}. That is to say, suppose
  that there is a constant $\epsilon > 0$ and a Boolean function
  computable in time $2^{O(n)}$ that has hardness $2^{\epsilon
    n}$. Then there exists a function $G\colon \{0,1\}^{O(\log n)} \to
  \{0,1\}^n$ computable in time polynomial in $n$ whose output (when
  given uniformly random bits at input) is indistinguishable from the
  uniform distribution for all algorithms in $\DSPACE[n]/O(n)$.
\end{thm}
\begin{proof}{\cite{NW}}
  Let $f$ be a function satisfying Assumption~\ref{assum:main} for
  some fixed $\epsilon > 0$, and recall that we intend to generate $n$
  pseudorandom bits from a truly random seed of length $\ell$ which is
  only logarithmically long in $n$.

  The idea of the construction is as follows: We evaluate the hard
  function $f$ in $n$ carefully chosen points, each of the same length
  $m$, where $m$ is to be determined shortly. Each of these
  $m$-bit~long inputs is obtained from a particular subset of the
  $\ell$ bits provided by the random seed.  This can be conveniently
  represented in a matrix form: Let $\D$ be an $n \times \ell$ binary
  matrix, each row of which having the same weight $m$. Now the
  pseudorandom generator $G$ is described as follows: The $i\Th$ bit
  generated by $G$ is the evaluation of $f$ on the projection of the
  $\ell$-bit long input sequence to those coordinates indicated by the
  $i\Th$ row of $\D$.  Note that because $f$ is in $\sE$, the output
  sequence can be computed in time polynomial in $n$, as long as $m$
  is logarithmically small.

  As we will shortly see, it turns out that we need $\D$ to satisfy a
  certain \emph{small-overlap} property. Namely, we require the
  bitwise product of each pair of the rows of $\D$ to have weight at
  most $\log n$. A straightforward counting argument shows that, for a
  logarithmically large value of $m$, the parameter $\ell$ can be kept
  logarithmically small as well. In particular, for the particular
  choice of $m \eqdef \frac{2}{\epsilon} \log n$, the matrix $\D$
  exists with $\ell = O(\log n)$.  Moreover, rows of the matrix can be
  constructed (in time polynomial in $n$) using a simple greedy
  algorithm.

  To show that our construction indeed gives us a pseudorandom
  generator, suppose that there is an algorithm $A$ working in
  $\DSPACE[n]/O(n)$ which is able to distinguish the output of $G$
  from a truly random sequence with a bias of at least $1/n$.  That
  is, for all large enough $n$ it holds that
  \[ \delta \eqdef
  |\Pr_y[A^{\alpha(n)}(y)=1]-\Pr_x[A^{\alpha(n)}(G(x))=1]| \geq
  1/n, \] where $x$ and $y$ are distributed uniformly in
  $\{0,1\}^\ell$ and $\{0,1\}^n$, respectively, and $\alpha(n)$ in the
  superscript denotes an advice string of linear length (that only
  depends on $n$).  The goal is to transform $A$ into a
  space-efficient (and non-uniform) algorithm that approximates $f$,
  obtaining a contradiction.

  Without loss of generality, let the quantity inside the absolute
  value be non-negative (the argument is similar for the negative
  case).  Let the distribution $D_i$ (for $0 \leq i \leq n$) over
  $\{0,1\}^n$ be defined by concatenation of the length-$i$ prefix of
  $G(x)$, when $x$ is chosen uniformly at random from $\{0,1\}^\ell$,
  with a Boolean string of length $n-i$ obtained uniformly at random.
  Define $p_i$ as $\Pr_z[A^{\alpha(n)}(z)=1]$, where $z$ is sampled
  from $D_i$, and let $\delta_i \eqdef p_{i-1}-p_i$.  Note that $D_0$
  is the uniform distribution and $D_n$ is uniformly distributed over
  the range of $G$. Hence, we have $\sum_{i=1}^{n}\delta_i = p_0-p_n =
  \delta \geq 1/n$, meaning that for some $i$, $\delta_{i} \geq
  1/n^2$. Fix this $i$ in the sequel.

  Without loss of generality, assume that the $i\Th$ bit of $G(x)$
  depends on the first $m$ bits of the random seed. Now consider the
  following randomized procedure $B$: Given $i-1$ input bits $u_1,
  \ldots, u_{i-1}$, choose a binary sequence $r_i, \ldots, r_n$
  uniformly at random and compute
  $A^{\alpha(n)}(u_1,\ldots,u_{i-1},r_i,\ldots,r_n)$. If the output
  was $1$ return $r_i$, otherwise, return the negation of $r_i$. It is
  straightforward to show that
  \begin{equation}
    \label{eqn:distinguisher}
    \Pr_{x,r}[B(G(x)_1^{i-1})=G(x)_i]\geq \frac{1}{2}+\delta_i.
  \end{equation}
  Here, $G(x)_1^{i-1}$ and $G(x)_i$ are shorthands for the $(i-1)$-bit
  long prefix of $G(x)$ and the $i\Th$ bit of $G(x)$, respectively,
  and the probability is taken over the choice of $x$ and the internal
  coins of $B$.

  So far we have constructed a linear-time probabilistic procedure for
  \emph{guessing} the $i\Th$ pseudorandom bit from the first $i-1$
  bits. By averaging, we note that there is a particular choice of
  $r_i, \ldots, r_n$, independent of $x$, that preserves the bias
  given in \eqref{eqn:distinguisher}.  Furthermore, note that the
  function $G(x)_i$ we are trying to guess, which is in fact
  $f(x_1,\ldots,x_m)$, does \emph{not} depend on $x_{m+1}, \ldots,
  x_\ell$. Therefore, again by averaging we see that these bits can
  also be fixed. Therefore, for a given sequence $x_{1}, \ldots, x_m$,
  one can compute $G(x)_1^{i-1}$, feed it to $B$ (having known the
  choices we have fixed), and guess $G(x)_i$ with the same bias as in
  \eqref{eqn:distinguisher}.  The problem is of course that
  $G(x)_1^{i-1}$ does not seem to be easily computable. However, what
  we know is that each bit of this sequence depends only on $\log n$
  bits of $x_1,\ldots,x_m$, followed by the construction of
  $\D$. Hence, having fixed $x_{m+1},\ldots,x_\ell$, we can trivially
  describe each bit of $G(x)_1^{i-1}$ by a Boolean formula (or a
  Boolean circuit) of exponential size (that is, of size $O(2^{\log
    n})=O(n)$). These $i-1=O(n)$ Boolean formulae can be encoded as an
  additional advice string of length $O(n^2)$ (note that their
  descriptions only depend on $n$), implying that $G(x)_1^{i-1}$ can
  be computed in linear space using $O(n^2)$ bits of advice.

  All the choices we have fixed so far (namely, $i$, $r_{i}, \ldots,
  r_{n}$, $x_{m+1}, \ldots, x_\ell$) only depend on $n$ and can be
  \emph{absorbed} into the advice string as
  well\footnote{Alternatively, one can avoid using this additional
    advice by enumerating over all possible choices and taking a
    majority vote. However, this does not decrease the total advice
    length by much.}.  Combined with the \emph{bit-guessing} algorithm
  we just described, this gives us a linear-space algorithm that needs
  an advice of quadratic length and correctly computes
  $f(x_1,\ldots,x_m)$ on at least a $\frac{1}{2}+\delta_i$ fraction of
  inputs, which is off from $1/2$ by a bias of at least $1/n^2$. But
  this is not possible by the hardness of $f$, which is assumed to be
  at least $2^{\epsilon m} = n^2$. Thus, $G$ must be a pseudorandom
  generator.
\end{proof}

The above proof uses a function that is completely unpredictable for
every efficient algorithm. Impagliazzo and Wigderson \cite{IW} improve
the construction to show that this requirement can be relaxed to one
that only requires a \emph{worst case} hardness, meaning that the
function computed by any \emph{efficient} (non-uniform) algorithm
needs to differ from the hard function on at least one input.  In our
application, this translates into the following hardness assumption:

\begin{assum}
  There is a constant $\epsilon > 0$ and a function $f$ in $\sE$ such
  that every algorithm in $\DSPACE[S(n)]/O(S(n))$ that correctly
  computes $f$ requires $S(n)=\Omega(2^{\epsilon n})$.
  \label{assum:IW}
\end{assum}

The idea of their result (which was later reproved in \cite{STV} using
a coding-theoretic argument) is to \emph{amplify} the given hardness,
that is, to transform a worst-case hard function in $\sE$ to another
function in $\sE$ which is hard on average. In our setting, this gives
us the following (since the proof essentially carries over without
change, we only sketch the idea):

\begin{thm}
  Assumption~\ref{assum:IW} implies Assumption~\ref{assum:main} and
  hence, the existence of pseudorandom generators.
\end{thm}

\begin{proof}[Proof Idea]{\cite{STV}}
  Let a function $f$ be hard in worst case. Consider the truth table
  of $f$ as a string $x$ of length $N \eqdef 2^n$. The main ingredient
  of the proof is a linear code $\C$ with dimension $N$ and length
  polynomial in $N$, which is obtained by concatenation of a
  Reed-Muller code with the Hadamard code. The code is list-decodable
  up to a fraction $\frac{1}{2}-\epsilon$ of errors, for arbitrary
  $\epsilon > 0$. Moreover, decoding can be done in sub-linear time,
  that is, by querying the \emph{received word} only at a small number
  of (randomly chosen) positions. Then the truth table of the
  transformed function $g$ can be simply defined as the encoding of
  $x$ with $\C$. Hence $g$ can be evaluated at any point in time
  polynomial in $N$, which shows that $g \in \sE$. Further, suppose
  that an algorithm $A$ can space-efficiently compute $g$ correctly in
  a fraction of points non-negligibly bounded away from $1/2$
  (possibly using an advice string). Then the function computed by $A$
  can be seen as a \emph{corrupted} version of the \emph{codeword} $g$
  and can be efficiently \emph{recovered} using the list-decoding
  algorithm.  From this, one can obtain a space-efficient algorithm
  for computing $f$, contradicting the hardness of $f$. Hence $g$ has
  to be hard on average.
\end{proof}



While the above result seems to require hardness against non-uniform
algorithms (as phrased in Assumption~\ref{assum:IW}), we will see that
the hardness assumption can be further relaxed to the following, which
only requires hardness against uniform algorithms:

\begin{assum}
  \label{assum:strong}
  The complexity class $\sE$ is not contained in $\DSPACE[2^{o(n)}]$.
\end{assum}

\begin{rems} A result by Hopcroft~et~al.~\cite{HPV77} shows a
  deterministic simulation of time by space. Namely, they prove that
  \[ \DTIME[t(n)] \subseteq \DSPACE[t(n)/\log t(n)].\] However, this
  result is not strong enough to influence the hardness assumption
  above. To violate the assumption, a much more space-efficient
  simulation in the form \[\DTIME[t(n)] \subseteq
  \DSPACE[t(n)^{o(1)}]\] is required.
\end{rems}

Before we show the equivalence of the two assumptions (namely,
Assumption~\ref{assum:IW} and Assumption~\ref{assum:strong}), we
address the natural question of how to construct an \emph{explicit}
function to satisfy the required hardness assumption (after all,
evaluation of such a function is needed as part of the pseudorandom
generator construction). One possible candidate (which is a canonical
hard function for $\sE$) is proposed in the following lemma:

\begin{lem}
  \label{lem:candidate}
  Let $\LE$ be the set (encoded in binary)
  \begin{multline*}
\{\langle M, x, t, i \rangle \mid \text{$M$ is a Turing machine, where given} \\ \text{input $x$ at
    time $t$ the $i\Th$ bit of its configuration is $1$}\},
    \end{multline*}
    and let
  the Boolean function $\fE$ be its characteristic function.  Then if
  Assumption~\ref{assum:strong} is true, it is satisfied by $\fE$.
\end{lem}

\begin{proof}
  First we show that $\LE$ is complete for $\sE$ under Turing
  reductions bounded in linear space. The language being in $\sE$
  directly follows from the efficient constructions of universal
  Turing machines.  Namely, given a properly-encoded input $\langle M,
  x, t, i\rangle$, one can simply simulate the Turing machine $M$ on
  $x$ for $t$ steps and decide according to the configuration obtained
  at time $t$.  This indeed takes exponential time. Now let $L$ be any
  language in $\sE$ which is computable by a Turing machine $M$ in
  time $2^{cn}$, for some constant $c>0$.  For a given $x$ of length
  $n$, using an oracle for solving $\fE$, one can query the oracle
  with inputs of the form $\langle M, x, 2^{cn}, i\rangle$ (where the
  precise choice of $i$ depends on the particular encoding of the
  configurations) to find out whether $M$ is in an accepting state,
  and hence decide $L$. This can obviously be done in space linear in
  $n$, which concludes the completeness of $\LE$.  Now if
  Assumption~\ref{assum:strong} is true and is not satisfied by $\fE$,
  this completeness result allows one to compute all problems in $\sE$
  in sub-exponential time, which contradicts the assumption.
\end{proof}

The following lemma shows that this seemingly weaker assumption is in
fact sufficient for our pseudorandom generator:

\begin{lem}
  Assumptions~\ref{assum:IW} and \ref{assum:strong} are equivalent.
\end{lem}

\begin{proof}
  This argument is based on \cite{peter}*{Section~5.3}.  First we
  observe that, given a \emph{black box} $C$ that receives $n$ input
  bits and outputs a single bit, it can be verified in linear space
  whether $C$ computes the restriction of $\fE$ to inputs of length
  $n$.  To see this, consider an input of the form $\langle M, x, t, i
  \rangle$, as in the statement of Lemma~\ref{lem:candidate}.  The
  correctness of $C$ can be explicitly checked when the time parameter
  $t$ is zero (that is, $C$ has to agree with the initial
  configuration of $M$). Moreover, for every time step $t > 0$, the
  answer given by $C$ has to be consistent with that of the previous
  time step (namely, the transition made at the location of the head
  should be legal and every other position of the tape should remain
  unchanged). Thus, on can verify $C$ simply by enumerating all
  possible inputs and verifying whether the answer given by $C$
  remains consistent across subsequent time steps. This can obviously
  be done in linear space.

  Now suppose that Assumption~\ref{assum:strong} is true and hence, by
  Lemma~\ref{lem:candidate}, is satisfied by $\fE$. That is, there is
  a constant $\epsilon > 0$ such that every algorithm for computing
  $\fE$ requires space $O(2^{\epsilon n})$. Moreover, assume that
  there is an algorithm $A$ working in $\DSPACE[S(n)]/O(S(n))$ that
  computes $\fE$. Using the verification procedure described above,
  one can (uniformly) simulate $A$ in space $O(S(n))$ by enumerating
  all choices of the advice string and finding the one that makes the
  algorithm work correctly. Altogether this requires space
  $O(S(n))$. Combined with the hardness assumption, we conclude that
  $S(n)=\Omega(2^{\epsilon n})$.  The converse direction is obvious.
\end{proof}

Putting everything together, we obtain a very strong pseudorandom
generator as follows:

\begin{coro}
  \label{coro:strong}
  Assumption~\ref{assum:strong} implies the existence of pseudorandom
  generators whose output of length $n$ is
  $(n,n,1/n)$-indistinguishable from the uniform distribution. \qed
\end{coro}

\section{Derandomized Code Construction}
\label{se:codes}

As mentioned before, the bound given by Gilbert and
Varshamov\cites{gilbert,varshamov} states that, for a $q$-ary
alphabet, large enough $n$, and for any value of $0 \leq \delta \leq
(q-1)/q$, there are codes with length $n$, relative distance at least
$\delta$ and rate $r \geq 1-h_q(\delta)$, where $h_q$ is the $q$-ary
entropy function.
Moreover, a random linear code (having each entry of its generator
matrix chosen uniformly at random) achieves this bound.  In fact, for
all $0 \leq r \leq 1$, in the family of linear codes with length $n$
and (designed) dimension $nr$, all but only a sub-constant fraction of
the codes achieve the bound when $n$ grows to infinity. However, the
number of codes in the family is exponentially large ($q^{nr}$) and we
do not have an \emph{a priori} indication on which codes in the family
are good. Putting it differently, a randomized algorithm that merely
outputs a random generator matrix \emph{succeeds} in producing a code
on the GV~bound with probability $1-o(1)$.  However, the number of
random bits needed by the algorithm is $nk\log q$.  For simplicity, in
the sequel we only focus on binary codes, for which no explicit
construction approaching the GV~bound is known.

The randomized procedure above can be considerably derandomized by
considering a more restricted family of codes. Namely, fix a length
$n$ and a basis for the finite field $\F_{m}$, where $m \eqdef
2^{n/2}$. Then over such a basis there is a natural isomorphism
between the elements of $\F_{m}$ and the elements of the vector space
$\F_2^{n/2}$. Now for each $\alpha \in \F_m$, define the code
$\C_\alpha$ as the set $\{\langle x, \alpha x \rangle \mid x \in \F_m
\}$, where the elements are encoded in binary\footnote{These codes are
  attributed to J.~M.~Wozencraft (see~\cite{massey}).}. This binary
code has rate $1/2$.  Further, it is well known that $\C_\alpha$
achieves the GV~bound for all but $1-o(1)$ fraction of the choices of
$\alpha$. Hence in this family a randomized construction can obtain
very good codes using only $n/2$ random bits.  Here we see how the
pseudorandom generator constructed in the last section can
dramatically reduce the amount of randomness needed in all code
constructions.  Our observation is based on the composition of the
following facts:

\begin{enumerate}
\item \textit{Random codes achieve the Gilbert-Varshamov bound:} It is
  well known that a simple randomized algorithm that chooses the
  entries of a generator matrix uniformly at random obtains a linear
  code satisfying the Gilbert-Varshamov bound with overwhelming
  probability \cite{varshamov}.

\item \textit{Finding the minimum distance of a (linear) code can be
    performed in linear space:} One can simply enumerate all the
  codewords to find the minimum weight codeword, and hence, the
  distance of the code. This only requires linear amount of memory
  with respect to the block length.

\item \textit{Provided a \emph{hardness condition}, namely that
    sub-exponential space algorithms cannot compute all the problems
    in $\sE$, every linear space algorithm can be \emph{fooled} by an
    explicit pseudorandom generator:} This is what we obtained in
  Corollary~\ref{coro:strong}.
\end{enumerate}

Now we formally propose a general framework that can be employed to
derandomize a wide range of combinatorial constructions.

\begin{lem}
  \label{lem:derandom}
  Let $\mathcal{S}$ be a family of combinatorial objects of
  (binary-encoded) length $n$, in which an $\epsilon$ fraction of the
  objects satisfy a property $P$. Moreover, suppose that the family is
  efficiently samplable, that is, there is a polynomial-time algorithm
  (in $n$) that, for a given $i$, generates the $i\Th$ member of the
  family. Further assume that the property $P$ is verifiable in
  polynomial space. Then for every constant $k > 0$, under
  Assumption~\ref{assum:strong}, there is a constant $\ell$ and an
  efficiently samplable subset of $\mathcal{S}$ of size at most
  $n^\ell$ in which at least an $\epsilon - n^{-k}$ fraction of the
  objects satisfy $P$.
\end{lem}

\begin{proof}
  Let $A$ be the composition of the sampling algorithm with the
  verifier for $P$.  By assumption, $A$ needs space $n^s$, for some
  constant $s$. Furthermore, when the input of $A$ is chosen randomly,
  it outputs $1$ with probability at least $\epsilon$. Suppose that
  the pseudorandom generator of Corollary~\ref{coro:strong} transforms
  $c \log n$ truly random bits into $n$ pseudorandom bits, for some
  constant $c>0$. Now it is just enough to apply the pseudorandom
  generator on $c \cdot \max\{s, k\} \cdot \log n$ random bits and
  feed $n$ of the resulting pseudorandom bits to $A$. By this
  construction, when the input of the pseudorandom generator is chosen
  uniformly at random, $A$ must still output $1$ with probability
  $\epsilon-n^{-k}$ as otherwise the pseudorandomness assumption would
  be violated. Now the combination of the pseudorandom generator and
  $A$ gives the efficiently samplable family of the objects we want,
  for $\ell \eqdef c \cdot \max\{s, k\}$, as the random seed runs over
  all the possibilities.
\end{proof}

As the distance of a code is obviously computable in linear space by
enumeration of all the codewords, the above lemma immediately implies
the existence of a (constructible) polynomially large family of codes
in which at least $1-n^{-k}$ of the codes achieve the GV~bound, for
arbitrary $k$.

\begin{rems}
  As shown in the original work of Nisan and Wigderson \cite{NW}
  (followed by the hardness amplification of Impagliazzo and Wigderson
  \cite{IW}) all randomized polynomial-time algorithms (namely, the
  complexity class $\BPP$) can be fully derandomized under the
  assumption that $\sE$ cannot be computed by Boolean circuits of
  sub-exponential size. This assumption is also sufficient to
  derandomize probabilistic constructions that allow a (possibly
  non-uniform) polynomial-time verification procedure for deciding
  whether a particular object has the desirable properties. For the
  case of good error-correcting codes, this could work if we knew of a
  procedure for computing the minimum distance of a linear code using
  circuits of size polynomial in the length of the code.  However, it
  turns out that (the decision version of) this problem is
  $\NP$-complete \cite{alex}, and even the approximation version
  remains $\NP$-complete \cite{madhu}. This makes such a possibility
  unlikely.

  However, a key observation, due to Klivans and van~Melkebeek
  \cite{KvM}, shows that the Nisan-Wigderson construction (as well as
  the Impagliazzo-Wigderson amplification) can be \emph{relativized}.
  Namely, starting from a hardness assumption for a certain family of
  \emph{oracle circuits} (i.e., Boolean circuits that can use special
  gates to compute certain Boolean functions as \emph{black box}) one
  can obtain pseudorandom generators secure against oracle circuits of
  the same family. In particular, this implies that any probabilistic
  construction that allows polynomial~time verification using $\NP$
  oracles (including the construction of good error-correcting codes)
  can be derandomized by assuming that $\sE$ cannot be computed by
  sub-exponential sized Boolean circuits that use $\NP$ oracle
  gates. However, the result given by Lemma~\ref{lem:derandom} can be
  used to derandomize a more general family of probabilistic
  constructions, though it needs a slightly stronger hardness
  assumption which is still plausible.
\end{rems}

\musicBoxGV


\Chapter{Concluding Remarks}
\epigraphhead[70]{\epigraph{\textsl{``To achieve great things, two
      things are needed; a plan, and not quite enough
      time.''}}{\textit{--- Leonard Bernstein}}}
\label{chap:conclude}

In this thesis, we investigated the role of objects studied at the
core of theoretical computer science--namely, randomness extractors,
condensers and pseudorandom generators--in efficient construction of
combinatorial objects suitable for more practical applications. The
applications being considered all share a coding-theoretic flavor and
include:
\begin{enumerate}
\item Wiretap coding schemes, where the goal is to provide
  infor\-mation-theo\-retic secrecy in a communication channel that is
  partially observable by an adversary (Chapter~\ref{chap:wiretap});

\item Combinatorial group testing schemes, that allow for efficient
  identification of sparse binary vectors using potentially unreliable
  disjunctive measurements (Chapter~\ref{chap:testing});

\item Capacity achieving codes, which provide optimally efficient and
  reliable transmission of information over unreliable discrete
  communication channels (Chapter~\ref{chap:capacity});

\item Codes on the Gilbert-Varshamov bound, which are error-correcting
  codes whose rate-distance trade-off matches what achieved by
  probabilistic constructions (Chapter~\ref{chap:gv}).
\end{enumerate}

We conclude the thesis by a brief and informal discussion of the
obtained results, open problems and possible directions for future
research.

\section*{Wiretap Protocols}

In Chapter~\ref{chap:wiretap} we constructed rate-optimal wiretap
schemes from optimal affine extractors. The combinatorial structure of
affine extractors guarantees almost perfect privacy even in presence
of linear manipulation of information. This observation was the key
for our constructions of infor\-mation-theoretically optimal schemes
in presence of noisy channels, active intruders, and linear network
coding.

Despite being sufficiently general for a wide range of practical
applications, it makes sense to consider different types of
intermediate processing.  We showed in Section~\ref{subsec:arbitrary}
that, at the cost of giving up zero leakage, it is possible to use
seeded extractors to provide secrecy in presence of arbitrary forms of
transformations. However, in order to attain zero leakage, it becomes
inevitable to construct seedless, invertible extractors for a class of
random sources that capture the nature of post-processing being
allowed.

For example, suppose that the encoded information is transmitted
through a packet network towards a destination, where information is
arbitrarily manipulated by intermediate routers, but is routed from
the source to the destination through $k \geq 2$ separated paths. In
this case, the intruder may learn a limited amount of information from
each of the $k$ components of the network. Similar arguments as what
presented in Chapter~\ref{chap:wiretap} can now be used to show that
the object needed for ensuring secrecy in this ``route-disjoint''
setting is invertible, $k$-source extractors. Shaltiel
\cite{ref:mileage} demonstrates that his method for boosting the
output size of extractors using output-length optimal seeded
extractors (that is the basis of our technique for making seedless
extractors invertible) can be extended to the case of two-source
extractors as well.

On the other hand, if the route-disjointness condition that is assumed
in the above example is not available, zero leakage can no longer be
guaranteed without imposing further restrictions (since, as discussed
in Section~\ref{subsec:arbitrary}, this would require seedless
extractors for general sources, which do not exist).  However, assume
that the intermediate manipulations are carried out by computationally
bounded devices (a reasonable assumption to model the real world).  A
natural candidate for modeling resource-bounded computation is the
notion of small-sized Boolean circuits.  The secrecy problem for this
class of transformations leads to invertible extractors for the
following class of sources:

\begin{quote}
  For an arbitrary Boolean function $C\colon \zo^n \to \zo$ that is
  computable by Boolean circuits of bounded size, the source is
  uniformly distributed on the set of inputs $x \in \zo^n$ such that
  $C(x) = 0$ (assuming that this set has a sufficiently large size).
\end{quote}

In a recent work of Shaltiel \cite{ref:Sha09}, this type of extractors
have been studied under the notion of ``extractors for
\emph{recognizable} sources'' (a notion that can be specialized to
different sub-classes depending on the bounded model of computation
being considered.

On the other hand, Trevisan and Vadhan \cite{TV00} introduce the
related notion of extractors for \emph{samplable sources}, where a
samplable source is defined as the image of a small-sized circuit
(having multiple outputs) when provided with a uniformly random
input. They proceed to show explicit constructions of such extractors
assuming suitable computational hardness assumptions (which turn out
to be to some extent necessary for such extractors to be
constructible).  It is straightforward to see that their techniques
can be readily extended to construction of explicit extractors for
sources recognizable by small-sized circuits (using even weaker
hardness assumptions). However, the technique works when the source
entropy is assured to be substantially large, and even so, is unable
to produce a nearly optimal output length.  To this date, explicit
construction of better extractors, under mild computational
assumptions, for sources that are samplable (or recognizable) by
small-sized circuits remains an important open problem.

Observe that the technique of using extractors for construction of
wiretap protocols as presented in Chapter~\ref{chap:wiretap} achieves
optimal rates only if the wiretap channel (i.e., the channel that
delivers intruder's information) is of \emph{erasure} nature.  That
is, we have so far assumed that, after some possible post-processing
of the encoded information, the intruder observes an arbitrarily
chosen, but bounded, subset of the bits being transmitted and remains
unaware of the rest. There are different natural choices of the
wiretap channel that can be considered as well. For example, suppose
that the intruder observes a \emph{noisy} version of the entire
sequence being transmitted (e.g., when a fraction of the encoded bits
get randomly flipped before being delivered to the intruder).  An
interesting question is to see whether invertible extractors (or a
suitable related notion) can be used to construct
infor\-mation-theoretically optimal schemes for such variations as
well.

\section*{Group Testing}

Non-adaptive group testing schemes are fundamental combinatorial
objects of both theoretical and practical interest. As we showed in
Chapter~\ref{chap:testing}, strong condensers can be used as building
blocks in construction of noise-resilient group testing and threshold
group testing schemes.

The factors that greatly influence the quality of our constructions
are the seed length and output length of the condenser being used. As
we saw, in order to obtain an asymptotically optimal number of
measurements, we need explicit constructions of extractors and
lossless condensers that achieve a logarithmic seed length, and output
length that is different from the source entropy by small additive
terms. While, as we saw, there are very good existing constructions of
both extractors and lossless condensers that can be used, they are
still sub-optimal in the above sense.  Thus, any improvement on the
state of the art in explicit construction of extractors and lossless
condensers will immediately improve the qualities of our explicit
constructions.

Moreover, our constructions of noise-resilient schemes with sublinear
decoding time demonstrates a novel application for list-decodable
extractors and condensers. This motivates further investigation of
these objects for improvement of their qualities.

In Section~\ref{sec:threshold}, we introduced the combinatorial notion
of $(d,e;u)$-regular matrices, that is used as an intermediate tool
towards obtaining threshold testing designs. Even though our
construction, assuming an optimal lossless condenser, matches the
probabilistic upper bound for regular matrices, the number of
measurements in the resulting threshold testing scheme will be larger
than the probabilistic upper bound by a factor of $\Omega(d \log n)$.
Thus, an outstanding question is coming up with a direct construction
of disjunct matrices that match the probabilistic upper bound.

Despite this, the notion of regular matrices may be of independent
interest, and an interesting question is to obtain (nontrivial)
concrete lower bounds on the number of rows of such matrices in terms
of the parameters $d, e, u$.

Moreover, in our constructions we have assumed the threshold $u$ to be
a fixed constant, allowing the constants hidden in asymptotic notions
to have a poor dependence on $u$.  An outstanding question is whether
the number of measurements can be reasonably controlled when $u$
becomes large; e.g., $u = \Omega(d)$.

Another interesting problem is decoding in the threshold model.  While
our constructions can combinatorially guarantee identification of
sparse vectors, for applications it is important to have an efficient
reconstruction algorithm as well. Contrary to the case of strongly
disjunct matrices that allow a straightforward decoding procedure
(cf. \cite{ref:thresh2}), it is not clear whether in general our
notion of disjunct matrices allow efficient decoding, and thus it
becomes important to look for constructions that are equipped with
efficient reconstruction algorithms.

Finally, for clarity of the exposition, in this presentation we have
only focused on asymptotic trade-offs, and it would be nice to obtain
good, non-asymptotic, estimates on the obtained bounds that are useful
for applications.

\section*{Capacity Achieving Codes}

The general construction of capacity-achieving codes presented in
Chapter~\ref{chap:capacity} can be used to obtain a polynomial-sized
ensemble of codes of any given block length $n$, provided that nearly
optimal linear extractors or lossless condensers are available. In
particular, this would require a logarithmic seed length and an output
length which is different from the input entropy by an arbitrarily
small constant fraction of the entropy.  Both extractors and lossless
condensers constructed by Guruswami, Umans, and Vadhan
\cite{ref:GUV09} achieve this goal, and as we saw in
Chapter~\ref{chap:extractor}, their lossless condenser can be easily
made linear.  However, to the best of our knowledge, to this date no
explicit construction of a linear extractor with logarithmic seed
length that extracts even a constant fraction of the source entropy is
known.


Another interesting problem concerns the duality principle presented
in Section~\ref{sec:dualityAffine}.  As we showed, linear affine
extractors and lossless condensers are dual objects. It would be
interesting to see whether a more general duality principle exist
between extractors and lossless condensers. It is not hard to use
basic Fourier analysis to slightly generalize our result to linear
extractors and lossless condensers for more general (not necessarily
affine) sources. However, since condensers for general sources are
allowed to have a positive, but negligible error (which is \emph{not}
the case for linear affine condensers), controlling the error to a
reasonable level becomes a tricky task, and forms an interesting
problem for future research.

\section*{The Gilbert-Varshamov Bound}

As we saw in Chapter~\ref{chap:gv}, a suitable computational
assumption implies a deterministic polynomial-time algorithm for
explicit construction of polynomially many linear codes of a given
length $n$, almost all of which attaining the Gilbert-Varshamov
bound. That is, a randomly chosen code from such a short list
essentially behaves like a fully random code and in particular, is
expected to attain the same rate-distance tradeoff.

An important question that remains unanswered is whether a single code
of length $n$ attaining the Gilbert-Varshamov bound can be efficiently
constructed from a list of $\poly(n)$ codes in which an overwhelming
fraction attain the bound. In effect, we are looking for an efficient
\emph{code product} to combine a polynomially long list of codes (that
may contain a few \emph{unsatisfactory} codes) into a single code that
possesses the qualities of the overwhelming majority of the codes in
the ensemble. Since the computational problem of determining (or even
approximating) the minimum distance of a linear code is known to be
intractable, such a product cannot be constructed by simply examining
the individual codes. It is also interesting to consider impossibility
results, that is, models under which such a code product may become as
difficult to construct as finding a good code ``from scratch''.

Finally, a challenging problem which still remains open is explicit
construction of codes (or even small ensembles of codes) that attain
the Gilbert-Varsha\-mov bound without relying on unproven
assumptions. For sufficiently large alphabets (i.e., of size $49$ or
higher), geometric Goppa codes are known to even surpass the GV bound
\cite{tsvz:82}. However, for smaller alphabets, or rates close to zero
over constant-sized alphabets, no explicit construction attaining the
GV bound is known. It also remains unclear whether in such cases the
GV bound is optimal; that is, whether there are families of codes, not
necessarily explicit, that beat the bound.

\musicBoxConclusion


\appendix

\Chapter{A Primer on Coding Theory}
\epigraphhead[70]{\epigraph{\textsl{``A classic is a book that has never finished
saying what it has to say.''}}{\textit{--- Italo Calvino}}}
\label{app:coding}

In this appendix, we briefly overview the essential notions of coding
theory that we have used in the thesis. For an extensive treatment of
the theory of error-correcting codes (an in particular, the facts
collected in this appendix), we refer the reader to the books by
MacWilliams and Sloane \cite{ref:MS}, van~Lint \cite{ref:vanl}, and
Roth \cite{ref:Roth} on the topic.

\section{Basics}

Let $\Sigma$ be a finite alphabet of size $q > 1$. A code\index{code}
$\C$ of length\index{code!length} $n$ over $\Sigma$ is a non-empty
subset of $\Sigma^n$. Each element of $\C$ is called a
\emph{codeword}\index{codeword} and $|\C|$ defines the
\emph{size}\index{code!size} of the code. The
\emph{rate}\index{code!rate} of the $\C$ is defined as $\log_q |\C| /
n$. An important choice for the alphabet is $\Sigma = \zo$, which
results in a \emph{binary code}\index{code!binary}. Typically, we
assume that $q$ is a prime power and take $\Sigma$ to be the finite
field $\F_q$.

The Hamming distance\index{Hamming distance} between vectors $w :=
(w_1, \ldots, w_n) \in \Sigma^n$ and $w' := (w'_1, \ldots, w'_n) \in
\Sigma^n$ is defined as the number of positions at which $w$ and $w'$
differ. Namely,
\[
\dist(w, w') := |\{ i\in [n]\colon w_i \neq w'_i \}|.
\]
The \emph{Hamming weight}\index{Hamming weight} of a vector $w \in
\F_q^n$ (denoted by $\wgt(w)$) is the number of its nonzero
coordinates; i.e.,
\[
\wgt(w) := |\{ i\in [n]\colon w_i \neq 0 \}|.
\]
Therefore, when $w, w' \in \F_q^n$, we have \[\dist(w, w') = \wgt(w -
w').\]

\noindent The \emph{minimum distance}\index{code!minimum distance} of
a code $\C \subseteq \Sigma^n$ is the quantity
\[
\dist(\C) := \min_{w, w' \in C} \dist(w, w'),
\]
and the \emph{relative distance} of the code is defined as
$\dist(\C)/n$. A family of codes of growing block lengths $n$ is
called \emph{asymptotically good}\index{code!asymptotically good} if,
for large enough $n$, it achieves a positive constant rate (i.e.,
independent of $n$) and a positive constant relative distance.

A code $\C \in \F_q^n$ is called \emph{linear}\index{code!linear} if
it is a vector subspace of $\F_q^n$.  In this case, the
\emph{dimension}\index{code!linear!dimension} of the code is defined
as its dimension as a subspace, and the rate would be given by the
dimension divided by $n$. A code $\C$ with minimum distance $d$ is
denoted by the shorthand $(n, \log_q |\C|, d)_q$, and when $\C$ is
linear with dimension $k$, by $[n, k, d]_q$.  The subscript $q$ is
omitted for binary codes. Any linear code must include the all-zeros
word $0^n$.  Moreover, due to the linear structure of such codes, the
minimum distance of a linear code is equal to the minimum Hamming
weight of its nonzero codewords.

A \emph{generator matrix}\index{code!linear!generator matrix} $G$ for
a linear $[n,k,d]_q$-code $\C$ is a $k \times n$ matrix of rank $k$
over $\F_q$ such that
\[
\C = \{ xG\colon x \in \F_q^k \}.
\]
Moreover, a \emph{parity check matrix}\index{code!linear!parity check
  matrix} $H$ for $\C$ is an $r \times n$ matrix over $\F_q$ of rank
$n-k$, for some $r \geq n-k$, such that\footnote{Here we consider
  vectors as row vectors, and denote column vectors (e.g., $x^\top$)
  as transpose of row vectors.}
\[
\C = \{ x \in \F_q^n\colon H x^\top = 0 \}.
\]
Any two such matrices are orthogonal to one another, in that we must
have $G H^\top = 0$.  It is easy to verify that, if $\C$ has minimum
distance $d$, then every choice of up to $d-1$ columns of $H$ are
linearly independent, and there is a set of $d$ columns of $H$ that
are dependent (and the dependency is given by a codeword of minimum
weight).

The \emph{dual}\index{code!dual} of a linear code $\C$ of length $n$
over $\F_q$ (denoted by $\C^\top$) is defined as the dual vector space
of the code; i.e., the set of vectors in $\F_q^n$ that are all
orthogonal to every codeword in $\C$:
\[
\C^\perp := \{ c \in \F_q^n\colon (\forall w \in \C)\ c \cdot w^\top =
0 \}.
\]
The dual of a $k$-dimensional code has dimension $n-k$, and
$(\C^\perp)^\perp = \C$. Moreover, a generator matrix for the code
$\C$ is a parity check matrix for $\C^\perp$ and vice versa.

An \emph{encoder} for a code $\C$ with $q^k$ codewords is a function
$E\colon \Sigma^k \to \Sigma^n$ whose image is the code $\C$. In
particular, this means that $E$ must be injective (one-to-one).
Moreover, any generator matrix for a linear code defines the encoder
$E(x) := xG$. The input $x$ is referred to as the \emph{message}.  We
will consider a code \emph{explicit}\index{code!explicit} if it is
equipped with a polynomial-time computable encoder.  For linear codes,
this is equivalent to saying that there is a deterministic polynomial
time algorithm (in the length $n$) that outputs a generator, or parity
check, matrix for the code\footnote{There are more strict
  possibilities for considering a code explicit; e.g., one may require
  each entry of a generator matrix to be computable in logarithmic
  space.}.

Given a message $x \in \Sigma^k$, assume that an encoding of $x$ is
obtained using an encoder $E$; i.e., $y := E(x) \in
\Sigma^n$. Consider a \emph{communication channel} through which the
encoded sequence $y$ is communicated. The output of the channel
$\tilde{y} \in \Sigma^n$ is delivered to a receiver, whose goal is to
reconstruct $x$ from $\tilde{y}$. Ideally, if the channel is perfect,
we will $\tilde{y} = y$ and, since $E(x)$ is injective, deducing $x$
amounts to inverting the function $E$, which is an easy task for
linear codes (in general, this can be done using Gaussian
elimination).  However, consider a \emph{closest distance} decoder
$D\colon \Sigma^n \to \Sigma^k$ that, given $\tilde{y}$, outputs an $x
\in \Sigma^k$ for which $\dist(E(x), \tilde{y})$ is minimized. It is
easy to see that, even if we allow the channel to arbitrarily alter up
to $t := \floor{(d-1)/2}$ of the symbols in the transmitted sequence
$y$ (in symbols, if $\dist(y, \tilde{y}) \leq t$), then we can still
ensure that $x$ is uniquely deducible from $\tilde{y}$; in particular,
we must have $D(\tilde{y}) = x$.

For a linear code over $\F_q$ with parity check matrix $H$, a
\emph{syndrome} corresponding to a sequence $\tilde{y} \in \F_q^n$ is
the vector $H {\tilde{y}}^\top$. Thus, $\tilde{y}$ is a codeword if
and only if its corresponding syndrome is the zero vector. Therefore,
in the channel model above, if the syndrome corresponding to the
received word $\tilde{y}$ is nonzero, we can be certain that
$\tilde{y} \neq y$.  The converse is not necessarily true. However, it
is a simple exercise to see that if $\tilde{y}$ and $\tilde{y}' \in
\F_q^n$ are both such that $\tilde{y} \neq \tilde{y}'$ and moreover
$\dist(y, \tilde{y}) \leq t$ and $\dist(y, \tilde{y}) \leq t$, then
the corresponding syndromes must be different; i.e., $H
{\tilde{y}}^\top \neq H \mbox{$\tilde{y}'$}^\top$.  Therefore,
provided that the number of errors is no more than the
``unique-decoding threshold'' $t$, it is ``combinatorially'' possible
to uniquely reconstruct $x$ from the syndrome corresponding to the
received word.  This task is known as \emph{syndrome
  decoding}\index{syndrome decoding}.  However, ideally it is
desirable to have an efficient algorithm for syndrome decoding as well
that runs in polynomial time in the length of the code. In general,
syndrome decoding for a linear code defined by its parity check matrix
is $\mathsf{NP}$-hard (see \cite{ref:BMT78}). However, a variety of
explicit code constructions are equipped with efficient syndrome
decoding algorithms.

As discussed above, a code with minimum distance $d$ can tolerate up
to $t := \floor{(d-1)/2}$ errors.  Moreover, if the number of errors
can potentially be larger than $t$, then a confusion becomes
unavoidable and unique decoding can no longer be guaranteed. However,
the notion \emph{list decoding}\index{code!list decodable} allows to
control the ``amount of confusion'' when the number of errors is more
than $t$. Namely, for a \emph{radius} $\rho$ and integer $\ell$
(referred to as the \emph{list size}), a code $\C \subseteq [q]^n$ is
called $(\rho, \ell)$ list-decodable if the number of codewords within
a distance $\rho n$ of any vector in $[q]^n$ is at most $\ell$.  In
this view, unique decoding corresponds to the case $\ell=1$, and a
code with minimum distance $d$ is $(\frac{1}{n}\floor{(d-1)/2}, 1)$
list-decodable. However, for many theoretical and practical purposes,
a small (but possibly much larger than~$1$) list size may be
sufficient.

\section{Bounds on codes}

For positive integers $n, d, q$, denote by $A_q(n, d)$ the maximum
size of a code with length $n$ and minimum distance $d$ over a $q$-ary
alphabet, and define
\[
\alpha_q(\delta) := \lim_{n \to \infty} \frac{\log_q A(n, \delta
  n)}{n}
\]
as the ``highest'' rate a code with relative distance $\delta$ can
asymptotically attain.  The exact form of the function
$\alpha_q(\cdot)$ is not known for any $q$; however, certain lower and
upper bounds for this quantity exist. In this section, we briefly
review some important bounds on $\alpha_q(\delta)$.

\subsection*{The Gilbert-Varshamov bound} \index{Gilbert-Varshamov
  bound}

Using the probabilistic method, it can be shown that a random linear
code (constructed by picking the entries of its generator, or parity
check, matrix uniformly and independently at random) with overwhelming
probability attains a dimension-distance tradeoff given by
\[
k \geq n (1 - h_q(d/n)),
\]
where $h_q(\cdot)$ is the $q$-ary entropy function defined as
\begin{equation} \label{eqn:HQ} \index{entropy function} h_q(x) := x
  \log_q(q-1) - x \log_q(x) - (1-x) \log_q(1-x).
\end{equation}
Thus we get the lower bound \[ \alpha_q(\delta) \geq 1-h_q(\delta) \]
on the function $\alpha_q(\cdot)$, known as the Gilbert-Varshamov
bound.

\subsection*{The Singleton bound} \index{Singleton bound}

On the negative side, the Singleton bound states that the minimum
distance $d$ of any $q$-ary code with $q^k$ or more codewords must
satisfy $d \leq n-k+1$. Codes that attain this bound with equality are
known as \emph{maximum distance separable (MDS)} codes.
\index{code!MDS} Therefore we get that, regardless of the alphabet
size, one must have
\[
\alpha_q(\delta) \leq 1 - \delta.
\]

\subsection*{Lower bounds for fixed alphabet size}

When the alphabet size $q$ is fixed, there are numerous lower bounds
known for the function $\alpha_q(\cdot)$. Here we list several such
bounds.

\begin{itemize}
\item {\em Hamming (sphere packing) bound: } $\alpha_q(\delta) \leq 1
  - h_q(\delta/2).$

\item {\em Plotkin bound: } $ \alpha_q(\delta) \leq \max\{ 0,
  1-\delta(q/(q-1)) \}.  $

\item {\em McEliece, Rodemich, Ramsey, and Welch (MRRW) bound: }
  \[
  \alpha_2(\delta) \leq h_2
  \big(\frac{1}{2}-\sqrt{\delta(1-\delta)}\big).
  \]
\end{itemize}
For the binary alphabet, these bounds are depicted in
Figure~\ref{fig:codebounds}.

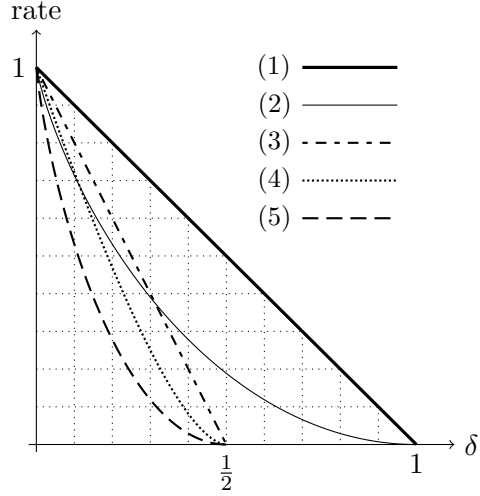
\begin{figure}[t]
  \centering
  \begin{tikzpicture}[xscale=5, yscale=5]
    \foreach \x in {1, ..., 9} { \draw[thin, dotted] (0.1*\x,
      1-0.1*\x) -- (0.1*\x, 0); \draw[thin, dotted] (0, 0.1*\x) -- (1
      - 0.1*\x, 0.1*\x); } \draw[->] (-0.02,0) -- (1.1,0) node[right]
    {$\delta$}; \draw[->] (0,-0.02) -- (0,1.1) node[above]
    {$\mathrm{rate}$}; \draw(0, 1) node[left, scale=1]{${1}$};
    \draw(1, 0) node[below, scale=1]{$1$};
    \draw[very thick] (0,1) -- (1,0); 
    \draw[thick, dash pattern=on 2pt off 3pt on 4pt off 3pt] (0,1) --
    (0.5,0) node[below]{$\frac{1}{2}$}; 
    \draw[thick, densely dotted] plot[raw gnuplot, id=MRRW] function{
      h(x) = (- x * log(x) - (1-x) * log(1-x) ) / log(2); mrrw(x) = (x
      == 0) ? 1 : ( (x==0.5) ? 0: h(0.5-sqrt(x*(1-x)))); plot
      [x=0:0.5] mrrw(x); }; \draw[solid] plot[raw gnuplot, id=HAMM]
    function{ h(x) = (- x * log(x) - (1-x) * log(1-x) ) / log(2);
      hamm(x) = (x == 0) ? 1 : ( (x==1) ? 0: 1 - h(x/2)); plot [x=0:1]
      hamm(x); }; \draw[thick, dash pattern=on 6pt off 3pt] plot[raw
    gnuplot, id=GV] function{ h(x) = (- x * log(x) - (1-x) * log(1-x)
      ) / log(2); gv(x) = (x == 0) ? 1 : ( (x==0.5) ? 0: 1 - h(x));
      plot [x=0:0.5] gv(x); };

    \draw[very thick] (0.7, 1) node[left]{{\small $(1)$}} -- (0.95,
    1); \draw (0.7, 0.9) node[left]{{\small $(2)$}} -- (0.95, 0.9);
    \draw[thick, dash pattern=on 2pt off 3pt on 4pt off 3pt](0.7, 0.8)
    node[left]{{\small $(3)$}} -- (0.95, 0.8); \draw[thick, densely
    dotted](0.7, 0.7) node[left]{{\small $(4)$}} -- (0.95, 0.7);
    \draw[thick, dash pattern=on 6pt off 3pt](0.7, 0.6)
    node[left]{{\small $(5)$}} -- (0.95, 0.6);
  \end{tikzpicture} \hspace{1cm}
  \caption[Bounds on binary codes]{ Bounds on binary codes: (1)
    Singleton bound, (2) Hamming bound, (3) Plotkin bound, (4) MRRW
    bound, (5) Gilbert-Varshamov bound.}
  \label{fig:codebounds}
\end{figure}

\subsection*{The Johnson Bound on List Decoding}

Intuitively, it is natural to expect that a code with large minimum
distance must remain a good list-decodable code when the list-decoding
radius exceeds half the minimum distance. The Johnson bound makes this
intuition rigorous.  Below we quote a strengthened version of the
bound.

\begin{thm} (cf.~\cite{ref:Venkat}*{Section~3.3}) \label{thm:johnson}
  Let $\C$ be a $q$-ary code of length $n$, and relative distance
  $\delta \geq (1-1/q)(1-\delta')$ for some $\delta' \in (0,1)$. Then
  for any $\gamma > \sqrt{\delta'}$, $\C$ is $((1-1/q)(1-\gamma),
  \ell)$ list-decodable for
  \[
  \ell = \min\{ n(q-1), \frac{1-\delta'}{\gamma^2 - \delta'}\}.
  \]
  Moreover, the code $\C$ is $((1-1/q)(1-\sqrt{\delta'}), 2n(q-1)-1)$
  list-decodable. \qed
\end{thm}

As an immediate corollary, we get that any binary code with relative
distance at least $\frac{1}{2}-\eps$ is $(\frac{1}{2}-\sqrt{\eps},
\frac{1}{2\eps})$ list-decodable.

\section{Reed-Solomon codes} \index{code!Reed-Solomon}

Let $p = (p_1, \ldots, p_n)$ be a vector consisting of $n$ distinct
elements of $\F_q$ (assuming $q \geq n$).  The \emph{evaluation
  vector} of a polynomial $f\colon \F_q \to \F_q$ with respect to $p$
is the vector $f(p) := (f(p_1), \ldots, f(p_n)) \in \F_q^n$.

A \emph{Reed-Solomon} code of length $n$ and dimension $k$ over $\F_q$
is the set of evaluation vectors of all polynomials of degree at most
$k-1$ over $\F_q$ with respect to a particular choice of $p$. The
dimension of this code is equal to $k$. A direct corollary of
Euclidean division algorithm states that, over any field, the number
of zeros of any nonzero polynomial is less than or equal to its
degree. Thus, we get that the minimum distance of a Reed-Solomon code
is at least $n-k+1$, and because of the Singleton bound, is in fact
equal to $n-k+1$.  Hence we see that a Reed-Solomon code is MDS. A
generator matrix for a Reed-Solomon code is given by the Vandermonde
matrix
\[
G := \begin{pmatrix}
  1&1&\ldots&1\\
  p_1&p_2&\ldots& p_n\\
  p_1^2&p_2^2&\ldots& p_n^2\\
  \vdots&\vdots&\ddots& \vdots \\
  p_1^{k-1}&p_2^{k-1}&\ldots& p_n^{k-1}
\end{pmatrix}
.
\]

\section{The Hadamard Code} \index{code!Hadamard}

The Hadamard code of dimension $n$ is a linear binary code of length
$2^n$ whose generator matrix can be obtained by arranging all binary
sequences of length $n$ as its columns. Each codeword of the Hadamard
code can thus be seen as the truth table of a linear form
\[
\ell(x_1, \ldots, x_n) = \sum_{i=1}^n \alpha_i x_i
\]
over the binary field. Therefore, each nonzero codeword must have
weight exactly $2^{n-1}$, implying that the relative distance of the
Hadamard code is $\frac{1}{2}$.

\section{Concatenated Codes}
\index{code!concatenation} \label{sec:concat}

Concatenation is a classical operation on codes that is mainly used
for reducing the alphabet size of a code.  Suppose that $\C_1$ (called
the \emph{outer code}) is an $(n_1, k_1, d_1)_Q$-code and $\C_2$
(called the \emph{inner code}) is a $(n_2, k_2, d_2)_{q}$-code, where
$Q = q^{k_2}$.  The concatenation of $\C_1$ and $\C_2$, that we
denote by $\C_1 \diamond \C_2$ is an $(n,k,d)_q$-code that can be
conveniently defined by its encoder mapping as follows.

Let $x = (x_1, \ldots, x_{k_1}) \in [Q]^{k_1}$ be the message given to
the encoder, and $C(x) = (c_1, \ldots, c_{n_1}) \in \C_1$ be its
encoding under $\C_1$.  Each $c_i$ is thus an element of $[q^{k_2}]$
and can thus be seen as a $q$-ary string of length $k_2$.  Denote by
$c'_i \in [q]^{n_2}$ be the encoding of this string under $\C_2$. Then
the encoding of $x$ by the concatenated code $\C_1 \diamond \C_2$ is
the $q$-ary string of length $n_1 n_2$
\[
(c''_1, \ldots, c''_{n_1})
\]
consisting of the string concatenation of symbol-wise encodings of
$C(x)$ using $\C_2$.

Immediately from the above definition, one can see that $n=n_1 n_2$,
and $k = k_1 k_2$.  Moreover, it is straightforward to observe that
the minimum distance of the concatenated code satisfies $d \geq d_1
d_2$. When $\C_1$ and $\C_2$ are linear codes, the so is $\C_1
\diamond \C_2$.

As an example, let $\C_1$ be a Reed-Solomon code of length $n_1 :=
2^{k_2}$ and dimension $k_1 := 2\delta n_1$ over $\F_Q$, where $Q :=
2^{k_2}$. Thus the relative distance of $\C_1$ equals $1-2\delta$.  As
the inner code $\C_2$, take the Hadamard code of dimension $k_2$ and
length $Q$.  The concatenated code $\C := \C_1 \diamond \C_2$ will
thus have length $n := Qn_1=2^{2k_2}$, dimension $k := \delta k_2
2^{k_2+1}$, and relative distance at least $\frac{1}{2}-\delta$.
Therefore, we obtain a binary $[n,k,d]$ code where $d \geq
(\frac{1}{2}-\delta)n$, and $n \leq (k/\delta)^2$. By the Johnson
bound (Theorem~\ref{thm:johnson}), this code must be
$(\frac{1}{2}-\delta, \ell)$ list-decodable with list size at most
$1/(2\delta)$.

We remark that binary codes with relative minimum distance
$\frac{1}{2}-\delta$ and rate $\Omega(\delta^3 \log(1/\delta))$ (which
only depends on the parameter $\delta$) can be obtained by
concatenating Geometric Goppa codes on the Tsfasman-Vl{\u a}du{\c
  t}-Zink bound (see Section~\ref{sec:KSext}) with the Hadamard
code. The Gilbert-Varshamov bound implies that binary codes with
relative distance $\frac{1}{2}-\delta$ and rate $\Omega(\delta^2)$
exist, and on the other hand, we know by the MRRW bound that
$O(\delta^2 \log(1/\delta))$ is the best rate one can hope for.

\musicBoxAppendixCode


%

\newpage
\bibliographystyle{plain}
\bibliography{bibliography}

\begin{bibdiv}
\begin{biblist}

\bib{ref:ADHP06}{inproceedings}{
      author={Adler, M.},
      author={Demaine, E.D.},
      author={Harvey, N.J.A.},
      author={P{\v a}tra{\c s}cu, M.},
       title={Lower bounds for asymmetric communication channels and
  distributed source coding},
        date={2006},
   booktitle={Proceedings of the $17$th symposium on discrete algorithms
  {(SODA)}},
       pages={251\ndash 260},
}

\bib{ref:AM01}{article}{
      author={Adler, M.},
      author={Maggs, B.},
       title={Protocols for asymmetric communication channels},
        date={2001},
     journal={Journal of Computer and System Sciences},
      volume={63},
      number={4},
       pages={573\ndash –596},
}

\bib{ref:AKS04}{article}{
      author={Agrawal, M.},
      author={Kayal, N.},
      author={Saxena, N.},
       title={{PRIMES} is in {P}},
        date={2004},
     journal={Annals of Mathematics},
      volume={160},
      number={2},
       pages={781\ndash 793},
}

\bib{ref:NetCod1}{article}{
      author={Ahlswede, R.},
      author={Cai, N.},
      author={Li, {S-Y.~R.}},
      author={Yeung, R.~W.},
       title={Network information flow},
        date={2000},
     journal={IEEE Transactions on Information Theory},
      volume={46},
      number={4},
       pages={1204\ndash 1216},
}

\bib{probMethod}{book}{
      author={Alon, N.},
      author={Spencer, J.H.},
       title={The probabilistic method},
   publisher={John Wiley and Sons},
        date={2000},
}

\bib{ref:ACS09}{inproceedings}{
      author={Ardestanizadeh, E.},
      author={Cheraghchi, M.},
      author={Shokrollahi, A.},
       title={Bit precision analysis for compressed sensing},
        date={2009},
   booktitle={Proceedings of {IEEE} international symposium on information
  theory ({ISIT})},
}

\bib{ref:Arikan09}{article}{
      author={Ar{\i}kan, E.},
       title={Channel polarization: A method for constructing
  capacity-achieving codes for symmetric binary-input memoryless channels},
        date={2009},
     journal={IEEE Transactions on Information Theory},
      volume={55},
      number={7},
       pages={3051\ndash 3073},
}

\bib{ref:AB09}{book}{
      author={Arora, S.},
      author={Barak, B.},
       title={Computational complexity: A modern approach},
   publisher={Cambridge University Press},
        date={2009},
}

\bib{BRSW06}{inproceedings}{
      author={Barak, B.},
      author={Rao, A.},
      author={Shaltiel, R.},
      author={Wigderson, A.},
       title={2-source dispersers for sub-polynomial entropy and {Ramsey}
  graphs beating the {Frankl-Wilson} construction},
        date={2006},
   booktitle={Proceedings of the $38$th annual {ACM} symposium on theory of
  computing ({STOC})},
       pages={671\ndash 680},
}

\bib{ref:BT09}{inproceedings}{
      author={{Ben-Aroya}, A.},
      author={{Ta-Shma}, A.},
       title={Constructing small-bias sets from algebraic-geometric codes},
        date={2009},
   booktitle={Proceedings of the $50$th annual {IEEE} symposium on foundations
  of computer science ({FOCS})},
}

\bib{ref:BBR85}{inproceedings}{
      author={Bennett, C.H.},
      author={Brassard, G.},
      author={Robert, J-M.},
       title={How to reduce your enemy's information},
        date={1985},
   booktitle={Proceedings of the $5$th annual international cryptology
  conference ({CRYPTO})},
      series={Lecture Notes in Computer Science},
      volume={218},
       pages={468\ndash 476},
}

\bib{ref:BMT78}{article}{
      author={Berlekamp, E.R.},
      author={McEliece, R.},
      author={{Tilborg}, H.C.A~van},
       title={On the inherent intractability of certain coding problems},
        date={1978},
     journal={IEEE Transactions on Information Theory},
      volume={24},
       pages={384\ndash 386},
}

\bib{ref:blahut}{book}{
      author={Blahut, R.E.},
       title={Theory and practice of error control codes},
   publisher={Addison-Wesley},
        date={1983},
}

\bib{ref:BG02}{article}{
      author={Blass, A.},
      author={Gurevich, Y.},
       title={Pairwise testing},
        date={2002},
     journal={Bulletin of the {EATCS}},
      volume={78},
       pages={100\ndash 132},
}

\bib{BM}{article}{
      author={Blum, M.},
      author={Micali, S.},
       title={How to generate cryptographically strong sequences of
  pseudorandom bits},
        date={1984},
     journal={SIAM Journal on Computing},
      volume={13},
      number={4},
       pages={850\ndash 864},
}

\bib{ref:Bourgain}{article}{
      author={Bourgain, J.},
       title={On the construction of affine extractors},
        date={2007},
     journal={Geometric and Functional Analysis},
      volume={17},
      number={1},
       pages={33\ndash 57},
}

\bib{ref:BourPriv}{misc}{
      author={Bourgain, J.},
       title={\textrm{Personal Communication}},
        date={2008},
}

\bib{ref:BKBB95}{article}{
      author={Bruno, W.J.},
      author={Knill, E.},
      author={Balding, D.J.},
      author={Bruce, D.C.},
      author={Doggett, N.A.},
      author={Sawhill, W.W.},
      author={Stallings, R.L.},
      author={Whittaker, C.C.},
      author={Torney, D.C.},
       title={Efficient pooling designs for library screening},
        date={1995},
     journal={Genomics},
      volume={26},
      number={1},
       pages={21\ndash 30},
}

\bib{ref:BMRV02}{article}{
      author={Buhrman, H.},
      author={Miltersen, P.B.},
      author={Radhakrishnan, J.},
      author={Venkatesh, S.},
       title={Are bitvectors optimal?},
        date={2002},
     journal={SIAM Journal on Computing},
      volume={31},
      number={6},
       pages={1723\ndash 1744},
}

\bib{ref:BKS06}{unpublished}{
      author={Buresh-Oppenheim, J.},
      author={Kabanets, V.},
      author={Santhanam, R.},
       title={Uniform hardness amplification in {$\mathsf{NP}$} via monotone
  codes},
        date={2006},
        note={\emph{ECCC} Technical Report TR06-154.},
}

\bib{ref:NetWT}{inproceedings}{
      author={Cai, N.},
      author={Yeung, R.~W.},
       title={Secure network coding},
        date={2002},
   booktitle={Proceedings of {IEEE} international symposium on information
  theory ({ISIT})},
}

\bib{ref:CDH}{inproceedings}{
      author={Canetti, R.},
      author={Dodis, Y.},
      author={Halevi, S.},
      author={Kushilevitz, E.},
      author={Sahai, A.},
       title={Exposure-resilient functions and all-or-nothing transforms},
        date={1999},
   booktitle={Proceedings of the $19$th annual international cryptology
  conference ({CRYPTO})},
      series={Lecture Notes in Computer Science},
      volume={1666},
       pages={503\ndash 518},
}

\bib{ref:CRVW02}{inproceedings}{
      author={Capalbo, M.},
      author={Reingold, O.},
      author={Vadhan, S.},
      author={Wigderson, A.},
       title={Randomness conductors and constant-degree expansion beyond the
  degree/2 barrier},
        date={2002},
   booktitle={Proceedings of the $34$th annual {ACM} symposium on theory of
  computing ({STOC})},
       pages={659\ndash 668},
}

\bib{ref:CCFS10}{article}{
      author={Chang, H.},
      author={Chen, H-B.},
      author={Fu, H-L.},
      author={Shi, C-H.},
       title={Reconstruction of hidden graphs and threshold group testing},
        date={2010},
     journal={Journal of Combinatorial Optimization},
}

\bib{ref:CGL}{article}{
      author={Charles, D.~X.},
      author={Lauter, K.~E.},
      author={Goren, E.~Z.},
       title={Cryptographic hash functions from expander graphs},
        date={2007},
     journal={Journal of Cryptology},
}

\bib{ref:CDH07}{article}{
      author={Chen, H-B.},
      author={Du, D-Z.},
      author={Hwang, F-K.},
       title={An unexpected meeting of four seemingly unrelated problems: graph
  testing, {DNA} complex screening, superimposed codes and secure key
  distribution},
        date={2007},
     journal={Journal of Combinatorial Optimization},
      volume={14},
      number={2-3},
       pages={121\ndash 129},
}

\bib{ref:thresh2}{article}{
      author={Chen, H-B.},
      author={Fu, H-L.},
       title={Nonadaptive algorithms for threshold group testing},
        date={2009},
     journal={Discrete Applied Mathematics},
      volume={157},
       pages={1581\ndash 1585},
}

\bib{ref:CFH08}{article}{
      author={Chen, H-B.},
      author={Fu, H-L.},
      author={Hwang, F-K.},
       title={An upper bound of the number of tests in pooling designs for the
  error-tolerant complex model},
        date={2008},
     journal={Optimization Letters},
      volume={2},
      number={3},
       pages={425\ndash 431},
}

\bib{ref:CD08}{article}{
      author={Cheng, Y.},
      author={Du, D-Z.},
       title={New constructions of one- and two-stage pooling designs},
        date={2008},
     journal={Journal of Computational Biology},
      volume={15},
      number={2},
       pages={195\ndash 205},
}

\bib{ref:Che09}{inproceedings}{
      author={Cheraghchi, M.},
       title={Noise-resilient group testing: Limitations and constructions},
        date={2009},
   booktitle={Proceedings of the $17$th international symposium on fundamentals
  of computation theory ({FCT})},
      series={Lecture Notes in Computer Science},
      volume={5699},
       pages={62\ndash 73},
}

\bib{ref:Che10}{inproceedings}{
      author={Cheraghchi, M.},
       title={Improved constructions for non-adaptive threshold group testing},
        date={2010},
   booktitle={Proceedings of the $37$th international colloquium on automata,
  languages and programming ({ICALP})},
        note={arXiv: cs.DM/1002.2244},
}

\bib{ref:CDS09}{inproceedings}{
      author={Cheraghchi, M.},
      author={Didier, F.},
      author={Shokrollahi, A.},
       title={Invertible extractors and wiretap protocols},
        date={2009},
   booktitle={Proceedings of {IEEE} international symposium on information
  theory ({ISIT})},
}

\bib{ref:CHKV09}{inproceedings}{
      author={Cheraghchi, M.},
      author={Hormati, A., A.~Karbasi},
      author={Vetterli, M.},
       title={Compressed sensing with probabilistic measurements: A group
  testing solution},
        date={2009},
   booktitle={Proceedings of the annual {Allerton} conference on communication,
  control, and computing},
}

\bib{ref:CKMS10}{inproceedings}{
      author={Cheraghchi, M.},
      author={Karbasi, A.},
      author={Mohajer, S.},
      author={Saligrama, V.},
       title={Graph-constrained group testing},
        date={2010},
   booktitle={Proceedings of {IEEE} international symposium on information
  theory ({ISIT})},
}

\bib{CG88}{article}{
      author={Chor, B.},
      author={Goldreich, O.},
       title={Unbiased bits from sources of weak randomness and probabilistic
  communication complexity},
        date={1988},
     journal={SIAM Journal of Computing},
      volume={17},
      number={2},
       pages={230\ndash 261},
}

\bib{ref:tResilient}{inproceedings}{
      author={Chor, B.},
      author={Goldreich, O.},
      author={H{\aa}stad, J.},
      author={Friedmann, J.},
      author={Rudich, S.},
      author={Smolensky, R.},
       title={The bit extraction problem or t-resilient functions},
        date={1985},
   booktitle={Proceedings of the $26$th annual {IEEE} symposium on foundations
  of computer science ({FOCS})},
       pages={396\ndash 407},
}

\bib{ref:CEPR07}{inproceedings}{
      author={Clifford, R.},
      author={Efremenko, K.},
      author={Porat, E.},
      author={Rothschild, A.},
       title={$k$-mismatch with don't cares},
        date={2007},
   booktitle={Proceedings of the $15$th european symposium on algorithm
  ({ESA})},
      series={Lecture Notes in Computer Science},
      volume={4698},
       pages={151\ndash 162},
}

\bib{ref:CM05}{article}{
      author={Cormode, G.},
      author={Muthukrishnan, S.},
       title={What's hot and what's not: tracking most frequent items
  dynamically},
        date={2005},
     journal={ACM Transactions on Database Systems},
      volume={30},
      number={1},
       pages={249\ndash 278},
}

\bib{ref:CM06}{inproceedings}{
      author={Cormode, G.},
      author={Muthukrishnan, S.},
       title={Combinatorial algorithms for compressed sensing},
        date={2006},
   booktitle={Proceedings of information sciences and systems},
       pages={198\ndash 201},
}

\bib{ref:cover}{book}{
      author={Cover, T.M.},
      author={Thomas, J.A.},
       title={Elements of information theory},
     edition={Second},
   publisher={John Wiley and Sons},
        date={2006},
}

\bib{ref:CK78}{article}{
      author={Csisz{\'a}r, I.},
      author={K{\"o}rner, J.},
       title={Broadcast channels with confidential messages},
        date={1978},
     journal={IEEE Transactions on Information Theory},
      volume={24},
      number={3},
       pages={339\ndash 348},
}

\bib{ref:thresh1}{inproceedings}{
      author={Damaschke, P.},
       title={Threshold group testing},
        date={2006},
   booktitle={General theory of information transfer and combinatorics},
      series={Lecture Notes in Computer Science},
      volume={4123},
       pages={707\ndash 718},
}

\bib{ref:DBGV05}{article}{
      author={De~Bonis, A.},
      author={Gasieniec, L.},
      author={Vaccaro, U.},
       title={Optimal two-stage algorithms for group testing problems},
        date={2005},
     journal={SIAM Journal on Computing},
      volume={34},
      number={5},
       pages={1253\ndash 1270},
}

\bib{ref:DG10}{inproceedings}{
      author={DeVos, M.},
      author={Gabizon, A.},
       title={Simple affine extractors using dimension expansion},
        date={2010},
   booktitle={Proceedings of the $25$th {IEEE} conference on computational
  complexity {(CCC)}},
}

\bib{ref:Dodis}{thesis}{
      author={Dodis, Y.},
       title={Exposure-resilient cryptography},
        type={Ph.D. Thesis},
        date={2000},
}

\bib{ref:Dod05}{inproceedings}{
      author={Dodis, Y.},
       title={On extractors, error-correction and hiding all partial
  information},
        date={2005},
   booktitle={Proceedings of the {IEEE} information theory workshop ({ITW})},
}

\bib{ref:DSS01}{inproceedings}{
      author={Dodis, Y.},
      author={Sahai, A.},
      author={Smith, A.},
       title={On perfect and adaptive security in exposure-resilient
  cryptography},
        date={2001},
   booktitle={Proceedings of {Eurocrypt}},
      series={Lecture Notes in Computer Science},
      volume={2045},
       pages={301\ndash 324},
}

\bib{ref:DS05}{inproceedings}{
      author={Dodis, Y.},
      author={Smith, A.},
       title={Entropic security and the encryption of high-entropy messages},
        date={2005},
   booktitle={Proceedings of the theory of cryptography conference ({TCC})},
      series={Lecture Notes in Computer Science},
      volume={3378},
       pages={556\ndash 577},
}

\bib{ref:Dor43}{article}{
      author={Dorfman, R.},
       title={The detection of defective members of large populations},
        date={1943},
     journal={Annals of Mathematical Statistics},
      volume={14},
       pages={436\ndash 440},
}

\bib{ref:groupTesting}{book}{
      author={Du, D-Z.},
      author={Hwang, F.},
       title={Combinatorial group testing and its applications},
     edition={Second},
   publisher={World Scientific},
        date={2000},
}

\bib{ref:DH06}{book}{
      author={Du, D-Z.},
      author={Hwang, F-K.},
       title={Pooling designs and nonadaptive group testing},
   publisher={World Scientific},
        date={2006},
}

\bib{madhu}{article}{
      author={Dumer, I.},
      author={Micciancio, D.},
      author={Sudan, M.},
       title={Hardness of approximating the minimum distance of a linear code},
        date={2003},
     journal={IEEE Transactions on Information Theory},
      volume={49},
      number={1},
       pages={22\ndash 37},
}

\bib{ref:DVMT02}{article}{
      author={D'yachkov, A.},
      author={Vilenkin, P.},
      author={Macula, A.},
      author={Torney, D.},
       title={Families of finite sets in which no intersection of $\ell$ sets
  is covered by the union of $s$ others},
        date={2002},
     journal={Journal of Combinatorial Theory, Series A},
      volume={99},
       pages={195\ndash 218},
}

\bib{ref:DR83}{article}{
      author={D'yachkov, A.G.},
      author={Rykov, V.V.},
       title={Bounds of the length of disjunct codes},
        date={1982},
     journal={Problems of Control and Information Theory},
      volume={11},
       pages={7\ndash 13},
}

\bib{ref:Emina}{inproceedings}{
      author={{El Rouayheb}, S.~Y.},
      author={Soljanin, E.},
       title={On wiretap networks {II}},
        date={2007},
   booktitle={Proceedings of {IEEE} international symposium on information
  theory ({ISIT})},
       pages={24\ndash 29},
}

\bib{ref:EGH07}{article}{
      author={Eppstein, D.},
      author={Goodrich, M.T.},
      author={Hirschberg, D.S.},
       title={Improved combinatorial group testing algorithms for real-world
  problem sizes},
        date={2007},
     journal={SIAM Journal on Computing},
      volume={36},
      number={5},
       pages={1360\ndash 1375},
}

\bib{Erdos}{article}{
      author={Erd{\H o}s, P.},
       title={Some remarks on the theory of graphs},
        date={1947},
     journal={Bulletin of the {American Mathematical Society}},
      volume={53},
       pages={292\ndash 294},
}

\bib{ref:KKM97}{inproceedings}{
      author={Farach, M.},
      author={Kannan, S.},
      author={Knill, E.},
      author={Muthukrishnan, S.},
       title={Group testing problems with sequences in experimental molecular
  biology},
        date={1997},
   booktitle={Proceedings of compression and complexity of sequences},
       pages={357\ndash 367},
}

\bib{ref:FMSS04}{inproceedings}{
      author={Feldman, J.},
      author={Malkin, T.},
      author={Servedio, R.},
      author={Stein, C.},
       title={On the capacity of secure network coding},
        date={2004},
   booktitle={Proceedings of the annual {Allerton} conference on communication,
  control, and computing},
}

\bib{ref:forney}{book}{
      author={Forney, G.D.},
       title={Concatenated codes},
   publisher={MIT Press},
        date={1966},
}

\bib{ref:FT00}{article}{
      author={Friedl, K.},
      author={Tsai, {S-C.}},
       title={Two results on the bit extraction problem},
        date={2000},
     journal={Discrete Applied Mathematics and Combinatorial Operations
  Research and Computer Science},
      volume={99},
}

\bib{ref:tResilient2}{inproceedings}{
      author={Friedmann, J.},
       title={On the bit extraction problem},
        date={1992},
   booktitle={Proceedings of the $33$rd annual {IEEE} symposium on foundations
  of computer science ({FOCS})},
       pages={314\ndash 319},
}

\bib{ref:Fue96}{article}{
      author={F{\"u}redi, Z.},
       title={On $r$-cover-free families},
        date={1996},
     journal={Journal of Combinatorial Theory, Series A},
      volume={73},
       pages={172\ndash 173},
}

\bib{ref:Gab85}{article}{
      author={Gabidulin, E.M.},
       title={Theory of codes with maximum rank distance (translation)},
        date={1985},
     journal={Problems of Information Transmission},
      volume={21},
      number={1},
       pages={1\ndash 12},
}

\bib{ref:affine}{inproceedings}{
      author={Gabizon, A.},
      author={Raz, R.},
       title={Deterministic extractors for affine sources over large fields},
        date={2005},
   booktitle={Proceedings of the $46$th annual {IEEE} symposium on foundations
  of computer science ({FOCS})},
       pages={407\ndash 418},
}

\bib{ref:GHTWZ06}{article}{
      author={Gao, H.},
      author={Hwang, F.K.},
      author={Thai, M.T.},
      author={Wu, W.},
      author={Znati, T.},
       title={Construction of {$d(H)$-disjunct} matrix for group testing in
  hypergraphs},
        date={2006},
     journal={Journal of Combinatorial Optimization},
      volume={12},
       pages={297\ndash 301},
}

\bib{gast:95}{article}{
      author={Garcia, A.},
      author={Stichtenoth, H.},
       title={A tower of {A}rtin-{S}chreier extensions of function fields
  attaining the {D}rinfeld-{V}l{\u a}du{\c t} bound},
        date={1995},
     journal={Invent. Math.},
      volume={121},
       pages={211\ndash 222},
}

\bib{gilbert}{article}{
      author={Gilbert, E.N.},
       title={A comparison of signaling alphabets},
        date={1952},
     journal={Bell System Technical Journal},
      volume={31},
       pages={504\ndash 522},
}

\bib{ref:GW95}{article}{
      author={Goemans, M.X.},
      author={Williamson, D.P.},
       title={Improved approximation algorithms for maximum cut and
  satisfiability problems using semidefinite programming},
        date={1995},
     journal={Journal of the ACM},
      volume={42},
       pages={1115\ndash 1145},
}

\bib{ref:Gol08}{book}{
      author={Goldreich, O.},
       title={Computational complexity: A conceptual perspective},
   publisher={Cambridge University Press},
        date={2008},
}

\bib{GZ97}{unpublished}{
      author={Goldreich, O.},
      author={Zuckerman, D.},
       title={Another proof that {$\mathsf{BPP} \subseteq \mathsf{PH}$} (and
  more)},
        date={1997},
        note={\emph{ECCC} Technical Report TR97-045 (available online at
  \url{http://eccc.hpi-web.de/eccc-reports/1997/TR97-045/index.html})},
}

\bib{gopp:81}{article}{
      author={Goppa, V.D.},
       title={Codes on algebraic curves},
        date={1981},
     journal={Soviet Mathematics Doklady},
      volume={24},
       pages={170\ndash 172},
}

\bib{GVZ06}{inproceedings}{
      author={Gradwohl, R.},
      author={Vadhan, S.},
      author={Zuckerman, D.},
       title={Random selection with an adversarial majority},
        date={2006},
   booktitle={Proceedings of the $26$th annual international cryptology
  conference ({CRYPTO})},
}

\bib{ref:Venkat}{thesis}{
      author={Guruswami, V.},
       title={List decoding of error-correcting codes},
        type={Ph.D. Thesis},
        date={2001},
}

\bib{ref:GG08}{inproceedings}{
      author={Guruswami, V.},
      author={Gopalan, P.},
       title={Hardness amplification within {$\mathsf{NP}$} against
  deterministic algorithms},
        date={2008},
   booktitle={Proceedings of the $23$rd {IEEE} conference on computational
  complexity {(CCC)}},
       pages={19\ndash 30},
}

\bib{ref:GI05}{article}{
      author={Guruswami, V.},
      author={Indyk, P.},
       title={Linear-time encodable/decodable codes with near-optimal rate},
        date={2005},
     journal={IEEE Transactions on Information Theory},
      volume={51},
      number={10},
       pages={3393\ndash 3400},
}

\bib{ref:GR08}{inproceedings}{
      author={Guruswami, V.},
      author={Rudra, A.},
       title={Concatenated codes can achieve list decoding capacity},
        date={2008},
   booktitle={Proceedings of the $19$th symposium on discrete algorithms
  {(SODA)}},
}

\bib{ref:GUV09}{article}{
      author={Guruswami, V.},
      author={Umans, C.},
      author={Vadhan, S.},
       title={Unbalanced expanders and randomness extractors from
  {Parvaresh-Vardy} codes},
        date={2009},
     journal={Journal of the ACM},
      volume={56},
      number={4},
}

\bib{HastadOptimal}{article}{
      author={H{\aa}stad, J.},
       title={Some optimal inapproximability results},
        date={2001},
     journal={Journal of the ACM},
      volume={48},
      number={4},
       pages={798\ndash 859},
        note={Preliminary version in \emph{Proceedings of STOC'97}.},
}

\bib{HILL}{article}{
      author={H{\aa}stad, J.},
      author={Impagliazzo, R.},
      author={Levin, L.},
      author={Luby, M.},
       title={A pseudorandom generator from any one-way function},
        date={1999},
     journal={SIAM Journal on Computing},
      volume={28},
      number={4},
       pages={1364\ndash 1396},
}

\bib{ref:HL00}{inproceedings}{
      author={Hong, E-S.},
      author={Ladner, R.E.},
       title={Group testing for image compression},
        date={2000},
   booktitle={Data compression conference},
       pages={3\ndash 12},
}

\bib{HLW06}{article}{
      author={Hoory, S.},
      author={Linial, N.},
      author={Wigderson, A.},
       title={Expander graphs and their applications},
        date={2006},
     journal={Bulletin of the {American Mathematical Society}},
      volume={43},
      number={4},
       pages={439\ndash 561},
}

\bib{HPV77}{article}{
      author={Hopcroft, J.},
      author={Paul, W.},
      author={Valiant, L.},
       title={On time versus space},
        date={1977},
     journal={Journal of the ACM},
      volume={24},
       pages={332\ndash 337},
}

\bib{ref:ILL89}{inproceedings}{
      author={Impagliazzo, R.},
      author={Levin, L.},
      author={Luby, M.},
       title={Pseudorandom generation from one-way functions},
        date={1989},
   booktitle={Proceedings of the $21$st annual {ACM} symposium on theory of
  computing ({STOC})},
       pages={12\ndash 24},
}

\bib{IW}{inproceedings}{
      author={Impagliazzo, R.},
      author={Wigderson, A.},
       title={{$\CP=\BPP$} unless {$\sE$} has sub-exponential circuits:
  derandomizing the {XOR} lemma},
        date={1997},
   booktitle={Proceedings of the $29$th annual {ACM} symposium on theory of
  computing ({STOC})},
       pages={220\ndash 229},
}

\bib{ref:Ind08}{inproceedings}{
      author={Indyk, P.},
       title={Explicit construction of compressed sensing of sparse signals},
        date={2008},
   booktitle={Proceedings of the $19$th symposium on discrete algorithms
  {(SODA)}},
}

\bib{ref:Justesen}{article}{
      author={Justesen, J.},
       title={A class of constructive asymptotically good algebraic codes},
        date={1972},
     journal={IEEE Transactions on Information Theory},
      volume={18},
       pages={652\ndash 656},
}

\bib{ref:KZ}{article}{
      author={Kamp, J.},
      author={Zuckerman, D.},
       title={Deterministic extractors for bit-fixing sources and
  exposure-resilient cryptography},
        date={2006},
     journal={SIAM Journal on Computing},
      volume={36},
       pages={1231\ndash 1247},
}

\bib{ref:KS64}{article}{
      author={Kautz, W.H.},
      author={Singleton, R.C.},
       title={Nonrandom binary superimposed codes},
        date={1964},
     journal={IEEE Transactions on Information Theory},
      volume={10},
       pages={363\ndash 377},
}

\bib{KKMOD04}{inproceedings}{
      author={Khot, S.},
      author={Kindler, G.},
      author={Mossel, E.},
      author={O'Donnell, R.},
       title={Optimal inapproximability results for {MAX-CUT} and other
  two-variable {CSPs}?},
        date={2004},
   booktitle={Proceedings of the $45$th annual {IEEE} symposium on foundations
  of computer science ({FOCS})},
       pages={146\ndash 154},
}

\bib{ref:KL04}{article}{
      author={Kim, H.K.},
      author={Lebedev, V.},
       title={On optimal superimposed codes},
        date={2004},
     journal={Journal of Combinatorial Designs},
      volume={12},
       pages={79\ndash 91},
}

\bib{KvM}{article}{
      author={Klivans, A.},
      author={Melkebeek, D.~van},
       title={Graph nonisomorphism has sub-exponential size proofs unless the
  polynomial-time hierarchy collapses},
        date={2002},
     journal={SIAM Journal on Computing},
      volume={31},
      number={5},
       pages={1501\ndash 1526},
}

\bib{ref:Kni95}{inproceedings}{
      author={Knill, E.},
       title={Lower bounds for identifying subset members with subset queries},
        date={1995},
   booktitle={Proceedings of the $6$th symposium on discrete algorithms
  {(SODA)}},
       pages={369\ndash 377},
}

\bib{ref:NetCod3}{article}{
      author={K\"otter, R.},
      author={M\'edard, M.},
       title={An algebraic approach to network coding},
        date={2003},
     journal={IEEE/ACM Transactions on Networking},
      volume={11},
      number={5},
       pages={782\ndash 795},
}

\bib{ref:RankMetric}{unpublished}{
      author={Kschischang, {F.~R.}},
      author={Silva, D.},
       title={Security for wiretap networks via rank-metric codes},
        date={2007},
        note={Unpublished manuscript (arXiv: cs.IT/0801.0061).},
}

\bib{ref:KJS}{article}{
      author={Kurosawa, K.},
      author={Johansson, T.},
      author={Stinson, D.},
       title={Almost $k$-wise independent sample spaces and their cryptologic
  applications},
        date={2001},
     journal={Journal of Cryptology},
      volume={14},
      number={4},
       pages={231\ndash 253},
}

\bib{ref:NetCod2}{article}{
      author={Li, {S-Y.~R.}},
      author={Yeung, R.~W.},
      author={Cai, N.},
       title={Linear network coding},
        date={2003},
     journal={IEEE Transactions on Information Theory},
      volume={49},
      number={2},
       pages={371\ndash 381},
}

\bib{ref:vanl}{book}{
      author={Lint, J.~H.~van},
       title={Introduction to coding theory},
     edition={Third},
      series={Graduate Texts in Mathematics},
   publisher={Springer Verlag},
        date={1998},
      volume={86},
}

\bib{ref:Lovasz}{article}{
      author={Lov\'asz, L.},
       title={Random walks on graphs: A survey},
        date={1996},
     journal={Combinatorics, Paul Erd\H{o}s is Eighty, Vol.\ 2 (ed. D.
  Mikl\'os, V. T. S\'os, T. Sz\H{o}nyi), J\'anos Bolyai Mathematical Society,
  Budapest},
       pages={353\ndash 398},
}

\bib{ref:ramanujan1}{article}{
      author={Lubotzky, A.},
      author={Phillips, R.},
      author={Sarnak, P.},
       title={Ramanujan graphs},
        date={1988},
     journal={Combinatorica},
      volume={8},
       pages={261\ndash 277},
}

\bib{ref:LT}{inproceedings}{
      author={Luby, M.},
       title={{LT}-codes},
        date={2002},
   booktitle={Proceedings of the $43$rd annual {IEEE} symposium on foundations
  of computer science ({FOCS})},
       pages={271\ndash 280},
}

\bib{ref:Mac99}{article}{
      author={Macula, A.J.},
       title={Probabilistic nonadaptive group testing in the presence of errors
  and {DNA} library screening},
        date={1999},
     journal={Annals of Combinatorics},
      volume={3},
      number={1},
       pages={61\ndash 69},
}

\bib{ref:MS}{book}{
      author={MacWilliams, F.J.},
      author={Sloane, N.J.},
       title={The theory of error-correcting codes},
   publisher={North Holand},
        date={1977},
}

\bib{ref:MR95}{inproceedings}{
      author={Mahajan, S.},
      author={Ramesh, H.},
       title={Derandomizing semidefinite programming based approximation
  algorithms},
        date={1995},
   booktitle={Proceedings of the $36$th annual {IEEE} symposium on foundations
  of computer science ({FOCS})},
       pages={162\ndash 169},
}

\bib{mani:81}{article}{
      author={Manin, Yu.I.},
       title={What is the maximum number of points on a curve over ${\F_2}$?},
        date={1981},
     journal={Journal of the Faculty of Science, University of Tokyo},
      volume={28},
       pages={715\ndash 720},
}

\bib{massey}{book}{
      author={Massey, J.L.},
       title={Threshold decoding},
   publisher={MIT Press},
     address={Cambridge, Massachusetts, USA},
        date={1963},
}

\bib{ref:Mil76}{article}{
      author={Miller, G.L.},
       title={{Riemann's} hypothesis and tests for primality},
        date={1976},
     journal={Journal of Computer and System Sciences},
      volume={13},
      number={3},
       pages={300\ndash 317},
}

\bib{peter}{article}{
      author={Miltersen, P.B.},
       title={Derandomizing complexity classes},
        date={2001},
     journal={Book chapter in \emph{Handbook of Randomized Computing}, Kluwer
  Academic Publishers},
}

\bib{ref:MU}{book}{
      author={Mitzenmacher, M.},
      author={Upfal, E.},
       title={Probability and computing},
   publisher={Cambridge University Press},
        date={2005},
}

\bib{ref:ramanujan2}{article}{
      author={Morgenstern, M.},
       title={Existence and explicit constructions of $q+1$ regular ramanujan
  graphs for every prime power $q$},
        date={1994},
     journal={Journal of Combinatorial Theory, Series~B},
      volume={62},
       pages={44\ndash 62},
}

\bib{MU01}{inproceedings}{
      author={Mossel, E.},
      author={Umans, C.},
       title={On the complexity of approximating the {VC} dimension},
        date={2001},
   booktitle={Proceedings of the $16$th {IEEE} conference on computational
  complexity {(CCC)}},
       pages={220\ndash 225},
}

\bib{ref:MR}{book}{
      author={Motwani, R.},
      author={Raghavan, P.},
       title={Randomized algorithms},
   publisher={Cambridge University Press},
        date={1995},
}

\bib{ref:ND00}{article}{
      author={Ngo, H-Q.},
      author={Du, D-Z.},
       title={A survey on combinatorial group testing algorithms with
  applications to {DNA} library screening},
        date={2000},
     journal={DIMACS Series on Discrete Math.\ and Theoretical Computer
  Science},
      volume={55},
       pages={171\ndash 182},
}

\bib{ref:nilli}{article}{
      author={Nilli, A.},
       title={On the second eigenvalue of a graph},
        date={1991},
     journal={Discrete Mathematics},
      volume={91},
       pages={207\ndash 210},
}

\bib{NW}{article}{
      author={Nisan, N.},
      author={Wigderson, A.},
       title={Hardness {vs.} randomness},
        date={1994},
     journal={Journal of Computer and Systems Sciences},
      volume={49},
      number={2},
       pages={149\ndash 167},
}

\bib{ref:Wyner2}{article}{
      author={Ozarow, L.~H.},
      author={Wyner, A.~D.},
       title={Wire-tap channel {II}},
        date={1984},
     journal={AT\&T Bell Laboratories Technical Journal},
      volume={63},
       pages={2135\ndash 2157},
}

\bib{papa}{book}{
      author={Papadimitriou, C.H.},
       title={Computational complexity},
   publisher={Addison-Wesley},
        date={1994},
}

\bib{ref:PV05}{inproceedings}{
      author={Parvaresh, F.},
      author={Vardy, A.},
       title={Correcting errors beyond the {Guruswami-Sudan} radius in
  polynomial time},
        date={2005},
   booktitle={Proceedings of the $46$th annual {IEEE} symposium on foundations
  of computer science ({FOCS})},
       pages={285\ndash 294},
}

\bib{ref:pizer}{article}{
      author={Pizer, A.~K.},
       title={Ramanujan graphs and {H}ecke operators},
        date={1990},
     journal={Bulletin of the {American Mathematical Society}},
      volume={23},
      number={1},
       pages={127\ndash 137},
}

\bib{ref:PR08}{inproceedings}{
      author={Porat, E.},
      author={Rothschild, A.},
       title={Explicit non-adaptive combinatorial group testing schemes},
        date={2008},
   booktitle={Proceedings of the $35$th international colloquium on automata,
  languages and programming ({ICALP})},
      series={Lecture Notes in Computer Science},
      volume={5125},
       pages={748\ndash 759},
}

\bib{ref:Rab80}{article}{
      author={Rabin, M.},
       title={Probabilistic algorithm for testing primality},
        date={1980},
     journal={Journal of Number Theory},
      volume={12},
      number={1},
       pages={128\ndash 138},
}

\bib{ref:lowerbounds}{inproceedings}{
      author={Radhakrishan, J.},
      author={{Ta-Shma}, A.},
       title={Tight bounds for depth-two superconcentrators},
        date={1997},
   booktitle={Proceedings of the $38$th annual {IEEE} symposium on foundations
  of computer science ({FOCS})},
       pages={585\ndash 594},
}

\bib{ref:RRV}{article}{
      author={Raz, R.},
      author={Reingold, O.},
      author={Vadhan, S.},
       title={Extracting all the randomness and reducing the error in
  {Trevisan's} extractor},
        date={2002},
     journal={Journal of Computer and System Sciences},
      volume={65},
      number={1},
       pages={97\ndash 128},
}

\bib{ref:RY08}{inproceedings}{
      author={Raz, R.},
      author={Yehudayoff, A.},
       title={Multilinear formulas, maximal-partition discrepancy and
  mixed-sources extractors},
        date={2008},
   booktitle={Proceedings of the $49$th annual {IEEE} symposium on foundations
  of computer science ({FOCS})},
}

\bib{ref:modern}{book}{
      author={Richardson, T.},
      author={Urbanke, R.},
       title={Modern coding theory},
   publisher={Cambridge University Press},
        date={2008},
}

\bib{ref:Riv97}{inproceedings}{
      author={Rivest, R.},
       title={All-or-nothing encryption and the package transform},
        date={1997},
   booktitle={Proceedings of the international workshop on fast software
  encryption},
      series={Lecture Notes in Computer Science},
      volume={1267},
       pages={210\ndash 218},
}

\bib{ref:Roth}{book}{
      author={Roth, R.M.},
       title={Introduction to coding theory},
   publisher={Cambridge University Press},
        date={2006},
}

\bib{RZ01}{article}{
      author={Russell, A.},
      author={Zuckerman, D.},
       title={Perfect information leader election in $\log^{\ast} n + o(1)$
  rounds},
        date={2001},
     journal={Journal of Computer and Systems Sciences},
      volume={63},
       pages={612\ndash 626},
}

\bib{ref:Rus94}{article}{
      author={Ruszink{\'o}},
       title={On the upper bound of the size of the $r$-cover-free families},
        date={1994},
     journal={Journal of Combinatorial Theory, Series~A},
      volume={66},
       pages={302\ndash 310},
}

\bib{ref:STR03}{inproceedings}{
      author={Schliep, A.},
      author={Torney, D.},
      author={Rahmann, S.},
       title={Group testing with {DNA} chips: Generating designs and decoding
  experiments},
        date={2003},
   booktitle={Proceedings of computational systems bioinformatics},
}

\bib{ref:Schwartz}{article}{
      author={Schwartz, J.~T.},
       title={Fast probabilistic algorithms for verification of polynomial
  identities},
        date={1980},
     journal={Journal of the ACM},
      volume={27},
      number={4},
       pages={701\ndash 717},
}

\bib{ref:mileage}{inproceedings}{
      author={Shaltiel, R.},
       title={How to get more mileage from randomness extractors},
        date={2006},
   booktitle={Proceedings of the $21$st annual conference on computational
  complexity},
       pages={46\ndash 60},
}

\bib{ref:Sha09}{inproceedings}{
      author={Shaltiel, R.},
       title={Weak derandomization of weak algorithms: explicit versions of
  yao's lemma},
        date={2009},
   booktitle={Proceedings of the $24$th annual conference on computational
  complexity},
}

\bib{ref:SU}{article}{
      author={Shaltiel, R.},
      author={Umans, C.},
       title={Simple extractors for all min-entropies and a new pseudorandom
  generator},
        date={2005},
     journal={Journal of the ACM},
      volume={52},
      number={2},
       pages={172\ndash 216},
}

\bib{shamir}{article}{
      author={Shamir, A.},
       title={On the generation of cryptographically strong pseudorandom
  sequences},
        date={1983},
     journal={ACM Transactions on Computer Systems},
      volume={1},
      number={1},
       pages={38\ndash 44},
}

\bib{ref:Shannon}{article}{
      author={Shannon, C.E.},
       title={A mathematical theory of communication},
        date={1948},
     journal={The Bell System Technical Journal},
      volume={27},
       pages={379\ndash 423 and 623\ndash 656},
}

\bib{ref:Raptor}{article}{
      author={Shokrollahi, A.},
       title={Raptor codes},
        date={2006},
     journal={IEEE Transactions on Information Theory},
      volume={52},
       pages={2551\ndash 2567},
}

\bib{ref:Shoup}{article}{
      author={Shoup, V.},
       title={New algorithms for finding irreducible polynomials over finite
  fields},
        date={1990},
     journal={Mathematics of Computation},
      volume={54},
       pages={435\ndash 447},
}

\bib{ref:sipser}{book}{
      author={Sipser, M.},
       title={Introduction to the theory of computation},
     edition={Second},
   publisher={Course Technology},
        date={2005},
}

\bib{ref:SS77}{article}{
      author={Solovay, R.M.},
      author={Strassen, V.},
       title={A fast {Monte-Carlo} test for primality},
        date={1977},
     journal={SIAM Journal on Computing},
      volume={6},
      number={1},
       pages={84\ndash 85},
}

\bib{ref:Spielman}{article}{
      author={Spielman, D.},
       title={Linear-time encodable and decodable error-correcting codes},
        date={1996},
     journal={IEEE Transactions on Information Theory},
      volume={42},
       pages={1723\ndash 1731},
}

\bib{stich}{book}{
      author={Stichtenoth, H.},
       title={Algebraic function fields and codes},
   publisher={Springer},
        date={1993},
}

\bib{ref:Stinson}{article}{
      author={Stinson, D.},
       title={Resilient functions and large set of orthogonal arrays},
        date={1993},
     journal={Congressus Numerantium},
      volume={92},
       pages={105\ndash 110},
}

\bib{ref:SW04}{article}{
      author={Stinson, D.R.},
      author={Wei, R.},
       title={Generalized cover-free families},
        date={2004},
     journal={Discrete Mathematics},
      volume={279},
       pages={463\ndash 477},
}

\bib{ref:SW00}{article}{
      author={Stinson, D.R.},
      author={Wei, R.},
      author={Zhu, L.},
       title={Some new bounds for cover-free families},
        date={2000},
     journal={Journal of Combinatorial Theory, Series A},
      volume={90},
       pages={224\ndash 234},
}

\bib{STV}{article}{
      author={Sudan, M.},
      author={Trevisan, L.},
      author={Vadhan, S.},
       title={Pseudorandom generators without the {XOR} lemma},
        date={2001},
     journal={Journal of Computer and Systems Sciences},
      volume={62},
      number={2},
       pages={236\ndash 266},
}

\bib{ref:Ta02}{article}{
      author={{Ta-Shma}, A.},
       title={Storing information with extractors},
        date={2002},
     journal={Information Processing Letters},
      volume={83},
      number={5},
       pages={267\ndash 274},
}

\bib{ref:TUZ01}{inproceedings}{
      author={{Ta-Shma}, A.},
      author={Umans, C.},
      author={Zuckerman, D.},
       title={Lossless condensers, unbalanced expanders, and extractors},
        date={2001},
   booktitle={Proceedings of the $33$rd annual {ACM} symposium on theory of
  computing ({STOC})},
       pages={143\ndash 152},
}

\bib{ref:TZ04}{article}{
      author={{Ta-Shma}, A.},
      author={Zuckerman, D.},
       title={Extractor codes},
        date={2004},
     journal={IEEE Transactions on Information Theory},
      volume={50},
      number={12},
       pages={3015\ndash 3025},
}

\bib{ref:RMExtractor}{article}{
      author={{Ta-Shma}, A.},
      author={Zuckerman, D.},
      author={Safra, S.},
       title={Extractors from {Reed-Muller} codes},
        date={2006},
     journal={Journal of Computer and System Sciences},
      volume={72},
       pages={786\ndash 812},
}

\bib{ref:TZ08}{inproceedings}{
      author={Tillich, J-P.},
      author={Z{\'e}mor, G.},
       title={Collisions for the lps expander graph hash function},
        date={2008},
   booktitle={Proceedings of {Eurocrypt}},
      series={Lecture Notes in Computer Science},
      volume={4965},
       pages={254\ndash 269},
}

\bib{ref:Tre}{article}{
      author={Trevisan, L.},
       title={Extractors and pseudorandom generators},
        date={2001},
     journal={Journal of the ACM},
      volume={48},
      number={4},
       pages={860\ndash 879},
}

\bib{TV00}{inproceedings}{
      author={Trevisan, L.},
      author={Vadhan, S.},
       title={Extracting randomness from samplable distributions},
        date={2000},
   booktitle={Proceedings of the $41$st annual {IEEE} symposium on foundations
  of computer science ({FOCS})},
       pages={32\ndash 42},
}

\bib{tsvz:82}{article}{
      author={Tsfasman, M.A.},
      author={Vl{\u a}du{\c t}, S.G.},
      author={Zink, Th.},
       title={Modular curves, {S}himura curves, and {G}oppa codes better than
  the {V}arshamov-{G}ilbert bound},
        date={1982},
     journal={Mathematische Nachrichten},
      volume={109},
       pages={21\ndash 28},
}

\bib{ref:Vad10}{inproceedings}{
      author={Vadhan, S.},
       title={The unified theory of pseudorandomness},
        date={2010},
   booktitle={Proceedings of the international congress of mathematicians},
}

\bib{alex}{article}{
      author={Vardy, A.},
       title={The intractability of computing the minimum distance of a code},
        date={1997},
     journal={IEEE Transactions on Information Theory},
      volume={43},
      number={6},
       pages={1757\ndash 1766},
}

\bib{varshamov}{article}{
      author={Varshamov, R.~R.},
       title={Estimate of the number of signals in error correcting codes},
        date={1957},
     journal={Doklady Akademii Nauk SSSR},
      volume={117},
       pages={739\ndash 741},
}

\bib{ref:Vaz87}{article}{
      author={Vazirani, U.V.},
       title={Towards a strong communication complexity theory or generating
  quasi-random sequences from two communicating semi-random sources},
        date={1987},
     journal={Combinatorica},
      volume={7},
      number={4},
       pages={375\ndash 392},
}

\bib{vN51}{article}{
      author={von Neumann, J.},
       title={Various techniques used in connection with random digits},
        date={1951},
     journal={Applied Math Series},
      volume={12},
       pages={36\ndash 38},
}

\bib{ref:WAF01}{inproceedings}{
      author={Watkinson, J.},
      author={Adler, M.},
      author={Fich, F.E.},
       title={New protocols for asymmetric communication channels},
        date={2001},
   booktitle={Proceedings of international colloquium on structural information
  and communication complexity ({SIROCCO})},
}

\bib{WZ99}{article}{
      author={Wigderson, A.},
      author={Zuckerman, D.},
       title={Expanders that beat the eigenvalue bound: explicit construction
  and applications},
        date={1999},
     journal={Combinatorica},
      volume={19},
      number={1},
       pages={125\ndash 138},
}

\bib{ref:Wol85}{article}{
      author={Wolf, J.},
       title={Born-again group testing: multiaccess communications},
        date={1985},
     journal={IEEE Transactions on Information Theory},
      volume={31},
       pages={185\ndash 191},
}

\bib{ref:WHL06}{article}{
      author={Wu, W.},
      author={Huang, Y.},
      author={Huang, X.},
      author={Li, Y.},
       title={On error-tolerant {DNA} screening},
        date={2006},
     journal={Discrete Applied Mathematics},
      volume={154},
      number={12},
       pages={1753\ndash 1758},
}

\bib{ref:WLHD08}{article}{
      author={Wu, W.},
      author={Li, Y.},
      author={Huang, C.H.},
      author={Du, D.Z.},
       title={Molecular biology and pooling design},
        date={2008},
     journal={Data Mining in Biomedicine},
      volume={7},
       pages={133\ndash 139},
}

\bib{ref:Wyner1}{article}{
      author={Wyner, A.~D.},
       title={The wire-tap channel},
        date={1975},
     journal={The Bell System Technical Journal},
      volume={54},
       pages={1355\ndash 1387},
}

\bib{yao}{inproceedings}{
      author={Yao, A.},
       title={Theory and applications of trapdoor functions},
        date={1982},
   booktitle={Proceedings of the $23$rd annual {IEEE} symposium on foundations
  of computer science ({FOCS})},
       pages={80\ndash 91},
}

\bib{ref:Yeh09}{unpublished}{
      author={Yehudayoff, A.},
       title={Affine extractors over prime fields},
        date={2009},
        note={Unpublished manuscript.},
}

\bib{ref:NetCodBook}{book}{
      author={Yeung, R.~W.},
      author={Li, {S-Y.~R.}},
      author={Cai, N.},
      author={Zhang, Z.},
       title={Network coding theory},
      series={Foundations and Trends in Communications and Information Theory},
        date={2005},
}

\bib{ref:Zippel}{article}{
      author={Zippel, R.~E.},
       title={Probabilistic algorithms for sparse polynomials},
        date={1979},
     journal={Springer Lecture Notes in Computer Science (EUROSCAM'79)},
      volume={72},
       pages={216\ndash 226},
}

\bib{Zuc96}{article}{
      author={Zuckerman, D.},
       title={On unapproximable versions of {$\mathsf{NP}$}-complete problems},
        date={1996},
     journal={SIAM Journal of Computing},
      volume={25},
       pages={1293\ndash 1304},
}

\bib{Zuc97}{article}{
      author={Zuckerman, D.},
       title={Randomness-optimal oblivious sampling},
        date={1997},
     journal={Random Structures and Algorithms},
      volume={11},
      number={4},
       pages={345\ndash 367},
}

\bib{ref:Zuc07}{article}{
      author={Zuckerman, D.},
       title={Linear degree extractors and the inapproximability of {Max
  Clique} and {Chromatic Number}},
        date={2007},
     journal={Theory of Computing},
      volume={3},
      number={6},
       pages={103\ndash 128},
}

\end{biblist}
\end{bibdiv}

\printindex


\end{document}